\documentclass[11pt,letterpaper]{amsart}
\usepackage[margin=1in]{geometry}
\usepackage{mathpazo,times}
\usepackage{amssymb,xcolor,colortbl,diagbox,comment,graphicx,cite}
\usepackage{MnSymbol}
\usepackage{color,float,fullpage,mathrsfs,mathtools}
\usepackage{enumerate,enumitem}
\usepackage[colorlinks=true, pdfstartview=FitV, linkcolor=blue, citecolor=blue, urlcolor=blue,unicode]{hyperref}
\usepackage[pageref]{backref}
\usepackage{multirow,array}
\usepackage{makecell}

\def\Xint#1{\mathchoice
    {\XXint\displaystyle\textstyle{#1}}%
    {\XXint\textstyle\scriptstyle{#1}}%
    {\XXint\scriptstyle\scriptscriptstyle{#1}}%
    {\XXint\scriptscriptstyle\scriptscriptstyle{#1}}%
    \!\int}
\def\XXint#1#2#3{{\setbox0=\hbox{$#1{#2#3}{\int}$}
        \vcenter{\hbox{$#2#3$}}\kern-.5\wd0}}

\def\dashint{\Xint-}

\def\pvint{\mathop{\int\kern-1.1em-\kern0.2em}\limits}

\makeatletter
\DeclareFontFamily{U}{mathx}{\hyphenchar\font45}
\DeclareFontShape{U}{mathx}{m}{n}{
      <5> <6> <7> <8> <9> <10>
      <10.95> <12> <14.4> <17.28> <20.74> <24.88>
      mathx10
      }{}
\DeclareSymbolFont{mathx}{U}{mathx}{m}{n}
\DeclareFontSubstitution{U}{mathx}{m}{n}
\DeclareMathAccent{\widecheck}{0}{mathx}{"71}
\DeclareMathAccent{\wideparen}{0}{mathx}{"75}

\def\bpm{\begin{pmatrix}}
\def\epm{\end{pmatrix}}
\def\gl{\gtrless}
\def\lg{\lessgtr}
\def\Real{\mathbb{R}}
\def\Complex{\mathbb{C}}
\def\Re{\mathop{\rm Re}\nolimits}
\def\Im{\mathop{\rm Im}\nolimits}
\def\arg{\mathop{\rm arg}\nolimits}

\def\bg{\mathrm{bg}}
\def\bchi{\boldsymbol{\chi}}
\def\d{\mathrm{d}}
\def\rdo{\mathrm{do}}
\def\e{\mathrm{e}}
\def\ext{\mathrm{ext}}
\def\dd{\mathrm{dd}}
\def\diag{\mathrm{diag}}

\def\ii{\mathrm{i}}
\def\m{\mathbf{m}}
\def\bmu{\boldsymbol{\mu}}

\def\o{\mathrm{o}}
\def\bphi{\boldsymbol{\phi}}

\def\brho{\boldsymbol{\rho}}
\def\bvarrho{\boldsymbol{\varrho}}
\def\sech{\mathrm{sech}}

\def\sign{\mathrm{sign}}
\def\tr{\mathrm{tr}}

\def\A{\mathbf{A}}
\def\B{\mathbf{B}}
\def\C{\mathbf{C}}
\def\E{\mathbf{E}}

\def\bGamma{\boldsymbol{\Gamma}}
\def\H{\mathcal{H}}
\def\bH{\mathbf{H}}
\def\I{\mathrm{I}}
\def\bbI{\mathbb{I}}
\def\II{\mathrm{II}}
\def\III{\mathrm{III}}
\def\J{\mathbf{J}}
\def\L{\mathbf{L}}
\def\bLambda{\boldsymbol{\Lambda}}
\def\M{\mathbf{M}}

\def\O{\mathcal{O}}

\def\P{\mathcal{P}}
\def\bPhi{\boldsymbol{\Phi}}
\def\bPi{\boldsymbol{\Pi}}
\def\Q{\mathbf{Q}}
\def\R{\mathbf{R}}
\def\Res{\mathop{\mathrm{Res}}}
\def\S{\mathbf{S}}
\def\U{\mathbf{U}}
\def\V{\mathbf{V}}

\def\X{\mathbf{X}}
\def\Y{\mathbf{Y}}
\def\Z{\mathbf{Z}}

\let\^=\hat
\let\==\bar
\let\@=\mathbf

\let\le=\leq
\let\ge=\geq

\newtheorem{definition}{Definition}
\newtheorem{lemma}{Lemma}
\newtheorem{remark}{Remark}
\newtheorem{rhp}{Riemann-Hilbert problem}
\newtheorem{theorem}{Theorem}
\newtheorem{proposition}{Proposition}

\def\~#1{\widetilde{\mathbf{#1}}}
\makeatother

\def\bse{\begin{subequations}}
\def\ese{\end{subequations}}

\definecolor{darkred}{rgb}{0.9,0,0}
\definecolor{darkblue}{rgb}{0,0,0.8}
\definecolor{darkorange}{rgb}{0.7,0.2,0}

\let\ul=\underline

\begin{document}

\title{On the coupled Maxwell-Bloch system of equations with non-decaying fields at infinity}
\author{Sitai Li$^*$}
\address[Sitai Li]{Xiamen University, School of Mathematical Sciences, Xiamen, Fujian 361005, P. R. China}
\email{sitaili@xmu.edu.cn}
\author{Gino Biondini}
\address[Gino Biondini]{State University of New York at Buffalo, Department of Mathematics, Buffalo, NY 14260, USA}
\email{biondini@buffalo.edu}
\author{Gregor Kova\v{c}i\v{c}}
\address[Gregor Kova\v{c}i\v{c}]{Rensselaer Polytechnic Institute, Department of Mathematical Sciences, Troy, NY 12180, USA}
\email{kovacg@rpi.edu}

\begin{abstract}
We study an initial-boundary-value problem (IBVP) for a system of coupled Maxwell-Bloch equations (CMBE) that model two colors or polarizations of light resonantly interacting with a degenerate, two-level, active optical medium with an excited state and a pair of degenerate ground states. We assume that the electromagnetic field approaches non-vanishing plane waves in the far past and future.
This type of interaction has been found to underlie nonlinear optical phenomena including electromagnetically induced transparency, slow light, stopped light, and quantum memory.
Under the assumptions of unidirectional, lossless propagation of slowly-modulated plane waves, the resulting CMBE become completely integrable in the sense of possessing a Lax Pair.
In this paper, we formulate an inverse scattering transform (IST) corresponding to these CMBE and their Lax pair, allowing for the spectral line of the atomic transitions in the active medium to have a finite width.
The scattering problem for this Lax pair is the same as for the Manakov system.
The main advancement in this IST for CMBE is calculating the nontrivial spatial propagation of the spectral data and determining the state of the optical medium in the distant future from that in the distant past, which is needed for the complete formulation of the IBVP.
The Riemann-Hilbert problem is used to extract the spatio-temporal dependence of the solution from the evolving spectral data.
We further derive and analyze several types of solitons and determine their velocity and stability, as well as find dark states of the medium which fail to interact with a given pulse.

\end{abstract}

\maketitle

\tableofcontents

\section{Introduction}
\label{s:intro}

This paper considers nonlinear resonant interaction between a light beam and an active optical medium with three working levels, arranged in what is suggestively known as the $\Lambda$-configuration.  This configuration consists of two ground levels and an excited level, with a forbidden atomic dipole transition between the two ground levels.  The two allowed dipole transitions interact with light of either two different colors or opposite circular polarizations.  In the latter case, the medium can be considered as having two working levels, with the ground level being degenerate, as shown in Figure~\ref{f:lambdaddiagram}.   This type of interaction underlies laser operation~\cite{shimoda86} and a version of self-induced transparency (SIT)~\cite{bm1984}.    It is also believed to underlie phenomena such as  inversionless laser operation~\cite{ih1989,szg1989},
electromagnetically induced transparency (EIT)~\cite{Fleischhauer00,wr2006,w2009},
slow light~\cite{rvb2005-2,rvb2005-3,rvb2005-4,rvb2005-1,lr2006},  and quantum memory~\cite{fim2005}.

After assuming lossless, unidirectional light propagation and slow modulation of the plane carrier light waves, we find the interaction between light of two different colors or opposite circular polarizations and a $\Lambda$-configuration active optical medium to be described by the following $\Lambda$-configuration coupled Maxwell-Bloch equations (CMBE)~\cite{Konopnicki81a,Konopnicki81,bm1984,maimistov84,%
maimistov85a,basharov88,Basharov90}:
\begin{equation}
\label{e:cmbe}
\displaystyle
\begin{aligned}
\brho_t
 & = \big[\ii k \J + \Q,\brho\big],\qquad &
\Q_z
 & = -\frac{1}{2}\int_{-\infty}^\infty \big[\J,\brho\big]\,g(k)\d k,\qquad
(t,z,k) \in \Real\times\Real^+\times\Real,\\
\Q(t,z)
 & \coloneq \bpm0 & -\E^\top\\ \E^* & \@O_{2\times2} \epm,\qquad &
\E(t,z)
 & \coloneq (E_1(t,z),E_2(t,z))^\top\in\Complex^2,\\
\brho(t,z,k)^\dagger
 & = \brho(t,z,k)\in\Complex^{3\times3},\quad &
\J & \coloneq \diag(1,-1,-1),
\end{aligned}
\end{equation}
where the subscripts $t$ and $z$ denote differentiation with respect to $t$ and $z$, respectively;
$\@O_{2\times2}$ is the $2\times2$ zero matrix;
the superscripts $\top$, $*$, and $\dagger$ denote the matrix transpose, complex conjugate, and conjugate transpose, respectively; and
$[\mathbf{A},\mathbf{B}] \coloneq \mathbf{A}\mathbf{B}-\mathbf{B}\mathbf{A}$ is the matrix commutator.
The variable $z = z_\mathrm{lab}$ is the propagation distance and the
variable $t = t_\mathrm{lab} - z_\mathrm{lab}/c$ is the retarded time; $c$ denotes the speed of light.
The complex vector $\E(t,z)$ contains the envelopes of the two modulated, plane-wave, light (electric field) components interacting with the dipole transitions between each ground state and the excited state.
The density matrix $\brho(t,z,k)$ is a $3\times3$ Hermitian matrix representing the state of the medium.
The diagonal entries of $\brho(t,z,k)$ denote the populations of the atoms in the excited state and the two ground states, respectively.
The off-diagonal elements $\rho_{12}$ and $\rho_{13}$ denote the complex-valued envelopes of the medium polarizability contributions corresponding to the two dipole transitions, while $\rho_{23}$ denotes the average coherence between the two ground states.
The probability density function $g(k)$ describes the shape of the spectral line.   In particular, $g(k)\,dk$ gives the proportion of the atoms in the medium whose transition frequencies are detuned from the resonance with the carrier frequencies of the two impinging light components by the amount $k$.    This spectral-line shape corresponds to the \emph{inhomogeneous broadening} of the spectral line and is caused by phenomena such as the Doppler effect (if the medium is a rarefied gas)~\cite{allen87}.       Frequently, it is modeled by a Lorentzian function, which we will do in this paper whenever a specific $g(k)$ is required.
Equations~\eqref{e:cmbe} were shown to be completely integrable in the sense of being derived from a Lax pair~\cite{bm1984,maimistov84,maimistov85a}, as reviewed in Section~\ref{s:lax}.

\begin{figure}[t!]
\kern-\medskipamount
\centering
\includegraphics[scale=0.45]{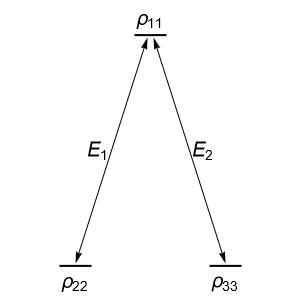}
\caption{
    Quantum transition diagram for the $\Lambda$-configuration CMBE~\eqref{e:cmbe} in the case of a degenerate, two-level, active optical medium interacting with two opposite circular polarizations of monocromatic light.
Please see the discussion following Equations~\eqref{e:cmbe} for an explanation of the variables.
}
\label{f:lambdaddiagram}
\end{figure}

Equations~\eqref{e:cmbe} are a generalization of the classic, two-level Maxwell-Bloch equations (MBE), which have been studied extensively for over three quarters of a century~\cite{Feynman57,Jaynes63,ab1965,risken:4662}.  (See~\cite{allen87} and references therein for many of the optical effects this model describes.)
The integrable nature of this model was gradually revealed in~\cite{LambJr1967181,lamb69,Lamb71,lamb73,lamb73a,Lamb74}.  The Lax pair was discovered in~\cite{Ablowitz74}, where the SIT was explained from the viewpoint of the IST, with additional explanations of physical effects following in~\cite{kaup77}.  (See also the description in~\cite{as1981}.)  Improvements of the IST used for describing more general physical phenomena were developed for photon echo~\cite{zm1982}, nonlinear amplification~\cite{m1982,z1980},
and superfluorescence~\cite{gzm1983,gzm1984}.   The development of the complete  IST for MBE with vanishing asymptotic values of the electric field in the far past and future was accomplished in~\cite{gzm1985}.   Additionally, a self-similar Bessel-function solution related to superfluorescence was discovered in~\cite{gzm1983,gzm1984}.    Self similar solutions of MBE and CMBE belonging to families of Painlev\'{e}-III functions and related to Bessel functions were also studied in~\cite{chk2003,lm2021}.      In~\cite{lbkg2018,bgkl2019}, we developed the IST to study the IBVP for the classic MBE with symmetric, non-vanishing asymptotic values of the electric-field envelope in the distant past and future, and described several families of soliton solutions.

The nontrivial nature of the phenomena described by the CMBE versus MBE is exemplified by the fact that even the single soliton described by the CMBE need not be a traveling wave.   In fact, typically, when the asymptotic values of the electric field in the distant past and future vanish, the interaction induced by a soliton will switch from one transition to the other and thus exhibit an internal degree of freedom~\cite{Maimistov85b,bgk2003}.

The IST for the $\Lambda$-configuration CMBE with vanishing asymptotic values of the electric-field envelopes in the far past and future was gradually developed in~\cite{maimistov84,Basharov90,bgk2003,cpa2014}.
Meanwhile, special solutions of CMBE with vanishing and non-vanishing asymptotic values in the past and future were obtained using symmetry-based methods such as Darboux transformations or dressing, primarily in the sharp-line limit of $g(k)=\delta(k)$, the Dirac delta.   These include descriptions of SIT~\cite{steudel88,Clader08}, slow light~\cite{rvb2005-2,rvb2005-3,rvb2005-4,rvb2005-1,lr2006}, and soliton and rational solutions~\cite{wwwgk2021,wwlgw2022}.

In this paper, we develop the IST for CMBE with symmetric, non-vanishing asymptotic values of the electric-field envelopes in the far past and future.   Typically, non-vanishing asymptotic values in IBVPs for integrable problems are given at spatial infinity, and so such problems are described as having non-zero boundary conditions (NZBC).    This is not the case for CMBE, and the best way to interpret the physical phenomena described here is as optical pulses riding on top of a non-zero (continuous-wave) background (NZBG), which is the terminology we will use in the rest of this paper.  (The terminology for the zero background will be ZBG.)  Additionally, in most of our discussion, we will maintain an arbitrary, unspecified shape of the spectral line, $g(k)$, since our results do not depend on it.

For the scattering and inverse-scattering steps of the IST we have devised, a close relation exists between CMBE and the focusing coupled nonlinear Schr\"{o}dinger (NLS) equation (the focusing Manakov system).
Namely, both systems share an identical scattering problem,
the $3\times3$ non-self-adjoint Zakharov-Shabat problem described in Section~\ref{s:lax} below, albeit in the spatial variable for the Manakov system and the temporal variable in CMBE.
As a consequence, the same initial/input data yield
the same initial/input spectra for both systems.
Moreover, since the scattering problem also governs the inverse-scattering step of the IST, we see
that the calculation procedure used in this step is the same for both systems.
We therefore simply reuse the scattering/inverse-scattering approach developed for the IBVP for the focusing Manakov system with NZBC at infinity~\cite{kbk2015}.

In contrast to the scattering problem, the propagation/evolution problems in the Lax pairs for CMBE and the focusing Manakov systems differ dramatically, due to the non-trivial initial values of the matrix~$\brho$, representing the initial state of the medium in the CMBE. As a result, the evolution/propagation of the spectral data (reflection coefficients and norming constants) in the ISTs for the two systems differ in a fundamental way.   For the Manakov system, the (temporal) evolution of the spectral data is trivial, and is thus often skipped altogether by finding simultaneous eigenfunctions of both equations in the Lax pair at the very beginning of the IST procedure. On the other hand, the (spatial) propagation (along the $z$-variable) of the spectral data for CMBE is complicated and can only be found using a rather involved set of calculations.  Note that this complication occurs even for MBE and CMBE with ZBG~\cite{Ablowitz74,gzm1985,cpa2014} and for MBE with NZBG~\cite{bgkl2019}, and that the presence of an inhomogeneously-broadened spectral line $g(k)$ with finite width helps in deciding how to take crucial steps in the derivation of the propagation equations.   Therefore, a large portion of this work focuses on deriving the equations for the propagation of the spectral data for CMBE, and this includes some detailed calculations that are relegated to the Appendix.

A byproduct of the IST developed in this paper is a general, explicit formula for $N$-soliton solutions valid for any  shape of the spectral line, of which formulae previously derived by Darboux transformations and dressing are special cases.   Moreover, even single CMBE solitons display a variety of different dynamical behaviors, and we find a rich family of both soliton solutions and solutions that can be obtained as their limits.   In particular, we show that there exist three types of solitons,   whose properties
depend on the loci of the corresponding discrete eigenvalues in the scattering problem.
The first type is a direct generalization of the classic two-level MBE solitons,
but the other two types exhibit entirely new features.
In addition, we find limiting rational solutions, periodic solutions, and a new plane-wave solution, different from the one used in our initial formulation of the IST. We emphasize that carrying out the calculations using a general, inhomogeneously-broadened spectral-line shape $g(k)$ is crucial in establishing the forms of these solitons;  calculations performed directly in the sharp-line limit may inadvertently lead to unphysical solutions. (Cf. the discussion in~\cite{lbkg2018,bgkl2019}.)
Finally, we carry out parameter studies of the soliton dynamics and their limits for the case of the Lorenztian shape of the spectral line.    In particular, we investigate how the (subluminal or superluminal) pulse velocity and stability depend on the input values of the spectral data and the initial state of the medium in the distant past.

The remainder of this paper is organized as follows:
In Section~\ref{s:setup}, we set up the IBVP, review the Lax pair for the CMBE, and discuss the background solutions used in the IBVP.
In Section~\ref{s:directproblem}, we solve the direct scattering problem, which includes defining the Jost eigenfunctions and scattering data, investigating the auxiliary direct problem, and calculating the symmetries and asymptotics of the eigenfunctions.
In Section~\ref{s:propagation}, we discuss the propagation of the spectral data along the medium, i.e., the propagation stage of the IST.   In Section~\ref{s:inverseproblem}, we cast the inverse-scattering problem in terms of an RHP and present implicit integral formulas for the solutions.
In Section~\ref{s:generalsoliton}, we discusses a number of exact solutions of CMBE with NZBG and an arbitrary spectral-line shape, including three types of soliton solutions, rational solutions, periodic solutions and a nontrivial plane-wave solution.
In Section~\ref{s:IB} we focus on the case of the inhomogeneous broadening for which the spectral-line shape is a Lorentzian, and investigate numerical examples of the exact solutions we had obtained in detail.   Appendices contain the lengthier calculations and proofs.

\section{Formulation and background}\label{s:setup}

In this section, we present the formulation of IBVP for CMBE~\eqref{e:cmbe} with NZBG to lay the groundwork for the IST. In particular, we describe the Lax pair, the Riemann surface and the uniformization variable, the background solutions with correct asymptotic behavior in the distant past and future, reductions of the NZBG problem to the ZBG problem, and dark states that do not interact with a given light beam.


\subsection{Initial-boundary values and the Lax pair}
\label{s:lax}

In this section, we present the Lax pair and the formulation of the IBVP for CMEB~\eqref{e:cmbe} with NZBC. This Lax pair is given by the equations
\bse
\label{e:laxpair}
\begin{gather}
\label{e:laxpair1}
\bphi_t = \X\,\bphi\,,\qquad
\X(t,z,k) \coloneq  \ii k \J + \Q(t,z)\,,\\
\label{e:laxpair2}
\bphi_z = \V\,\bphi\,,\qquad
\V(t,z,k) \coloneq \frac{\ii \pi}{2}\H_{k}[\brho(t,z,k)g(k)]\,,
\end{gather}
\ese
where $\bphi=\bphi(t,z,k)$ is the (auxiliary) wave function and $\H_{k}$ denotes the Hilbert transform
\begin{equation}
    \label{e:hilbert}
    \H_{k}[f(k)] \coloneq \frac1\pi \pvint\nolimits_{\!\!-\infty}^\infty \frac{f(k')}{k'-k}\d k'\,.
\end{equation}

The physical problem we consider is that of a beam consisting of two colors or polarizations of light, described by the complex vector envelope $\E(t, z)$, traveling inside a narrow, semi-infinite tube filled with the active optical medium located along the non-negative $z$-axis, $z \ge 0$.
The light beam is injected into the medium at $z = 0$, and propagates along it as $z$ increases. In the reference frame co-moving with the speed of light, $t\to-\infty$ denotes the distant past and $t\to+\infty$ denotes the distant future. Without loss of generality, we take the light cone as $t > 0$ and $z > 0$.

In CMBE~\eqref{e:cmbe}, the physical meaning of $t$ and $z$ is that of a temporal and spatial variable, respectively. However, reflecting the initial-value-signaling nature of the problem we are investigating for CMBE~\eqref{e:cmbe}, in the formulation of IST governed by the Lax pair~\eqref{e:laxpair}, $t$ plays the role usually reserved for the spatial variable and $z$ that for the temporal variable.
Moreover, we adopt a change in the terminology describing the IST.
In particular, the three major stages of IST are typically referred to as the direct problem, the evolution, and the inverse problem. Here, we rename the second stage as the propagation, in order to reflect the fact that, in a signaling problem, $z$ represents the propagation distance of the light into the medium.

To finalize the problem formulation as well as to avoid certain technical difficulties in the further development, we make the following assumptions:
\begin{enumerate}
    \item
    \textbf{Symmetric nonzero backgroud amplitudes}:
    The electric field has the same intensity in the distant past and in the distant future, i.e.,
    $\|\E(t,z)\|_2\to E_0 > 0$ as $t\to\pm\infty$.
    In particular,
    we write $\E(t,z)\to\E_{\pm}(z)$ as $t\to\pm\infty$,
    where $\E_{\pm}(z) \coloneq (E_{\pm,1}(z),E_{\pm,2}(z))^\top$.
    \item
    \textbf{Initial conditions (injected light beam)}:
    At the injection point $z = 0$, the vector-valued electric-field envelope $\E(t,0)$ decays towards the background as $t\to\pm\infty$ sufficiently fast.
    In other words,
    $\E(t,0) - \E_{\pm}(z)\to0$ sufficiently fast as $t\to\pm\infty$,
    respectively.
    \item
    \textbf{Discrete spectrum of the scattering operator $\X$ in Equation~\eqref{e:laxpair1}}:
    \begin{enumerate}
        \item
        There are no spectral singularities and no embedded eigenvalues on the continuous spectrum.
        \item
        All discrete eigenvalues are simple.
    \end{enumerate}
\end{enumerate}

We now show that, without loss of generality, the asymptotic value $\E_-(0)$ of the injected electric-field envelope can be chosen as a real vector. To do so, we present the following:
\begin{lemma}
\label{thm:CMBE-U-invariance}
Let $\U$ be a constant $2\times2$ unitary matrix,
i.e., $\U \U^\dagger = \bbI$.
Define a $3\times3$ unitary matrix $\U_\ext$
\begin{equation}
\nonumber
\U_\ext = \bpm
1 & \@0 \\ \@0 & \U^*
\epm\,,\qquad
\U_\ext^{-1} = \bpm
1 & \@0 \\ \@0 & \U^\top
\epm\,.
\end{equation}
Let $\E(t,z)$ and $\brho(t,z,k)$ form a solution of CMBE~\eqref{e:cmbe}.
Then the functions
\begin{equation}
\label{e:solphaseinvariance}
\widetilde \E(t,z) \coloneqq \U^{-1} \E(t,z)\,,\qquad
\widetilde \Q(t,z) \coloneqq \U_\ext \Q(t,z) \U_\ext^{-1}\,,\qquad
\widetilde \brho(t,z,k) \coloneqq \U_\ext^{-1} \brho(t,z,k) \U_\ext\,,
\end{equation}
also form a solution of CMBE~\eqref{e:cmbe}.
\end{lemma}
\begin{remark}[On the simplified distant-past value of the injected light beam]
In general, one can write the distant-past value of the injected optical field as
$\E_-(0) = (E_0 \e^{\ii\alpha_1}\cos\alpha,E_0\e^{\ii\alpha_2}\sin\alpha)$
with $\alpha_j\in[0,2\pi)$ and $\alpha\in[0,\pi/2]$.
Applying Lemma~\ref{thm:CMBE-U-invariance} with $\U = \diag(\e^{\ii \alpha_1},\e^{\ii \alpha_2})$ and dropping the tilde,
the above value $\E_-(0)$ can be transformed into
\begin{equation}
\label{e:bc-q}
\E_-(0) = (E_0\cos\alpha,E_0\sin\alpha)^\top\,,\qquad
\mbox{with }\alpha\in[0,\pi/2].
\end{equation}
Hence,
without loss of generality,
in the rest of the work we take the condition~\eqref{e:bc-q} as the asymptotic injected background value in the distant past.
\end{remark}

\subsection{Riemann surface and uniformization variable}
\label{s:uniformization}

Similarly to other integrable systems with NZBC such as the non-self-adjoint Zakharov-Shabat spectral problem,
the formulation of IST can be simplified by utilizing a uniformization variable $\zeta(k)$ in the spectral plane.
We only present a quick review of this approach here,
and also take this opportunity to discuss two equivalent formulations of IST.
In particular, a complex-valued square-root function, $\lambda(k)$, appears in the IST,
\begin{equation}
\label{e:lambda-def}
\lambda(k) \coloneq (E_0^2 + k^2)^{\frac12}\,,
\qquad k\in\Complex\,.
\end{equation}
To complete the definition of $\lambda(k)$,
one needs to specify the branch cut corresponding to this square root in the complex plane.
For this cut, we can use two topologically distinct choices:
\begin{itemize}
\item
\textbf{The interval $\ii[-E_0, E_0]$ along the imaginary axis:}
On the real line, one can then write $\lambda(k) = \sign(k)\sqrt{E_0^2 + k^2}$, with $k\in\Real$.
This definition of $\lambda(k)$ is discontinuous at $k = 0$,
but has the advantage that $\lambda(k)\to k$ as $E_0\to0$.
It can be shown that this choice of the branch cut allows one to take the limit $E_0\to0$ directly and continuously throughout the formulation of IST to recover the IST for CMBE~\eqref{e:cmbe} with ZBG~\cite{bgkl2019}.

\item
\textbf{The interval union $\ii(-\infty,-E_0]\cup \ii[E_0,\infty)$ (or any other pair of non-intersecting curves connecting $\pm \ii E_0$ and $\infty$):}
On the real line, we write $\lambda(k) = \sqrt{E_0^2 + k^2}$ with $k\in\Real$.
This definition of $\lambda(k)$ is continuous on the entire real line,
but in the limit $E_0\to0$, does not recover $k$ from $\lambda(k)$.
Without showing details, to recover the ZBG case of IST,
more care must be taken in the limiting process with this choice of the branch cut~\cite{bgkl2019}.
\end{itemize}
\begin{remark}[Equivalence between choices of branch cut]
Despite the previous discussion and comparison, it is worth pointing out that the two choices of the branch cut for $\lambda(k)$ in Equation~\eqref{e:lambda-def}, and the two corresponding versions of IST, can be shown to be equivalent,
just as in the classic two-level case~\cite{bgkl2019}.
\end{remark}

Because of this equivalence and its advantages, in this paper, we use the first definition of the branch cut,
i.e., along the interval $\ii[-q_0,q_0]$.
This choice allows us to take the limit $E_0\to0$ and recover the case of ZBG in a trivial way.

Once we have settled on the branch cut,
we naturally introduce two sheets of the complex $k$-plane,
$\Complex_\I$ and $\Complex_\II$,
with the square root in $\lambda(k)$ taking different signs on each sheet.
We introduce subscripts I and II to denote a quantity evaluated on the first and second $k$-sheet, respectively.

In order to simplify calculations on two separate copies of the complex plane,
we further introduce a two-sheeted Riemann surface by defining a uniformization variable $\zeta(k)$,
\begin{equation}
\label{e:zeta-def}
\zeta(k) \coloneq k + \lambda(k)\,,\qquad k\in\Complex\,.
\end{equation}
The Riemann surface $\zeta\in\Complex$ is obtained by gluing the two copies of the complex $k$-plane together along the branch cut $\ii[-E_0,E_0]$.
Correspondingly, the origin of the first sheet is mapped to $\infty$ in the $\zeta$-plane,
i.e., $0_\I\mapsto\infty$,
whereas the origin of the second sheet is mapped to the origin of the $\zeta$-plane,
i.e., $0_\II\mapsto0$.
The corresponding inverse transforms are given by
\begin{equation}
\label{e:k-lambda}
k(\zeta) = \frac{\zeta + \hat\zeta}{2}\,,\qquad
\lambda(\zeta) = \frac{\zeta - \hat\zeta}{2}\,,\quad
\mbox{where }
\hat\zeta \coloneq -\frac{E_0^2}{\zeta}\,,\qquad
\zeta\in\Complex\backslash\{0\}\,.
\end{equation}
We point out that the above transforms are only valid on the punctured complex plane,
because $\infty_\II$ is mapped to the origin in the $\zeta$-plane.


In the formulation of the IST,
we use the uniformization variable $\zeta$ instead of the original spectral variable $k$ for most of the calculations.

\subsection{Background solutions}

We present the simplest solution of CMBE~\eqref{e:cmbe} corresponding to the IBVP with NZBG~\eqref{e:bc-q} at infinity.
The calculations are presented in Appendix~\ref{s:background}.
This solution is given by the formulas
\begin{equation}
\label{e:background}
\everymath{\displaystyle}
\begin{aligned}
E_{\bg,1}(z)
 & \coloneqq E_0\e^{\tfrac\ii2 \int_0^z w(s)\d s}\cos\alpha\,,\\
E_{\bg,2}(z)
 & \coloneqq E_0\e^{\tfrac\ii2 \int_0^z w(s)\d s}\sin\alpha\,,\\
\brho_\bg(t,z,k)
 & \coloneqq \varrho_{\bg,1,1}\brho_{\bg}^{(1)}(t,z,k) + \varrho_{\bg,2,2}\brho_{\bg}^{(2)}(t,z,k) + \varrho_{\bg,3,3}\brho_{\bg}^{(3)}(t,z,k)\,,
\end{aligned}
\end{equation}
which contain three free parameters, $\varrho_{\bg,j,j}\ge 0$, for $j = 1,2,3$.
The quantity $w(z)$ is given by
\begin{equation}
w(z)\coloneqq \int_{-\infty}^\infty (\varrho_{\bg,1,1}(z) - \varrho_{\bg,3,3}(z))\frac{g(k)}{\lambda(k)}\d k\in\Real\,,
\end{equation}
and the three matrices $\brho_\bg^{(j)}(t,z,k)$ are given by
\begin{equation}
\everymath{\displaystyle}
\begin{aligned}
\brho_\bg^{(1)}(t,z,k)
& \coloneqq \frac{1}{2\lambda}\bpm
\zeta & \ii E_{\bg,1} & \ii E_{\bg,2}
\\
-\ii E_{\bg,1}^* & (\lambda - k)\cos^2\alpha &
(\lambda - k) \sin\alpha\cos\alpha
\\
-\ii E_{\bg,2}^* & (\lambda - k)\sin\alpha\cos\alpha &
(\lambda - k)\sin^2\alpha
\epm\,,\\
\brho_\bg^{(2)}(t,z,k)
& \coloneqq \bpm
0 & 0 & 0
\\
0 & \sin^2\alpha &
-\sin\alpha\cos\alpha
\\
0 & -\sin\alpha\cos\alpha &
\cos^2\alpha
\epm\,,\\
\brho_\bg^{(3)}(t,z,k)
& \coloneqq \frac{1}{2\lambda}\bpm
\lambda - k & -\ii E_{\bg,1} & -\ii E_{\bg,2}
\\
\ii E_{\bg,1}^* &  \zeta\cos^2\alpha &
 \zeta\sin\alpha\cos\alpha
\\
\ii E_{\bg,2}^* & \zeta\sin\alpha\cos\alpha &
\zeta\sin^2\alpha
\epm\,,
\end{aligned}
\end{equation}


\begin{remark}
The background solution $\brho_\bg$ for the state of the medium is a linear combination of the three components $\brho_\bg^{(j)}(t,z,k)$ in general. Note that for each component $\tr\brho_\bg^{(j)}(t,z,k) = 1$ for all $t$, $z$ and $k$. All diagonal entries for $\brho_\bg^{(j)}(t,z,k)$ are always nonnegative. Therefore, $\tr\brho_\bg(t,z,k) = \varrho_{\bg,1,1} + \varrho_{\bg,2,2} + \varrho_{\bg,3,3}$, which is the total population of atoms in the optical medium.
\end{remark}

\subsection{Zero boundary condition reductions of the background solutions}

In order to better understand the physical meaning of the background solution and the roles played by the coefficients $\varrho_{\bg,j,j}$,
we relate the solution~\eqref{e:background} to its counterpart in the case of ZBG.
Thus, we next consider the limiting case as $\E_\bg\to0$, i.e., $E_0\to0$.
This also enables us to relate the two $k$-sheets and explore any potential symmetries.

We know from the definition that $\lambda_\I\to k$ as $E_0\to0$, and $\lambda_\II\to-k$.
We now examine the background solution~\eqref{e:background} in the same limit.
\begin{itemize}[leftmargin = *]
\item
On sheet I, the background density matrix reduces to
\begin{equation}
\nonumber
\brho_\bg \to \bpm
  \varrho_{\bg,1,1} & 0 & 0 \\
  0 & \varrho_{\bg,2,2}\sin^2\alpha + \varrho_{\bg,3,3}\cos^2\alpha  & (\varrho_{\bg,3,3} - \varrho_{\bg,2,2})\sin\alpha\cos\alpha \\
  0 & (\varrho_{\bg,3,3} - \varrho_{\bg,2,2})\sin\alpha\cos\alpha & \varrho_{\bg,2,2}\cos^2\alpha  + \varrho_{\bg,3,3}\sin^2\alpha
\epm\,,\qquad E_0\to0\,,
\end{equation}
Recall that $\rho_{1,1}$ corresponds to the population of the atoms in the excited state,
and $\rho_{2,2}$ and $\rho_{3,3}$ correspond to the populations of the two ground states, respectively.
Thus, in the limit $E_0\to0$,
$\varrho_{\bg,1,1}$ indicates the population of excited atoms,
and $\varrho_{\bg,2,2}$ and $\varrho_{\bg,3,3}$ relate to the populations in the two ground states.
This situation is identical to the case of ZBG.

\item
On sheet II, the density matrix becomes
\begin{equation}
\nonumber
\brho_\bg \to \bpm
  \varrho_{\bg,3,3} & 0 & 0 \\
  0 & \varrho_{\bg,2,2}\sin^2\alpha + \varrho_{\bg,1,1}\cos^2\alpha  & (\varrho_{\bg,1,1} - \varrho_{\bg,2,2})\sin\alpha\cos\alpha  \\
  0 & (\varrho_{\bg,1,1} - \varrho_{\bg,2,2})\sin\alpha\cos\alpha & \varrho_{\bg,2,2}\cos^2\alpha  + \varrho_{\bg,1,1}\sin^2\alpha
\epm\,,\qquad E_0\to0\,.
\end{equation}
By a similar discussion as in the previous case,
we find that $\varrho_{\bg,3,3}$ indicates the population in the excited state,
and $\varrho_{\bg,1,1}$ and $\varrho_{\bg,2,2}$ correspond to the population in the two ground states.
\end{itemize}

We conclude that \textit{the roles of $\varrho_{\bg,1,1}$ and $\varrho_{\bg,3,3}$ are interchanged between the two $k$-sheets}.
We will use this symmetry later to define proper asymptotics for the density matrix $\brho(t,z,\zeta)$ as $t\to\pm\infty$ (cf. Section~\ref{s:boundary}).
Also recall that a similar symmetry occurs between the two sheets in the classic two-level case with NZBG~\cite{bgkl2019}.

\begin{remark}
The coefficients $\varrho_{\bg,j,j}$ are not the atomic level populations in each state observed from Equation~\eqref{e:background}.
Instead, the atomic level population in each state $\rho_{j,j}$ can be written as a linear combination of the coefficients $\varrho_{\bg,j,j}$ using Equation~\eqref{e:background} with additional parameters such as $k$ and $E_0$. This complicated situation is quite different from the case of zero asymptotic values, in which $\rho_{j,j} = \varrho_{\bg,j,j}$ holds due to the above calculations. (One needs to enforce additional conditions $\alpha = 0$ or $\pi/2$ due to the additional internal freedom from NZBG.)
In other words, the nontrivial relation between $\rho_{j,j}$ and $\varrho_{\bg,j,j}$ is a direct consequence of the nonzero background, hence a novel feature of this work.
\end{remark}

\subsection{Nonzero background solution in a dark state}
\label{s:dark-state-background}

We note from the background solution~\eqref{e:background} that if we choose $\varrho_{\bg,1,1} = \varrho_{\bg,3,3} = 0$ and $\varrho_{\bg,2,2}\ne0$, i.e., $\brho_\bg(t,z,k) = \brho_\bg^{(2)}(t,z,k)$, then the excited state is never occupied.
The electric field envelopes become $\E_{\bg,1}(t,z) = E_0\cos\alpha$ and $\E_{\bg,2}(t,z) = E_0\sin\alpha$.
In this particular situation, the electric field and the medium do not interact at all, and medium is in a ``dark state"~\cite{fim2005}.   Despite a slight risk of clashing with this established physics terminology, we refer to  entire such configurations as \textit{dark-state solutions} of the CMBE~\eqref{e:cmbe}.

The above special background solution is the first and the simplest dark-state solution appearing in this work. We shall look for other kinds of dark-state solutions later, in Section~\ref{s:dark-state}.

\subsection{Notation}

In the rest of the paper, we will use the following notation.
The superscript $\bot$ is defined for every $\@v=(v_1,v_2)^\top \in \Complex^2$ as
\begin{equation}
\label{e:perp}
\@ v^\bot \coloneq (v_2,-v_1)^\dagger\,,\qquad
\end{equation}
with the superscript $\dagger$ denoting conjugate transpose.

We introduce the notation $\A_{j}$ as the $j$-th column of a matrix $\A$.

For every complex quantity $\xi$,
we define the following shorthand notation,
\begin{equation}
\hat\xi \coloneqq -\frac{E_0^2}{\xi}\,,\qquad \xi\in\Complex\backslash\{0\}\,.
\end{equation}

For an arbitrary $3\times3$ matrix $\C = (c_{i,j})$,
we define the following subscripted matrices
\begin{equation}
\label{e:dodef}
\begin{aligned}
\C_{\d} & \coloneq \bpm
c_{1,1} & 0 & 0 \\ 0 & c_{2,2} & c_{2,3} \\ 0 & c_{3,2} & c_{3,3}
\epm\,,\qquad&
\C_{\o} & \coloneq \bpm
0 & c_{1,2} & c_{1,3} \\ c_{2,1} & 0 & 0 \\ c_{3,1} & 0 & 0
\epm\,,\\
\C_{\dd} & \coloneq \bpm
c_{1,1} & 0 & 0 \\ 0 & c_{2,2} & 0 \\ 0 & 0 & c_{3,3}
\epm\,,\qquad&
\C_{\rdo} & \coloneq \bpm
0 & 0 & 0 \\ 0 & 0 & c_{2,3} \\ 0 & c_{3,2} & 0
\epm\,,\qquad&
\C_{[1,1]} \coloneq \bpm
 c_{2,2} & c_{2,3} \\ c_{3,2} & c_{3,3}
\epm\,,
\end{aligned}
\end{equation}
where the subscripts $\d$ and $\o$ indicate the block-diagonal and block-off-diagonal parts of the matrix, respectively,
as well as subscript $\dd$ indicates the main diagonal of the matrix.
Using this notation,
one can show that for two given $3\times3$ matrices $\A$ and $\B$ the following identities hold
\begin{equation}
\label{e:matrixidentity}
\begin{aligned}
[\A\B]_{\d} & = \A_{\d}\B_{\d} + \A_{\o}\B_{\o}\,,\qquad&
[\A\B]_{\o} & = \A_{\d}\B_{\o} + \A_{\o}\B_{\d}\,,\\
[\A_{\d}\B_{\d}]_{\dd} & = \A_{\dd}\B_{\dd} + \A_{\rdo}\B_{\rdo}\,,\qquad&
[\A_{\d}\B_{\d}]_{\rdo} & = \A_{\dd}\B_{\rdo} + \A_{\rdo}\B_{\dd}\,.
\end{aligned}
\end{equation}

\section{IST: Direct scattering problem}
\label{s:directproblem}

We now begin formulating the IST for CMBE~\eqref{e:cmbe} with NZBGs.
This section focuses on the direct problem,
which consists of analyzing Equation~\eqref{e:laxpair1} in the Lax pair
and obtaining the scattering data at the injection point $z = 0$.

We recall that CMBE and the focusing Manakov system share the same scattering problem. Hence, the direct problem is formulated mainly following what has been done in~\cite{kbk2015}. Nonetheless, certain differences still appear between the current work and that concerning the focusing Manakov system. One of the main differences is that~\cite{kbk2015} utilized a simultaneous solution of both equations in the Lax pair, so that the scattering data only depends on the spectral parameter $\zeta$. However, due to the complicated propagation in $z$ of CMBE solutions with NZBG, we only consider a solution of the scattering problem, and so the scattering data must depend on both $z$ and $\zeta$. In other words, while the evolution stage of the IST for the focusing Manakov system in~\cite{kbk2015} is trivial, the propagation stage of CMBE with NZBG is highly nontrivial, and so is discussed in a separate section.

We recall the discussion in Section~\ref{s:uniformization} on the uniformization variable.
Starting from this section,
the original complex variable $k$ will be interpreted as a function of this variable, $\zeta$, via Equation~\eqref{e:k-lambda}.
The explicit $\zeta$-dependence is frequently omitted for brevity.

As usual, we first consider the asymptotic scattering problem as $t\to\pm\infty$
\begin{equation}
\label{e:asymp-scattering}
\bphi_t = \X_\pm\,\bphi\,,\qquad
\X_\pm(z,\zeta)
 \coloneqq \ii k\J+\Q_{\pm}(z)\,,\qquad
\Q_{\pm}(z)
 \coloneqq \bpm 0 & -\E_{\pm}(z)^\top \\
\E_{\pm}(z)^* & \@O \epm.
\end{equation}
Recall $\E_\pm(z) = \lim_{t\to\pm\infty}\E(t,z)$ are the asymptotic conditions in the distant past and future.

The eigenvalues of $\X_{\pm}(z,\zeta)$ are $\{-\ii k, \pm \ii\lambda\}$ with $\lambda$ given in Equation~\eqref{e:lambda-def}.
The continuous spectrum $\Sigma$ is defined as the set of $\zeta\in\Complex$ in which all the eigenvalues are purely imaginary,
i.e., all the eigenfunctions are bounded.
From Equation~\eqref{e:asymp-scattering}, we find $\Sigma$ to be
\begin{equation}
\Sigma \coloneq \Real\cup \Sigma_\circ\,,\quad
\mbox{where }
\Sigma_\circ \coloneq \{\zeta\in\Complex: |\zeta| = E_0\}\,.
\end{equation}
The continuous spectrum $\Sigma$ with its orientation is shown in Figure~\ref{f:analyticity} (left), and is similar to its analog in~\cite{kbk2015}.


%
\begin{figure}[t!]
\centering
\includegraphics[width = 0.375\textwidth]{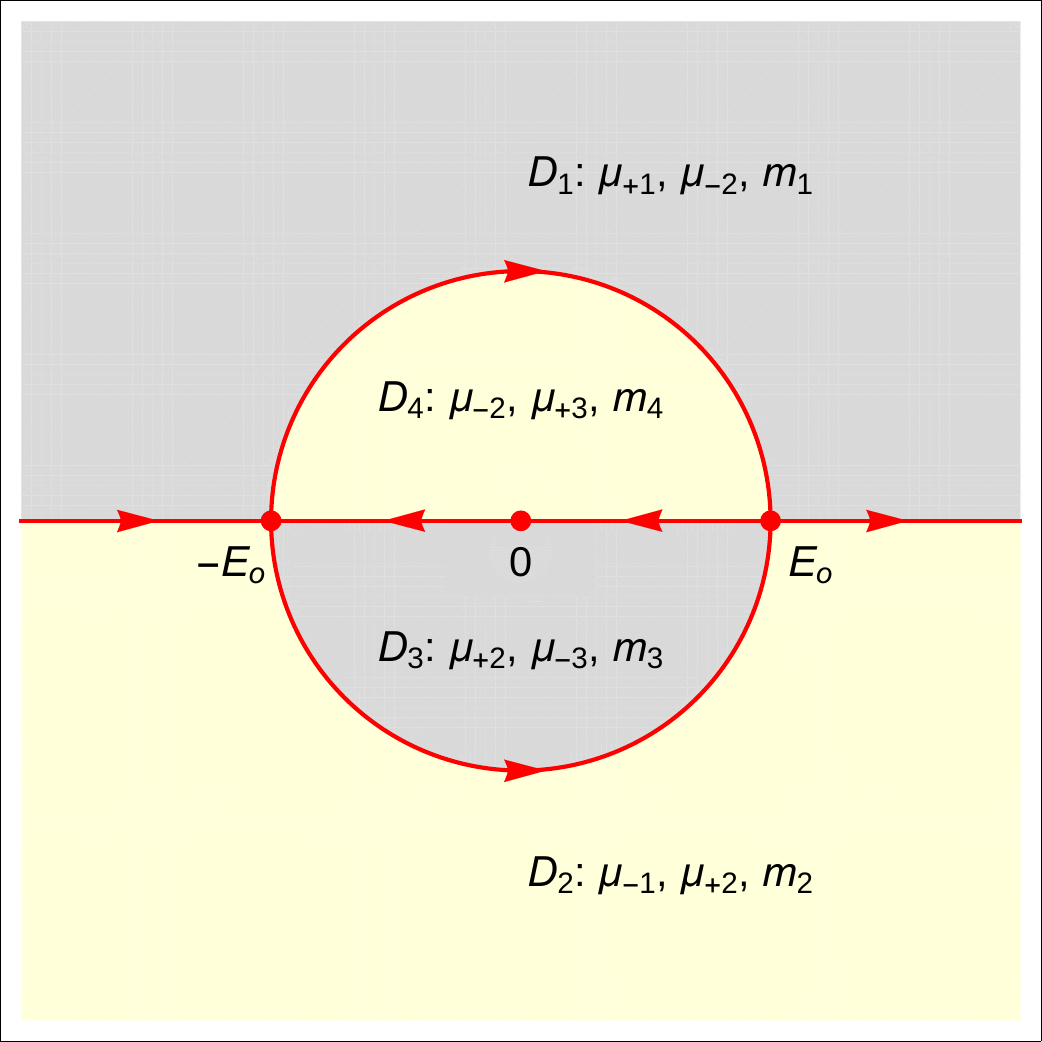}\qquad
\includegraphics[width = 0.375\textwidth]{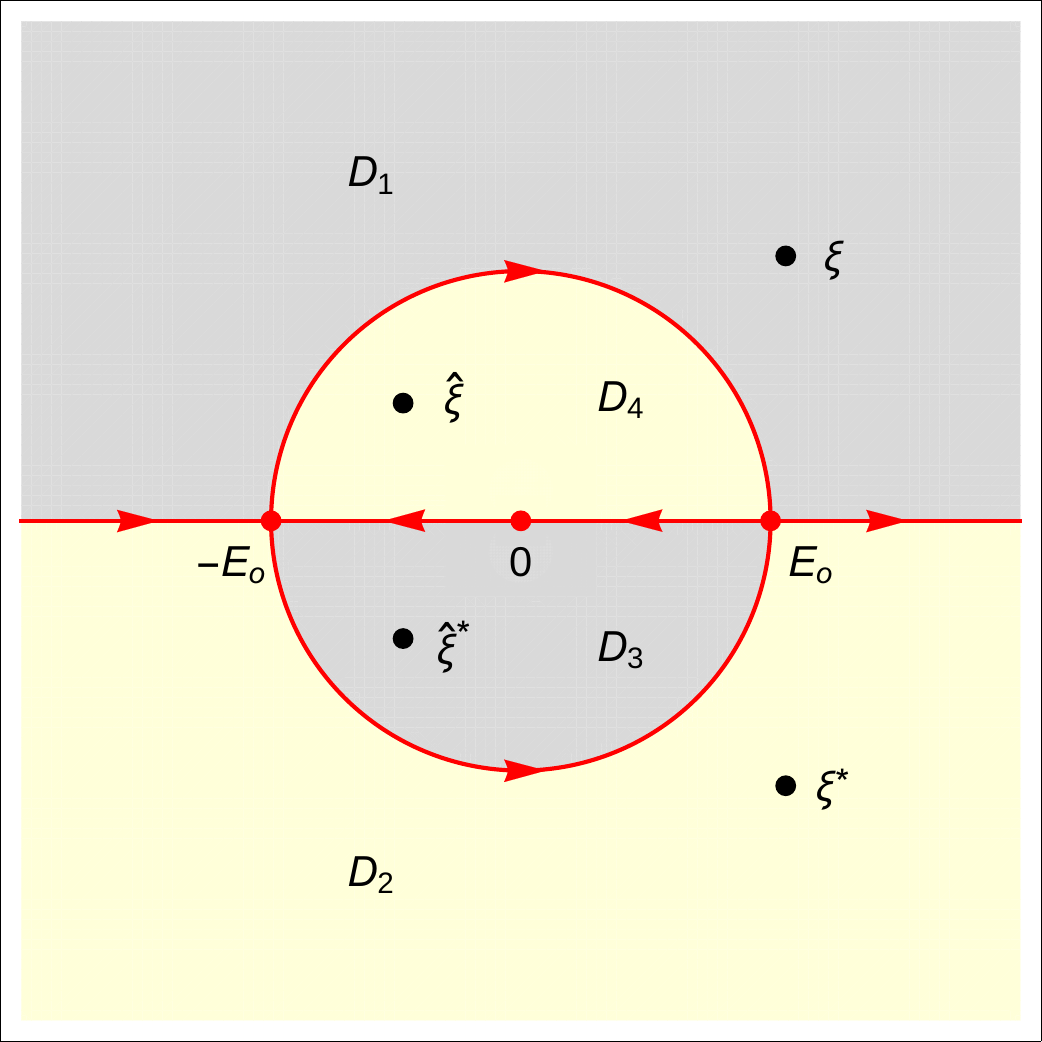}
\caption{
\textbf{Left:} Eigenfunctions in their analytic regions.
\textbf{Right:} The quartet of a discrete eigenvalue $\zeta = \xi$ in the complex $\zeta$-plane.
}
\label{f:analyticity}
\end{figure}

For future use, we also define the continuous spectrum without the two points $\pm\ii E_0$,
\begin{equation}
\label{e:Sigma*-def}
\Sigma^*\coloneqq \Sigma\backslash\{\pm\ii E_0\}\,.
\end{equation}
These two points are excluded because certain quantities become singular there, as we will see in the next section. Similar singularities have has also appeared in other integrable systems with NZBC/NZBG at infinity, e.g., the focusing NLS equation and the focusing Manakov system~\cite{kbk2015}. The local behavior of various quantities near these singular points in the IST will not be discussed in this work. In the context of the focusing NLS equation, this phenomenon was studied extensively and shown to be closely related to rogue-wave solutions~\cite{bm2019}.

\subsection{Jost solutions and scattering matrix}
\label{s:Jost}

\subsubsection{Jost solution.}

The asymptotic scattering problem~\eqref{e:asymp-scattering} can be solved by diagonalizing the matrix $\X_\pm(z,\zeta)$, which yields the following eigenvalue problem:
\begin{equation}
\label{e:Xpm-eigen}
\X_\pm \Y_\pm = \ii\Y_\pm \bLambda\,,\qquad \zeta\in\Sigma^*\,,
\end{equation}
where $\bLambda(\zeta)$ is the diagonal matrix containing the eigenvalues multiplied by the factor $-\ii$, and $\Y_\pm(z,\zeta)$ contains the corresponding eigenvectors:
\begin{equation}
\label{e:eigen}
\bLambda(\zeta) \coloneq \diag(\lambda,-k,-\lambda)\,,\qquad
\Y_{\pm}(z,\zeta) \coloneq \bpm1 & 0 & -\ii E_0/\zeta \\
-\ii\E_\pm^*/\zeta & (\E_\pm^\bot)^*/E_0 & \E_\pm^*/E_0
\epm\,.
\end{equation}
The inverse of the eigenvector matrix can be expressed as
\begin{equation}
\nonumber
\Y_\pm^{-1}(z,\zeta) = \frac{1}{\gamma(\zeta)}
\bpm
1 & \ii\E_\pm^\top/\zeta \\
0 & \gamma(\zeta)(\E_\pm^\bot)^\top/E_0 \\
\ii E_0/\zeta & \E_\pm^\top/E_0
\epm
\end{equation}
where $\gamma(\zeta)$ is defined as
\begin{equation}
\label{e:gamma-def}
\gamma(\zeta) \coloneq \det \Y_\pm(z,\zeta) =1+\frac{E_0^2}{\zeta^2}\,.
\end{equation}
The eigenvalue problem~\eqref{e:Xpm-eigen} is only valid on $\Sigma^*$ because the eigenvector matrix $\Y(z,\zeta)$ is singular at $\zeta = \pm\ii E_0$.

From the asymptotic scattering problem~\eqref{e:asymp-scattering}, one expects that the scattering problem~\eqref{e:laxpair1} admits solutions $\bphi_\pm(t,z,\zeta)$ with the asymptotic behavior
\begin{equation}
\label{e:Jostsol}
\begin{aligned}
\bphi_-(t,z,\zeta)
 & \coloneqq \Y_-(z,\zeta)\e^{\ii \bLambda t}+o(1)\,,\qquad && t\to-\infty\,,\\
\bphi_+(t,z,\zeta)
 & \coloneqq \Y_+(z,\zeta)\e^{\ii \bLambda t}+o(1)\,,\qquad && t\to\infty\,.
\end{aligned}
\end{equation}
The solutions $\bphi_\pm(t,z,\zeta)$ are called Jost solutions. In order to further study the Jost solutions, we introduce the modified eigenfunctions by removing oscillations,
\begin{equation}
\label{e:mudef}
\bmu_\pm(t,z,\zeta) \coloneq \bphi_\pm(t,z,\zeta) \e^{-\ii \bLambda t}\,,\qquad
\zeta\in\Sigma\,.
\end{equation}
These new functions solve the following integral equations on the continuous spectrum $\Sigma$,
\begin{equation}
\label{e:mu-integral}
\begin{aligned}
\bmu_-(t,z,\zeta)
 & = \Y_-(z,\zeta) + \int_{-\infty}^t \Y_-(z,\zeta)\e^{\ii \bLambda(t-\tau)}\Y_-(z,\zeta)^{-1}\Delta \Q_-(\tau,z)\bmu_-(\tau,z,\zeta) \e^{-\ii \bLambda(t-\tau)}\d \tau\,,\\
\bmu_+(t,z,\zeta)
 & = \Y_+(z,\zeta) - \int_{t}^\infty \Y_+(z,\zeta)\e^{\ii \bLambda(t-\tau)}\Y_+(z,\zeta)^{-1}\Delta \Q_+(\tau,z)\bmu_+(\tau,z,\zeta) \e^{-\ii \bLambda(t-\tau)}\d \tau\,,
\end{aligned}
\end{equation}
where we use the shorthand notation
\begin{equation}
\Delta \Q_\pm(t,z) \coloneq \Q(t,z) - \Q_\pm(z)\,.
\end{equation}
The existence of $\bmu_\pm(t,z,\zeta)$ can be derived from the above integral equations assuming sufficiently fast decay of $\Delta\Q_\pm(t,z)$ as $t\to\pm\infty$, respectively.

Correspondingly, the existence of the Jost eigenfunctions $\bphi_\pm(t,z,\zeta)$ on the continuous spectrum $\Sigma$ under certain assumptions on the injected pulse $\E(t,0)$ is obtained from Equation~\eqref{e:mudef}.

\begin{proposition}
\label{thm:Jostfundamental}
The two Jost eigenfunctions $\bphi_\pm(t,z,\zeta)$ are two fundamental matrix solutions on $\Sigma^*$ of the scattering problem.
In particular, $\det\bphi_\pm(t,z,\zeta) = \e^{\ii kt}\gamma(\zeta)$ on $\Sigma$.
\end{proposition}
\begin{proof}
Besides the aforementioned existence of $\bphi_\pm(t,z,\zeta)$,
it is sufficient to show $\det\bphi_\pm(t,z,\zeta)\ne0$ on $\Sigma^*$.
Since $\bphi_\pm(t,z,\zeta)$ solves the scattering problem,
Abel's identity yields
\begin{equation}
\nonumber
\frac{\partial}{\partial t}(\det \bphi_\pm)
 = \tr\X \det\bphi_\pm
 = -\ii k\det\bphi_\pm\,,
\end{equation}
which can be rewritten as
\begin{equation}
\nonumber
\frac{\partial}{\partial t}(\det(\bphi_\pm \e^{-\ii \bLambda t})) = 0.
\end{equation}
In other words, the determinant of $\bphi_\pm(t,z,\zeta)\e^{-\ii \bLambda t}$ is independent of $t$.
Combining this fact with the asymptotic behavior $\bphi_\pm= \Y_\pm \e^{\ii \bLambda t} + o(1)$ as $t\to\pm\infty$, respectively,
we obtain
\begin{equation}
\nonumber
\det(\bphi_\pm \e^{-\ii \bLambda t}) = \gamma(\zeta)\,,\qquad
\zeta\in\Sigma\,.
\end{equation}
Therefore, we conclude that
\begin{equation}
\nonumber
\det\bphi_\pm(t,z,\zeta) = \e^{\ii kt}\gamma(\zeta) \ne 0\,,\qquad
\zeta\in\Sigma^*.
\end{equation}
\end{proof}

Similarly to what happens in the focusing Manakov system with NZBC~\cite{kbk2015},
the eigenfunctions can be analytically extended off the continuous spectrum.
Recall that $\bmu_{\pm,j}$ denotes the $j$-th column of the matrix $\bmu_\pm$.
The columns of the eigenfunctions then exhibit the following analyticity properties:
\bse
\label{e:muanalyticity}
\begin{gather}
\bmu_{+,1}:\zeta\in D_1,\qquad
\bmu_{+,2}:\Im \zeta< 0,\qquad
\bmu_{+,3}:\zeta\in D_4\,,\\
\bmu_{-,1}:\zeta\in D_2,\qquad
\bmu_{-,2}:\Im \zeta> 0,\qquad
\bmu_{-,3}:\zeta\in D_3\,,
\end{gather}
\ese
where the domains of analyticity are given by
\bse
\label{e:region}
\begin{gather}
D_1 \coloneq \{\zeta:\Im \zeta> 0\wedge|\zeta|> E_0\},\qquad
D_2 \coloneq \{\zeta:\Im \zeta< 0\wedge |\zeta|> E_0\},\\
D_3 \coloneq \{\zeta:\Im \zeta< 0\wedge|\zeta|< E_0\},\qquad
D_4 \coloneq \{\zeta:\Im \zeta> 0\wedge |\zeta|< E_0\}.
\end{gather}
\ese
These regions are shown in Figure~\ref{f:analyticity} (left) together with the oriented continuous spectrum $\Sigma$.
The same analyticity properties also hold for the columns of the matrix $\bphi_\pm(t,z,\zeta)$ via the relation~\eqref{e:mudef}.

\begin{remark}
In each region $D_j$ there are only two analytic eigenfunctions,
so they are not enough to formulate the inverse problem by building a $3\times3$ Riemann-Hilbert problem (RHP).
Therefore, unlike in the two-level MBE case,
we need to seek additional eigenfunctions.
These can be formed using the adjoint problem,
which is presented in Section~\ref{s:adjoint}.
\end{remark}

\subsubsection{Scattering matrix.}

Proposition~\ref{thm:Jostfundamental} implies that there exists an invertible $3\times3$ matrix $S(z,\zeta)$ such that
\begin{equation}
\label{e:Sdef}
\bphi_+(t,z,\zeta) = \bphi_-(t,z,\zeta)\, \S(z,\zeta)\,,\qquad
\zeta\in\Sigma^*.
\end{equation}
Proposition~\ref{thm:Jostfundamental} also implies the equation
\begin{equation}
\nonumber
\det \S(z,\zeta) = 1\,,\qquad \zeta\in\Sigma^*\,.
\end{equation}
Let us write the entries of the scattering matrix as $\S(z,\zeta) = (a_{i,j})_{3\times 3}$ and of its inverse as $\S^{-1}(\zeta,z)=(b_{i,j})_{3\times3}$.
Following the same procedure as in~\cite{kbk2015},
the scattering coefficients can be analytically extended off the continuous spectrum $\Sigma$ into the following regions:
\begin{equation}
\label{e:abanalyticity}
\begin{aligned}
a_{1,1} & :\,\zeta\in D_1\,,\qquad
a_{2,2}:\,\Im\zeta<0\,,\qquad
a_{3,3}:\,\zeta\in D_4\,,\\
b_{1,1} & :\,\zeta\in D_2\,,\qquad
b_{2,2}:\,\Im \zeta>0\,,\qquad
b_{3,3}:\,\zeta\in D_3\,.
\end{aligned}
\end{equation}

\subsection{Adjoint problem and auxiliary eigenfunctions}
\label{s:adjoint}

In this subsection, we construct additional eigenfunctions analytic in each region $D_j$ in order to formulate a $3\times3$ RHP later on.
To do so, let us consider the adjoint scattering problem
\begin{equation}
\label{e:adjoinscattering}
\widetilde \bphi_t=\widetilde \X\,\widetilde \bphi\,,\qquad
\widetilde \X(t,z,\zeta) \coloneq -\ii k\J + \Q(t,z)^*\,.
\end{equation}
Similarly to what was done in~\cite{kbk2015},
it can be shown that if $\~v(t,z,\zeta)$ and $\~w(t,z,\zeta)$ are two solutions of the adjoint problem~\eqref{e:adjoinscattering},
then the following combination
\begin{equation}
\@u(t,z,\zeta) \coloneq \e^{-\ii kt}[\~v \times \~w](t,z,\zeta)\,,
\end{equation}
is a solution of the original scattering problem~\eqref{e:laxpair1}.
We use this relation to construct additional eigenfunctions of the original scattering problem.
One thus needs to solve the adjoint scattering problem~\eqref{e:adjoinscattering} first, and find its eigenfunctions. Following the same procedure as before,
we know that the matrix $\widetilde \X_{\pm}(z,\zeta)$ of the asymptotic adjoint scattering problem has eigenvalues $\ii k$ and $\pm \ii\lambda$.
The eigenvalue matrix can be written as $-\ii\bLambda(\zeta)$ where $\bLambda(\zeta)$ is given in Equation~\eqref{e:eigen}.
One can choose the eigenvector matrix as
$\widetilde \Y_\pm(z,\zeta) \coloneq \Y_\pm(z,\zeta^*)^*$ with $\Y_\pm(z,\zeta)$ given in Equation~\eqref{e:eigen}.
Note that $\det\widetilde \Y_\pm(z,\zeta) = \gamma(\zeta)$ as well.

As in Section~\ref{s:Jost}, for all $\zeta\in\Sigma$, we define the Jost solutions of the adjoint scattering problem with the asymptotic behavior
\begin{equation}
\widetilde\bphi_\pm(t,z,\zeta)
 = \widetilde \Y_\pm(z,\zeta)\,\e^{-\ii \bLambda(\zeta) t}+o(1)\,,\qquad t\to\pm\infty\,.
\end{equation}
Then, we introduce the modified adjoint eigenfunctions by removing the oscillations
\begin{equation}
\label{e:tildemu-def}
\widetilde\bmu_\pm(t,z,\zeta)
 \coloneq \widetilde\bphi_\pm(t,z,\zeta) \,\e^{\ii \bLambda(\zeta) t}\,,\qquad
\zeta\in\Sigma\,,
\end{equation}
so that $\widetilde\bmu_\pm(t,z,\zeta) = \widetilde \Y_\pm(z,\zeta) + o(1)$ as $t\to\pm\infty$, respectively.
Again, similarly to what we did in Section~\ref{s:Jost},
one can show that the columns of $\widetilde\bmu_\pm(t,z,\zeta)$ can be analytically extended into the following regions of the complex plane:
\begin{gather*}
\widetilde\bmu_{-,1}:\,\zeta\in D_1\,,\qquad
\widetilde\bmu_{-,2}:\,\Im\zeta<0\,,\qquad
\widetilde\bmu_{-,3}:\,\zeta\in D_4\,,\\
\widetilde\bmu_{+,1}:\,\zeta\in D_2\,,\qquad
\widetilde\bmu_{+,2}:\,\Im\zeta>0\,,\qquad
\widetilde\bmu_{+,3}:\,\zeta\in D_3\,.
\end{gather*}

We then define the adjoint scattering matrix $\widetilde \S(z,\zeta)$ as
\begin{equation}
\widetilde \bphi_+(t,z,\zeta) = \widetilde\bphi_-(t,z,\zeta)\,\widetilde \S(z,\zeta)\,,\qquad
\zeta\in\Sigma\,.
\end{equation}
Following the same notation $\widetilde\S = (\widetilde a_{i,j})_{3\times3}$ and $\widetilde\S^{-1} = (\widetilde b_{i,j})_{3,3}$ and arguments similar to those in Section~\ref{s:Jost},
the scattering coefficients can be analytically extended off the continuous spectrum:
\begin{gather*}
\widetilde b_{1,1}:\,\zeta\in D_1\,,\qquad
\widetilde b_{2,2}:\,\Im\zeta<0\,,\qquad
\widetilde b_{3,3}:\,\zeta\in D_4\,,\\
\widetilde a_{1,1}:\,\zeta\in D_2\,,\qquad
\widetilde a_{2,2}:\,\Im\zeta>0\,,\qquad
\widetilde a_{3,3}:\,\zeta\in D_3\,.
\end{gather*}
Finally, we define four new solutions of the original Lax pair,
\begin{gather*}
\bchi_1(t,z,\zeta) \coloneq \e^{-\ii kt}[\widetilde\bphi_{-,1}\times\widetilde\bphi_{+,2}](t,z,\zeta)\,,\qquad
\bchi_2(t,z,\zeta) \coloneq \e^{-\ii kt}[\widetilde\bphi_{+,1}\times\widetilde\bphi_{-,2}](t,z,\zeta)\,,\\
\bchi_3(t,z,\zeta) \coloneq \e^{-\ii kt}[\widetilde\bphi_{-,2}\times\widetilde\bphi_{+,3}](t,z,\zeta)\,,\qquad
\bchi_4(t,z,\zeta) \coloneq \e^{-\ii kt}[\widetilde\bphi_{+,2}\times\widetilde\bphi_{-,3}](t,z,\zeta)\,.
\end{gather*}
\begin{remark}
\label{rmk:adjoint-analyticity}
The new eigenfunctions $\bchi_j(t,z,\zeta)$ with $j = 1,\dots,4$ are called auxiliary eigenfunctions, and can be shown to be analytic in $D_j$,
using the relation~\eqref{e:tildemu-def} and the analyticity of $\widetilde \bmu_{\pm,j}(t,z,\zeta)$ from above.
\end{remark}
Moreover, the scattering matrix $\S(z,\zeta)$ and the adjoint scattering matrix $\widetilde \S(z,\zeta)$ are related by
\begin{equation}
\widetilde \S(z,\zeta)^{-1}
 = \bGamma(\zeta)\, \S(z,\zeta)^\top\,\bGamma(\zeta)^{-1}\,,\qquad
\zeta\in\Sigma^*\,,
\end{equation}
where $\bGamma(\zeta) \coloneq \diag\left(1,\gamma(\zeta),1\right)$.

For all $\zeta\in\Sigma$, the Jost eigenfunctions have the following decompositions:
\begin{equation}
\label{e:phidecompose}
\begin{aligned}
\bphi_{-,1}(t,z,\zeta) & = \frac{1}{b_{2,2}(z,\zeta)}\left[\bchi_4(t,z,\zeta)+b_{2,1}(z,\zeta)\bphi_{-,2}(t,z,\zeta)\right]\\
& =\frac{1}{b_{3,3}(z,\zeta)}[b_{3,1}(z,\zeta)\bphi_{-,3}(t,z,\zeta)+\bchi_3(t,z,\zeta)]\,,\\
\bphi_{+,1}(t,z,\zeta) & =\frac{1}{a_{2,2}(z,\zeta)}\left[\bchi_3(t,z,\zeta)+a_{2,1}(z,\zeta)\bphi_{+,2}(t,z,\zeta)\right]\\
& =\frac{1}{a_{3,3}(z,\zeta)}[a_{3,1}(z,\zeta)\bphi_{+,3}(t,z,\zeta)+\bchi_4(t,z,\zeta)]\,,\\
\bphi_{-,3}(t,z,\zeta) & =\frac{1}{b_{2,2}(z,\zeta)}\left[\bchi_1(t,z,\zeta)+b_{2,3}(z,\zeta)\bphi_{-,2}(t,z,\zeta)\right]\\
& =\frac{1}{b_{1,1}(z,\zeta)}[b_{1,3}(z,\zeta)\bphi_{-,1}(t,z,\zeta)+\bchi_2(t,z,\zeta)]\,,\\
\bphi_{+,3}(t,z,\zeta) & =\frac{1}{a_{2,2}(z,\zeta)}\left[\bchi_2(t,z,\zeta)+a_{2,3}(z,\zeta)\bphi_{+,2}(t,z,\zeta)\right]\\
& =\frac{1}{a_{1,1}(z,\zeta)}[a_{1,3}(z,\zeta)\bphi_{+,1}(t,z,\zeta)+\bchi_1(t,z,\zeta)]\,.
\end{aligned}
\end{equation}

In addition, we remove the exponential oscillations and define the modified auxiliary eigenfunctions
\begin{equation}
\label{e:mdef}
\begin{aligned}
\m_j(t,z,\zeta) & \coloneq \bchi_j(t,z,\zeta)\e^{\ii \lambda t}\,,\qquad& j=1,2,\\
\m_j(t,z,\zeta) & \coloneq \bchi_j(t,z,\zeta)\e^{-\ii \lambda t}\,,\qquad& j=3,4.
\end{aligned}
\end{equation}
Next, we list the asymptotic behavior of the modified auxiliary eigenfunctions as $t\to\pm\infty$.
For all $\zeta$ in their corresponding domain of analyticity,
the modified auxiliary eigenfunctions have the following asymptotic behavior
as $t\to\pm\infty$:
\begin{equation}
\label{e:masym}
\begin{aligned}
\lim_{t\to-\infty}\m_1(t,z,\zeta) & = b_{1,1}(z,\zeta)\Y_{-3}(z,\zeta)\,,\qquad&
\lim_{t\to\infty}\m_1(t,z,\zeta) & = a_{2,2}(z,\zeta)\Y_{+3}(z,\zeta)\,,\\
\lim_{t\to-\infty}\m_2(t,z,\zeta) & = b_{2,2}(z,\zeta)\Y_{-3}(z,\zeta)\,,\qquad&
\lim_{t\to\infty}\m_2(t,z,\zeta) & = a_{1,1}(z,\zeta)\Y_{+3}(z,\zeta)\,,\\
\lim_{t\to-\infty}\m_3(t,z,\zeta) & = b_{2,2}(z,\zeta)\Y_{-1}(z,\zeta)\,,\qquad&
\lim_{t\to\infty}\m_3(t,z,\zeta) & = a_{3,3}(z,\zeta)\Y_{+1}(z,\zeta)\,,\\
\lim_{t\to-\infty}\m_4(t,z,\zeta) & = b_{3,3}(z,\zeta)\Y_{-1}(z,\zeta)\,,\qquad&
\lim_{t\to\infty}\m_4(t,z,\zeta) & = a_{2,2}(z,\zeta)\Y_{+1}(z,\zeta)\,.
\end{aligned}
\end{equation}

Now, for each region $D_j$, there are three independent eigenfunctions as shown in Figure~\ref{f:analyticity} (left). The independence can be verified directly by calculating Wronskians.

\subsection{Symmetries}

Here we discuss the symmetries of all the eigenfunctions and scattering coefficients derived from the scattering problem.
These symmetries play an important role throughout the entire IST procedure.
There are two symmetries to be discussed in this section.

\subsubsection{First symmetry.}
\label{s:first-symmetry}

Let us consider the transform $\zeta\mapsto\zeta^*$,
implying $(k,\lambda)\mapsto(k^*,\lambda^*)$.
This transform maps a point from upper half plane into the lower half plane or vice versa in the $\zeta$-plane or on the same sheet of the $k$-plane.

\begin{lemma}
\label{thm:symmetry1}
If $\bphi(t,z,\zeta)$ is a non-singular square matrix solution of the Lax pair~\eqref{e:laxpair},
then so is $\bphi(t,z,\zeta^*)^{-\dagger}$,
where the superscript $-\dagger$ denotes the inversion and conjugate transpose.
\end{lemma}

Lemma~\ref{thm:symmetry1} can be proved by direct substitution.
Applying Lemma~\ref{thm:symmetry1} to the Jost eigenfunctions $\bphi_\pm(t,z,\zeta)$,
and noticing the uniqueness of solutions,
one concludes that there exits an invertible matrix $\C$ such that
\begin{equation}
\label{e:phisymmetry1}
\bphi_\pm(t,z,\zeta^*)^{-\dagger}\C(\zeta)=\bphi_\pm(t,z,\zeta)\,, \qquad
\zeta\in\Sigma^*\,.
\end{equation}
It is not obvious that (i) the connection matrix $\C(\zeta)$ is the same for both Jost eigenfunction matrices $\bphi_\pm(t,z,\zeta)$;
(ii) $\C(\zeta)$ is independent of $z$,
but one can use the asymptotic behavior of $\bphi_\pm(t,z,\zeta)$ as $t\to\pm\infty$ and obtain $\C(\zeta)$ explicitly as
\begin{equation}
\C(\zeta) = \diag(\gamma(\zeta),1,\gamma(\zeta))\,.
\end{equation}

Furthermore, the identity
\begin{equation}
\bphi_\pm(t,z,\zeta)^{-\top}
 = \frac{1}{\det \bphi_\pm(t,z,\zeta)}(\bphi_{\pm,2}\times\bphi_{\pm,3},\bphi_{\pm,3}\times\bphi_{\pm,1},\bphi_{\pm,1}\times\bphi_{\pm,2})(t,z,\zeta)\,,
\end{equation}
together with the symmetry~\eqref{e:phisymmetry1}, imply
\begin{equation}
\begin{aligned}
\bphi_{-,1}(t,z,\zeta^*)^* & = \frac{\e^{\ii kt}}{b_{2,2}(z,\zeta)}[\bphi_{-,2}\times\bchi_1](t,z,\zeta)\,,\\
\bphi_{+,1}(t,z,\zeta^*)^* & = \frac{\e^{\ii kt}}{a_{2,2}(z,\zeta)}[\bphi_{+,2}\times\bchi_2](t,z,\zeta)\,,\\
\bphi_{-,2}(t,z,\zeta^*)^* & = \frac{\e^{\ii kt}}{\gamma(\zeta)b_{1,1}(z,\zeta)}[\bchi_2\times\bphi_{-,1}](t,z,\zeta)
 = \frac{\e^{\ii kt}}{\gamma(\zeta)b_{3,3}(z,\zeta)}[\bphi_{-,3}\times\bchi_3](t,z,\zeta)\,,\\
\bphi_{+,2}(t,z,\zeta^*)^* & = \frac{\e^{\ii kt}}{\gamma(\zeta)a_{1,1}(z,\zeta)}[\bchi_1\times\bphi_{+,1}](t,z,\zeta)
 = \frac{\e^{\ii kt}}{\gamma(\zeta)a_{3,3}(z,\zeta)}[\bphi_{+,3}\times\bchi_3](t,z,\zeta)\,,\\
\bphi_{-,3}(t,z,\zeta^*)^* & = \frac{\e^{\ii kt}}{b_{2,2}(z,\zeta)}[\bchi_4\times\bphi_{-,2}](t,z,\zeta)\,,\\
\bphi_{+,3}(t,z,\zeta^*)^* & = \frac{\e^{\ii kt}}{a_{2,2}(z,\zeta)}[\bchi_3\times\bphi_{+,2}](t,z,\zeta)\,.
\end{aligned}
\end{equation}
These relations involving $\bchi_j(t,z,\zeta)$ are also valid for $\zeta\in D_j$, with $j=1,\dots,4$, besides the continuous spectrum $\Sigma^*$.

Thereafter, by using the definition of the scattering matrix~\eqref{e:Sdef},
we also conclude that the scattering matrix and its inverse satisfy the symmetry relation:
\begin{equation}
\S(z,\zeta^*)^{-\dagger}
 = \C(\zeta) \S(z,\zeta) \C(\zeta)^{-1}\,,\qquad \zeta\in\Sigma^*\,.
\end{equation}
Componentwise, for all $\zeta\in\Sigma^*$, the above equation yields
\begin{equation}
\label{e:asymmetry1}
\begin{aligned}
a_{1,1}(z,\zeta) & = b_{1,1}(z,\zeta^*)^*\,,\qquad&
a_{1,2}(z,\zeta) & = \frac{b_{2,1}(z,\zeta^*)^*}{\gamma(\zeta)}\,,\qquad&
a_{1,3}(z,\zeta) & = b_{3,1}(z,\zeta^*)^*\,,\\
a_{2,1}(z,\zeta) & = \gamma(\zeta)b_{1,2}(z,\zeta^*)^*\,,\qquad&
a_{2,2}(z,\zeta) & = b_{2,2}(z,\zeta^*)^*\,,\qquad&
a_{2,3}(z,\zeta) & = \gamma(\zeta)b_{3,2}(z,\zeta^*)^*\,,\\
a_{3,1}(z,\zeta) & = b_{1,3}(z,\zeta^*)^*\,,\qquad&
a_{3,2}(z,\zeta) & = \frac{b_{2,3}(z,\zeta^*)^*}{\gamma(\zeta)}\,,\qquad&
a_{3,3}(z,\zeta) & = b_{3,3}(z,\zeta^*)^*\,.
\end{aligned}
\end{equation}
The Schwarz reflection principle then allows us to conclude
\begin{equation}
\begin{aligned}
a_{1,1}(z,\zeta) & = b_{1,1}(z,\zeta^*)^*\,,\qquad && \zeta\in D_1\,,\\
a_{2,2}(z,\zeta) & = b_{2,2}(z,\zeta^*)^*\,,\qquad && \Im \zeta<0\,,\\
a_{3,3}(z,\zeta) & = b_{3,3}(z,\zeta^*)^*\,,\qquad && \zeta\in D_4\,.
\end{aligned}
\end{equation}

We can also obtain symmetry relations for the auxiliary eigenfunctions.
It can be shown that the auxiliary eigenfunctions satisfy the following symmetry relations~\cite{kbk2015},
\begin{equation}
\begin{aligned}
\bchi_1(t,z,\zeta^*)^*
 & =\e^{\ii kt}[\bphi_{-,1}\times\bphi_{+,2}](t,z,\zeta)\,,\qquad && \zeta\in D_2\,,\\
\bchi_2(t,z,\zeta^*)^*
 & =\e^{\ii kt}[\bphi_{+,1}\times\bphi_{-,2}](t,z,\zeta)\,,\qquad && \zeta\in D_1\,,\\
\bchi_3(t,z,\zeta^*)^*
 & =\e^{\ii kt}[\bphi_{-,2}\times\bphi_{+,3}](t,z,\zeta)\,,\qquad && \zeta\in D_4\,,\\
\bchi_4(t,z,\zeta^*)^*
 & =\e^{\ii kt}[\bphi_{+,2}\times\bphi_{-,3}](t,z,\zeta)\,,\qquad && \zeta\in D_3\,.
\end{aligned}
\end{equation}
In other words, we have
\begin{equation}
\bphi_{\pm,j}(t,z,\zeta^*)^*
 = \e^{\ii kt}[\bphi_{\pm,l}\times\bphi_{\pm,m}](t,z,\zeta)/\gamma(\zeta)\,,
\end{equation}
where $j,l$ and $m$ are cyclic indices.

\subsubsection{Second symmetry.}

Next, we consider the transformation $\zeta\mapsto \hat \zeta$ in the scattering problem,
mapping the exterior of the circle $\Sigma_\circ$ into its interior,
and vice versa,
and implying $(k,\lambda)\mapsto(k,-\lambda)$.

\begin{lemma}
\label{thm:symmetry2}
If $\bphi(t,z,\zeta)$ is a matrix solution of the Lax pair~\eqref{e:laxpair},
so is $\bphi(t,z,\hat\zeta)$.
\end{lemma}
This lemma can be proved by direct substitution.
Applying Lemma~\ref{thm:symmetry2} to the Jost eigenfunctions and recalling the uniqueness of solutions,
as proved in~\cite{kbk2015}, it can be shown that there exists an invertible matrix $\bPi(\zeta)$ such that the Jost eigenfunctions satisfy
\begin{equation}
\label{e:phisymmetry2}
\bphi_\pm(t,z,\zeta) = \bphi_\pm(t,z,\hat\zeta)\bPi(\zeta)\,,\qquad
\zeta\in\Sigma\,.
\end{equation}
Similarly to what has been done in discussing the first symmetry,
we consider the asymptotic behavior of $\bphi_\pm(t,z,\zeta)$ as $t\to\pm\infty$, respectively.
By comparison, we obtain an explicit form of the matrix
\begin{equation}
\label{e:Pi}
\bPi(\zeta)
 = \bpm 0 & 0 & -\ii E_0/\zeta \\ 0 & 1 & 0 \\ -\ii E_0/\zeta & 0 & 0 \epm\,.
\end{equation}

As before, the analyticity properties of the eigenfunctions allow us to analytically extend all of the above relations off the continuous spectrum:
\begin{equation}
\label{e:eigenfunctionsymmetry1}
\begin{aligned}
\bphi_{\pm,1}(t,z,\zeta) & = -\frac{\ii E_0}{\zeta}\bphi_{\pm,3}(t,z,\hat\zeta)\,,\qquad
&&\Im \zeta\gl0\wedge|\zeta|>E_0\,,\\
\bphi_{\pm,2}(t,z,\zeta) & = \bphi_{\pm, 2}(t,z,\hat\zeta)\,,\qquad &&\Im\zeta\lg0\,,\\
\bphi_{\pm,3}(t,z,\zeta) & = -\frac{\ii E_0}{\zeta}\bphi_{\pm,1}(t,z,\hat\zeta)\,,\qquad
&&\Im\zeta\lg0\wedge|\zeta|<E_0\,.
\end{aligned}
\end{equation}
The definition of the scattering matrix~\eqref{e:Sdef} yields the corresponding symmetry
\begin{equation}
\S(z,\hat\zeta) = \bPi(\zeta)\S(z,\zeta)\bPi^{-1}(\zeta)\,,\qquad \zeta\in\Sigma^*\,.
\end{equation}
Componentwise, on the continuous spectrum $\Sigma^*$, we have the equations
\begin{equation}
\label{e:asymmetry2}
\begin{aligned}
a_{1,1}(z,\zeta) & =a_{3,3}(z,\hat\zeta)\,,\qquad&
a_{1,2}(z,\zeta) & =-\frac{\ii E_0}{\zeta}a_{3,2}(z,\hat\zeta)\,,\qquad&
a_{1,3}(z,\zeta) & =a_{3,1}(z,\hat\zeta)\,,\\
a_{2,1}(z,\zeta) & =\frac{\ii\zeta}{E_0}a_{2,3}(z,\hat\zeta)\,,\qquad&
a_{2,2}(z,\zeta) & =a_{2,2}(z,\hat\zeta)\,,\qquad&
a_{2,3}(z,\zeta) & =\frac{\ii\zeta}{E_0}a_{2,1}(z,\hat\zeta)\,,\\
a_{3,1}(z,\zeta) & =a_{1,3}(z,\hat\zeta)\,,\qquad&
a_{3,2}(z,\zeta) & =-\frac{\ii E_0}{\zeta}a_{1,2}(z,\hat\zeta)\,,\qquad&
a_{3,3}(z,\zeta) & =a_{1,1}(z,\hat\zeta)\,.
\end{aligned}
\end{equation}

The analyticity of the scattering coefficients~\eqref{e:abanalyticity} allows us to conclude
\begin{equation*}
\begin{aligned}
b_{1,1}(z,\zeta) & = b_{3,3}(z,\hat\zeta)\,,\qquad&& \zeta\in D_2\,,\\
a_{1,1}(z,\zeta) & = a_{3,3}(z,\hat\zeta)\,,\qquad&& \zeta\in D_1\,,\\
a_{2,2}(z,\zeta) & = a_{2,2}(z,\hat\zeta)\,,\qquad&& \Im\zeta\le0\,,\\
b_{2,2}(z,\zeta) & = b_{2,2}(z,\hat\zeta)\,,\qquad&& \Im\zeta\ge0\,,
\end{aligned}
\end{equation*}
where we recall that $a_{i,j}$ and $b_{i,j}$ are the entries of the scattering matrix and its inverse, respectively.
Moreover, repeating the procedure in the discussion of the first symmetry in Section~\ref{s:first-symmetry},
we also conclude that the auxiliary eigenfunctions satisfy the relations
\begin{equation}
\label{e:eigenfunctionsymmetry2}
\begin{aligned}
\bchi_1(t,z,\zeta)
 & = -\frac{\ii E_0}{\zeta}\bchi_4(t,z,\hat\zeta)\,,\qquad&& \zeta\in D_1\,,\\
\bchi_2(t,z,\zeta)
 & = -\frac{\ii E_0}{\zeta}\bchi_3(t,z,\hat\zeta)\,,\qquad&& \zeta\in D_2\,.
\end{aligned}
\end{equation}

\subsubsection{Combined symmetry and reflection coefficients.}

We define the following quantities $r_j(\zeta)$ with $j = 1,2,3$ on the continuous spectrum,
which are the reflection coefficients appearing in the inverse problem:
\bse
\label{e:reflection-def}
\begin{equation}
\begin{aligned}
r_1(z,\zeta)
 & \coloneq \frac{a_{2,1}(z,\zeta)}{a_{1,1}(z,\zeta)}
 = \gamma(\zeta)\frac{b_{1,2}(z,\zeta^*)^*}{b_{1,1}(z,\zeta^*)^*}\,,\\
r_2(z,\zeta)
 & \coloneq \frac{a_{3,1}(z,\zeta)}{a_{1,1}(z,\zeta)}
 = \frac{b_{1,3}(z,\zeta^*)^*}{b_{1,1}(z,\zeta^*)^*}\,,\\
r_3(z,\zeta)
 & \coloneq \frac{a_{3,2}(z,\zeta)}{a_{2,2}(z,\zeta)}
 = \frac{1}{\gamma(\zeta)}\frac{b_{2,3}(z,\zeta^*)^*}{b_{2,2}(z,\zeta^*)^*}\,.
\end{aligned}
\end{equation}
The symmetries~\eqref{e:asymmetry2} of the scattering matrix yield
\begin{equation}
\begin{aligned}
r_1(z,\hat\zeta)
 & = -\frac{\ii E_0}{\zeta}\frac{a_{2,3}(z,\zeta)}{a_{3,3}(z,\zeta)}
 = -\gamma(\zeta)\frac{\ii E_0}{\zeta}\frac{b_{3,2}(z,\zeta^*)^*}{b_{3,3}(z,\zeta^*)^*}\,,\\
r_2(z,\hat\zeta)
 & = \frac{a_{1,3}(z,\zeta)}{a_{3,3}(z,\zeta)}
 = \frac{b_{3,1}(z,\zeta^*)^*}{b_{3,3}(z,\zeta^*)^*}\,,\\
r_3(z,\hat\zeta)
 & = \frac{\ii\zeta}{E_0}\frac{a_{1,2}(z,\zeta)}{a_{2,2}(z,\zeta)}
 = \frac{\ii\zeta}{E_0\gamma(\zeta)}\frac{b_{2,1}(z,\zeta^*)^*}{b_{2,2}(z,\zeta^*)^*}\,.
\end{aligned}
\end{equation}
\ese
Furthermore, the definition of $\S(z,\zeta)$ yields the following relations
\begin{equation}
a_{3,2}(z,\zeta) = b_{1,2}(z,\zeta)b_{3,1}(z,\zeta) - b_{1,1}(z,\zeta)b_{3,2}(z, \zeta)\,,\qquad
\zeta\in\Sigma^*\,.
\end{equation}
This equation shows that, in fact, the three reflection coefficients are not independent,
and they are related by the equation
\begin{equation}
r_3(z,\zeta) = \frac{b_{1,1}(z,\zeta)b_{1,1}(z,\hat\zeta)}{a_{2,2}(z,\zeta)\gamma(\zeta)}
\left[r_1(z,\zeta^*)^* r_2(z,\hat\zeta^*)^* + \frac{\ii\zeta}{E_0}r_1(z,\hat\zeta^*)^*\right]\,,\qquad
\zeta\in\Sigma^*\,.
\end{equation}

\subsection{Discrete spectrum, norming constants and their symmetries}
\label{s:eigen}

With the eigenfunctions and their behavior on the continuous spectrum established,
we next discuss their behavior on the discrete spectrum.
Recall that we have assumed there to be no spectral singularities or embedded eigenvalues,
so the continuous and discrete spectra of the current problem are well separated.

\subsubsection{Discrete spectrum and norming constants.}

The discrete spectrum of the direct problem is the set of all values $\zeta\in\Complex\backslash\Sigma$ for which bounded eigenfunctions exist.
In order to characterize the discrete spectrum,
it is convenient to introduce the following $3\times3$ matrices:
\begin{equation}
\label{e:bPhi(j)-def}
\begin{aligned}
\bPhi^{(1)}(t,z,\zeta) & \coloneq (\bphi_{+,1}(t,z,\zeta),\bphi_{-,2}(t,z,\zeta),\bchi_1(t,z,\zeta))\,,\qquad&&
 \zeta\in D_1\,,\\
\bPhi^{(2)}(t,z,\zeta) & \coloneq
(\bphi_{-,1}(t,z,\zeta),\bphi_{+,2}(t,z,\zeta),\bchi_2(t,z,\zeta))\,,\qquad&&
 \zeta\in D_2\,,\\
\bPhi^{(3)}(t,z,\zeta) & \coloneq
(\bchi_3(t,z,\zeta),\bphi_{+,2}(t,z,\zeta),\bphi_{-,3}(t,z,\zeta))\,,\qquad&&
 \zeta\in D_3\,,\\
\bPhi^{(4)}(t,z,\zeta) & \coloneq
(\bchi_4(t,z,\zeta),\bphi_{-,2}(t,z,\zeta),\bphi_{+,3}(t,z,\zeta))\,,\qquad&&
 \zeta\in D_4\,.
\end{aligned}
\end{equation}
It follows immediately that
\begin{equation}
\label{e:det-Phi}
\begin{aligned}
\det\bPhi^{(1)}(t,z,\zeta) & = b_{2,2}(z,\zeta)a_{1,1}(z,\zeta)\e^{-\ii kt}\,,\qquad&& \zeta\in D_1\,,\\
\det\bPhi^{(2)}(t,z,\zeta) & = b_{1,1}(z,\zeta)a_{2,2}(z,\zeta)\e^{-\ii kt}\,,\qquad&&\zeta\in D_2\,,\\
\det\bPhi^{(3)}(t,z,\zeta) & = b_{3,3}(z,\zeta)a_{2,2}(z,\zeta)\e^{-\ii kt}\,,\qquad&&\zeta\in D_3\,,\\
\det\bPhi^{(4)}(t,z,\zeta) & = b_{2,2}(z,\zeta)a_{3,3}(z,\zeta)\e^{-\ii kt}\,,\qquad&&\zeta\in D_4\,.
\end{aligned}
\end{equation}
If the above determinants are zero,
then some of the eigenfunctions are linearly dependent.
Similarly to~\cite{kbk2015},
it can be shown that in these cases bounded eigenfunctions exist.
Therefore, the zeros of the determinants, i.e., the zeros of the scattering coefficients from Equation~\eqref{e:det-Phi},
are the discrete eigenvalues.

Therefore, we now discuss the zeros of the scattering coefficients.
For simplicity, we assume that these zeros are simple.
Due to the symmetries of the problem, it is unnecessary to discuss separately the zeros of all five scattering coefficients.
In fact, there are only two equivalence classes of zeros,
in particular:
\begin{itemize}
\item
Let $\zeta_0\in D_1$ be a zero of $a_{1,1}(z,\zeta)$;
then by symmetries~\eqref{e:asymmetry1} and~\eqref{e:asymmetry2} we obtain
\begin{equation}
b_{1,1}(z,\zeta_0^*) = 0
\Longleftrightarrow a_{1,1}(z,\zeta_0) = 0
\Longleftrightarrow a_{3,3}(z,\hat\zeta_0) = 0
\Longleftrightarrow b_{3,3}(z,\hat\zeta_0^*) = 0\,.
\end{equation}
\item
Let $\zeta_0\in \Complex^+$ be a zero of $b_{2,2}(z,\zeta)$;
then by symmetries~\eqref{e:asymmetry1} and~\eqref{e:asymmetry2} we obtain
\begin{equation}
a_{2,2}(z,\zeta_0^*) = 0
\Longleftrightarrow b_{2,2}(z,\zeta_0) = 0
\Longleftrightarrow b_{2,2}(z,\hat\zeta_0) = 0
\Longleftrightarrow a_{2,2}(z,\hat\zeta_0^*) = 0\,.
\end{equation}
\end{itemize}

\begin{remark}
(i)
It is sufficient to study the zeros of $a_{1,1}(z,\zeta)$ and $b_{2,2}(z,\zeta)$ in only one region, $\zeta \in D_1$.
(ii)
For each discrete eigenvalue in $D_1$ when $a_{1,1}(z,\zeta)$ or $b_{2,2}(z,\zeta)$ vanishes,
there are three corresponding zeros in the other three regions $D_j$ with $j = 2,3,4$, respectively.
Hence, each discrete eigenvalue in $D_1$ generates a quartet of discrete eigenvalues in the entire complex plane.
(iii)
A discrete eigenvalue in $D_1$ can be categorized by probing which quantity, $a_{1,1}(z,\zeta)$ or $b_{2,2}(z,\zeta)$, vanishes there.
\end{remark}

The above remark implies that all the discrete eigenvalues in the complex plane can be classified based on which scattering coefficient $a_{1,1}$ or $b_{2,2}$ vanishes in a single region, $D_1$.

\begin{definition}
\label{def:eigenvalue}
There are three kinds of eigenvalue quartets corresponding to a given eigenvalue in $D_1$:
\begin{itemize}
\item
\textbf{Eigenvalue of the first kind, $w_n$ for $n = 1,\dots,N_\I$}:
$a_{1,1}(z,w_n) = 0$ and $b_{2,2}(z,w_n)\ne0$.
\item
\textbf{Eigenvalue of the second kind, $z_n$ for $n = 1,\dots,N_\II$}:
$a_{1,1}(z,z_n)\ne0$ and $b_{2,2}(z,z_n) = 0$.
\item
\textbf{Eigenvalue of the third kind, $\zeta_n$ for $n = 1,\dots,N_\III$}:
$a_{1,1}(z,\zeta_n) = b_{2,2}(z,\zeta_n) = 0$.
\end{itemize}
\end{definition}
Then we have
\begin{proposition}
For every eigenvalue $\zeta_0\in D_1$ the following statements are equivalent:
\begin{enumerate}
\item
$\bchi_2(t,z,\zeta_0^*) = \@0$.
\item
$\bchi_3(t,z,\hat\zeta_0^*) = \@0$.
\item
There exists a constant $b_0\in\Complex$ such that
$\bphi_{-,2}(t,z,\zeta_0) = b_0\,\bphi_{+,1}(t,z,\zeta_0)$.
\item
These exists a constant $\widehat b_0\in\Complex$ such that $\bphi_{-,2}(t,z,\hat\zeta_0) = \widehat b_0\,\bphi_{+,3}(t,z,\hat\zeta_0)$.
\end{enumerate}
\end{proposition}

\begin{proposition}
For every eigenvalue $\zeta_0\in D_1$ the following statements are equivalent:
\begin{enumerate}
\item
$\bchi_1(t,z,\zeta_0) = \@0$.
\item
$\bchi_4(t,z,\hat\zeta_0) = \@0$.
\item
There exists a constant $\overline b_0\in\Complex$ such that
$\overline b_0\,\bphi_{+,2}(t,z,\zeta_0^*) = \bphi_{-,1}(t,z,\zeta_0^*)$.
\item
These exists a constant $\widecheck b_0\in\Complex$ such that
$\widecheck b_0\,
\bphi_{+,2}(t,z,\hat\zeta_0^*)=\bphi_{-,3}(t,z,\hat\zeta_0^*)$.
\end{enumerate}
\end{proposition}

The above two propositions together with Definition~\ref{def:eigenvalue}, yield the following result:
\begin{proposition}
Let $\zeta_0\in D_1$ be a discrete eigenvalue of the scattering problem, that is, $a_{1,1}(\zeta_0)b_{2,2}(\zeta_0)=0$.
Then the following statements are true.
\begin{enumerate}
\item
If $w_n$ is an eigenvalue of the first kind for $n = 1,\dots,N_\I$,
there exist complex constants $c_n$,
$\widehat c_n$, $\widecheck c_n$ and $\overline c_n$ such that
\begin{equation}
\label{e:norming1}
\begin{aligned}
\bchi_2(t,z,w_n^*) & = \overline c_n\, \bphi_{-,1}(t,z,w_n^*)\,,\qquad&
\bphi_{+,1}(t,z,w_n) & = c_n\,\bchi_1(t,z,w_n)\,,\\
\bphi_{+,3}(t,z,\hat w_n) & = \widehat c_n\,\bchi_4(t,z,\hat w_n)\,,\qquad&
\bchi_3(t,z,\hat w_n^*) & = \widecheck c_n\,\bphi_{-,3}(t,z,\hat\zeta_0^*)\,.
\end{aligned}
\end{equation}
\item
If $z_n$ is an eigenvalue of the second kind for $n = 1,\dots,N_\II$,
there exist complex constants $d_n$,
$\widehat d_n$, $\widecheck d_n$ and $\overline d_n$ such that
\begin{equation}
\label{e:norming2}
\begin{aligned}
\bphi_{+,2}(t,z,z_n^*) & = \overline d_n\,\bchi_2(t,z,z_n^*)\,,\qquad&
\bchi_1(t,z,z_n) & = d_n\, \bphi_{-,2}(t,z,z_n)\,,\\
\bchi_4(t,z,\hat z_n) & = \widehat d_n\, \bphi_{-,2}(t,z,\hat z_n)\,,\qquad&
\bphi_{+,2}(t,z,\hat z_n^*) & = \widecheck d_n\,\bchi_3(t,z,\hat z_n^*)\,.
\end{aligned}
\end{equation}
\item
If $\zeta_n$ is an eigenvalue of the third kind for $n = 1,\dots,N_\III$,
then there exist complex constants
$f_n$, $\widehat f_n$, $\widecheck f_n$ and $\overline f_n$ such that
\begin{equation}
\label{e:norming3}
\begin{aligned}
\bphi_{+,2}(t,z,\zeta_n^*) & = \overline f_n\,\bphi_{-,1}(t,z,\zeta_n^*)\,,\qquad&
\bphi_{+,1}(t,z,\zeta_n) & = f_n\,\bphi_{-,2}(t,z,\zeta_n)\,,\\
\bphi_{+,3}(t,z,\hat\zeta_n) & = \widehat f_n\,\bphi_{-,2}(t,z,\hat\zeta_n)\,,\qquad&
\bphi_{+,2}(t,z,\hat\zeta_n^*) & = \widecheck f_n\,\bphi_{-,3}(t,z,\hat\zeta_n^*)\,.
\end{aligned}
\end{equation}
\end{enumerate}
\end{proposition}
The above twelve constants $c_n$, $\widehat c_n$, $\widecheck c_n$ and $\overline c_n$ (and similar quantities for $d_n$ and $f_n$)
are the norming constants.

We further define the modified norming constants that will be used directly in the inverse problem.
\begin{definition}
\label{def:norming}
The modified norming constants are given by
\begin{equation}
\label{e:mnormingdef}
\begin{aligned}
C_n(z) & \coloneq \frac{c_n \e^{-2\ii \lambda(w_n)t}}{a'_{1,1}(z,w_n)}\,,&
\overline C_n(z) & \coloneq \frac{\overline c_n \e^{2\ii \lambda(w_n^*)t}}{b'_{1,1}(z,w_n^*)}\,,&
\widehat C_n(z) & \coloneq \frac{\widehat c_n \e^{-2\ii \lambda(w_n)t}}{a'_{3,3}(z,\widehat w_n)}\,,&
\widecheck C_n(z) & \coloneq \frac{\widecheck c_n \e^{2\ii \lambda (w_n^*)t}}{b'_{3,3}(z,\hat w_n^*)}\,,\\
D_n(z) & \coloneq \frac{d_n \e^{-\ii  \hat z_n t}}{b'_{2,2}(z,z_n)}\,,&
\overline D_n(z) & \coloneq \frac{\overline d_n \e^{\ii  \hat z_n^* t}}{a'_{2,2}(z,z_n^*)}\,,&
\widehat D_n(z) & \coloneq \frac{\widehat d_n \e^{-\ii  \hat z_n t}}{b'_{2,2}(z,\hat z_n)}\,,&
\widecheck D_n(z) & \coloneq \frac{\widecheck d_n \e^{\ii  \hat z_n^* t}}{a'_{2,2}(z,\hat z_n^*)}\,,\\
F_n(z) & \coloneq \frac{f_n \e^{-\ii  \zeta_n t}}{a'_{1,1}(z,\zeta_n)}\,,&
\overline F_n(z) & \coloneq \frac{\overline f_n \e^{\ii  \zeta_n^* t}}{a'_{2,2}(z,\zeta_n^*)}\,,&
\widehat F_n(z) & \coloneq \frac{\widehat f_n \e^{-\ii  \zeta_n t}}{a'_{3,3}(z,\hat \zeta_n)}\,,&
\widecheck F_n(z) & \coloneq \frac{\widecheck f_n \e^{\ii  \zeta_n^* t}}{a'_{2,2}(z,\hat \zeta_n^*)}\,,
\end{aligned}
\end{equation}
where the primes denote partial differentiation with respect to $\zeta$,
$n=1,\dots,N_\I$, for the four quantities in the first row,
$n=1,\dots, N_\II$, for the second row,
and $n=1,\dots, N_\III$, for the last row.
\end{definition}

\subsubsection{Symmetries of the norming constants.}

Let $\{w_n\}_{n=1}^{N_\I}$ be the set of all the eigenvalues of the first kind.
With the definitions of the modified eigenfunctions~\eqref{e:mudef} and~\eqref{e:mdef},
the definitions of the first kind of the norming constants~\eqref{e:norming1} reduce to
\bse
\label{e:normingdef2}
\begin{equation}
\begin{aligned}
\m_2(t,z,w_n^*)
 & = \overline c_n\, \bmu_{-,1}(t,z,w_n^*)\e^{2\ii \lambda(w_n^*) t}\,,\qquad&
\bmu_{+,1}(t,z,w_n)
 & = c_n\, \m_1(t,z,w_n)\e^{-2\ii \lambda(w_n) t}\,,\\
\bmu_{+,3}(t,z,\hat w_n)
 & = \widehat c_n\, \m_4(t,z,\hat w_n)\e^{-2\ii \lambda(w_n) t}\,,\qquad&
\m_3(t,z,\hat w_n^*)
 & = \widecheck c_n \, \bmu_{-,3}(t,z,\hat w_n^*)\e^{2\ii \lambda(w_n^*) t}\,.
\end{aligned}
\end{equation}
Let $\{z_n\}_{n=1}^{N_\II}$ be the set of all the eigenvalues of the second kind.
Using definitions~\eqref{e:mudef} and~\eqref{e:mdef}, the definitions of the second type of the norming constants~\eqref{e:norming2} reduce to
\begin{equation}
\begin{aligned}
\bmu_{+,2}(t,z,z_n^*)
 & = \overline d_n\, \m_2(t,z,z_n^*)\e^{\ii  \hat z_n^* t}\,,\qquad&
\m_1(t,z,z_n)
 & = d_n\, \bmu_{-,2}(t,z,z_n)\e^{-\ii \hat z_n t}\,,\\
\m_4(t,z,\hat z_n)
 & = \widehat d_n\, \bmu_{-,2}(t,z,\hat z_n)\e^{-\ii  \hat z_n t}\,,\qquad&
\bmu_{+,2}(t,z,\hat z_n^*)
 & = \widecheck d_n\, \m_3(t,z,\hat z_n^*)\e^{\ii  \hat z_n^* t}\,.
\end{aligned}
\end{equation}
Let $\{\zeta_n\}_{n=1}^{N_\III}$ be the set of all the eigenvalues of the third kind.
Using definitions~\eqref{e:mudef} and~\eqref{e:mdef}, the definitions of the third type of the norming constants~\eqref{e:norming3} reduce to
\begin{equation}
\begin{aligned}
\bmu_{+,2}(t,z,\zeta_n^*)
 & = \overline f_n\, \bmu_{-,1}(t,z,\zeta_n^*)\e^{\ii  \zeta_n^* t}\,,\qquad&
\bmu_{+,2}(t,z,\hat\zeta_n^*)
 & = \widecheck f_n\, \bmu_{-,3}(t,z,\hat\zeta_n^*)\e^{\ii \zeta_n^* t}\,,\\
\bmu_{+,1}(t,z,\zeta_n)
 & = f_n\, \bmu_{-,2}(t,z,\zeta_n)\e^{-\ii  \zeta_n t}\,,\qquad&
\bmu_{+,3}(t,z,\hat\zeta_n)
 & = \widehat f_n\, \bmu_{-,2}(t,z,\hat\zeta_n)\e^{-\ii  \zeta_n t}\,.
\end{aligned}
\end{equation}
\ese

We are now ready to discuss the symmetries of the norming constants.
By using the symmetries of the eigenfunctions~\eqref{e:eigenfunctionsymmetry1} and~\eqref{e:eigenfunctionsymmetry2},
we obtain the symmetries among all the norming constants:
\begin{equation}
\label{e:normingsymmetry}
\begin{aligned}
\widecheck c_n & = \overline c_n\,,\qquad&
c_n & =\widehat c_n = -\frac{\overline c_n^*}{b_{2,2}(z,w_n)}\,,\\
\overline d_n & = \frac{\ii z_n^*}{E_0}\widecheck d_n\,,\qquad&
\overline d_n^* & = -\frac{d_n}{\gamma(z_n)a_{1,1}(z,z_n)}\,,\qquad&
\overline d_n^* & = \frac{\ii E_0}{z_n}\frac{\widehat d_n}{\gamma(z_n)b_{1,1}(z,z_n^*)^*}\,,\\
\overline f_n & = \frac{\ii  \zeta_n^*}{E_0}\widecheck f_n\,,\qquad&
\overline f_n^* & = -\frac{b_{2,2}'(z,\zeta_n)}{a_{1,1}'(z,\zeta_n)}\frac{f_n}{\gamma(\zeta_n)}\,,\qquad&
\overline f_n^* & = \frac{\ii  E_0}{\zeta_n}\frac{b_{2,2}'(z,\zeta_n)}{a_{1,1}'(z,\zeta_n)}\frac{\widehat f_n}{\gamma(\zeta_n)}\,.
\end{aligned}
\end{equation}
Using the symmetries~\eqref{e:normingsymmetry} and Definition~\ref{def:norming},
we also obtain the following symmetries:
\begin{equation}
\label{e:mnormingsymmetry}
\begin{aligned}
C_n & =-\frac{\overline C_n^*}{b_{2,2}(z,w_n)} \,,\qquad&
\widehat C_n & =-\frac{E_0^2}{w_n^2}\frac{\overline C_n^*}{b_{2,2}(z,w_n)} \,,\qquad&
\widecheck C_n & = \frac{E_0^2}{(\bar w_n^*)^2}C_n\,,\\
D_n & = -\gamma(z_n)a_{1,1}(z,z_n) \overline D_n^*\,,\qquad&
\widehat D_n & = -\frac{\ii E_0}{z_n}\gamma(z_n)a_{1,1}(z,z_n)\overline D_n^*\,,\qquad&
\widecheck D_n & = -\frac{\ii E_0^3}{(z_n^*)^3}\overline D_n\,,\\
\widecheck F_n & = -\frac{\ii E_0^3}{(\zeta_n^*)^3}\overline F_n\,,\qquad&
\widehat F_n & = -\frac{\ii E_0}{\zeta_n}\gamma(\zeta_n) \overline F_n^*\,,\qquad&
F_n & = -\gamma(\zeta_n) \overline F_n^*\,.
\end{aligned}
\end{equation}

\subsection{Asymptotic behavior as \texorpdfstring{$\zeta\to\infty$}{ζ→∞} and \texorpdfstring{$\zeta\to0$}{ζ→0}}
\label{s:asympbehavior}

The asymptotic behavior of all eigenfunctions and scattering data will be used in both the propagation and the inverse problem stages of the IST.
Therefore, in this section we study the asymptotic behavior of the eigenfunctions and scattering data as $k\to\infty$.
In terms of the uniformization variable $\zeta = k+\lambda$,
this requires studying the behavior both as $\zeta\to\infty$ and as $\zeta\to0$.

Consider the following formal expansion for $\bmu_+(t,z,\zeta)$:
\begin{equation}
\bmu_+(t,z,\zeta) = \sum_{n=0}^\infty\bmu_n(t,z,\zeta)\,,
\end{equation}
where Equation~\eqref{e:mu-integral} implies
\begin{equation}
\label{e:muexpansion}
\begin{aligned}
\bmu_0(t,z,\zeta) & = \Y_+(z,\zeta)\,,\\
\bmu_{n+1}(t,z,\zeta) & = -\int_t^\infty \Y_+(z,\zeta)\e^{\ii (t-s)\bLambda(\zeta)}\Y_+(z,\zeta)^{-1}
\Delta \Q_+(s,z)\bmu_n(t,z,\zeta)\e^{-\ii (t-s)\bLambda(\zeta)}\d s\,.
\end{aligned}
\end{equation}
For all $n\ge 0$,
Equation~\eqref{e:muexpansion} provides an asymptotic expansion for the columns of $\bmu_+(t,z,\zeta)$ as
$\zeta\to\infty$ in the appropriate region of the complex $\zeta$-plane, with
\begin{equation}
\label{e:bmu-asymptotic-infty}
\begin{aligned}
[\bmu_{2n}]_{\d} &  = \O(\zeta^{-n})\,,\qquad&
[\bmu_{2n}]_{\o} & = \O(\zeta^{-(n+1)})\,,\\
[\bmu_{2n+1}]_{\d} & = \O(\zeta^{-(n+1)})\,,\qquad&
[\bmu_{2n+1}]_{\o} & = \O(\zeta^{-(n+1)})\,,
\end{aligned}
\end{equation}
where the subscripts $\d$ and $\o$ are defined in Equation~\eqref{e:dodef}.
Moreover, for all $n\ge 0$, Equation~\eqref{e:muexpansion} provides an asymptotic expansion for the columns of $\bmu_+(t,z,\zeta)$
as $\zeta\to0$ in the appropriate region of the complex $\zeta$-plane, with
\begin{gather}
\label{e:bmu-asymptotic-0}
[\bmu_{2n}]_{\d} = \O(\zeta^n)\,,\qquad
[\bmu_{2n}]_{\o} = \O(\zeta^{n-1})\,,\qquad
[\bmu_{2n+1}]_{\d} = \O(\zeta^{n})\,,\qquad
[\bmu_{2n+1}]_{\o} = \O(\zeta^{n})\,.
\end{gather}

As a result of Equation~\eqref{e:bmu-asymptotic-infty},
as $\zeta\to\infty$ in the appropriate regions of the $\zeta$-plane,
the leading order terms of the eigenfunctions $\bmu_{\pm,j}(t,z,\zeta)$ are given by
\begin{equation}
\label{e:muinfty}
\everymath{\displaystyle}
\begin{aligned}
\bmu_{\pm,1}(t,z,\zeta)
 & = \bpm 1 \\ 0 \epm
  - \frac{\ii}{\zeta} \bpm 0 \\ \E(t,z)^* \epm + \O(\zeta^{-2})\,,\\
\bmu_{\pm,2}(t,z,\zeta)
 & = \frac{1}{E_0}\bpm 0 \\ \E_\pm(z)^{\bot,*}\epm
  + \frac{\ii}{E_0\zeta}\bpm \E(t,z)^\top\E_\pm(z)^{\bot,*} \\ 0 \epm + \O(\zeta^{-2})\,,\\
\bmu_{\pm,3}(t,z,\zeta)
 & = \frac{1}{E_0}\bpm 0 \\ \E_\pm(z)^*\epm
  - \frac{\ii}{E_0\zeta}\bpm \E(t,z)^\top\E_\pm(z)^* \\ 0 \epm + \O(\zeta^{-2})\,.
\end{aligned}
\end{equation}
Similarly,
as the result of Equation~\eqref{e:bmu-asymptotic-0},
as $\zeta\to0$ in the appropriate regions of the $\zeta$-plane,
the leading order terms of $\bmu_{\pm,j}(t,z,\zeta)$ are
\begin{equation}
\begin{aligned}
\bmu_{\pm,1}(t,z,\zeta)
 & = -\frac{\ii}{\zeta}\bpm 0 \\ \E_\pm(z)^*\epm
  + \frac{1}{E_0^2} \bpm \E(t,z)^\top\E_\pm(z)^* \\ 0 \epm + \O(\zeta)\,,\\
\bmu_{\pm,2}(t,z,\zeta)
 & = \frac{1}{E_0} \bpm 0 \\ \E_\pm(z)^{\bot,*} \epm + \O(\zeta)\,,\\
\bmu_{\pm,3}(t,z,\zeta)
 & = -\frac{\ii}{\zeta}\bpm E_0 \\ 0 \epm + \frac{1}{E_0} \bpm 0 \\ \E(t,z)^* \epm + \O(\zeta)\,.
\end{aligned}
\end{equation}

Similar arguments yield that as $\zeta\to\infty$ in the appropriate regions of the $\zeta$-plane, the leading order terms of the auxiliary eigenfunctions $\m_j(t,z,\zeta)$ are
\begin{equation}
\begin{aligned}
\m_1(t,z,\zeta)
 & = \frac{1}{E_0}\bpm 0 \\ \E_+(z)^*\epm
  - \frac{\ii}{E_0\zeta}\bpm \E(t,z)^\top\E_+(z)^* \\ 0 \epm + \O(\zeta^{-2})\,,\\
\m_2(t,z,\zeta)
 & = \frac{1}{E_0} \bpm 0 \\ \E_-(z)^*\epm + \frac{1}{E_0\zeta}\bpm-\ii\E(t,z)^\top\E_-(z)^*\\ 0 \epm + \O(\zeta^{-2})\,,\\
\m_3(t,z,\zeta)
 & = \frac{1}{E_0^2}\bpm\E_+(z)^\top\E_-(z)^*\\ 0 \epm + \frac{\ii}{E_0^2\zeta}\bpm 0 \\ \E_+(z)^\top\E_-(z)^*\E(t,z)\epm + \O(\zeta^{-2})\,,\\
\m_4(t,z,\zeta)
 & = \frac{1}{E_0^2}\bpm\E_-(z)^\top\E_+(z)^* \\ 0 \epm + \frac{\ii}{E_0^2\zeta}\bpm 0 \\ \E_-(z)^\top\E_+(z)^*\E(t,z)\epm + \O(\zeta^{-2})\,.
\end{aligned}
\end{equation}
Similarly, as $\zeta\to0$ in the appropriate regions of the $\zeta$-plane,
the leading order terms of the auxiliary eigenfunctions $\m_j(t,z,\zeta)$ are
\begin{equation}
\begin{aligned}
\m_1(t,z,\zeta)
 & = \frac{1}{E_0\zeta}\bpm-\ii\E_-(z)^\top\E_+(z)^*\\0\epm + \O(1)\,,&\quad
\m_2(t,z,\zeta)
 & = \frac{1}{E_0\zeta}\bpm-\ii\E_+(z)^\top\E_-(z)^*\\ 0\epm + \O(1)\,,\\
\m_3(t,z,\zeta)
 & = -\frac{\ii}{\zeta}\bpm 0 \\ \E_-(z)^* \epm + \O(1)\,,&\quad
\m_4(t,z,\zeta)
 & = -\frac{\ii}{\zeta}\bpm 0 \\ \E_+(z)^* \epm + \O(1)\,.
\end{aligned}
\end{equation}

One can also compute the expansions of the scattering coefficients,
as $\zeta\to\infty$, in the appropriate regions of the $\zeta$-plane,
\begin{equation}
\begin{aligned}
a_{1,1}(z,\zeta) & = 1 + \O(\zeta^{-1})\,,\quad&
a_{2,2}(z,\zeta) & = \E_+^\top\E_-^*/E_0^2 + \O(\zeta^{-1})\,,\quad&
a_{3,3}(z,\zeta) & = \E_-^\top\E_+^*/E_0^2 + \O(\zeta^{-1})\,,\\
b_{1,1}(z,\zeta) & = 1 + \O(1/\zeta)\,,\quad&
b_{2,2}(z,\zeta) & = \E_-^\top\E_+^*/E_0^2 + \O(\zeta^{-1})\,,\quad&
b_{3,3}(z,\zeta) & = \E_+^\top\E_-^*/E_0^2 + \O(\zeta^{-1})\,.
\end{aligned}
\end{equation}
Similarly, as $\zeta\to0$ in the appropriate regions of the $\zeta$-plane, we find
\begin{equation}
\begin{aligned}
a_{1,1}(z,\zeta) & =\E_-^\top\E_+^*/E_0^2 + \O(\zeta)\,,\quad&
a_{2,2}(z,\zeta) & =\E_+^\top\E_-^*/E_0^2 + \O(\zeta)\,,\quad&
a_{3,3}(z,\zeta) & =1 + \O(\zeta)\,,\\
b_{1,1}(z,\zeta) & =\E_+^\top\E_-^*/E_0^2 + \O(\zeta)\,,\quad&
b_{2,2}(z,\zeta) & =\E_-^\top\E_+^*/E_0^2 + \O(\zeta)\,,\quad&
b_{3,3}(z,\zeta) & =1 + \O(\zeta)\,.
\end{aligned}
\end{equation}

\subsection{Asymptotics of the density matrix (part I)}
\label{s:boundary}

Recall that this work addresses the IBVP for CMBE with NZBG.
Similarly to the classic two-level MBE with ZBG or NZBG and similarly to the CMBE with ZBG,
it is not possible to solve the problem by using only two quantities, $q(t,0)$ and $q_-(z)$.
In fact,
additional quantities are needed,
which turn out to be the asymptotic expressions for the elements of the density matrix as $t\to-\infty$.
Therefore, we analyze these asymptotic expressions in the current section.
We will discuss $q_-(z)$ in the next section.

We first derive the asymptotics for the density matrix $\brho(t,z,\zeta)$ as $t\to\pm\infty$.
From the CMBE~\eqref{e:cmbe} and the scattering problem~\eqref{e:laxpair1}, it is easy to show that
\begin{equation}
\frac{\partial}{\partial t}\big[\bphi^{-1}(t,z,\zeta)\brho(t,z,\zeta)\bphi(t,z,\zeta)\big]=0\,,\qquad
\zeta\in\Real\,.
\end{equation}
Consequently, the quantity inside the square brackets is independent of $t$,
so naturally we define
\begin{equation}
\label{e:defrhopm}
\bvarrho_\pm(z,\zeta) \coloneq \bphi_\pm^{-1}(t,z,\zeta)\brho(t,z,\zeta)\bphi_\pm(t,z,\zeta)\,,\qquad
\zeta\in\Real\,,
\end{equation}
where $\bphi_\pm(t,z,\zeta)$ are the Jost eigenfunction.
Conversely, one can write
\begin{equation}
\label{e:rhorhopm}
\brho(t,z,\zeta) = \bphi_\pm(t,z,\zeta)\bvarrho_\pm(z,\zeta)\bphi_\pm^{-1}(t,z,\zeta)\,.
\end{equation}
Note that the density matrix $\brho(t,z,\zeta)$ itself has no limits as $t\to\pm\infty$,
but $\bvarrho_\pm(z,\zeta)$ obviously does.

\begin{remark}
In this work, the matrices $\bvarrho_\pm(z,\zeta)$ are regarded as the limiting values underlying the limiting expressions for the density matrix $\brho(t,z,\lambda)$ in the limits $t\to\pm\infty$.
Besides Equation~\eqref{e:defrhopm},
the following equivalent expression is more useful to compute $\bvarrho_\pm(z,\zeta)$ in practice:
\begin{equation}
\label{e:rhopminfty}
\bvarrho_\pm(z,\zeta)
 = \lim_{t\to\pm\infty}\e^{-\ii \bLambda(\zeta) t}\Y_\pm^{-1}(z,\zeta)\brho(t,z,\zeta)\Y_\pm(z,\zeta)\e^{\ii \bLambda(\zeta) t}\,.
\end{equation}
This expression is obtained by taking the limit of Equation~\eqref{e:defrhopm} as $t\to\pm\infty$ and using the asymptotics~\eqref{e:Jostsol}.
\end{remark}

\begin{remark}
Similarly to the quantities $\varrho_{\bg,j,j}$ in the background solution,
the entries $\{\varrho_{-,j,j}\}$ play an important role in the formulation of IST.
In fact, as we show later, these parameters appear in many spectral data,
and control the behavior of solutions.
Similarly to Equation~\eqref{e:background},
$\varrho_{-,j,j}$ are not the initial atomic level populations $\rho_{-,j,j}$ of the medium in each state,
but they determine these initial populations via a complicated combination, discussed in the rest of this section.
Consequently, even if one assumes that all the coefficients $\{\varrho_{-,j,j}\}$ are independent of the spectral parameter $k$,
the matrix $\rho_-(z,k)$ still depends on $k$ in a nontrivial way.
\end{remark}

In order to form a complete data set of initial-boundary values for CMBE~\eqref{e:cmbe}, in addition to the asymptotic initial condition $\E_\pm(z)$ for the optical pulse and the input condition $E(t,0)$ for the injected pulse,
one needs to impose a proper asymptotic initial condition for the medium,
i.e., the density matrix $\brho(t,z,\lambda)$.
One natural choice is the newly found asymptotics $\bvarrho_\pm(z,\zeta)$.
However, due to the first-order nature of the differential equations for $\brho(t,z,\lambda)$ in CMBE~\eqref{e:cmbe},
one has to determine:
(i) whether $\bvarrho_\pm(z,\zeta)$ are independent of each other and;
(ii) how to pick the proper (and mathematically correct) asymptotic condition.

To address these concerns,
we first show that $\bvarrho_-(z,\zeta)$ and $\bvarrho_+(z,\zeta)$ are dependent,
meaning that only one of them needs to be specified for the solution of the initial-boundary value problem.
Since $\bphi_+(t,z,\zeta)=\bphi_-(t,z,\zeta) \S(z,\zeta)$, by Equation~\eqref{e:rhorhopm}, one sees that
$\S(z,\zeta) \bvarrho_+(z,\zeta) = \bvarrho_-(z,\zeta) \S(z,\zeta)$,
or equivalently,
\begin{equation}
\label{e:rhopmsymmetry1}
\bvarrho_+(z,\zeta) = \S^{-1}(z,\zeta)\bvarrho_-(z,\zeta)\S(z,\zeta)\,,\qquad
\zeta\in\Real\,,
\end{equation}
which relates the asymptotic behaviors of the density matrix as $t\to\pm\infty$.
This relation allows one to obtain the asymptotic behavior $\bvarrho_+(z,\zeta)$ from $\bvarrho_-(z,\zeta)$ using the scattering matrix $S(z,\zeta)$,
which can be calculated from the optical pulse envelope $\Q(t,z)$.
Hence, only one of $\bvarrho_\pm(z,\zeta)$ can be chosen at will.
The causality of CMBE thus ensures that only $\bvarrho_-(z,\zeta)$ is needed to solve the initial-boundary value problem considered in this work.

\begin{definition}[Asymptotic conditions for the density matrix $\brho(t,z,\lambda)$]
\label{def:rho-BC}
In this work, the quantity $\bvarrho_-(z,\zeta)$ from Equation~\eqref{e:rhopminfty} is chosen as the asymptotic value associated with the behavior of the density matrix $\brho(t,z,\lambda)$ in the distant past,
i.e., in the limits $t\to-\infty$,
in order to form a complete data set of the initial-boundary problem for CMBE~\eqref{e:cmbe}.
\end{definition}

\subsubsection{Properties of $\bvarrho_\pm(z,\zeta)$}

We proceed to discuss the properties of the newly found asymptotic initial data $\bvarrho_-(z,\zeta)$.
It turns out that the properties of both $\bvarrho_\pm(z,\zeta)$ are almost identical,
so we instead discuss them together for completeness.
However, one should always remember that only $\bvarrho_-(z,\zeta)$ is required in the IST.

Looking at Equation~\eqref{e:rhopminfty},
we note that $\bvarrho_\pm(z,\zeta)$ is not necessarily Hermitian,
even though $\brho(t,z,\zeta)$ is.
Since the solution $\brho(t,z,\zeta)$ only depends on the real values of $k$,
we can consider the Schwarz reflection of the matrix $\brho(t,z,\zeta)$ and easily verify that $\brho(t,z,\zeta^*)^\dagger = \brho(t,z,\zeta)$ for all
$\zeta\in\Real$.
Similarly, it can also be verified that $\Y_\pm^\dagger(z,\zeta^*) = \bPi_0(\zeta) \Y_\pm^{-1}(z,\zeta)$ from Equation~\eqref{e:eigen} with
\begin{equation}
\bPi_0(\zeta) \coloneq \diag(\gamma(\zeta),1,\gamma(\zeta))\,.
\end{equation}
Correspondingly,
we obtain a symmetry for the asymptotic value $\bvarrho_\pm(z,\zeta)$ given by
\begin{equation}
\label{e:rhopmdagger}
\bvarrho_\pm^\dagger(z,\zeta) = \bPi_0(\zeta^*) \bvarrho_\pm(z,\zeta^*) \bPi_0^{-1}(\zeta^*)\,,\qquad
\zeta\in\Real\,.
\end{equation}

Besides the above symmetry~\eqref{e:rhopmdagger},
one more connection between $\bvarrho_\pm(z,\zeta)$ exists,
which is the relation between the data sets on the two $k$-sheets,
i.e., between $\zeta$ and $\hat\zeta$.
Applying the second symmetry of the eigenfunctions~\eqref{e:phisymmetry2},
we have
\begin{equation}
\label{e:rhopmsymmetry2}
\bvarrho_\pm(z,\zeta) = \bPi^{-1}(\zeta)\bvarrho_\pm(z,\hat\zeta)\bPi(\zeta)\,,\qquad
\zeta\in\Real\,,
\end{equation}
where the matrix $\bPi(\zeta)$ is defined in Equation~\eqref{e:Pi}.

\begin{remark}
\label{rmk:rho-BC-symmetry}
The symmetry~\eqref{e:rhopmsymmetry2} implies that one cannot choose the data $\bvarrho_-(z,\zeta)$ arbitrarily as a function of $\zeta$ on the entire real line.
Indeed, one only has the freedom to pick $\bvarrho_-(z,\zeta)$ with $k\in\Real$ on the first sheet,
or $\zeta\in(-\infty,-E_0]\cup[E_0,\infty)$,
after which the full expression of $\bvarrho_-(z,\zeta)$ is determined on the entire real $zeta$ line.
\end{remark}

Following the above discussion,
Definition~\ref{def:rho-BC} and Remark~\ref{rmk:rho-BC-symmetry},
we thus write the proper asymptotic data for the density matrix as $t\to-\infty$,
in the $k$-plane, in the element-wise form, as
\begin{equation}
\label{e:BC}
\begin{aligned}
\bvarrho_\pm(z,k_\I)
 & = (\varrho_{\pm,i,j}(z,k_\I))_{3\times3}\,,\\
\bvarrho_\pm(z,k_\II)
 & =
 \bpm
\varrho_{\pm,3,3}(z,k_\I) & -\ii E_0\varrho_{\pm,3,2}(z,k_\I)/\zeta & \varrho_{\pm,3,1}(z,k_\I)\\
\ii\zeta\varrho_{\pm,2,3}(z,k_\I)/E_0 & \varrho_{\pm,2,2}(z,k_\I) & \ii \zeta\varrho_{\pm,2,1}(z,k_\I)/E_0\\
\varrho_{\pm,1,3}(z,k_\I) & -\ii E_0\varrho_{\pm,1,2}(z,k_\I)/\zeta & \varrho_{\pm,1,1}(z,k_\I)
\epm\,,
\end{aligned}
\end{equation}
where the subscript $\I$ or $\II$ denotes evaluation on sheet I or II, respectively.
Also, Equation~\eqref{e:BC} needs to satisfy the symmetry~\eqref{e:rhopmdagger}.
One can thus also write an equivalent form of $\bvarrho_\pm(z,\zeta)$ in terms of $\zeta$ as follows:
\bse
\label{e:BC2}
\begin{equation}
\begin{aligned}
\varrho_{\pm,1,1}(z,\zeta) & = \varrho_{\pm,3,3}(z,\hat\zeta)\,,&
\varrho_{\pm,1,2}(z,\zeta) & = \frac{\ii\zeta}{E_0}\varrho_{\pm,3,2}(z,\hat\zeta)\,,&
\varrho_{\pm,1,3}(z,\zeta) & = \varrho_{\pm,3,1}(z,\hat\zeta)\,,\\
\varrho_{\pm,2,1}(z,\zeta) & = -\frac{\ii E_0}{\zeta}\varrho_{\pm,2,3}(z,\hat\zeta)\,,&
\varrho_{\pm,2,2}(z,\zeta) & = \varrho_{\pm,2,2}(z,\hat\zeta)\,,&
\varrho_{\pm,2,3}(z,\zeta) & = -\frac{\ii E_0}{\zeta}\varrho_{\pm,2,1}(z,\hat\zeta)\,,\\
\varrho_{\pm,3,1}(z,\zeta) & = \varrho_{\pm,1,3}(z,\hat\zeta)\,,&
\varrho_{\pm,3,2}(z,\zeta) & = \frac{\ii\zeta}{E_0}\varrho_{\pm,1,2}(z,\hat\zeta)\,,&
\varrho_{\pm,3,3}(z,\zeta) & = \varrho_{\pm,1,1}(z,\hat\zeta)\,,
\end{aligned}
\end{equation}
and
\begin{equation}
\begin{aligned}
\varrho_{\pm,1,1}(z,\zeta) & = \varrho_{\pm,1,1}(z,\zeta^*)^*\,,&
\varrho_{\pm,1,2}(z,\zeta) & = \frac{1}{\gamma(\zeta)}\varrho_{\pm,2,1}(z,\zeta^*)^*\,,&
\varrho_{\pm,1,3}(z,\zeta) & = \varrho_{\pm,3,1}(z,\zeta^*)^*\,,\\
\varrho_{\pm,2,1}(z,\zeta) & = \gamma(\zeta)\varrho_{\pm,1,2}(z,\zeta^*)^*\,,&
\varrho_{\pm,2,2}(z,\zeta) & = \varrho_{\pm,2,2}(z,\zeta^*)^*\,,&
\varrho_{\pm,2,3}(z,\zeta) & = \gamma(\zeta)\varrho_{\pm,3,2}(z,\zeta^*)^*\,,\\
\varrho_{\pm,3,1}(z,\zeta) & = \varrho_{\pm,1,3}(z,\zeta^*)^*\,,&
\varrho_{\pm,3,2}(z,\zeta) & = \frac{1}{\gamma(\zeta)}\varrho_{\pm,2,3}(z,\zeta^*)^*\,,&
\varrho_{\pm,3,3}(z,\zeta) & = \varrho_{\pm,3,3}^*(z,\zeta^*)\,.
\end{aligned}
\end{equation}
\ese

We have now fully determined the asymptotic behavior of the density matrix $\brho(t,z,\zeta)$ in the distant past.
We next explore the physical meaning of $\bvarrho_-(z,\zeta)$ with regard to $\brho(t,z,\zeta)$.
In particular, since we know that the diagonal entries of $\brho(t,z,\zeta)$ (namely $\brho_\dd$) denote the populations of atoms in the three states,
we would like to relate $\bvarrho_{-,\dd}$ to the initial populations of atoms in all three states.

The relation between $\bvarrho_\dd$ and $\brho_\dd$ are simple in the case of ZBG,
where one can easily show that $\bvarrho_{-,\dd}(z,k) = \lim_{t\to-\infty}\brho_\dd(t,z,k)$.
However, we will show that this relation is much more complicated for NZBG due to the presence of the nonzero background $\E_-(z)$.
Indeed this also happens in the classic two-level MBE with NZBGs~\cite{bgkl2019}.
This can be seen in the discussion below Equation~\eqref{e:background} in context of background solutions.

Next, we show that there are no simple relations between each $\rho_{i,j}$ and $\varrho_{-,i,j}$.
We write the diagonal entries of Equation~\eqref{e:rhopminfty} explicitly as
\begin{align*}
\varrho_{-,1,1}(z,\zeta) = & \frac{1}{\zeta  \lambda }\lim_{t\to-\infty}[\Re(E_1(t,z) E_2^*(t,z)\rho_{2,3}(t,z,\zeta)) -\zeta\Im(E_1(t,z)\rho_{2,1}(t,z,\zeta) + E_2(t,z) \rho_{3,1}(t,z,\zeta))]\\
    & +\frac{\zeta}{2 \lambda }\lim_{t\to-\infty}\rho_{1,1}(t,z,\zeta)
    +\frac{|E_{-1}(z)|^2}{2 \zeta  \lambda }\lim_{t\to-\infty}\rho_{2,2}(t,z,\zeta)
    +\frac{|E_{-2}(z)|^2}{2 \zeta  \lambda }\lim_{t\to-\infty}\rho_{3,3}(t,z,\zeta)\,,\\
\varrho_{-,2,2}(z,\zeta) = & -\frac{2}{E_0^2} \lim_{t\to-\infty}\Re(E_2(t,z) E_1^*(t,z) \rho_{3,2}(t,z,\zeta))
    + \frac{|E_{-,1}(z)|^2 }{E_0^2}\lim_{t\to-\infty}\rho_{3,3}(t,z,\zeta)\\
    & + \frac{|E_{-,2}(z)|^2 }{E_0^2}\lim_{t\to-\infty}\rho_{2,2}(t,z,\zeta)\,,\\
\varrho_{-,3,3}(z,\zeta) = & \frac{1}{\lambda }\lim_{t\to-\infty}[E_0^{-2}\Re(\zeta E_1(t,z) E_2^*(t,z) \rho_{2,3}(t,z,\zeta))
    -\Im(\rho_{1,2}(t,z,\zeta) E_1^*(t,z) + \rho_{1,3}(t,z,\zeta) E_2^*(t,z))]\\
    & + \frac{E_0^2 }{2 \zeta  \lambda }\lim_{t\to-\infty}\rho_{1,1}(t,z,\zeta)
    +\frac{|E_{-,1}(z)|^2 \zeta}{2 E_0^2 \lambda }\lim_{t\to-\infty}\rho_{2,2}(t,z,\zeta)
    +\frac{|E_{-,2}(z)|^2 \zeta  }{2 E_0^2 \lambda }\lim_{t\to-\infty}\rho_{3,3}(t,z,\zeta)\,.
\end{align*}
Inverting the system~\eqref{e:rhopminfty} yields $\brho(t,z,\zeta) = Y_\pm(z,\zeta) \e^{\ii \Lambda t}\bvarrho_\pm(z,\zeta) \e^{-\ii \Lambda t} Y_\pm^{-1}(z,\zeta) + o(1)$ as $t\to\pm\infty$.
Therefore, the diagonal entries $\brho_\dd$ can be expressed in the limit $t\to-\infty$ as
\begin{equation}
\label{e:rho-state}
\begin{aligned}
\rho_{1,1}(t,z,\zeta)
 = & D_{-,1}(z,\zeta) + \frac{\ii  E_0}{2 \lambda } (\varrho_{-,1,3}(z,\zeta) - \varrho_{-,3,1}(z,\zeta)) + o(1)\,,\\
\rho_{2,2}(t,z,\zeta)
 = & D_{-,2}(z,\zeta)
    +\frac{E_{-,1}^*(z) E_{-,2}^*(z)}{E_0^2}\varrho_{-,3,2}(z,\zeta)
    -\frac{\ii  E_{-,1}^*(z) E_{-,2}^*(z)}{E_0 \zeta }\varrho_{-,1,2}(z,\zeta)\\
 & +\frac{\ii  |E_{-,1}(z)|^2}{2 E_0 \lambda } (\varrho_{-,3,1}(z,\zeta)-\varrho_{-,1,3}(z,\zeta))
    +\frac{E_{-,1}(z) E_{-,2}(z) \zeta}{2 E_0^2 \lambda } \varrho_{-,2,3}(z,\zeta)\\
 & +\frac{\ii  E_{-,1}(z) E_{-,2}(z)}{2 E_0 \lambda } \varrho_{-,2,1}(z,\zeta) + o(1)\,,\\
\rho_{3,3}(t,z,\zeta)
 = & D_{-,3}(z,\zeta)
    -\frac{E_{-,1}^*(z) E_{-,2}^*(z)}{E_0^2}\varrho_{-,3,2}(z,\zeta)
    +\frac{\ii  E_{-,1}^*(z) E_{-,2}^*(z)}{E_0 \zeta }\varrho_{-,1,2}(z,\zeta)\\
 & +\frac{\ii  |E_{-,2}(z)|^2}{2 E_0 \lambda }(\varrho_{-,3,1}(z,\zeta) -\varrho_{-,1,3}(z,\zeta) )
    -\frac{E_{-,1}(z) E_{-,2}(z) \zeta}{2 E_0^2 \lambda } \varrho_{-,2,3}(z,\zeta)\\
 & -\frac{\ii  E_{-,1}(z) E_{-,2}(z)}{2 E_0 \lambda} \varrho_{-,2,1}(z,\zeta) + o(1)\,,
\end{aligned}
\end{equation}
with $D_{-,j}(z,\zeta)$ given by
\begin{equation}
\label{e:rho-D1}
\begin{aligned}
D_{-,1}(z,\zeta)
 & \coloneq \frac{\zeta}{2 \lambda }\varrho_{-,1,1}(z,\zeta)
    +\frac{E_0^2}{2 \zeta  \lambda }\varrho_{-,3,3}(z,\zeta)\,,\\
D_{-,2}(z,\zeta)
 & \coloneq \frac{|E_{-,1}(z)|^2}{2 \zeta  \lambda } \varrho_{-,1,1}(z,\zeta)
    +\frac{|E_{-,2}(z)|^2}{E_0^2} \varrho_{-,2,2}(z,\zeta)
    +\frac{|E_{-,1}(z)|^2 \zeta}{2 E_0^2 \lambda } \varrho_{-,3,3}(z,\zeta)\,,\\
D_{-,3}(z,\zeta)
 & \coloneq \frac{|E_{-,2}(z)|^2 }{2 \zeta  \lambda }\varrho_{-,1,1}(z,\zeta)
    +\frac{|E_{-,1}(z)|^2}{E_0^2} \varrho_{-,2,2}(z,\zeta)
    +\frac{|E_{-,2}(z)|^2 \zeta}{2 E_0^2 \lambda } \varrho_{-,3,3}(z,\zeta)\,.
\end{aligned}
\end{equation}

\begin{remark}[Importance of $D_{-,j}$]
\label{rmk:D-j}
We point out that if $\bvarrho(z,\zeta)$ is diagonal, then the quantities $D_{-,j}$ are precisely the asymptotic initial values of $\rho(t,z,\lambda)$ as $t\to-\infty$,
i.e., $\rho_{j,j}(t,z,\zeta) = D_{-,j}(z,\zeta)+o(1)$ as $t\to-\infty$.
This is a direct consequence of Equation~\eqref{e:rho-state}.
As we will see, a diagonal $\bvarrho(z,\zeta)$ is one of the requirements for pure soliton solutions of CMBE.
\end{remark}

For later use,
it is also useful to invert the system~\eqref{e:rho-D1}
\begin{equation}
\label{e:Drho-}
\begin{aligned}
\varrho_{-,1,1}(z,\zeta) & = \frac{\zeta }{2 k}D_{-,1}
    +\frac{\hat\zeta |E_{-,1}|^2}{2  k (| E_{-,1}| ^2-| E_{-,2}| ^2)}D_{-,2}
    -\frac{\hat\zeta |E_{-,2}| ^2}{2 k (| E_{-,1}| ^2-| E_{-,2}| ^2)}D_{-,3}\\
\varrho_{-,2,2}(z,\zeta) & = -\frac{| E_{-,2}| ^2}{| E_{-,1}| ^2-| E_{-,2}| ^2}D_{-,2}
    +\frac{| E_{-,1}| ^2}{| E_{-,1}| ^2-| E_{-,2}| ^2}D_{-,3}\,,\\
\varrho_{-,3,3}(z,\zeta) & = \frac{ \hat\zeta}{2  k}D_{-,1}
    +\frac{\zeta  | E_{-,1}| ^2}{2 k (| E_{-,1}| ^2-| E_{-,2}| ^2)}D_{-,2}
    -\frac{\zeta  | E_{-,2}| ^2}{2 k (| E_{-,1}| ^2-| E_{-,2}| ^2)}D_{-,3}\,.
\end{aligned}
\end{equation}
Note that if $|E_{-,1}| = |E_{-,2}|$, the system~\eqref{e:rho-D1} is singular and cannot be inverted.

\section{IST: Propagation}
\label{s:propagation}

In this section,
we compute the propagation of the asymptotic initial data,
norming constants and reflection coefficients.
Recall that in CMBE the propagation variable is the spatial variable $z$ instead of $t$, the usual evolution variable in equations such as KdV, nonlinear Schr\"{o}dinger equation and the Manakov systems.
Here we discuss the $z$ dependence of various quantities appearing in the IST.

\subsection{Propagaton of the limiting data in the distant past and future}
\label{s:qpmevolution}

First, let us discuss the dependence on $z$ of the asymptotic values of the background electric field envelope $q_\pm(z)$.
We assume that the limit as $t\to\pm\infty$ and the partial derivative $\partial/\partial_z$ commute.
Then the second equation in CMBE~\eqref{e:cmbe} yields
\begin{equation}
\label{e:dQpmdz}
\frac{\partial}{\partial z}\Q_\pm(z) = -\frac{1}{2}\lim_{t\to\pm\infty}\int_{-\infty}^\infty\big[\J,\brho(t,z,k)\big]g(k)\d k\,.
\end{equation}
Also, Equation~\eqref{e:rhopminfty} can be written equivalently as
\begin{equation}
\label{e:rhoinfty}
\brho(t,z,\zeta)
 = \Y_\pm(z) \e^{\ii \bLambda(\zeta) t} \bvarrho_\pm(z,\zeta)\e^{-\ii \bLambda(\zeta) t}\Y_\pm(z)^{-1} + o(1)\,,\qquad
t\to\pm\infty\,.
\end{equation}

We would like to combine Equations~\eqref{e:rhoinfty} and~\eqref{e:dQpmdz} to derive a propagation equation for the quantity $\Q_\pm(z)$.
We outline the calculations below without showing all the details. The exact calculations are cumbersome and similar to those in the classic two-level system~\cite{bgkl2019}.
\begin{enumerate}
\item
We decompose the matrix $\bvarrho_\pm(z,\zeta)$ into diagonal and off-diagonal parts:
$\bvarrho_\pm = \bvarrho_{\pm,\dd}+\bvarrho_{\pm,\o}+\bvarrho_{\pm,\rdo}$, with subscripts $\dd$, $\o$ and $\rdo$, defined in Equation~\eqref{e:dodef}.
\item
The integrand $[\J,\brho(t,z,k)]$ in Equation~\eqref{e:dQpmdz}, combined with the asymptotics~\eqref{e:rhoinfty} and the decomposition in the previous step, simplifies to become the equation
(with temporary omission of $t$, $z$ and $\zeta$ dependence in all quantities)
\begin{equation*}
\begin{aligned}
\big[\J,\brho\big]
 = \big[\J,\Y_\pm\bvarrho_{\pm,\dd}\Y_\pm^{-1}\big]
 + \big[\J,\Y_\pm\e^{\ii \bLambda t}\bvarrho_{\pm,\o}\e^{-\ii \bLambda t}\Y_\pm^{-1}\big]
 + \big[\J,\Y_\pm\e^{\ii \bLambda t}\bvarrho_{\pm,\rdo}\e^{-\ii \bLambda t}\Y_\pm^{-1}\big] + o(1)\,.
\end{aligned}
\end{equation*}
The last two terms before the error term in this equation contain oscillatory exponentials,
which cancel out as $t\to\pm\infty$ inside the integral in Equation~\eqref{e:dQpmdz} by the Riemann-Lebesgue lemma.
Hence, one only needs to consider the very first term inside the integral in the limit.
\item
Using the definition of $\Y_\pm(z,\zeta)$ from Equation~\eqref{e:eigen},
the effective part of $[\J,\brho(t,z,k)]$ from the previous step (i.e., the part that survives after integration and limits) becomes, as $t\to\pm\infty$,
\begin{equation*}
[\J,\brho(t,z,\zeta)]_\mathrm{eff}
 = \big[\J,\Y_\pm(z,\zeta)\bvarrho_{\pm,\dd}(z,\zeta)\Y_\pm^{-1}(z,\zeta)\big]
 = - \frac{\ii}{\lambda}\big(\varrho_{\pm,1,1}(z,\zeta) - \varrho_{\pm,3,3}(z,\zeta)\big)\J \Q_\pm(z)\,.
\end{equation*}
\item
The propagation equation~\eqref{e:dQpmdz} for the asymptotic conditions $\Q_\pm(z)$ becomes
\begin{equation}
\label{e:dQpmdz2}
\partial_z\Q_\pm(z)
 = \frac{\ii}{4}\int_{-\infty}^\infty \big(\varrho_{\pm,1,1}(z,k) - \varrho_{\pm,3,3}(z,k)\big)g(k)\frac{\d k}{\lambda(k)}\, \big[\J,\Q_\pm(z)\big]\,.
\end{equation}

\end{enumerate}

\begin{remark}
With known asymptotic data $\bvarrho_-(z,\zeta)$,
$\Q_-(z)$ can be obtained from Equation~\eqref{e:dQpmdz2}.
The quantity $\Q_+(z)$ turns out to be unnecessary in the formulation of IST, which is consistent with causality.
Of course, with $\bvarrho_-(z,\zeta)$ known, $\bvarrho_+(z,\zeta)$ and consequently $\Q_+(z)$ can be found using Equations~\eqref{e:rhopmsymmetry1} and~\eqref{e:dQpmdz2} as well.
\end{remark}

Before solving for $\Q_\pm(z)$ from Equation~\eqref{e:dQpmdz2},
we must discuss one potential issue:
on which $k$-sheet should we compute the integral in Equation~\eqref{e:dQpmdz2}?
Since $\lambda$ takes different values on each sheet,
it seems that the integral~\eqref{e:dQpmdz2} is ambiguous and the result depends on the sheet choice.
However, recall the definition~\eqref{e:BC} that $\bvarrho_\pm(z,\zeta)$ also take different values on each sheet as seen from Equation~\eqref{e:BC2}.
In particular, one needs to interchange the values for $\varrho_{\pm,1,1}(z,\zeta)$ and $\varrho_{\pm,3,3}(z,\zeta)$ between the sheets.
Because both quantities, $\lambda(k)$ and $\varrho_{\pm,1,1} - \varrho_{\pm,3,3}$, change signs between the two sheets,
the integral~\eqref{e:dQpmdz2} is uniquely defined and can be evaluated on either sheet.

It is convenient to define
\begin{equation}
\label{e:wpm}
w_\pm(z)
 \coloneq \frac{1}{4}\int\left(\varrho_{\pm,1,1}(z,k) - \varrho_{\pm,3,3}(z,k)\right)g(k)\frac{\d k}{\lambda(k)}\in\Real\,.
\end{equation}
This quantity simplifies the propagation equation~\eqref{e:dQpmdz2} to become
\begin{equation}
\partial_z\Q_\pm(z) = \ii w_\pm(z)\big[\J,\Q_\pm(z)\big]\,.
\end{equation}
By solving this ODE,
we obtain the $z$-dependence of the background optical pulse
\begin{equation}
\label{e:QpmQ0}
\Q_\pm(z) = \e^{\ii  W_\pm(z) \J}\Q_\pm(0)\e^{-\ii W_\pm(z) \J}\,,\qquad W_\pm(z) \coloneq \int_0^z w_\pm(z')\d z'\in\Real\,.
\end{equation}
It is evident that $W_\pm(z)$ is independent of $k$ from Equations~\eqref{e:wpm} and~\eqref{e:QpmQ0}.
Recall that $\|\E_\pm(z)\| = \|\E_\pm(0)\| = E_0$ for $z>0$ by our assumption before Lemma~\ref{thm:CMBE-U-invariance}.
We can therefore write Equation~\eqref{e:QpmQ0} explicitly, element-wise, as
\begin{equation}
\nonumber
E_{\pm,1}(z) = \e^{2\ii W_\pm(z)}E_{\pm,1}(0)\,,\qquad
E_{\pm,2}(z) = \e^{2\ii W_\pm(z)}E_{\pm,2}(0)\,.
\end{equation}
Equation~\eqref{e:bc-q} further simplifies the $E_{-,j}(z)$ to become
\begin{equation}
\label{e:E-j}
E_{-,1}(z) = E_0\e^{2\ii W_\pm(z)}\cos\alpha\,,\qquad
E_{-,2}(z) = E_0\e^{2\ii W_\pm(z)}\sin\alpha\,.
\end{equation}


\subsection{Asymptotics for the density matrix (part II)}

Before continuing discussing the propagation of the physical quantities,
we make a detour and simplify the asymptotics for the density matrix using the explicit expressions~\eqref{e:E-j}.

Combining equations~\eqref{e:rho-D1} and~\eqref{e:E-j}, we obtain
\begin{equation}
\label{e:rho-D}
\begin{aligned}
D_{-,1}(z,\zeta)
 & \coloneq \frac{\zeta}{2 \lambda }\varrho_{-,1,1}(z,\zeta)
    +\frac{E_0^2}{2 \zeta  \lambda }\varrho_{-,3,3}(z,\zeta)\,,\\
D_{-,2}(z,\zeta)
 & \coloneq \frac{E_0^2\cos^2\alpha}{2 \zeta  \lambda } \varrho_{-,1,1}(z,\zeta)
    + \varrho_{-,2,2}(z,\zeta)\sin^2\alpha
    + \frac{\zeta}{2 \lambda} \varrho_{-,3,3}(z,\zeta)\cos^2\alpha\,,\\
D_{-,3}(z,\zeta)
 & \coloneq \frac{E_0^2\sin^2\alpha}{2 \zeta  \lambda }\varrho_{-,1,1}(z,\zeta)
    + \varrho_{-,2,2}(z,\zeta)\cos^2\alpha
    + \frac{\zeta}{2 \lambda} \varrho_{-,3,3}(z,\zeta)\sin^2\alpha\,.
\end{aligned}
\end{equation}
As mentioned in Remark~\ref{rmk:D-j},
these three quantities play crucial roles in pure soliton solutions,
as they determine the initial populations of atoms in the three states of the medium.

\subsection{Propagation of reflection coefficients and norming constants}
\label{s:reflectionnorming}

The next step is to calculate the propagation of all scattering data,
including the reflection coefficients on the continuous spectrum and norming constants for discrete eigenvalues.
To do so, we need to use the $z$-dependence of the eigenfunctions $\bphi_\pm(t,z,\zeta)$ whose $t$-dependence is discussed in the direct problem.

Recall that the eigenfunctions are \textit{not} simultaneous solutions for the Lax pair, so they only solve the scattering problem.
As such, they contain nontrivial $z$-dependence, which implies that all the scattering data contain nontrivial $z$-dependence as well.
Hence, we need to address the $z$-dependence of the eigenfunctions before moving on to the scattering data. We do so by employing simultaneous solutions of the Lax pair.

\subsubsection{Simultaneous solutions of the Lax pair and auxiliary matrix.}

Because the asymptotic behavior of the eigenfunctions $\bphi_\pm(t,z,\zeta)$ as $t\to\pm\infty$ is fixed,
in general,
they are not solutions of Equation~\eqref{e:laxpair2}.
However, because both $\bphi_+(t,z,\zeta)$ and $\bphi_-(t,z,\zeta)$ are fundamental matrix solutions of the scattering problem,
every other solution $\bPhi(t,z,\zeta)$ of the scattering problem~\eqref{e:laxpair1} can be written as
\begin{equation}
\nonumber
\bPhi(t,z,\zeta)
 = \bphi_+(t,z,\zeta) \C_+(z,\zeta)
 = \bphi_-(t,z,\zeta) \C_-(z,\zeta)\,,\qquad
\zeta\in\Sigma\,,
\end{equation}
where $\C_\pm(z,\zeta)$ are $3\times3$ matrices independent of $t$.

We assume that $\bPhi(t,z,\zeta)$ is a simultaneous solution of both equations in the Lax pair~\eqref{e:laxpair},
so $\bPhi(t,z,\zeta)$ satisfies Equation~\eqref{e:laxpair2},
i.e., $\Phi_z(t,z,\zeta) = \V(t,z,\zeta)\Phi(t,z,\zeta)$ for $\zeta\in\Sigma$ as well,
where the matrix $\V(t,z,\zeta)$ is given in Equation~\eqref{e:laxpair2}.
Therefore, we find
\begin{equation}
\nonumber
\partial_z \C_\pm(z,\zeta) = \frac{\ii}{2} \R_\pm(z,\zeta) \C_\pm(z,\zeta)\,,\qquad
\zeta\in\Sigma\,,
\end{equation}
where the auxiliary matrices $\R_\pm(z,\zeta)$ are given by
\begin{equation}
\label{e:defRpm}
\frac{\ii}{2} \R_\pm(z,\zeta) \coloneq \bphi_\pm^{-1}(t,z,\zeta)\bigg[\V(t,z,\zeta)\bphi_\pm(t,z,\zeta) - \frac{\partial}{\partial z}\bphi_\pm(t,z,\zeta)\bigg]\,,\qquad
\zeta\in\Sigma.
\end{equation}
It will be shown later that, in order to determine the $z$ dependence of other scattering data,
we need to compute $\R_\pm(z,\zeta)$ explicitly.
Let us again assume that $z$-derivatives and the limits as $t\to\pm\infty$ commute,
and we know that $\R_\pm(\zeta,z)$ is independent of $t$, so we can write
\begin{equation}
\label{e:Rpm-1}
\R_\pm(z,\zeta)
 = -2\ii \lim_{t\to\pm\infty} \bigg[\bphi_\pm^{-1}(t,z,\zeta)\V(t,z,\zeta)\bphi_\pm(t,z,\zeta) - \bphi_\pm^{-1}(t,z,\zeta)\frac{\partial}{\partial z}\bphi_\pm(t,z,\zeta)\bigg]\,.
\end{equation}
Equation~\eqref{e:rhorhopm} implies that $\V(t,z,\zeta)$ becomes
\begin{equation}
\label{e:V-1}
\V(t,z,\zeta)
 = \frac{\ii \pi}{2} \H_k\big[\bphi_\pm(t,z,\zeta')\bvarrho_\pm(z,\zeta')\bphi_\pm^{-1}(t,z,\zeta')g(k')\big]\,,
\end{equation}
where we use the shorthand notation $\zeta' = \zeta(k')$,
and $\H_k(\cdot)$ is the Hilbert transform~\eqref{e:hilbert}.
After a lengthy calculation presented in Appendix~\ref{a:computeRpm},
we thus obtain the explicit expression for the matrix $\R_\pm(z,\zeta)$ as follows:
\begin{equation}
\label{e:Rpm}
\everymath{\displaystyle}
\begin{aligned}
\R_{\pm,\dd}(z,\zeta) = &
 \pi \lambda\,\H_{k}\big[\rho_\pm^{-}(k')g(k')/\lambda'\big]\diag\big( 1\,,\,	0\,,\, - 1 \big) + \{\pi\H_{k}\big[\rho_\pm^{+}(k')\,g(k')\big] + 2w_\pm(z)\} \diag(1 , 0, 1)\,,\\
 &\qquad\qquad\qquad + \{\pi\H_{k}[\varrho_{\pm,2,2}(k')g(k')]-4w_\pm(z)\}\diag(0,1,0)\,,\qquad  \zeta\in\Sigma\,,\\
\R_{\pm,\o}(z,\zeta) = & \begin{cases}
 \pm \ii\pi g(k)\diag(1,-1,-\nu)\bvarrho_{\pm,\o}(z,\zeta)\diag(1,1,\nu)\,, & \zeta\in\Real\,,\\
 \@0_{3\times3}\,, & \zeta\in \Sigma_\circ\,,
 \end{cases}\\
\R_{\pm,\rdo}(z,\zeta) = & \begin{cases}
 \pm \ii\pi g(k) \bvarrho_{\pm,\rdo}(z,\zeta)\diag(1,1,-1)\,, & \zeta\in\Real\,,\\
 \@0_{3\times3}\,, & \zeta\in \Sigma_\circ\,,
 \end{cases}
\end{aligned}
\end{equation}
where $\varrho_\pm^+$ and $\varrho_\pm^-$ are defined as,
\begin{equation}
\label{e:rhopmpm}
\begin{aligned}
\varrho_\pm^+ & \coloneq (\varrho_{\pm,1,1} + \varrho_{\pm,3,3})/2\,,\qquad&
\varrho_\pm^- & \coloneq (\varrho_{\pm,1,1} - \varrho_{\pm,3,3})/2\,,\\
\varrho_{\pm,1,1} & = \varrho_\pm^+ + \varrho_\pm^-\,,\qquad&
\varrho_{\pm,3,3} & = \varrho_\pm^+ - \varrho_{\pm}^-\,,
\end{aligned}
\end{equation}
and where we define
\begin{equation}
\label{e:nu-def}
\nu(\zeta) \coloneqq \begin{cases}
1 \,, & \zeta\in(-\infty,-E_0]\cup[E_0,\infty)\,,\\
-1\,, & \zeta\in(-E_0,E_0)\,.
\end{cases}
\end{equation}

Note that the matrix $\R_\pm(z,\zeta)$ is computed via a Hilbert transform,
meaning that in general it is discontinuous when $\zeta$ (or $k$) crosses the real axis.
Note also that the results~\eqref{e:Rpm} seemingly depend on the choice of the $k$-sheet on which one performs the calculation,
as can be seen from $\nu$ having different values on each sheet.
However, recall what happens to Equation~\eqref{e:wpm} and recall our definition~\eqref{e:BC} for the asymptotic data $\bvarrho_\pm(z,\zeta)$.
The potential ``ambiguity'' inside the above results~\eqref{e:Rpm} is resolved by interchanging the values for $\varrho_{\pm,i,j}(z,\zeta)$ between the $k$-sheets:
\begin{itemize}
\item
For the diagonal portion $\R_{\pm,\dd}(z,\zeta)$:
Recall that $\varrho_\pm^-$ is defined in Equation~\eqref{e:rhopmpm}. We find that $\varrho_\pm^-(z,\hat\zeta) = -\varrho_\pm^-(z,\zeta)$ from the symmetries~\eqref{e:BC2}. The switching of signs of $\varrho_\pm^-$ between the two $k$-sheets cancels the additional negative sign from $\lambda$ in the first integral of $\R_{\pm,\dd}(z,\zeta)$ in Equation~\eqref{e:Rpm}.
Also, $\varrho_\pm^+(z,\hat\zeta) = \varrho_\pm^+(z,\zeta)$ and $\varrho_{\pm,2,2}(z,\hat\zeta) = \varrho_{\pm,2,2}(z,\hat\zeta)$ by the symmetries~\eqref{e:BC2}. Thus, all integrals in $\R_{\pm,\dd}(z,\zeta)$ can be evaluated on either $k$-sheet with identical outcomes.
\item
For the $(1,3)$ and $(3,1)$ components:
The symmetries~\eqref{e:BC2} yield that $\varrho_{\pm,1,3}(z,\hat\zeta) = -\varrho_{\pm,1,3}(z,\zeta)$ and $\varrho_{\pm,3,1}(z,\hat\zeta) = -\varrho_{\pm,3,1}(z,\zeta)$. Thus, these two quantities evaluated on the two $k$-sheets take opposite signs,
which cancels the effect of $\nu$.
\item
For the rest of the off-diagonal entries: The symmetries~\eqref{e:BC2} yield $\varrho_{\pm,j,k}(z,\hat\zeta) = \varrho_{\pm,j,k}(z,\zeta)$, i.e., these entries evaluated on either $k$-sheet give identical results.
\end{itemize}
With the above discussion, we conclude that Equation~\eqref{e:Rpm} can be used on either $k$-sheet and yield identical results.
In other words,
the expressions~\eqref{e:Rpm} are uniquely defined on both sheets and yield consistent results.

Besides the values on the continuous spectrum $\Sigma$,
the matrices $\R_\pm(z,\zeta)$ at certain complex points are also needed when calculating the propagation of the norming constants.
Hence, we next discuss the analyticity of the entries of $\R_\pm(z,\zeta)$.

Note that the explicit expressions for the off-diagonal entries~\eqref{e:Rpm} are only valid on the continuous spectrum.
However, via the integral expression~\eqref{e:defRpm}, it is possible to extend some entries off the continuous spectrum.
More precisely,
in Appendix~\ref{a:computeRpm} the analytical extension is found by analyzing the exponential in each integral that produces each entry.
It can be shown that one integral can be analytically continued (and correspondingly becomes zero) if and only if the real part of the corresponding exponent is negative.
The explicit expressions for all the off-diagonal entries are presented in Appendix~\ref{a:computeRpm}.
Finally, we obtain
\begin{equation}
\label{e:Roff}
\begin{aligned}
R_{-,1,2}(z,\zeta) & = 0\,,\quad& \zeta\in\Complex^+\,,\qquad
R_{+,1,2}(z,\zeta) & = 0\,,\quad& \zeta\in\Complex^-\,,\\
R_{-,1,3}(z,\zeta) & = 0\,,\quad& \zeta\in D_3\cup D_4\,,\qquad
R_{+,1,3}(z,\zeta) & = 0\,,\quad& \zeta\in D_1\cup D_2\,,\\
R_{-,2,1}(z,\zeta) & = 0\,,\quad& \zeta\in\Complex^-\,,\qquad
R_{+,2,1}(z,\zeta) & = 0\,,\quad& \zeta\in\Complex^+\,,\\
R_{-,3,1}(z,\zeta) & = 0\,,\quad& \zeta\in D_1\cup D_2\,,\qquad
R_{+,3,1}(z,\zeta) & = 0\,,\quad& \zeta\in D_3\cup D_4\,,\\
R_{-,2,3}(z,\zeta) & = 0\,,\quad& \zeta\in\Complex^-\,,\qquad
R_{+,2,3}(z,\zeta) & = 0\,,\quad& \zeta\in\Complex^+\,,\\
R_{-,3,2}(z,\zeta) & = 0\,,\quad& \zeta\in\Complex^+\,,\qquad
R_{+,3,2}(z,\zeta) & = 0\,,\quad& \zeta\in\Complex^-\,.
\end{aligned}
\end{equation}
These conditions are all we need to compute the propagation of all the scattering data.

\subsubsection{Propagation equations for the reflection coefficients.}

By Equation~\eqref{e:Sdef} and the relation $\bphi_+(t,z,\zeta) \C_+(z,\zeta) = \bphi_-(t,z,\zeta) \C_-(z,\zeta)$,
we find that
\begin{equation}
\label{e:evoeqnS}
\S(z,\zeta) = \C_-(z,\zeta) \C_+^{-1}(z,\zeta)\,,\qquad\zeta\in\Sigma\,,
\end{equation}
together with
\begin{equation}
\label{e:dSdz}
\begin{aligned}
\frac{\partial}{\partial z}\S(z,\zeta)
 & = \frac{\ii}{2}\left[\R_-(z,\zeta)\S(z,\zeta) - \S(z,\zeta) \R_+(z,\zeta)\right]\,,\\
\frac{\partial}{\partial z}\S^{-1}(z,\zeta)
 & = \frac{\ii}{2}\left[\R_+(z,\zeta)\S^{-1}(z,\zeta) - \S^{-1}(z,\zeta)\R_-(z,\zeta)\right]\,.
\end{aligned}
\end{equation}
Note that the scattering coefficients $a_{i,j}(z,\zeta)$ from $\S(z,\zeta)$ and $b_{i,j}(z,\zeta)$ from $\S(z,\zeta)^{-1}$ do not appear in the inverse problem directly,
but in the form of a combination among them,
i.e., the reflection coefficients $r_j(z,\zeta)$ with $j = 1,2,3$, appear,
as defined in Equation~\eqref{e:reflection-def}.
Therefore, Equation~\eqref{e:dSdz} is not enough to complete the formulation of IST, and we need to compute the propagation equation for the reflection coefficients instead.
Note that by our definition~\eqref{e:BC},
the matrix $\R_\pm(z,\zeta)$ takes the same values on the two $k$-sheets,
but the reflection coefficients do not.
Thus, we also need to perform the calculations on $k$-sheets I and II separately.
\begin{itemize}[leftmargin=*]
\item
First, let us consider the reflection coefficients $r_1(z,\zeta)$ and $r_2(z,\zeta)$.
We define the matrix $\B^{(1)}(z,\zeta)$ such that
\begin{equation}
\label{e:B1}
\B^{(1)}(z,\zeta)
 \coloneq \S_\o(z,\zeta) \S_\d(z,\zeta)^{-1}
 = \bpm 0 & (a_{1,2}(z,\zeta), a_{1,3}(z,\zeta))\cdot\S_{[1,1]}(z,\zeta)^{-1} \\
  r_1(z,\zeta) & \@0_{1\times2} \\ r_2(z,\zeta) & \@0_{1\times2}\epm\,,
\end{equation}
where the subscripts $\d$, $\o$ and $[1,1]$ are defined in Equation~\eqref{e:dodef}.
Then,
by differentiating Equation~\eqref{e:B1} with respect to $z$,
one obtains an ODE for the matrix $\B^{(1)}(z,\zeta)$,
\begin{equation}
\nonumber
\frac{\partial \B^{(1)}}{\partial z}(z,\zeta)
 = \frac{\partial \S_{\o}}{\partial z}(z,\zeta)\S_{\o}(z,\zeta)^{-1} \B^{(1)}(z,\zeta) - \B^{(1)}(z,\zeta) \frac{\partial \S_{\d}}{\partial z}(z,\zeta) \S_{\d}(z,\zeta)^{-1}\,.
\end{equation}
By using Equation~\eqref{e:evoeqnS}, separating the block-diagonal and block-off-diagonal parts of the matrix,
applying the identity $\S(z,\zeta)\S^{-1}(z,\zeta) = \bbI$ and performing some tedious calculations,
we find the propagation equations for the two reflection coefficients when
$\zeta\in\Sigma\backslash(-E_0,E_0)$.
Let $\@r(z,\zeta) \coloneqq (r_1(z,\zeta),r_2(z,\zeta))^\top$; then the propagation equations for the two reflection coefficients can be written compactly as
\bse
\label{e:timeevoeqnB1}
\begin{gather}
\frac{\partial \@r}{\partial z}(z,\zeta)
 = \frac{\ii}{2} \A(z,\zeta)\@r(z,\zeta) + \@b(z,\zeta)\,,\qquad
\zeta\in\Sigma\backslash(-E_0,E_0)\,,
\end{gather}
where
\begin{equation}
\begin{aligned}
\A(z,\zeta)
 & \coloneq
  \bpm
  H_{2,2}(z,\zeta) - H_{1,1}(z,\zeta) & H_{2,3}(z,\zeta) \\
  0 & H_{3,3}(z,\zeta) - H_{1,1}(z,\zeta) \\
  \epm\,,\\
\@b(z,\zeta)
 & \coloneq -\nu_0 \pi g(k)
  \bpm \varrho_{-,2,1}(z,\zeta) \\ \varrho_{-,3,1}(z,\zeta) \epm\,,\\
\bH(z,\zeta) & = (H_{i,j}(z,\zeta))_{2\times2} \coloneq \R_-(z,\zeta) + \ii\nu_0 \pi g(k) \bvarrho_-(z,\zeta)\,,
\end{aligned}
\end{equation}
\ese
with
\begin{equation}
\label{e:nu0-def}
\nu_0 \coloneqq
\begin{cases}
1\,, & \zeta\in(-\infty,-E_0]\cup[E_0,\infty)\,,\\
0\,, & \zeta \in \Sigma_\circ\,.
\end{cases}
\end{equation}
Note that the two ODEs for $r_1(z,\zeta)$ and $r_2(z,\zeta)$ are coupled, which is similar to the case of ZBG.
Note also that,
unlike the two-level case with NZBG,
the propagation equation~\eqref{e:timeevoeqnB1} is only valid when $k\in\Sigma_k$ on sheet I or $\zeta\in\Sigma\backslash(-E_0,E_0)$,
not on the entire continuous spectrum $\Sigma$.
We will see that this is all we need to reconstruct our solution in the inverse problem.

In fact, one could also compute the propagation equations of $r_1(z,\zeta)$ and $r_2(z,\zeta)$ with $\zeta\in(-E_0,E_0)$.
However, in this case,
the equations are not as simple as the propagation equation~\eqref{e:timeevoeqnB1}.
This does not happen to both ZBG cases (two-level and coupled) and to the two-level system with NZBG.
As we will show later,
these equations are not necessary for the IST.
Thus, we skip this derivation and omit explicit formulas.
\item
Next, we compute the propagation equation for the last reflection coefficient $r_3(z,\zeta)$.
To do so, we define another matrix $\B^{(2)}(z,\zeta)$ as
\begin{equation}
\label{e:B2}
\B^{(2)}(z,\zeta)
 \coloneq \S_{\rdo}(z,\zeta)\S_{\dd}(z,\zeta)^{-1}
 = \bpm 0 & 0 & 0 \\ 0 & 0 & a_{2,3}(z,\zeta)/a_{3,3}(z,\zeta) \\ 0 & r_3(z,\zeta) & 0 \epm\,,\qquad
\zeta\in\Sigma\,,
\end{equation}
where the subscripts $\dd$ and $\rdo$ are defined in Equation~\eqref{e:dodef}.
By differentiating Equation~\eqref{e:B2} with respect to $z$,
we obtain the ODE
\begin{equation}
\nonumber
\frac{\partial \B^{(2)}}{\partial z}(z,\zeta)
 = \frac{\partial \S_{\rdo}}{\partial z}(z,\zeta)\S_{\rdo}(z,\zeta)^{-1} \B^{(2)}(z,\zeta) - \B^{(2)}(z,\zeta) \frac{\partial \S_{\dd}}{\partial z}(z,\zeta)\S_{\dd}(z,\zeta)^{-1}\,.
\end{equation}
By using Equation~\eqref{e:matrixidentity} to separate the diagonal and off-diagonal parts,
and using similar techniques as for computing Equation~\eqref{e:timeevoeqnB1},
the above equation reduces to
[with $\widetilde{\@ r}(z,\zeta) \coloneq (r_3(z,\zeta),r_3(z,\hat\zeta))^\top$]
\bse
\label{e:timeevoeqnB2}
\begin{equation}
\frac{\partial \widetilde{\@ r}}{\partial z}(z,\zeta)
 = \frac{\ii}{2}\widetilde \A(z,\zeta) \widetilde{\@ r}(z,\zeta) + \widetilde{\@ b}(z,\zeta)\,,\qquad
\zeta\in\Sigma\backslash(-E_0,E_0)\,,
\end{equation}
where
\begin{equation}
\begin{aligned}
\widetilde \A(z,\zeta)
 & \coloneq
  \bpm
  \widetilde H_{3,3}(z,\zeta) - \widetilde H_{2,2}(z,\zeta) &  0 \\
  \ii\zeta\, \widetilde H_{1,3}(z,\zeta)/E_0 & \widetilde H_{1,1}(z,\zeta) - \widetilde H_{2,2}(z,\zeta)
  \epm\,,\\
\widetilde{\@ b}(z,\zeta)
 & \coloneq \nu_0 \pi g(k) \bpm\varrho_{-,3,2}(z,\zeta) \\ \ii\zeta\varrho_{-,1,2}(z,\zeta)/E_0\epm\,,\\
\widetilde \bH(z,\zeta) & = (\widetilde H_{i,j}(z,\zeta))_{2\times2} \coloneq \R_-(z,\zeta) - \ii\nu_0 \pi g(k) \bvarrho_-(z,\zeta)\,,
\end{aligned}
\end{equation}
\ese
with $\nu_0$ defined in Equation~\eqref{e:nu0-def}.
Note that if $\zeta\in(-\infty,-E_0]\cup[E_0,\infty)$, then $\hat{\zeta}\in(-E_0,E_0)$,
so Equation~\eqref{e:timeevoeqnB2} gives the propagation of $r_3(z,\zeta)$ on the entire continuous spectrum $\Sigma$.
Moreover, the third reflection coefficient $r_3(z,\zeta)$ has different expressions for $k$ on sheet I or II,
and the propagation equations are coupled.
Also, if $\zeta\in \Sigma_\circ$, so is $\hat{\zeta}$.
In this case, by the choice of the asymptotic condition~\eqref{e:BC},
the system~\eqref{e:timeevoeqnB2} can be shown to be self-compatible.
\end{itemize}
Similarly to the classic two-level MBE with NZBG~\cite{bgkl2019},
the reflection coefficients take different forms depending on whether $\zeta$ is real or complex.
Note also that all the equations are determined solely by the asymptotic condition $\bvarrho_-(z,\zeta)$  via $\R_-(z,\zeta)$,
which is consistent with causality.

\subsubsection{Propagation equations for the norming constants.}

Here, we present the propagation equations for all three types of the norming constants.
The calculations can be found in the Appendix~\ref{a:norming}. These equations are
\begin{equation}
\label{e:normingtimeevolution}
\begin{aligned}
\frac{\partial \overline C_n}{\partial z}(z)
 & = \frac{\ii}{2}\left[R_{-,1,1}(w_n^*) - R_{-,3,3}(w_n^*)\right]\overline C_n(z)\,,\quad& n & = 1,2,\dots, N_\I\,,\\
\frac{\partial D_n}{\partial z}(z)
 & = \frac{\ii}{2}\left[R_{-,2,2}(z_n) - R_{-,3,3}(z_n)\right] D_n(z)\,,\quad& n & = 1,2,\dots,N_\II\,,\\
\frac{\partial \overline F_n}{\partial z}(z)
 & = \frac{\ii}{2}\left[R_{-,1,1}(\zeta_n^*) - R_{-,2,2}(\zeta_n^*)\right]\overline F_n(z)\,,\quad& n & = 1,2,\dots,N_\III\,,
\end{aligned}
\end{equation}
where $R_{-,j,j}(z,\zeta)$, with $j = 1,2,3$, are defined in Equation~\eqref{e:Rpm}.
Recall that $w_n^*$,
$z_n$ and $\zeta_n^*$ are the three types of discrete eigenvalues corresponding to the three modified norming constants $\overline C_n$,
$D_n$ and $\overline F_n$ as discussed in Section~\ref{s:eigen}.
Note that these equations only contain the asymptotic conditions as $t\to-\infty$, consistent with causality.

\begin{remark}
The three propagation equations~\eqref{e:normingtimeevolution} (one for each type of the norming constants) are sufficient for the formulation of the entire IST,
even though there are a total of twelve norming constants (c.f. Section~\ref{s:eigen}).
After one uses Equation~\eqref{e:normingtimeevolution} to compute the $z$-dependence of three norming constants,
one is able to use the symmetries~\eqref{e:mnormingsymmetry} to compute the other nine norming constants.
\end{remark}
\begin{remark}
The propagation Equation~\eqref{e:normingtimeevolution} for the norming constants can also be calculated via the trace formulae in Equation~\eqref{e:trace-formula} (derived in Section~\ref{s:trace-formula} below), similarly to what has been done for IST with asymmetric background in Ref.~\cite{abkp}. For the symmetric background case with the two-level MBEs studied in Ref.~\cite{bgkl2019}, it can be shown that the two approaches are equivalent and produce identical propagation equations. We believe the same to be true here, but we have not carried out the tedious calculation explicitly.
\end{remark}

\section{IST: Inverse problem}
\label{s:inverseproblem}

\subsection{Riemann-Hilbert problem and reconstruction formula}
\label{s:RHP}

We formulate the inverse problem by constructing a RHP,
after which the solution of CMBE~\eqref{e:cmbe} can be reconstructed from the solution of the RHP.
To achieve this, we first rearrange the eigenfunctions and divide them into four groups,
so that each group is meromorphic in a particular region of the complex $\zeta$-plane.
Then, we compute the jump matrix that characterizes the difference between two groups of eigenfunctions on the shared boundary.

\begin{definition}
Using the analycity of the eigenfunctions (cf. Equation~\eqref{e:muanalyticity} and Remark~\ref{rmk:adjoint-analyticity}) and the analyticity of the scattering data in Equation~\eqref{e:abanalyticity}, we define a sectionally meromorphic matrix $\M(t,z,\zeta)$,
\begin{equation}
\label{e:defM}
\displaystyle
\begin{aligned}
\M(t,z,\zeta)
 \coloneq \begin{cases}
 \displaystyle
 \bPhi^{(1)}\,\e^{-\ii \bLambda t}[\diag(a_{1,1},1,b_{2,2})]^{-1}
  = \left(\frac{\bmu_{+,1}}{a_{1,1}},\bmu_{-,2},\frac{\m_1}{b_{2,2}}\right)\,,&\quad
  \zeta\in D_1\,,\\
 \displaystyle
 \bPhi^{(2)}\,\e^{-\ii \bLambda t}[\diag(1,a_{2,2},b_{1,1})]^{-1}
  = \left(\bmu_{-,1},\frac{\bmu_{+,2}}{a_{2,2}},\frac{\m_2}{b_{1,1}}\right)\,,&\quad
  \zeta\in D_2\,,\\
 \displaystyle
 \bPhi^{(3)}\,\e^{-\ii \bLambda t}[\diag(b_{3,3},a_{2,2},1)]^{-1}
  = \left(\frac{\m_3}{b_{3,3}},\frac{\bmu_{+,2}}{a_{2,2}},\bmu_{-,3}\right)\,,&\quad
  \zeta\in D_3\,,\\
 \displaystyle
 \bPhi^{(4)}\,\e^{-\ii \bLambda t}[\diag(b_{2,2},1,a_{3,3})]^{-1}
  = \left(\frac{\m_4}{b_{2,2}},\bmu_{-,2},\frac{\bmu_{+,3}}{a_{3,3}}\right)\,,&\quad
  \zeta\in D_4\,,
 \end{cases}
\end{aligned}
\end{equation}
where the matrices $\bPhi^{(j)}(t,z,\zeta)$ are defined in Equation~\eqref{e:bPhi(j)-def}.
\end{definition}

\begin{definition}
\label{def:DSigmaL}
As shown in Figure~\ref{f:jump-contour},
we define two regions $D^\pm$ as $D^+ \coloneqq D_1\cup D_3$ and $D^- \coloneqq D_2\cup D_4$,
and decompose the continuous spectrum as $\Sigma = \cup_{j =1}^4\Sigma_j$,
with
$\Sigma_1 \coloneqq (-\infty, -E_0]\cup[E_0,+\infty)$,
$\Sigma_2 \coloneqq \{\lambda = E_0\e^{\ii \alpha}|\alpha\in[\pi,2\pi]\}$,
$\Sigma_3 \coloneqq \{\lambda = -z|z\in[-E_0,E_0]\}$,
and $\Sigma_4 \coloneqq \{\lambda = E_0\e^{\ii(\pi - \alpha)}|\alpha\in[0,\pi]\}$.
With the designated orientation,
$D^+$ is always to the contour $\Sigma$'s left and $D^-$ is always to its right.
Moreover, we define the matrix $\L(z,\zeta)$ as follows:
\begin{equation}
\label{e:jumpmatrix}
\begin{aligned}
\L(z,\zeta) & \coloneq
 \begin{cases}
 \e^{\ii \bLambda t}\bpm
 -\ii E_0 r_1\hat r_3/\zeta+(-r_2+r_1r_3)R_2 & -R_1/\gamma & -R_2+R_1R_3\\
 -r_1 & 0 & \gamma R_3\\
 -r_2+r_1r_3 & r_3 & -\gamma r_3R_3
 \epm\e^{-\ii \bLambda t}\,, &
 \zeta\in\Sigma_1\,,\\
 \e^{\ii \bLambda t}\bpm
 r_2\hat r_2 & 0 & \hat r_2 \\
 0 & 0 & 0 \\
 -r_2 & 0 & 0
 \epm\e^{-\ii \bLambda t}\,,&
 \zeta\in\Sigma_2\,,\\
 \e^{\ii \bLambda t}\bpm
 -\hat r_2\hat R_2 & r_3\hat r_2 + \ii E_0\hat r_3/\zeta & \hat r_2 \\
 \gamma R_3\hat R_2 - \ii E_0 \gamma \hat R_3/\zeta & -\gamma  r_3R_3-\gamma(\gamma-1)\hat r_3\hat R_3 & -\gamma R_3 \\
 \hat R_2 & -r_3 & 0
 \epm\e^{-\ii \bLambda t}\,, &
 \zeta\in\Sigma_3\,,\\
 \e^{\ii \bLambda t}\bpm
 R_2\hat R_2 & 0 & -R_2 \\
 0 & 0 & 0 \\
 \hat R_2 & 0 & 0
 \epm\e^{-\ii \bLambda t}\,,&
 \zeta\in\Sigma_4\,,
\end{cases}\\
\hat r_j & \coloneq r_j(z,\hat\zeta)\,,\qquad
R_j \coloneq r_j^*(z,\zeta^*)\,,\qquad
\hat R_j \coloneq r_j^*(z,\hat\zeta^*)\,,\qquad
\hat\zeta \coloneq -E_0^2/\zeta\,,
\end{aligned}
\end{equation}
where $r_j = r_j(z,\zeta)$ are reflection coefficients.

\end{definition}

\begin{figure}[t!]
\centering
\includegraphics[width = 0.375\textwidth]{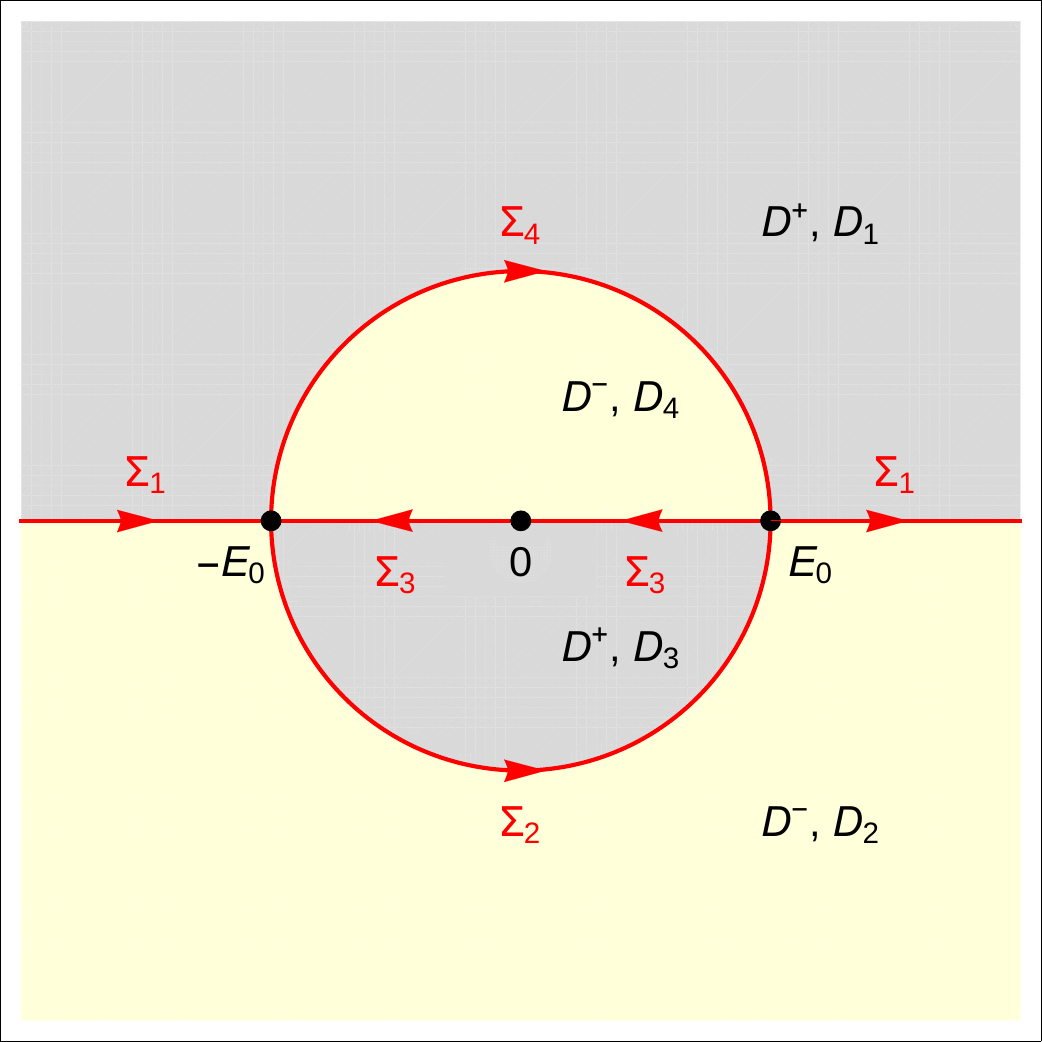}
\caption{
\textbf{Left:} The regions $D^+ = D_1\cup D_3$ and $D^- = D_2\cup D_4$,
and the decomposition of the oriented continuous spectrum $\Sigma = \cup_{j =1}^4 \Sigma_j$.
}
\label{f:jump-contour}
\end{figure}

\begin{rhp}
\label{rhp:inverse}
One seeks a $3\times3$ matrix $\M(t,z,\zeta)$ satisfying the conditions:
\begin{itemize}[align=left]

\item[\textbf{Asymptotics:}]
\begin{equation*}
\begin{aligned}
\M(t,z,\zeta) & = \M_\infty(t,z,\zeta) + \O(1/\zeta)\,,\qquad& \zeta\to\infty\,,\\
\M(t,z,\zeta) & = -\frac{\ii}{\zeta} \M_0(t,z,\zeta) + \O(1)\,,\qquad&
\zeta\to0\,,
\end{aligned}
\end{equation*}
where
\begin{equation}
\nonumber
\M_\infty(t,z,\zeta) \coloneq \bpm
1 & 0 & 0 \\ \@0 & (\E^\bot_-)^*/E_0 & \E_-^*/E_0
\epm\,,\qquad
\M_0(t,z,\zeta) \coloneq \bpm
1 & 0 & E_0 \\ \E_-^* & \@0 & \@0
\epm\,.
\end{equation}
Note that $\M_\infty + (\ii/\zeta) \M_0 = \Y_-$,
which is the asymptotic eigenfunction matrix~\eqref{e:eigen} of the scattering problem.

\item[\textbf{Jump Condition:}]
\begin{equation*}
\M^+(t,z,\zeta) =  \M^-(t,z,\zeta)[\bbI - \L(z,\zeta)]\,,\qquad
\zeta\in\Sigma\,.
\end{equation*}

\item[\textbf{Analyticity:}]
$\M(t,z,\zeta)$ is sectionally meromorphic on $\zeta\in\Complex$.

\item[\textbf{Residue Condition:}]
$\M(t,z,\zeta)$ satisfies the following conditions at the nonanalytic points:
\begin{align*}
\Res_{\zeta = w_n}\M(t,z,\zeta)
 & = \lim_{\zeta\to w_n}\M\bpm 0 & 0 & 0 \\ 0 & 0 & 0 \\ -\overline C_n^* & 0 & 0 \epm\,,\quad&
\Res_{\zeta = w_n^*}\M(t,z,\zeta)
 & = \lim_{\zeta\to w_n^*}\M\bpm 0 & 0 & \overline C_n \\ 0 & 0 & 0 \\ 0 & 0 & 0 \epm\,,\\
\Res_{\zeta = \hat w_n^*}\M(t,z,\zeta)
 & = \lim_{\zeta\to \hat w_n^*}\M \bpm 0 & 0 & 0 \\ 0 & 0 & 0 \\ \widecheck C_n & 0 & 0 \epm\,,\quad&
\Res_{\zeta = \hat w_n}\M(t,z,\zeta)
 & = \lim_{\zeta\to \widehat w_n} \M \bpm 0 & 0 & -\frac{E_0^2}{w_n^2}\overline C_n^* \\ 0 & 0 & 0 \\ 0 & 0 & 0 \epm\,,\\
\Res_{\zeta = z_n}\M(t,z,\zeta)
 & = \lim_{\zeta\to z_n}\M \bpm 0 & 0 & 0 \\ 0 & 0 & D_n \\ 0 & 0 & 0 \epm\,,\quad&
\Res_{\zeta = z_n^*}\M(t,z,\zeta)
 & = \lim_{\zeta\to z_n^*} \M \bpm 0 & 0 & 0 \\ 0 & 0 & 0 \\ 0 & -\frac{D_n}{\gamma(\zeta)} & 0 \epm\,,\\
\Res_{\zeta = \hat z_n^*}\M(t,z,\zeta)
 & = \lim_{\zeta\to \hat z_n^*}\M
    \bpm -\frac{\widehat D_n^*}{\gamma(\zeta)} & 0 & 0 \\ 0 & 0 & 0 \\ 0 & 0 & 0 \epm \,,\quad&
\Res_{\zeta = \hat z_n}\M(t,z,\zeta)
 & = \lim_{\zeta\to \hat z_n} \M \bpm 0 & 0 & 0 \\ \widehat D_n & 0 & 0 \\ 0 & 0 & 0 \epm\,,\\
\Res_{\zeta = \zeta_n}\M(t,z,\zeta)
 & = \lim_{\zeta\to\zeta_n}\M \bpm 0 & 0 & 0 \\ F_n & 0 & 0 \\ 0 & 0 & 0 \epm\,,\quad&
\Res_{\zeta = \zeta_n^*}\M(t,z,\zeta)
 & = \lim_{\zeta\to\zeta_n^*} \M\bpm 0 & \overline F_n & 0 \\ 0 & 0 & 0 \\ 0 & 0 & 0 \epm\,,\\
\Res_{\zeta = \hat\zeta_n^*}\M(t,z,\zeta)
 & = \lim_{\zeta\to\hat\zeta_n^*}\M \bpm 0 & 0 & 0 \\ 0 & 0 & 0 \\ 0 & \widecheck F_n & 0 \epm\,,\quad&
\Res_{\zeta = \hat \zeta_n}\M(t,z,\zeta)
 & = \lim_{\zeta\to\hat \zeta_n} \M \bpm 0 & 0 & 0 \\ 0 & 0 & \widehat F_n \\ 0 & 0 & 0 \epm\,,
\end{align*}
where all the modified norming constants are defined in Equation~\eqref{e:mnormingdef},
and the loci of a in appropriately chosen quartet of discrete eigenvalues are shown in Figure~\ref{f:analyticity} (right).

\end{itemize}

\end{rhp}

\begin{remark}
This RHP can be derived via direct calculations using the relation~\eqref{e:Sdef}
and the definition of the reflection coefficients~\eqref{e:reflection-def}.
For the \textbf{jump condition},
note the differences among all three reflection coefficients in the jump condition in RHP~\ref{rhp:inverse}:
i) $r_1(z,\zeta)$ only appears on $\zeta\in\Sigma_1$;
ii) $r_2(z,\zeta)$ appears on all portions of the continuous spectrum $\Sigma$,
but the portion of $\zeta\in\Sigma_3$ only requires $r_2(z,\hat\zeta)$ and $r_2^*(z,\hat{\zeta}^*)$.
Therefore, the explicit expression for $r_2(z,\zeta)$ is only needed on
$\zeta\in\Sigma\backslash\Sigma_3$;
iii) $r_3(z,\zeta)$ appears on $\Sigma_1\cup\Sigma_3$.
Combining all the parts,
we conclude that the propagation equations~\eqref{e:timeevoeqnB1} and~\eqref{e:timeevoeqnB2} are sufficient to close the inverse problem.
The \textbf{asymptotics} is exactly the leading order term in the asymptotic expansion of $M(t,z,\zeta)$ as $\zeta\to\infty$ and as $\zeta\to0$, obtained from Section~\ref{s:asympbehavior}.
The \textbf{residue condition} can be proved by using the definition~\eqref{e:mnormingdef} of the modified norming constants.
\end{remark}

One can then formally solve RHP~\ref{rhp:inverse} and obtain its solution.

\begin{theorem}
\label{thm:RHP-sol}
The solution of RHP~\ref{rhp:inverse} can be written as the implicit integral equation
\begin{multline}
\label{e:RHPsol}
\M(t,z,\zeta)
 = \Y_-(z,\zeta) + \sum_{n}\left(
  \frac{\Res_{\zeta = v_n}\M(\zeta)}{\zeta-v_n}
  + \frac{\Res_{\zeta = v_n^*}\M(\zeta)}{\zeta-v_n^*}
  + \frac{\Res_{\zeta = \hat v_n^*}\M(\zeta)}{\zeta-\hat v_n^*}
  + \frac{\Res_{\zeta = \hat v_n}\M(\zeta)}{\zeta-\hat v_n}\right)\\
 - \frac{1}{2\pi\ii }\int_\Sigma\frac{\M^-(t,z,\eta) \L(z,\eta)}{\eta-\zeta}\d\eta\,,\qquad
\zeta\in\Complex\,,
\end{multline}
where $v_n$ denotes all the discrete eigenvalues in the region $D_1$,
and the last term contains the Cauchy projector
\begin{equation}
\label{e:Cauchyprojector}
\P(f)(\zeta) \coloneq \frac{1}{2\pi\ii }
 \int_\Sigma \frac{f(\eta)}{\eta-\zeta}\d\eta\,,\qquad
\zeta\in\Complex\,.
\end{equation}
The evaluation of the Cauchy projector on the jump contour $\Sigma$ should be distinguished as from the left or right of the contour,
i.e., using the limit $\zeta\pm \ii\epsilon$ as $\epsilon\to0^+$, respectively.
\end{theorem}

Theorem~\ref{thm:RHP-sol} is proved in two steps.
First, we subtract the asymptotic behavior $\M_\infty \coloneqq \M(t,z,\infty)$,
$\M_0 \coloneqq \M(t,z,0)$ and the residues at the discrete spectrum from the jump condition,
in order to normalize it:
\begin{multline}
\label{e:RHP}
\M^+ -  \M_\infty + \frac{\ii}{\zeta} \M_0
 - \sum_{n}\left(
   \frac{\Res_{\zeta = v_n}\M(\zeta)}{\zeta-v_n}
   + \frac{\Res_{\zeta = v_n^*}\M(\zeta)}{\zeta-v_n^*}
   + \frac{\Res_{\zeta = \hat v_n^*}\M(\zeta)}{\zeta-\hat v_n^*}
   + \frac{\Res_{\zeta = \hat v_n}\M(\zeta)}{\zeta-\hat v_n}\right)\\
 = \M^- - \M_\infty + \frac{\ii}{\zeta} \M_0
 - \sum_{n}\left(
   \frac{\Res_{\zeta = v_n}\M(\zeta)}{\zeta-v_n}
   + \frac{\Res_{\zeta = v_n^*}\M(\zeta)}{\zeta-v_n^*}
   + \frac{\Res_{\zeta = \hat v_n^*}\M(\zeta)}{\zeta-\hat v_n^*}
   + \frac{\Res_{\zeta = \hat v_n}\M(\zeta)}{\zeta-\hat v_n}\right)\\
 - \M^- \L\,,
\end{multline}
where $\zeta\in\Sigma$ and $v_n$ denotes all the discrete eigenvalues.
Second, we apply the Cauchy projector~\eqref{e:Cauchyprojector} to Equation~\eqref{e:RHP}.

Once the solution to RHP~\ref{rhp:inverse} is obtained,
we are able to reconstruct the solutions to CMBE~\eqref{e:cmbe},
which is formalized in the following theorem:

\begin{theorem}
The solution to CMBE~\eqref{e:cmbe} can be reconstructed from RHP~\ref{rhp:inverse} and Theorem~\ref{thm:RHP-sol} as,
\begin{equation}
\label{e:cmbe-sol}
\begin{aligned}
E_j(t,z)
 & = E_{-,j} +
\sum_{n = 1}^{N_\I}\left[\ii\overline C_n M_{j+1,3}(w_n)^* - \frac{  w_n}{E_0}\widecheck C_n^*M_{(j+1),1}(w_n^*)^*\right]
 - \ii\sum_{n=1}^{N_\II}\widehat D_n^*M_{(j+1),2}(z_n)^*\\
 &\hspace*{10em} - \ii\sum_{n=1}^{N_\III}F_n^*M_{(j+1),2}(\zeta_n)^*
 + \frac{1}{2\pi}\int_{\Sigma}[\M^-\L]_{j+1,1}^*\d \eta\,,\qquad j = 1,2,\\
\brho(t,z,\zeta)
 & = \bmu_-(t,z,\zeta)\,\e^{\ii \bLambda t}\,\bvarrho_-(\zeta,z)\,\e^{-\ii \bLambda t}\,\bmu_-^{-1}(t,z,\zeta)\,,\qquad
\zeta\in\Real\,,
\end{aligned}
\end{equation}
where the eigenfunction matrix $\bmu_-(t,z,\zeta)$ is given by
\begin{equation}
\label{e:musol}
\begin{aligned}
\bmu_-(t,z,\zeta)
 & = \Y_-(z,\zeta) + \left(\hat\bmu_1,\hat\bmu_2,\hat\bmu_3\right)
 - \frac{1}{2\pi\ii }\int_\Sigma
    \bigg(\frac{[\M^- \L]_1}{\eta-\zeta + \ii0},
    \frac{[\M^- \L]_2}{\eta-\zeta - \ii0},
    \frac{[\M^- \L]_3}{\eta-\zeta + \ii0}\bigg)\d\eta\,,\qquad
\zeta\in\Real\,,\\
\hat\bmu_1
 & \coloneq
 \sum_{n=1}^{N_\I}\left(\frac{\ii  w_n^*}{E_0}\frac{\widecheck C_n \M_{1}(w_n^*)}{\zeta-\hat w_n^*} -\frac{\overline C_n^* \M_3(w_n)}{\zeta-w_n} \right)
 + \sum_{n=1}^{N_\II}\frac{\widehat D_n \M_{2}(z_n)}{\zeta-\hat z_n}
 + \sum_{n=1}^{N_\III}\frac{F_n \M_{2}(\zeta_n)}{\zeta-\zeta_n}\,,\\
\hat\bmu_2
 & \coloneq
 - \sum_{n=1}^{N_\II}\left(\frac{1}{\zeta-z_n^*}
 + \frac{(\gamma(z_n^*)-1)}{\zeta-\hat z_n^*}\right)\frac{D_n^*}{\gamma(z_n^*)}\M_3(z_n^*)
 + \sum_{n=1}^{N_\III}\left(\frac{\overline F_n}{\zeta-\zeta_n^*}
 + \frac{\ii\zeta_n^*}{E_0}\frac{\widecheck F_n}{\zeta-\hat\zeta_n^*}\right)\M_{1}(\zeta_n^*)\,,\\
\hat\bmu_3
 & \coloneq
 \sum_{n=1}^{N_\I}\left[\frac{\overline C_n \M_{1}(w_n^*)}{\zeta-w_n^*}
 - \frac{\ii  E_0}{w_n}\frac{\overline C_n^* \M_3(w_n)}{\zeta-\hat w_n}\right]
 + \sum_{n=1}^{N_\II}\frac{D_n\M_{2}(z_n)}{\zeta-z_n}
 + \sum_{n=1}^{N_\III}\frac{\widehat F_n\M_{2}(\zeta_n)}{\zeta-\hat\zeta_n}\,,
\end{aligned}
\end{equation}
where $[\M^-\L]_j$ is the $j$th column of the matrix $\M^-\L$.
Note that $\Y_-(z,\zeta)$ is continuous across the real line.
\end{theorem}

\begin{proof}
First, we use the integral equation~\eqref{e:RHPsol} to compute the solution $\M(t,z,\zeta)$ at the discrete eigenvalues.
This can be achieved by:
i) evaluating the first column of Equation~\eqref{e:RHPsol} at the discrete eigenvalues $\zeta = w_n^*$ and $\zeta = \zeta_n^*$;
ii) evaluating the second column of Equation~\eqref{e:RHPsol} at the discrete eigenvalues $\zeta = z_n$ and $\zeta =  \zeta_n$;
iii) evaluating the third column of Equation~\eqref{e:RHPsol} at the discrete eigenvalues $\zeta = w_n$ and $\zeta = z_n^*$.

The resulting equations at the discrete spectrum are given by,
\begin{equation}
\label{e:Mdiscrete}
\begin{aligned}
\M_{1}(\zeta)
 & = \Y_{-,1}(\zeta)
 + \sum_{j = 1}^{N_\I} \left(\frac{\overline C_j^*\M_3(w_j)}{\zeta - w_j}
 + \frac{\ii  w_j^*}{E_0} \frac{\widecheck C_j \M_{1}(w_j^*)}{\zeta-\hat w_j^*}\right)
+ \sum_{j = 1}^{N_\II}\frac{\widehat D_j \M_{2}(z_j)}{\zeta-\hat z_j}\\
 &\hspace{13em} + \sum_{j = 1}^{N_\III}\frac{F_j \M_{2}(\zeta_j)}{\zeta-\zeta_j}
-\frac{1}{2\pi\ii }\int_{\Sigma}\frac{[\M^- \L]_1}{\eta-\zeta}\d \eta\,, \qquad \zeta = w_{n}^*\,,\,\, \zeta_{n}^*\,,\\
\M_{2}(\zeta)
 & = \Y_{-,2}(\zeta)
  - \sum_{j=1}^{N_\II} \left(\frac{1}{\zeta-z_j^*}
  + \frac{\gamma(z_j^*)-1}{\zeta-\hat z_j^*}\right)\frac{D_j^*}{\gamma(z_j^*)}\M_3(z_j^*) \\
 &\hspace{3.7em} + \sum_{j=1}^{N_\III} \left(\frac{1}{\zeta-\zeta_j^*} + \frac{\gamma(\zeta_j^*) - 1}{\zeta-\hat\zeta_j^*}\right)\overline F_j\M_{1}(\zeta_j^*)
 -\frac{1}{2\pi\ii }\int_{\Sigma}\frac{\left[\M^- \L\right]_2}{\eta-\zeta}\d \eta\,,\qquad \zeta = z_{n}\,,\,\,\zeta_{n}\,,\\
\M_3(\zeta)
 & = \Y_{-,3}(\zeta)
 + \sum_{j=1}^{N_\I}\left(\frac{\overline C_j \M_{1}(w_j^*)}{\zeta-w_j^*}
 - \frac{\ii E_0}{w_j}\frac{\overline C_j^*\M_3(w_j)}{\zeta-\hat w_j}\right)
 + \sum_{j=1}^{N_\II}\frac{D_j \M_{2}(z_j)}{\zeta-z_j}\\
 &\hspace{13em} + \sum_{j=1}^{N_\III}\frac{\widehat F_j\M_{2}(\zeta_j)}{\zeta-\hat \zeta_j}
-\frac{1}{2\pi\ii }\int_{\Sigma}\frac{[\M^- \L]_3}{\eta-\zeta}\d \eta\,,\qquad
\zeta = w_n\,,\,\, z_n^*\,.
\end{aligned}
\end{equation}
The six equations~\eqref{e:Mdiscrete} form a \textit{closed} algebraic system for a given set of scattering data.
Therefore, the unknown vectors
$\{\M_1(w_n^*), \M_1(\zeta_n^*), \M_2(z_n), \M_2(\zeta_n), \M_3(w_n), \M_3(z_n^*)\}$ can be solved for,
after which Equation~\eqref{e:RHPsol} yields an integral equation for the solution $\M(t,z,\zeta)$.

Now we are ready to formulate the solution of CMBE~\eqref{e:cmbe}.
First, we reconstruct the electric field envelope.
By comparing the asymptotic behavior of the first column of Equation~\eqref{e:RHPsol} as
$\zeta\to\infty$ to the asymptotic behavior of
$\bmu_{-,1}(t,z,\zeta)$ in Equation~\eqref{e:muinfty},
we obtain the reconstruction formula for the envelopes $E_j(t,z)$, with
$j = 1,2$, as in Equation~\eqref{e:cmbe-sol}.

Next, we reconstruct the density matrix.
In terms of the original spectral variable $k$,
we note that $\brho(t,z,k)$ should be constructed on the real line,
which translates to $\brho(t,z,\zeta)$ with $\zeta\in\Real$ as well.
Furthermore,
due to the double $k$-sheet,
one can reconstruct $\brho(t,z,\zeta)$ on
$\zeta\in(-\infty,-E_0]\cup[E_0,\infty)$, corresponding to $k$ on the first sheet,
or $\zeta\in[-E_0,E_0]$, corresponding to $k$ on the second sheet,
or mix the two cases together.
All the results are equivalent.
After all, the solution only concerns $k\in\Real$,
not $\lambda$ or $\zeta$.

By definition~\eqref{e:defM}, we know that
\begin{equation}
\nonumber
\bmu_-(t,z,\zeta) = \left([\M(t,z,\zeta - \ii0)]_1,[\M(t,z,\zeta + \ii0)]_2,[ \M(t,z,\zeta - \ii0)]_3\right)\,,\qquad \zeta\in\Real\,,
\end{equation}
Hence, by the solution formula~\eqref{e:RHPsol} of the RHP,
we obtain the reconstruction formula for $\rho$ in Equation~\eqref{e:cmbe-sol}.

\end{proof}

\subsection{Trace formul\ae}
\label{s:trace-formula}

Similarly to the two-level MBEs~\cite{bgkl2019} and the focusing NLS equation~\cite{kbk2015,pab2006},
it is also possible to obtain the so-called ``trace" formula for CMBE~\eqref{e:cmbe},
which express the scattering data in terms of the reflection coefficients and discrete eigenvalues.
We present this ``trace'' formul\ae~below and relegate the calculations to Appendix~\ref{a:tracetheta}. The formul\ae~are
\begin{equation}
\label{e:trace-formula}
\begin{aligned}
b_{1,1}(z,\zeta) & = \e^{-\frac{1}{2\pi\ii }\int_\Sigma\frac{J}{\eta-\zeta}\d\eta}
\prod_{n=1}^{N_\I}\frac{\zeta-w_n^*}{\zeta-w_n}\frac{\zeta-\hat w_n}{\zeta-\hat w_n^*}
\prod_{n=1}^{N_\II}\frac{\zeta-\hat z_n^*}{\zeta-\hat z_n}
\prod_{n=1}^{N_\III}\frac{\zeta-\zeta_n^*}{\zeta-\zeta_n}\,,\\
a_{2,2}(z,\zeta) & = \e^{-\ii \Delta\theta}\e^{\frac{1}{2\pi\ii }\int_\Real\frac{J_0}{\eta-\zeta}\d\eta}
\prod_{n=1}^{N_\II}\frac{\zeta-z_n^*}{\zeta-z_n}\frac{\zeta-\hat z_n^*}{\zeta-\hat z_n}
\prod_{n=1}^{N_\III}\frac{\zeta-\zeta_n^*}{\zeta-\zeta_n}\frac{\zeta-\hat\zeta_n^*}{\zeta-\hat\zeta_n}\,,\\
b_{2,2}(z,\zeta) & = \e^{\ii \Delta\theta}\e^{-\frac{1}{2\pi\ii }\int_\Real\frac{J_0}{\eta-\zeta}\d \eta}
\prod_{n=1}^{N_\II}\frac{\zeta-z_n}{\zeta-z_n^*}\frac{\zeta-\hat z_n}{\zeta-\hat z_n^*}
\prod_{n=1}^{N_\III}\frac{\zeta-\zeta_n}{\zeta-\zeta_n^*}\frac{\zeta-\hat\zeta_n}{\zeta-\hat\zeta_n^*}\,,
\end{aligned}
\end{equation}
where the integrands are defined by
\begin{equation}
\label{e:J0J-def}
\begin{aligned}
J_0 & \coloneq \log\left[1+\gamma(\zeta)(\gamma(\zeta)-1)r_3(\hat \zeta)r_3^*(\hat\zeta^*)
+\gamma(\zeta) r_3(\zeta)r_3^*(\zeta^*)\right]\,,\\
J & \coloneqq \begin{cases}
J_1\,,\qquad & \zeta\in\Sigma_1\,,\\
J_2\,,\qquad & \zeta\in\Sigma_2\,,\\
J_3\,,\qquad & \zeta\in\Sigma_3\,,\\
J_4\,,\qquad & \zeta\in\Sigma_4\,,
\end{cases}\\
J_1
 & \coloneqq -\log\left(1 + \frac{1}{\gamma(\zeta)}r_1(\zeta)r_1^*(\zeta^*) + r_2(\zeta)r_2^*(\zeta^*)\right)\,,\\
J_2
 & \coloneqq \frac{1}{2\pi\ii }\int_\Sigma\frac{J_0}{\eta-\zeta}\d\eta - \log\left(1-r_2^*(\zeta^*)r_2^*(\hat\zeta^*)\right)\,,\\
J_3
 & \coloneqq -\log\left(r_2(\hat\zeta)r_2^*(\hat\zeta^*) + \frac{1}{\gamma(\zeta)(\gamma(\zeta)-1)}
r_1(\hat\zeta)r_1^*(\hat\zeta^*)+1\right)\,,\\
J_4
 & \coloneqq \frac{1}{2\pi\ii }\int_\Sigma\frac{J_0}{\eta-\zeta}\d\eta - \log\left(1-r_2(\hat\zeta)r_2(\zeta)\right)\,.
\end{aligned}
\end{equation}

Also, in Appendix~\ref{a:tracetheta}, we show how to derive the ``theta'' formula,
which relates the asymptotic phase difference of the solution as $t\to\pm\infty$,
\begin{equation}
\nonumber
\Delta\theta
 \coloneq \theta_+-\theta_-
 = \frac{1}{2\pi}\int_\Sigma\frac{J}{\eta}\d\eta
 - 4\sum_{n=1}^{N_\I}\arg w_n
 + 2\sum_{n=1}^{N_\II}\arg z_n
 - 2\sum_{n=1}^{N_\III}\arg \zeta_n\,,
\end{equation}
where $\theta_\pm$ are defined in the asymptotic conditions
$\E_\pm(z)= \E_0\e^{\ii \theta_\pm}$ with $\|\E_0\| = E_0$.

\subsection{Reflectionless potentials}
\label{s:reflectionless}

In this section,
we discuss reflectionless potentials corresponding to pure soliton solutions.
\begin{remark}
Similarly to other cases of MBEs,
reflectionless potentials require not only $b_j(0,\zeta) \equiv 0$,
i.e., reflectionless input data $\E(t,0)$,
but also special asymptotic data,
i.e., $\bvarrho_-(z,\zeta)$ must be diagonal,
[cf. the ODEs~\eqref{e:timeevoeqnB1} and~\eqref{e:timeevoeqnB2}].
\end{remark}
Following this discussion,
we take the asymptotic conditions from Equation~\eqref{e:BC} as
\begin{equation}
\label{e:reflectionlessBC}
\bvarrho_-(z,\zeta) = \diag(\varrho_{-,1,1},\varrho_{-,2,2},\varrho_{-,3,3})\,,\qquad
\zeta\in \Sigma\,,
\end{equation}
with $\bvarrho_{-,j,j}(z,\zeta)\ge0$ for $j = 1,2,3$.
Importantly,
the above matrix satisfies the mandatory symmetry~\eqref{e:BC2}.

We assume that there are total $N_\I\ge0$ discrete eigenvalues of the first kind $w_n$,
$N_\II\ge0$ discrete eigenvalues of the second kind $z_n$,
and $N_\III\ge0$ discrete eigenvalues of the third kind $\zeta_n$.

Because all three reflection coefficients vanish,
the jump matrix~\eqref{e:jumpmatrix} is identically zero.
Therefore,
the RHP~\ref{rhp:inverse} can be solved explicitly,
via the closed algebraic system~\eqref{e:Mdiscrete} with all the integrals vanishing.
We define the following vector notation for $j = 1,2,3$,
in order to further simplify the system~\eqref{e:Mdiscrete}:
\begin{equation}
\begin{aligned}
\ul{M_{j,1}(w_n^*)}
 & \coloneq \bpm M_{j,1}(w_1^*),\dots, M_{j,1}(w_{N_\I}^*)\epm\,,\qquad
\ul{M_{j,1}(\zeta_n^*)}
 \coloneq \bpm M_{j,1}(\zeta_1^*),\dots, M_{j,1}(\zeta_{N_\III}^*)\epm\,,\\[1ex]
\ul{M_{j,2}(z_n)}
 & \coloneq \bpm M_{j,2}(z_1),\dots, M_{j,2}(z_{N_\II})\epm\,,\qquad
\ul{M_{j,2}(\zeta_n)}
 \coloneq \bpm M_{j,2}(\zeta_1),\dots, M_{j,2}(\zeta_{N_\III})\epm\,,\\[1ex]
\ul{M_{j,3}(z_n^*)}
 & \coloneq \bpm M_{j,3}(z_1^*),\dots,M_{j,3}(z_{N_\II}^*)\epm\,,\qquad
\ul{M_{j,3}(w_n)}
 \coloneq \bpm M_{j,3}(w_1),\dots,M_{j,3}(w_{N_\I})\epm\,,\\[1ex]
\ul X_j
 & \coloneq  \bpm \ul{M_{j,1}(w_n^*)},\ul{M_{j,1}(\zeta_n^*)}, \ul{M_{j,2}(z_n)},\ul{M_{j,2}(\zeta_n)},
\ul{M_{j,3}(z_n^*)},\ul{M_{j,3}(w_n)}\epm^\top\,.
\end{aligned}
\end{equation}
Then,
the system~\eqref{e:Mdiscrete} can be rewritten in the following compact form,
\begin{equation}
\nonumber
(\bbI - \A)\ul X_j = \ul B_j\,,
\end{equation}
where $\bbI$ is the $3\times3$ identity matrix,
and the block matrices $\A$ and $\ul B_j$ are defined as
\begin{equation}
\begin{aligned}
\A & \coloneq \big(\A_{i,j}\big)_{6\times6}\,,\\
\ul B_j & \coloneq \big(
\ul{Y_{-,j,1}(w_n^*)},\ul{Y_{-,j,1}(\zeta_n^*)},\ul{Y_{-,j,2}(z_n)},\ul{Y_{-,j,2}(\zeta_n)},
\ul{Y_{-,j,3}(z_n^*)},\ul{Y_{-,j,3}(w_n)}
\big)^\top\,,
\end{aligned}
\end{equation}
with $Y_{-,i,j}$ being the $(i,j)$-component of the matrix $\Y_-$ in Equation~\eqref{e:eigen} and the underline denoting
\begin{equation}
\begin{aligned}
\ul{Y_{-,j,1}(w_n^*)}
 & \coloneq \big(Y_{-,j,1}(w_1^*),\dots,Y_{-,j,1}(w_{N_\I}^*)\big)\,,&\qquad
\ul{Y_{-,j,1}(\zeta_n^*)}
 & \coloneq \big(Y_{-,j,1}(\zeta_1^*),\dots,Y_{-,j,1}(\zeta_{N_\III}^*)\big)\,,\\[1ex]
\ul{Y_{-,j,2}(z_n)}
 & \coloneq \big(Y_{-,j,2}(z_1),\dots,Y_{-,j,2}(z_{N_\II})\big)\,,&\qquad
\ul{Y_{-,j,2}(\zeta_n)}
 & \coloneq \big(Y_{-,j,2}(\zeta_1),\dots,Y_{-,j,2}(\zeta_{N_\III})\big)\,,\\[1ex]
\ul{Y_{-j,3}(z_n^*)}
 & \coloneq \big(Y_{-,j,3}(z_1^*),\dots,Y_{-,j,3}(z_{N_\II}^*)\big)\,,&\qquad
\ul{Y_{-,j,3}(w_n)}
 & \coloneq  \big(Y_{-,j,3}(w_1),\dots,Y_{-,j,3}(w_{N_\I})\big)\,.
\end{aligned}
\end{equation}
The blocks $\A_{i,j}$ are defined as ($m$ is the row index and $n$ is the column index)
\bse
\begin{equation}
\begin{aligned}
\A_{1,1} & \coloneq \bigg(\frac{\ii  w_n^*}{E_0}\frac{\widecheck C_n}{w_m^*-\hat w_n^*}\bigg)_{N_\I\times N_\I}\,,&
\A_{1,3} & \coloneq \bigg(\frac{\widehat D_n}{w_m^*-\hat z_n}\bigg)_{N_\I\times N_\II}\,,\\
\A_{1,4} & \coloneq \bigg(\frac{F_n}{w_m^*-\zeta_n}\bigg)_{N_\I\times N_\III}\,,&
\A_{1,6} & \coloneq \bigg(\frac{-\overline C_n^*}{w_m^*-w_n}\bigg)_{N_\I\times N_\I}\,,\\
\A_{2,1} & \coloneq \bigg(\frac{\ii  w_n^*}{E_0}\frac{\widecheck C_n}{\zeta_m^*-\hat w_n^*}\bigg)_{N_\III\times N_\I}\,,&
\A_{2,3} & \coloneq \bigg(\frac{\widehat D_n}{\zeta_m^*-\hat z_n}\bigg)_{N_\III\times N_\II}\,,\\
\A_{2,4} & \coloneq \bigg(\frac{F_n}{\zeta_m^*-\zeta_n}\bigg)_{N_\III\times N_\III}\,,&
\A_{2,6} & \coloneq \bigg(\frac{-\overline C_n}{\zeta_m^*-w_n}\bigg)_{N_\III\times N_\I}\,,\\
\A_{3,2} & \coloneq \bigg(\frac{\overline F_n}{z_m-\zeta_n^*}+\frac{\ii\zeta_n^*}{E_0}\frac{\widecheck F_n}{z_m-\hat \zeta_n^*}\bigg)_{N_\II\times N_\III}\,,&
\A_{3,5} & \coloneq \bigg(-\frac{D_n^*}{\gamma(z_n^*)}\bigg(\frac{1}{z_m-z_n^*} + \frac{\gamma(z_n^*) - 1}{z_m-\hat z_n^*}\bigg)\bigg)_{N_\II\times N_\II}\,,\\
\end{aligned}
\end{equation}
as well as
\begin{equation}
\begin{aligned}
\A_{4,2} & \coloneq \bigg(\frac{\overline F_n}{\zeta_m-\zeta_n^*}+\frac{\ii\zeta_n^*}{E_0}\frac{\widecheck F_n}{\zeta_m-\hat \zeta_n^*}\bigg)_{N_\III\times N_\III}\,,&
\A_{4,5} & \coloneq \bigg(- \frac{D_n^*}{\gamma(z_n^*)}\bigg(\frac{1}{\zeta_m-z_n^*} + \frac{\gamma(z_n^*) - 1}{\zeta_m-\hat z_n^*}\bigg)\bigg)_{N_\III\times N_\II}\,,\\
\A_{5,1} & \coloneq \bigg(\frac{\overline C_n}{z_m^*-w_n^*}\bigg)_{N_\II\times N_\I}\,,&
\A_{5,3} & \coloneq \bigg(\frac{D_n}{z_m^*-z_n}\bigg)_{N_\II\times N_\II}\,,\\
\A_{5,4} & \coloneq \bigg(\frac{\widehat F_n}{z_m^*-\hat\zeta_n}\bigg)_{N_\II\times N_\III}\,,&
\A_{5,6} & \coloneq \bigg(-\frac{\ii E_0}{w_n}\frac{\overline C_n^*}{z_m^*-\hat w_n}\bigg)_{N_\II\times N_\I}\,,\\
\A_{6,1} & \coloneq \bigg(\frac{\overline C_n}{w_m-w_n^*}\bigg)_{N_\I\times N_\I}\,,&
\A_{6,3} & \coloneq \bigg(\frac{D_n}{w_m-z_n})\bigg)_{N_\I\times N_\II}\,,\\
\A_{6,4} & \coloneq \bigg(\frac{\widehat F_n}{w_m-\hat\zeta_n}\bigg)_{N_\I\times N_\III}\,,&
\A_{6,6} & \coloneq \bigg(-\frac{\ii E_0}{w_n}\frac{\overline C_n^*}{w_m-\hat w_n}\bigg)_{N_\I\times N_\I}\,,
\end{aligned}
\end{equation}
\ese
with the remaining blocks being identically zero matrices of appropriate sizes.
Then the unknowns $X_{i,j}$ can be expressed as
\begin{equation}
\nonumber
X_{i,j} = \frac{\det(\overline{(\bbI - \A)_{i,j}})}{\det(\bbI - \A)}\,,
\end{equation}
where $\overline{(\bbI - \A)_{i,j}}$ is the matrix $\bbI - \A$ with the $i$-th column replaced by $\ul B_j$.
Thus, the electric field envelope $E_j$ is obtained from Equation~\eqref{e:cmbe-sol} as
\bse
\begin{equation}
E_j = E_{-,j} - \ii\,\bigg(\frac{\det(\widetilde{(\bbI - \A)}_{j+1})}{\det(\bbI - \A)}\bigg)^*\,,
\end{equation}
where
\begin{equation}
\begin{aligned}
\widetilde{(\bbI - \A)}_j
 & \coloneq \bpm0 & -\ul D \\ \ul B_j & \bbI - \A\epm\,,&
\ul D & \coloneq \bpm\ul D_1,\ul D_2,\dots, \ul D_6\epm\,,&
\ul D_1 & \coloneq \bigg(\frac{\ii  w_1^*}{E_0}\widecheck C_1,\dots,\frac{\ii  w_{N_\I}^*}{E_0}\widecheck C_{N_\I}\bigg)\,,\\
\ul D_3 & \coloneq \big(\widehat D_1,\dots, \widehat D_{N_\II}\big)\,,&
\ul D_4 & \coloneq \big(F_1,\dots, F_{N_\III}\big)\,,&
\ul D_6 & \coloneq \big(-\overline C_1^*,\dots, -\overline C_{N_\I}^*\big)\,,
\end{aligned}
\end{equation}
\ese
with $\ul D_2$ and $\ul D_5$ being zero vectors of appropriate sizes.

We can also reconstruct the eigenfunction matrix from the solution~\eqref{e:musol} as
\bse
\label{e:reflectionless-mu-sol}
\begin{equation}
\bmu_-(t,z,\zeta)
 = \Y_- +\left[\frac{\det(\Z_{i,j})}{\det(\bbI - \A)}\right]_{3\times3}\,,
\end{equation}
with
\begin{equation}
\begin{aligned}
\Z_{i,j}
 & \coloneq \bpm0 & -\@G_j \\ \B_i & \bbI - \A\epm \,,\\[1ex]
\@G_1 & \coloneq \bigg(\frac{\ii}{E_0}\ul{\frac{w_n^*\widecheck C_n}{\zeta-\hat w_n^*}},\@0,
\ul{\frac{\widehat D_n}{\zeta-\hat z_n}},\ul{\frac{F_n}{\zeta-\zeta_n}},
\@0,-\ul{\frac{\overline C_n^*}{\zeta-w_n}}\bigg)\,,\\[1ex]
\@G_2 & \coloneq \bigg(\@0,\ul{\frac{\overline F_n}{\zeta-\zeta_n^*}+\frac{\ii\zeta_n^*}{E_0}\frac{\widecheck F_n}{\zeta-\hat\zeta_n^*}},
\@0,\@0,-\ul{\left(\frac{1}{\zeta-z_n^*}
 + \frac{(\gamma(z_n^*)-1)}{\zeta-\hat z_n^*}\right)\frac{D_n^*}{\gamma(z_n^*)}},\@0\bigg)\,,\\[1ex]
\@G_3 & \coloneq \bigg(\ul{\frac{\overline C_n}{\zeta-w_n^*}},\@0,\ul{\frac{D_n}{\zeta-z_n}},
\ul{\frac{\widehat F_n}{\zeta-\hat\zeta_n}},\@0,-\ul{\frac{\ii E_0}{w_n}\frac{\overline C_n^*}{\zeta-\hat w_n}}\bigg)\,,
\end{aligned}
\end{equation}
\ese
where all $\@0$ in the above expression are zero vectors of appropriate sizes,
and the underline notation $\ul{x_n}$ denotes vectors $(x_1,\dots,x_N)$ of appropriate sizes $N$ for each type of discrete eigenvalues.

Finally, the reconstruction formula~\eqref{e:cmbe-sol} yields the density matrix $\brho(t,z,\zeta)$
\begin{equation}
\label{e:reflectionlessrhosol}
\brho(t,z,\zeta) = \bmu_-(t,z,\zeta)\,\bvarrho_-(z,\zeta)\,\bmu_-^{-1}(t,z,\zeta)\,,\qquad
\zeta\in\Real\,.
\end{equation}

\section{Solitons and their behavior, part I: general description}
\label{s:generalsoliton}

We are ready to calculate explicit soliton solutions of the MBEs~\eqref{e:cmbe} with NZBG,
and we present the discussion in two sections.
In this section,
we calculate all three kinds of one-soliton solutions of CMBE~\eqref{e:cmbe},
after which we compute various limiting cases.
In order to achieve maximal generality,
we present descriptions of soliton solutions without imposing explicit assumptions on the spectral line shape $g(k)$.
As a result, the $z$-dependence of soliton solutions cannot be calculated fully.
In Section~\ref{s:IB},
we analyze all aforementioned soliton solutions with two particular spectral-line shapes:
Lorentzian and sharp-line.
With an explicit expression for the shape of the spectral line,
we are able to discuss the $z$-dependence of solutions,
and also perform stability analysis on soliton solutions.
Recall the discussion in Section~\ref{s:reflectionless};
we take the diagonal asymptotic condition~\eqref{e:reflectionlessBC} throughout this section.

Recall the discussion of the asymptotic conditions in Section~\ref{s:intro}.
Without loss of generality,
we take the asymptotic input condition~\eqref{e:bc-q}.
Therefore, the asymptotic conditions for the electric-filed envelopes are
\begin{equation}
\label{e:soliton-backgroundE}
E_{-,1}(z) = E_0 \e^{2 \ii W_-(z)}\cos\alpha\,,\qquad
E_{-,2}(z) = E_0 \e^{2 \ii W_-(z)}\sin\alpha\,,
\end{equation}
where $W_-(z)$ is defined in Equation~\eqref{e:QpmQ0} and $\alpha\in[0,\pi/2]$.
Since there are three kinds of discrete eigenvalues (cf. Section~\ref{s:eigen}),
there will be three kinds of one-soliton solutions.
Because we discuss them separately,
for simplicity and for all three types of one-soliton solutions,
we take the discrete eigenvalue in the identical form:
\begin{equation}
\label{e:wzzeta}
w_1 = z_1 = \zeta_1 = E_0\, \eta\,\e^{\ii \beta} \in D_1\,,
\end{equation}
where $\eta>1$ and $\beta\in(0,\pi)$.
We also define the following quantities for later usage:
\begin{equation}
\Delta_\pm \coloneqq \eta \pm 1/\eta\,,\qquad
c_o \coloneqq \sqrt{\eta^4+2\eta^2\cos2\beta+1}\,,\qquad
c_\pm \coloneqq \eta^2 \pm 1/\eta^2\,.
\end{equation}
Furthermore,
it will be convenient to use the subscripts $\I$, $\II$ and $\III$ to denote quantities corresponding to soliton solutions of
the first kind,
the second kind and the third kind, respectively.
Correspondingly,
subscripts $\I$ and $\II$ in Sections~\ref{s:solitonI} and~\ref{s:solitonII} do not denote evaluations on the first and second complex $k$-planes from the direct problem, respectively.
It is also worth noting that, due to the complexity of the explicit expressions for the $3\times3$ density matrix $\brho(t,z,k)$, we mainly focus on analyzing the electric field envelope $\E(t,z)$ in Sections~\ref{s:generalsoliton} and~\ref{s:IB}.
On the other hand, it should be pointed out that $\brho(t,z,k)$ can always be constructed explicitly from Equations~\eqref{e:reflectionless-mu-sol} and~\eqref{e:reflectionlessrhosol}.

\subsection{One-soliton solution: type I}
\label{s:solitonI}

We take the discrete eigenvalue $w_1$ in Equation~\eqref{e:wzzeta},
and take the corresponding norming constant as
\begin{equation}
\label{e:solitonC1bar}
\overline C_1(t,z) = \exp(2\ii  \lambda_1^* t+ \xi_\I(z) + \ii\,\psi_\I(z))\,,
\end{equation}
where $\lambda_1\coloneq\lambda(\zeta = w_1) \coloneq E_0\cosh(\ln\eta + \ii\beta)$ and $\xi_\I(z)\in\Real$,
$\psi_\I(z)\in\Real$.
The $z$-propagation of $\overline C_1(t,z)$ is governed by Equation~\eqref{e:normingtimeevolution}.
It currently cannot be determined further due to the missing information about the shape of the spectral line, $g(k)$. Using the symmetries~\eqref{e:mnormingsymmetry},
one can compute the other three norming constants:
$C_1(t,z)$, $\widehat C_1(t,z)$ and $\widecheck C_1(t,z)$, as well.

Following Section~\ref{s:reflectionless},
we obtain the one-soliton solution formula,
\bse
\label{e:soliton1}
\begin{equation}
E_j(t,z) \coloneq q_{\I,j}(t,z)\,,\qquad  j = 1,2,
\end{equation}
where it can be shown that $q_{\I,j}(t,z)$ is a one-soliton solution of the two-level MBEs with NZBG
$q_{\I,j} \to E_{-j}$ as $t\to-\infty$,
namely,
\begin{equation}
q_{\I,j}(t,z)
 \coloneq \e^{-2 \ii \beta} E_{-,j}
\frac{\cosh (\chi_\I - 2\ii \beta) -  A \big[c_+ \big(\eta^2 \sin (s_\I + 2\beta)+\sin s_\I\big) - \ii c_- \big(\eta^2 \cos (s_\I + 2\beta) + \cos s_\I\big)\big]}
{\cosh\chi_\I + 2 A \big[\eta^2 \sin (s_\I + 2 \beta)+\sin s_\I\big]}\,,
\end{equation}
where we omit both $t$ and $z$ dependence, $j = 1,2$ and
\begin{gather}
\chi_\I(t,z) \coloneq 2t \Im(\lambda_1) + \xi_\I(z) + \ln\frac{\Delta_+\csc\beta}{2E_0 c_o}\,,\qquad
s_\I(t,z) \coloneq 2t \Re(\lambda_1) + \psi_\I(z)\,,\qquad
A \coloneq \frac{\sin\beta}{\Delta_+ c_o}\,.
\end{gather}
\ese
Importantly,
Equations~\eqref{e:soliton1} imply that the polarization state of the light is preserved under propagation onto the medium.

\subsubsection{Purely imaginary discrete eigenvalue.}

In particular,
for a purely imaginary discrete eigenvalue,
i.e., $\beta = \pi/2$,
the soliton solution simplifies to
\bse
\begin{equation}
\label{e:soliton1pureimaginary}
E_j(t,z) = E_{-,j}\frac{\cosh\widetilde\chi_\I - (c_+ \sin \psi_\I - \ii c_- \cos \psi_\I)/\Delta_+}
{\cosh\widetilde\chi_\I - 2/\Delta_+ \sin \psi_\I}\,,\qquad
j = 1,2,
\end{equation}
with
\begin{equation}
\widetilde\chi_\I(t,z)
 = E_0 \Delta_- t +\xi_\I(z) + \ln \frac{\Delta_+}{2E_0 \eta\Delta_-}\,,\qquad
\psi_\I = \psi_\I(z)\,.
\end{equation}
\ese

\subsubsection{Soliton velocity and amplitude.}

The soliton~\eqref{e:soliton1} is localized along the line $\chi_\I(t,z) = $ constant,
which is $2t \Im(\lambda_1) + \xi_\I(z) = $ constant.
Taking the asymptotic condition~\eqref{e:reflectionlessBC},
we know that $\xi_\I(z) = \xi_{\I}^{(1)} z + \xi_\I(0)$ where $\xi_{\I}^{(1)} $ is independent of $z$.
Then the soliton velocity is defined as
\begin{equation}
\label{e:soliton-typeI-velocity}
V_\I \coloneq -2\Im(\lambda_1)/\xi_\I^{(1)}\,.
\end{equation}
The soliton amplitude compared to the background is defined as
\begin{equation}
A_j(z)
 \coloneq \max_{t\in\Real}|E_j(t,z)| - |E_{-,j}(z)|\,.
\end{equation}

\subsection{Periodic solution}
\label{s:periodic-solution}

Firstly, we focus on the case where the discrete eigenvalue is in the first quadrant, i.e., $\beta\in(0,\pi/2)$.
Then, one can take the limit of Equation~\eqref{e:soliton1} as $\eta\to1$ and obtain the following:
\begin{equation}
\begin{aligned}
E_j(t,z) = \e^{-2\ii \beta}E_{-,j}
\frac{\cosh(\chi_\I-2\ii \beta) - \big[\sin(s_\I+2\beta) + \sin s_\I \big]/(2\tan\beta)}
{\cosh\chi_\I + \big[\sin(s_\I+2\beta) + \sin s_\I\big]/(2\tan\beta)}\,,\\
\chi_\I(z) = \xi_\I(z) - \ln (E_0\sin2\beta)\,,\qquad
s_\I(t,z) = 2E_0 t \cos\beta + \psi_\I(z)\,.
\end{aligned}
\end{equation}
This case is the analogue of the Akhmediev breather of the focusing NLS equation,
and is the vectorization of the periodic solution of the two-level MBE.

A detailed discussion of the periodic solution and why it bears this name is presented in Section~\ref{s:periodic-ih}.

\subsection{Rational solutions}
\label{s:rational-solution}

We use the soliton solution~\eqref{e:soliton1pureimaginary} with a purely imaginary discrete eigenvalue $\zeta = \ii E_0\eta$ to produce a rational solution,
which is similar to the one obtained from two-level MBE with NZBG or from the focusing NLS equation with NZBC,
i.e., the famous Peregrine soliton.
The idea for producing such a solution is to compute the nontrivial limit of Equation~\eqref{e:soliton1pureimaginary} as $\eta\to1$.

To do so, we first compute the limit of $R_{-1,1}(z,\zeta)$ and $R_{-3,3}(z,\zeta)$ as $\eta\to1$ to analyze the propagation of the soliton solution~\eqref{e:soliton1pureimaginary}.
Note that $R_{-2,2}(z,\zeta)$ is not needed in type $\I$ soliton solution.
It can be shown that
\begin{equation}
\begin{aligned}
R_{-1,1}(z,\ii E_0\eta) & = R^{(0)}(z) - (\eta-1)R^{(1)}(z) + \O(\eta-1)^2\,,\qquad \eta\to1^+\,,\\
R_{-3,3}(z,\ii E_0\eta) & = R^{(0)}(z) + (\eta-1)R^{(1)}(z) + \O(\eta-1)^2\,,\qquad \eta\to1^+\,,
\end{aligned}
\end{equation}
where
\begin{equation}
\begin{aligned}
R^{(0)}(z) & \coloneq \int_{-\infty}^\infty \varrho_-^+(z,\zeta(k')) g(k') \frac{\d k'}{k' + \ii E_0} + 2w_-(z)\,,\\
R^{(1)}(z) & \coloneq \int_{-\infty}^\infty \frac{\ii  E_0}{\lambda(k')}\varrho_-^-(z,\zeta(k')) \,g(k')\frac{\d k'}{k' + \ii E_0}\,,
\end{aligned}
\end{equation}
and $\varrho_-^-$ is defined in Equation~\eqref{e:rhopmpm}.
Therefore, we can write
\begin{equation}
R_{-1,1}(z,\ii E_0\eta) - R_{-3,3}(z,\ii E_0\eta)
 = -2(\eta-1)R^{(1)}(z)  + \O(\eta-1)^2\,,\qquad \eta\to1^+\,.
\end{equation}
Using Equations~\eqref{e:normingtimeevolution},~\eqref{e:solitonC1bar} and the above expansions,
we obtain the following expansion for the ODEs governing $\xi_\I(z)$ and $\psi_\I(z)$ as $\eta\to1^+$:
\begin{equation}
\partial_z \xi_\I(z) = (\eta-1) R_{\Im}^{(1)}(z) + \O(\eta-1)^2\,,\qquad
\partial_z \psi_\I(z) = - \ii(\eta-1) R_{\Re}^{(1)}(z) + \O(\eta-1)^2\,,
\end{equation}
where $R_{\Re}^{(1)}(z)$ and $R_{\Im}^{(1)}(z)$ are the real and imaginary parts of $R^{(1)}(z)$, respectively.
Thereafter,
we solve for $\xi_\I(z)$ and $\psi_\I(z)$ easily in the limit $\eta\to1^+$:
\begin{equation}
\begin{aligned}
\xi_\I(z) & = \widetilde\xi_0 + \bigg(\int_0^z R_{\Im}^{(1)}(z)\d z + \widetilde\xi_1\bigg)(\eta-1) + \O(\eta-1)^2\,,\\
\psi_\I(z) & = \widetilde\psi_0 +\bigg(- \ii \int_0^z R_{\Re}^{(1)}(z)\d z + \widetilde\psi_1\bigg)(\eta-1) + \O(\eta-1)^2\,,
\end{aligned}
\end{equation}
where $\widetilde\xi_0$, $\widetilde\xi_1$, $\widetilde\psi_0$ and $\widetilde\psi_1$ are all independent of $z$.
In particular, these four constants depend on the ``initial data" $\xi_\I(0)$ and $\psi_\I(0)$,
i.e., $\overline C_1(0,0)$,
so that they can be chosen arbitrarily.

Next, we compute the nontrivial limit of Equation~\eqref{e:soliton1pureimaginary} as $\eta\to1^+$.
Choosing $\widetilde\xi_0 =  -\ln({\Delta_+}/{2E_0 \eta\Delta_-})$ and $\widetilde\psi_0 = \pi/2$
yields the limit as $\eta\to1^+$ as follows:
\begin{equation}
\begin{aligned}
\widetilde\chi_\I(t,z)
 & = E_0 \Delta_- t +\bigg(\int_0^z R_{\Im}^{(1)}(z) \d z + \widetilde\xi_1\bigg)(\eta-1)+ \O(\eta-1)^2\,,\\
\psi_\I(z)
 & = \pi/2  + \bigg(- \ii \int_0^z R_{\Re}^{(1)}(z) \d z + \widetilde\psi_1\bigg)(\eta-1) + \O(\eta-1)^2\,.
\end{aligned}
\end{equation}
Correspondingly, the nontrivial limit of the soliton solution~\eqref{e:soliton1pureimaginary} as $\eta\to1^+$ is revealed as
\begin{equation}
E_j(t,z) = E_{-,j}(z) \frac{X^2 + s^2 - 4\ii s -3}{X^2 + s^2 +1}\,,
\label{e:rationalsolution}
\end{equation}
with
\begin{equation}
X(t,z) \coloneq 2E_0t + \int_0^z R_{\Im}^{(1)}(z)\d z + \widetilde\xi_1\,,\qquad
s(t,z) \coloneq - \ii \int_0^z R_{\Re}^{(1)}(z)\d z + \widetilde\psi_1\,.
\end{equation}
Note that the two real constants $\widetilde\xi_1$ and $\widetilde\psi_1$ are arbitrary.
They determine the displacement of this solution in the $(t,z)$ plane.
Again, this rational solution is the analogue to the one for the classic two-level MBE with NZBG and the one for the focusing NLS equation with NZBCs.
Note however that, while the corresponding solutions for the focusing NLS equation are rational both in $t$ and $z$,
in this case the solution \eqref{e:rationalsolution} is rational in $t$, but it need not necessarily be rational in $z$.

A detailed discussion under the assumption that the spectral line is Lorentzian is presented in Section~\ref{s:rational-ih}.


\subsection{One-soliton solution: type II}
\label{s:solitonII}

Let us discuss the one-soliton solution with a discrete eigenvalue $z_1$ given in Equation~\eqref{e:wzzeta}.
We take the modified norming constant
\begin{equation}
D_1(t,z) \coloneq \exp(-\ii\hat z_1 t + \xi_\II(z) + \ii\psi_\II(z))\,,
\end{equation}
where we recall that $D_1(t,z)$ can be obtained from the ODE~\eqref{e:normingtimeevolution} once the spectral line shape $g(k)$ is known.
The corresponding components of the soliton solution are
\begin{gather}
\label{e:soliton2}
E_j(t,z) = \e^{\ii \beta}\big[ E_{-,j} \cosh (\chi_\II + \ii\beta) +
(-1)^j \e^{-\ii  s_\II} E_{-,3-j}^*\sqrt{\Delta_+/\eta} \sin\beta  \big]\sech\chi_\II\,,\qquad
j = 1,2,
\end{gather}
where $E_{-,j} = E_{-,j}(z)$ from Equation~\eqref{e:soliton-backgroundE},
and
\begin{equation}
\label{e:chisII}
\begin{aligned}
\chi_\II(t,z)
 & \coloneq (E_0 t \sin\beta)/\eta + \xi_\II(z) -\ln(2E_0 \sqrt{\Delta_+\eta}\sin\beta)\,,\\
s_\II(t,z)
 & \coloneq (E_0 t \cos\beta)/\eta + \psi_\II(z)\,.
\end{aligned}
\end{equation}
Obviously, the electric field envelopes~\eqref{e:soliton2} are concentrated along the line $\chi_\II(t,z) = \,$constant.
With the chosen asymptotic condition~\eqref{e:reflectionlessBC}, the ODE~\eqref{e:normingtimeevolution} implies the form $\xi_\II(z) = \xi_\II^{(1)}z + \xi_\II(0)$.
As a result, the soliton velocity is defined as
\begin{equation}
\label{e:soliton2-velocity}
V_\II = -\frac{E_0\sin\beta}{\eta\xi_\II^{(1)}}\,.
\end{equation}
It is evident that the velocity $V_\II$ is affected by the location of the discrete eigenvalue and the asymptotic condition~\eqref{e:reflectionlessBC} via the quantities $\xi_\II^{(1)}$ and $\beta$.
As such, one expects more interesting and complex phenomena,
as we demonstrate using explicit examples in the next section.

\subsubsection{Purely imaginary discrete eigenvalue.}
In particular, for a purely imaginary discrete eigenvalue,
i.e., $\beta = \pi/2$,
the above one-soliton solution becomes
\begin{equation}
E_j(t,z) = - E_{-,j}(z) \tanh\chi_\II(t,z) + (-1)^j \sqrt{\Delta_+/\eta} E_{j-3}^* \e^{-\ii \psi_\II(z)} \sech\chi_\II(t,z)\,,
\end{equation}
with
\begin{equation}
\chi_\II(t,z) \coloneq E_0 t/\eta + \xi_\II(z) -\ln(2E_0 \sqrt{\Delta_+\eta})\,.
\end{equation}

\subsection{A nontrivial plane-wave solution}

Unlike for the soliton solution of type $\I$,
we discover a special limiting case as $\eta\to\infty$ of the second type of soliton solution~\eqref{e:soliton2}.

To uncover this nontrivial limiting solution,
we first examine the matrix $\R_{-,\d}(z,E_0\eta\e^{\ii \beta})$ defined in Equation~\eqref{e:Rpm} by computing its limit as $\eta\to\infty$ with $\beta\in(0,\pi)$.
It can be shown that
\begin{equation}
\lim_{\eta\to\infty}R_{-,1,1}(z,E_0\eta\e^{\ii \beta}) = 0\,,\quad
\lim_{\eta\to\infty}R_{-,2,2}(z,E_0\eta\e^{\ii \beta})
 = - \lim_{\eta\to\infty}R_{-,3,3}(z,E_0\eta\e^{\ii \beta}) = - 4w_-(z)\,.
\end{equation}
Therefore,
the propagation equation~\eqref{e:normingtimeevolution} yields that $\lim_{\eta\to\infty}D_1(t,z) = \e^{-4\ii W_-(z)}\lim_{\eta\to\infty}D_1(t,0)$ with $W_-(z)\in\Real$,
so
\begin{equation}
\lim_{\eta\to\infty}\xi_\II(z) = \lim_{\eta\to\infty}\xi_\II(0)\,,\qquad
\lim_{\eta\to\infty}\psi_\II(z) = -4W_-(z) + \lim_{\eta\to\infty}\psi_\II(0)\,.
\end{equation}
This calculation shows that, in the limit $\eta\to\infty$, the quantity $\xi_\II$ becomes independent of $z$.
Consequently,
if we choose the initial condition $\xi_\II(0) = \ln(2E_0 \sqrt{\Delta_+\eta}\sin\beta) + \widetilde\xi_\II$,
from Equation~\eqref{e:chisII}, we obtain the following expressions for all $t$ and $z$:
\begin{equation}
\lim_{\eta\to\infty}\chi_\II(t,z)
 = \widetilde\xi_\II\,,\qquad
\lim_{\eta\to\infty}s_\II(t,z)
 = \lim_{\eta\to\infty}\psi_\II(z) = -4W_-(z) + \widetilde\psi_\II\,,
\end{equation}
where $\widetilde\psi_\II \coloneq \lim_{\eta\to\infty}\psi_\II(0)$,
and both $\widetilde\xi_\II$ and $\widetilde\psi_\II$ are real constants.

The limit of Equation~\eqref{e:soliton2} can be computed directly and obtained as
\begin{equation}
\begin{aligned}
E_j(z)
 & = \e^{\ii  \beta }\big[E_{-,j}(z) \cosh (\widetilde\xi_\II + \ii\beta) +
 (-1)^j \e^{- \ii\widetilde\psi_\II} E_{-,3-j}(z) \sin\beta  \big]\sech\widetilde\xi_\II \,,\quad
j = 1,2.
\end{aligned}
\end{equation}
Clearly,
this soliton solution is independent of $t$.
Moreover,
the modulus of this soliton is constant, meaning that it is a plane-wave solution.

Recall that two solutions of the CMBE~\eqref{e:cmbe} may be equivalent by Lemma~\ref{thm:CMBE-U-invariance}. It is thus necessary to discuss if the two sets $\{E_1(z), E_2(z)\}$ and $\{E_{\bg,1}(z),E_{\bg,2}(z)\}$ from Equation~\eqref{e:background} are equivalent, the latter of which is the background solution used in the formulation of IST, i.e., chosen as the asymptotic condition~\eqref{e:soliton-backgroundE}. It is easy to check that the two sets of solutions are related by
\begin{equation}
\bpm E_1(z)\\ E_2(z) \epm = \V \bpm E_{\bg,1}(z) \\ E_{\bg,2}(z)\epm\,,\qquad
\V \coloneqq \e^{\ii\beta}\sech\widetilde\xi_\II \bpm \cosh(\widetilde\xi_\II + \ii\beta) & -\e^{-\ii\widetilde\psi_\II}\sin\beta \\ \e^{-\ii\widetilde\psi_\II}\sin\beta & \cosh(\widetilde\xi_\II + \ii\beta)\epm\,.
\end{equation}
One can verify that the matrix $\V$ is not a unitary matrix in general, meaning that the solutions $\{E_1(z), E_2(z)\}$ and $\{E_{\bg,1}(z),E_{\bg,2}(z)\}$ are not equivalent by Lemma~\ref{thm:CMBE-U-invariance}. Hence, we consider the new solution $\{E_1(z), E_2(z)\}$ a nontrivial plane-wave solution (compared to the background solution $\{E_{\bg,1}(z),E_{\bg,2}(z)\}$).


\subsection{One-soliton solution: type III}
\label{s:solitonIII}

Recall that type $\III$ discrete eigenvalue $\zeta_1$ is given in Equation~\eqref{e:wzzeta}.
The form of the norming constant is chosen as
\begin{equation}
\overline F_1(t,z) = \exp(\ii\zeta_1^*t + \xi_\III(z) + \ii\psi_\III(z))\,,
\end{equation}
where $\xi_\III(z)$ and $\psi_\III(z)$ are real quantities as in the previous cases.
The electric field envelopes are given by
\begin{equation}
\label{e:plane-wave}
E_j(t,z) = \e^{-\ii \beta }\big[E_{-,j}(z) \cosh (\chi_\III - \ii\beta) + \ii(-1)^{j+1} B\, E_{-,3-j}^*(z) \big(\eta^2 + \e^{2 \ii \beta }\big)\e^{\ii  s_\III} \big]\sech\chi_\III\,,\qquad
j = 1,2,
\end{equation}
where the parameters are given as
\begin{equation}
\begin{aligned}
\chi_\III(t,z) & \coloneq E_0\eta t \sin\beta + \xi_\III(z) - \ln(2E_0 \eta^2 B)\,,\\
s_\III(t,z) & \coloneq E_0\eta t \cos\beta + \psi_\III(z) + \beta\,,\\
B & \coloneq \frac{\sqrt{\eta\Delta_+}}{c_o}\sin\beta\,.
\end{aligned}
\end{equation}

\subsubsection{Purely imaginary discrete eigenvalue.}
In particular,
for a purely imaginary discrete eigenvalue,
i.e., $\beta = \pi/2$,
the soliton solution in Equation~\eqref{e:plane-wave} becomes
\begin{equation}
E_j(t,z)
 = - E_{-,j}(z) \tanh\widetilde\chi_\III(t,z) + \ii(-1)^{j+1}E_{-,3-j}^*(z) \e^{\ii \psi_\III(z)} \sqrt{\eta\Delta_+} \sech\widetilde\chi_\III(t,z)\,,
\end{equation}
with
\begin{equation}
\widetilde\chi_\III(t,z) \coloneq E_0\eta t + \xi_\III(z) - \ln(2E_0 \eta\sqrt{\eta\Delta_+}\big/\Delta_-)\,.
\end{equation}

\section{Solitons and their behavior, Part II: particular spectral line shape}
\label{s:IB}

We discuss a widely used case, in which the spectral-line shape, $g(k)$, is chosen as a Lorentzian
\begin{equation}
\label{e:Lorentzian}
g(k;\epsilon) \coloneq \frac{\epsilon}{\pi}\frac{1}{k^2+\epsilon^2}\,,\qquad
\epsilon>0\,.
\end{equation}
We take the asymptotic condition~\eqref{e:reflectionlessBC} with the assumption that the diagonal entries of $\bvarrho_-$ are independent of both $z$ and $k$.
Therefore, by examining Equation~\eqref{e:wpm} and noticing that $\lambda(k)$ and $g(k;\epsilon)$ are respectively odd and even functions of $k$,
one concludes $w_-(z) \equiv 0$.

Besides the quantity $w_-$,
the propagation of solutions is determined by the norming constant,
which is governed by Equation~\eqref{e:normingtimeevolution}.
In Appendix~\ref{a:Rnd},
we show how to compute the auxiliary matrix $\R_{-,\d}(z,\zeta)$,
which is actually independent of $z$.
The results are as follows:
\begin{equation}
\label{e:R-dexplicit}
\begin{aligned}
R_{-,1,1}(\zeta)
 & = \begin{cases}
\displaystyle
- \frac{\rho_-^+}{k + \ii\epsilon} - \varrho_-^-\,g(k)
\left[\log \left(\frac{E_0 - \lambda }{E_0 + \lambda}\right) +
\frac{\lambda}{\sqrt{E_0^2 - \epsilon^2}} \log\left(\frac{E_0 + \sqrt{E_0^2 - \epsilon^2}}{E_0 - \sqrt{E_0^2 - \epsilon ^2}}\right) \right]\,, & k\in\Complex^+\,,\\
\displaystyle
 -\frac{k\rho_-^+}{k^2 + \epsilon ^2} - \varrho_-^-\,g(k)
\left[\log \left(\frac{\lambda - E_0}{\lambda + E_0}\right) +
\frac{\lambda}{\sqrt{E_0^2 - \epsilon^2}} \log \left(\frac{E_0 + \sqrt{E_0^2 - \epsilon^2}}{E_0 - \sqrt{E_0^2 - \epsilon ^2}}\right)\right]\,, & k\in\Real\,,\\
\displaystyle
- \frac{\rho_-^+}{k - \ii \epsilon} - \varrho_-^- \,g(k)
\left[\log\left(\frac{E_0 - \lambda }{E_0 + \lambda}\right) +
\frac{\lambda}{\sqrt{E_0^2 - \epsilon^2}} \log\left(\frac{E_0 + \sqrt{E_0^2 - \epsilon^2}}{E_0 - \sqrt{E_0^2 - \epsilon ^2}}\right) \right]\,, & k\in\Complex^-\,,
\end{cases}\\
R_{-,2,2}(\zeta)
 & = \begin{cases}
\displaystyle
-\varrho_{-,2,2}/(k + \ii \epsilon)\,, & k\in\Complex^+\,,\\
\displaystyle
 -k\varrho_{-,2,2}/(k^2 + \epsilon^2)\,, & k\in\Real\,,\\
 \displaystyle
- \varrho_{-,2,2}/(k - \ii \epsilon)\,, & k\in\Complex^-\,,
\end{cases}\\
R_{-,3,3}(\zeta)
 & = \begin{cases}
\displaystyle
- \frac{\rho_-^+}{k + \ii \epsilon} + \varrho_-^-\,g(k)
\left[ \log \left(\frac{E_0-\lambda }{E_0 + \lambda}\right) +
\frac{\lambda}{\sqrt{E_0^2-\epsilon^2}} \log \left(\frac{E_0 + \sqrt{E_0^2-\epsilon^2}}{E_0-\sqrt{E_0^2-\epsilon ^2}}\right) \right]\,, & k\in\Complex^+\,,\\
 \displaystyle
 -\frac{k\rho_-^+}{k^2+\epsilon ^2} + \varrho_-^-\,g(k)
\left[ \log \left(\frac{\lambda -E_0}{\lambda + E_0}\right) +
\frac{\lambda}{\sqrt{E_0^2-\epsilon^2}} \log \left(\frac{E_0 + \sqrt{E_0^2-\epsilon^2}}{E_0-\sqrt{E_0^2-\epsilon ^2}}\right)\right]\,, & k\in\Real\,,\\
\displaystyle
- \frac{\rho_-^+}{k - \ii \epsilon} + \varrho_-^- \,g(k)
\left[ \log \left(\frac{E_0-\lambda }{E_0 + \lambda}\right) +
\frac{\lambda}{\sqrt{E_0^2 - \epsilon^2}} \log \left(\frac{E_0 + \sqrt{E_0^2-\epsilon^2}}{E_0-\sqrt{E_0^2-\epsilon^2}}\right) \right]\,, & k\in\Complex^-\,,
\end{cases}
\end{aligned}
\end{equation}
where we recall that $\varrho_-^\pm$ are defined in Equation~\eqref{e:rhopmpm}.

\begin{remark}
Similarly to the two-level MBE with NZBG,
we can show that the two limits $E_0\to0$ and $\epsilon\to0$ do not commute.
\end{remark}

\begin{proof}
\textit{$\bullet$ Zero background limit first.}
We first take the limit $E_0\to0$.
It is obvious that $E_{-,j}(z)\to0$ for $j=1,2$.
The matrix $\R_{-,\d}(\zeta)$ in Equation~\eqref{e:R-dexplicit} then becomes
\begin{equation}
\label{e:R-dexplicitZBC}
\begin{aligned}
R_{-,1,1}(\zeta) & = \begin{cases}\displaystyle
(-k\varrho_{-,1,1} + \ii\epsilon\varrho_{-,3,3})/(k^2 + \epsilon^2)\,, & k\in\Complex^+\,,\\
\displaystyle
 - k\varrho_{-,1,1}/(k^2 + \epsilon ^2)\,, & k\in\Real\,,\\
 \displaystyle
- \varrho_{-,1,1}/(k - \ii \epsilon)\,, & k\in\Complex^-\,,
\end{cases}\\
R_{-,2,2}(\zeta) & = \begin{cases}\displaystyle
-\varrho_{-,2,2}/(k + \ii \epsilon)\,, & k\in\Complex^+\,,\\
\displaystyle
 -k\varrho_{-,2,2}/(k^2 + \epsilon ^2)\,, & k\in\Real\,,\\
 \displaystyle
- \varrho_{-,2,2}/(k - \ii \epsilon)\,, & k\in\Complex^-\,,
\end{cases}\\
R_{-,3,3}(\zeta) & = \begin{cases}\displaystyle
(-k\varrho_{-,3,3} + \ii\epsilon\varrho_{-,1,1})/(k^2+\epsilon^2)\,, & k\in\Complex^+\,,\\
\displaystyle
 - k\varrho_{-,3,3}/(k^2 + \epsilon ^2)\,, & k\in\Real\,,\\
 \displaystyle
- \varrho_{-,3,3}/(k - \ii \epsilon)\,, & k\in\Complex^-\,.
\end{cases}
\end{aligned}
\end{equation}
Consequently, we take the limit $\epsilon\to0$ and obtain
\begin{equation}
\label{e:R-dZBCSL}
\R_{-,\d}(\zeta) = -\bvarrho_{-,\d}/k\,,\qquad
\zeta\in\Complex\,.
\end{equation}

$\bullet$
\textit{Sharp-line limit first.}
One the other hand,
after taking the limit $\epsilon\to0$ first,
the quantities in Equation~\eqref{e:R-dexplicit} become
\begin{equation}
\label{e:R-dSLZBC}
R_{-,1,1}(\zeta) = - \rho_-^+/k\,,\qquad
R_{-,2,2}(\zeta) = - \varrho_{-,2,2}/k\,,\qquad
R_{-,3,3}(\zeta) = - \rho_-^+/k\,,
\end{equation}
where $\zeta\in\Complex\backslash\Real$.
Evidently, $\R_{-,\d}$ is independent of $E_0$,
so the limit $E_0\to0$ does not affect the matrix $\R_{-,\d}$.

As a result,
by comparing the two matrices $\R_{-,\d}$ from the two cases~\eqref{e:R-dZBCSL} and~\eqref{e:R-dSLZBC},
we conclude that the two limits $\epsilon\to0$ and $E_0\to0$ \textit{do not commute}.

\end{proof}

\subsection{Solution parameters}

We point out that all solutions discussed in Section~\ref{s:generalsoliton} contain large numbers of free parameters,
including the discrete eigenvalue, the norming constant, the initial state of the density matrix, the background optical field, and more.
Hence, it is impossible to present solutions covering all possible combinations of parameters due to the space limitation of the manuscript.
Thus, in this section, we only consider solutions with a fixed subset of parameters, including a fixed background amplitude, a fixed width of the spectral line, and a fixed initial state of the norming constant:
\begin{equation}
E_0 = 1,\qquad
\epsilon = 2,\qquad
\xi(0) = \psi(0) = 0.
\end{equation}
We then vary other, more illustrative parameters, as shown in Table~\ref{tab:numerics}.

\begin{table}[h]
\caption{Six settings of parameters for all solutions.}
\begin{tabular}{|c|c|c|c|c|c|c|c|c|c|}
\hline
Setting & $\vphantom{\Big|}\varrho_{-1,1}$ & $\varrho_{-2,2}$ & $\varrho_{-3,3}$ & $\alpha$ & Discrete Eigenvalue \\
\hline\hline
(a) & $\vphantom{\Big|}0$  & $0.4$ & $0.6$ & $\pi/8$ & $2\ii$ \\
\hline
(A) & $\vphantom{\Big|}0$  & $0.4$ & $0.6$ & $\pi/8$ & $1+\ii$ \\
\hline
(b) & $\vphantom{\Big|}0.2$  & $0.1$ & $0.7$ & $\pi/4$ & $2\ii$ \\
\hline
(B) & $\vphantom{\Big|}0.2$  & $0.1$ & $0.7$ & $\pi/4$ & $1+\ii$ \\
\hline
(c) & $\vphantom{\Big|}0.7$ & $0.2$ & $0.1$ & $3\pi/8$ & $2\ii$ \\
\hline
(C) & $\vphantom{\Big|}0.7$ & $0.2$ & $0.1$ & $3\pi/8$ & $1+\ii$ \\
\hline
\end{tabular}
\label{tab:numerics}
\end{table}

The six parameter settings cover different situations based on:
i) the distribution of the background amplitude $E_0$ in the two electric-field envelope components ($E_0\cos\alpha$ versus $E_0\sin\alpha$);
ii) the initial state of the medium $D_{-,j}$ from $\varrho_{-,j,j}$ via Equation~\eqref{e:rho-D} [uninverted (a/A) versus partially inverted (b/B) versus fully inverted (c/C)];
iii) the location of the discrete eigenvalue [purely imaginary (a/b/c) versus generic complex number (A/B/C)].

In particular, the initial medium states are complicated due to the $k$ dependence. The graphs of $D_{-,j}$ are shown in Figure~\ref{f:initialstate}.
Recall that $D_{-,1}$ denotes the initial population in the excited state,
and $D_{-,2}$ and $D_{-,3}$ denote the initial populations in the ground states.

\begin{figure}[h]
\centering
\includegraphics[scale=0.42]{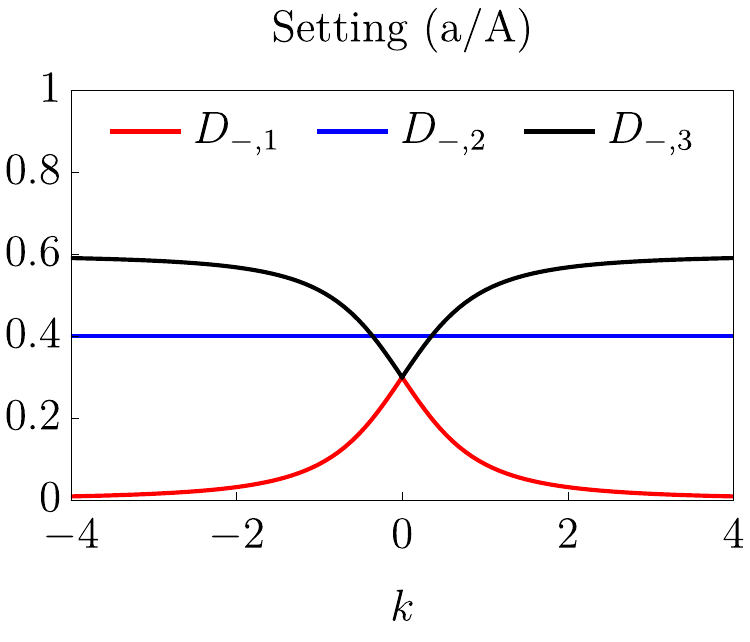}
\includegraphics[scale=0.42]{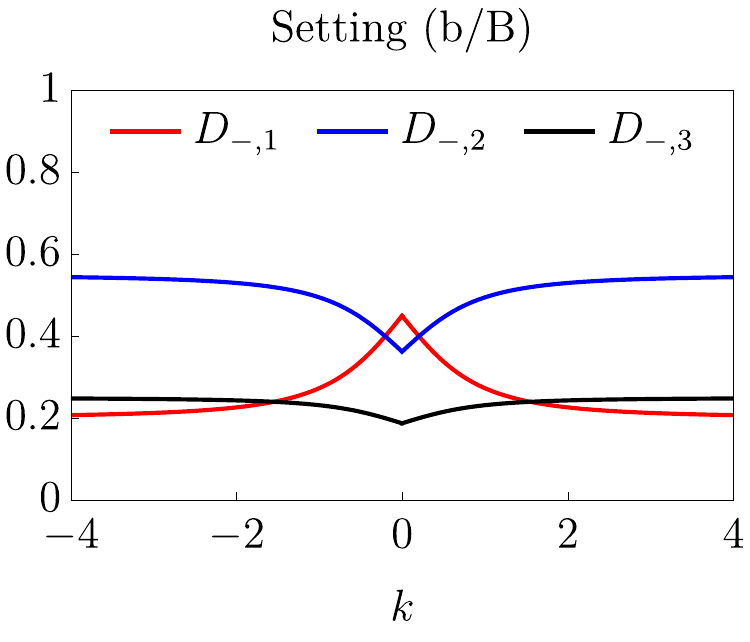}
\includegraphics[scale=0.42]{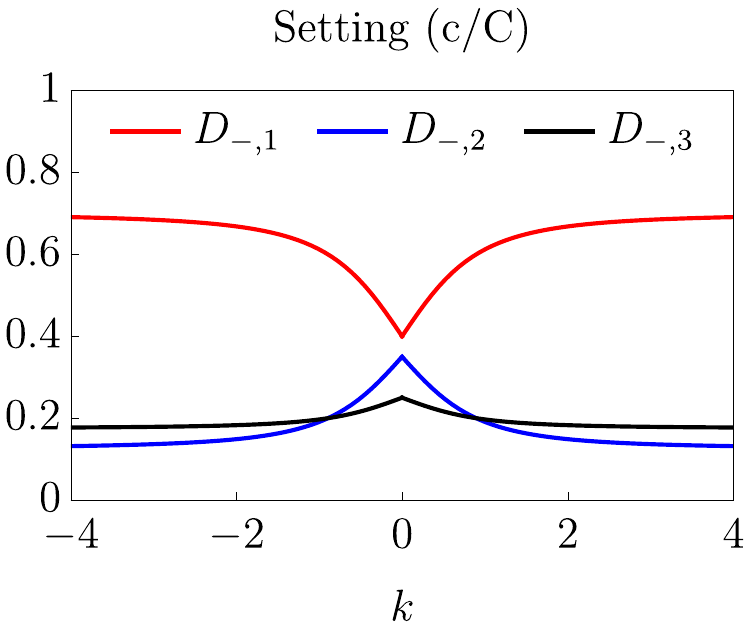}
\caption{The initial medium states in the six settings.
Recall that $D_{-,1}$ is the initial population of the excited state while $D_{-,2}$ and $D_{-,3}$ are those for the two ground states.
Settings (a/A) denote an uninverted medium as $D_{-,1}$ is always less than $D_{-,2}$ and $D_{-,3}$.
Settings (b/B) denote a partially inverted medium because $D_{-,1}$ is less/greater than $D_{-,2}$ and $D_{-,3}$ depending on the value of $k$.
Settings (c/C) denote a fully inverted medium since $D_{-,1}$ is always greater than $D_{-,2}$ and $D_{-,3}$.
}
\label{f:initialstate}
\end{figure}

\subsection{Solitary, periodic and rational solutions}

Since we now know the spectral-line shape $g(k)$ explicitly,
we are able to present all the formulas for the solutions and to investigate their properties.

\subsubsection{One-Soliton solutions: type I}

Equation~\eqref{e:normingtimeevolution} indicates that we need to substitute $w_1^*$ into $\R_{-,\d}(\zeta)$.
We therefore define the following quantity for $\zeta\in D_2$:
\begin{equation}
\label{e:RI}
R_\I(\zeta) \coloneq \frac{\ii}{2}(R_{-,1,1}(\zeta) - R_{-,3,3}(\zeta))
= \ii\varrho_-^- \,g(k)
\left[\log \left(\frac{E_0 + \lambda }{E_0 - \lambda}\right) -
\frac{\lambda}{\sqrt{E_0^2-\epsilon^2}}
\log \left(\frac{E_0 + \sqrt{E_0^2-\epsilon^2}}{E_0-\sqrt{E _o^2-\epsilon ^2}}\right) \right]\,.
\end{equation}
Then, by Equation~\eqref{e:normingtimeevolution}, we know that
\begin{equation}
\nonumber
\xi_\I(z) = \Re(R_\I(w_1^*))z + \xi_\I(0)\,,\qquad
\psi_\I(z) = \Im(R_\I(w_1^*))z + \psi_\I(0)\,.
\end{equation}
Substituting the above quantities into the soliton formula~\eqref{e:soliton1},
we obtain type I soliton completely.

Six type-I soliton solutions are shown in Figure~\ref{f:solitontypeI1}.
In particular, the top row includes solutions with settings (a) and (A);
the middle row contains solutions with settings (b) and (B);
and the bottom row contains solutions with settings (c) and (C).
The left two columns in Figure~\ref{f:solitontypeI1} contain settings (a/b/c),
whereas the right two columns contain (A/B/C).
Furthermore, the first and third columns show plots of $|E_1(t,z)|$ and the other two columns show plots of $|E_2(t,z)|$.

It is easy to observe that the solitons presented in Figure~\ref{f:solitontypeI1} are subluminal in the uninverted medium, i.e., settings (a/A), and in the partially inverted medium, i.e., settings (b/B),
but are superluminal in the fully inverted medium, i.e., settings (c/C).
It thus appears that the partially inverted case is similar to the uninverted case.
However, as we will see below, this rule does not hold for all cases of solitons.

For a purely imaginary discrete eigenvalue, i.e., settings (a/b/c), the solions are traveling waves, but for a general complex value, i.e., settings (A/B/C), the solitons exhibit internal oscillations.
For all cases, the group velocities are determined by Equation~\eqref{e:soliton-typeI-velocity},
while the phase velocities in settings (A/B/C) seemingly coincide with the velocities of the traveling waves in settings (a/b/c).

\begin{remark}
It is worth pointing out that type-I solitons are similar to their counterparts in the two-level case~\cite{bgkl2019},
because these types of solitons share similar structure and properties between the two Maxwell-Bloch systems with a nonzero background and with inhomogeneous broadening of the spectral line.
\end{remark}

\begin{figure}[t!]
\centering
\includegraphics[width=0.24\textwidth]{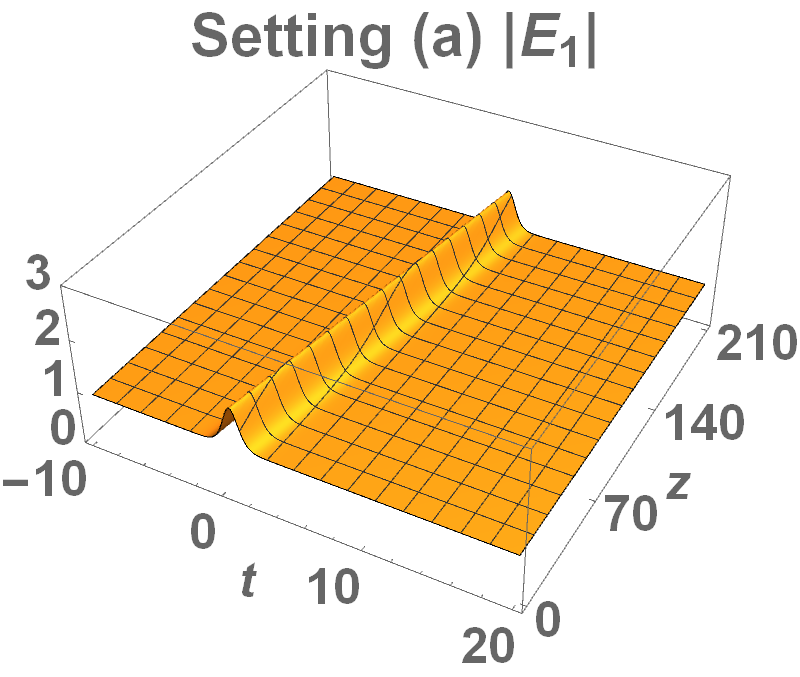}
\includegraphics[width=0.24\textwidth]{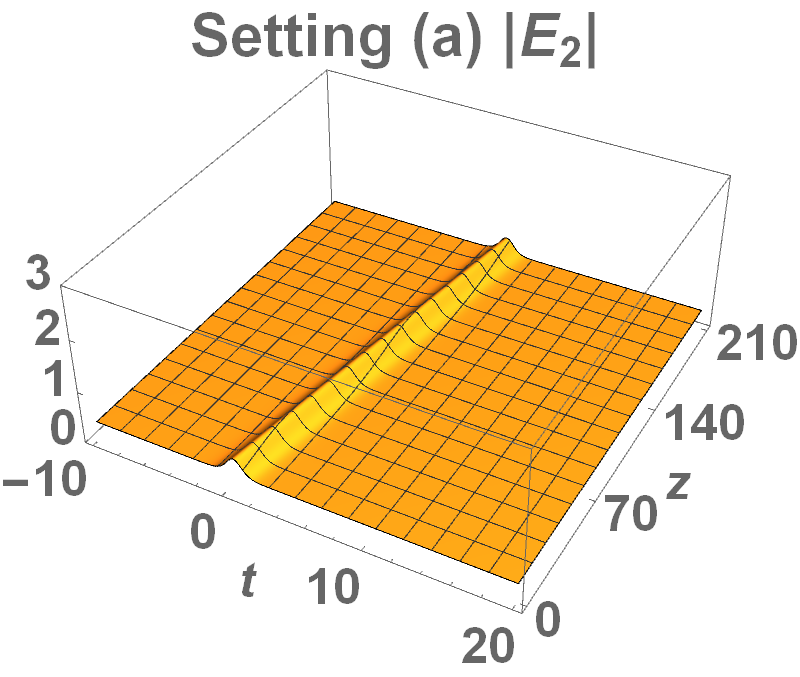}\quad
\includegraphics[width=0.24\textwidth]{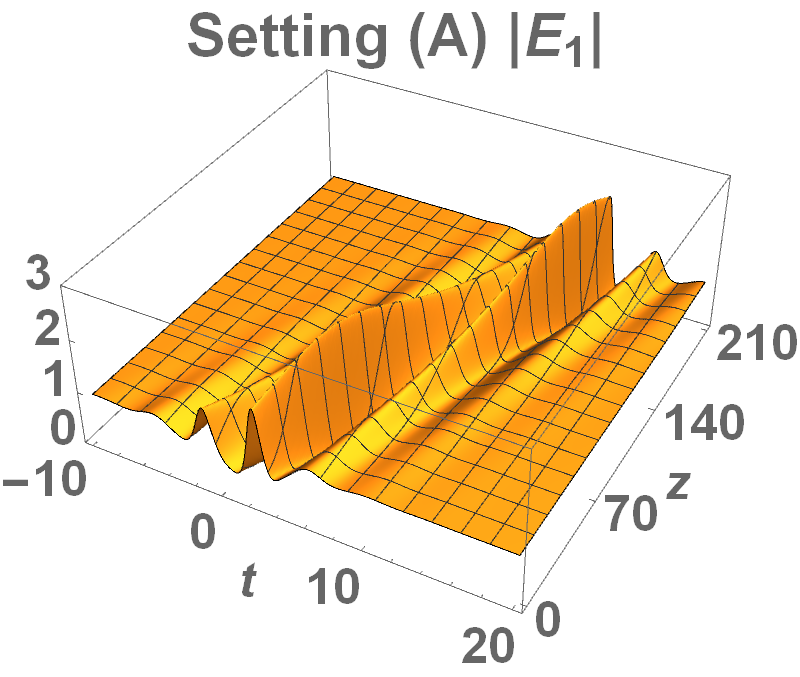}
\includegraphics[width=0.24\textwidth]{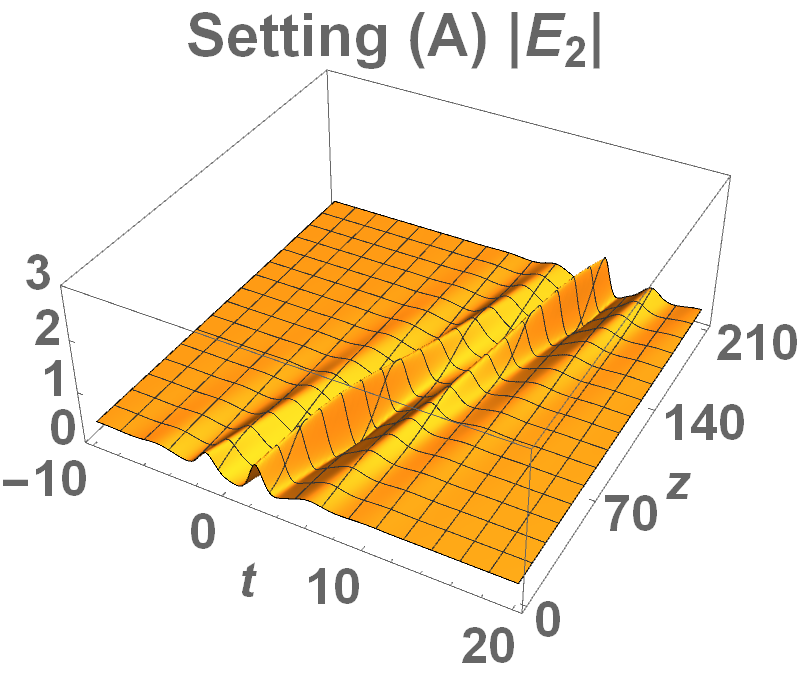}\\[2ex]
\includegraphics[width=0.24\textwidth]{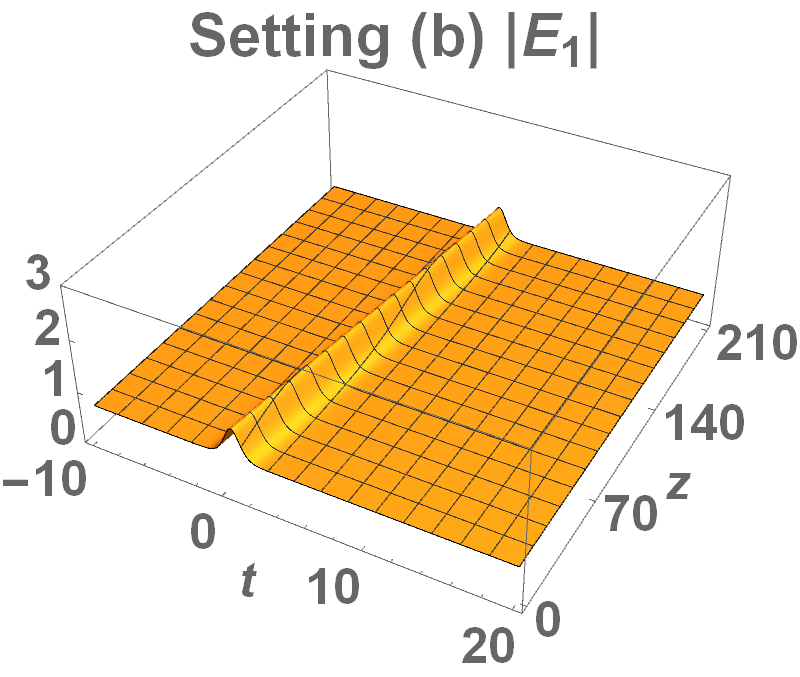}
\includegraphics[width=0.24\textwidth]{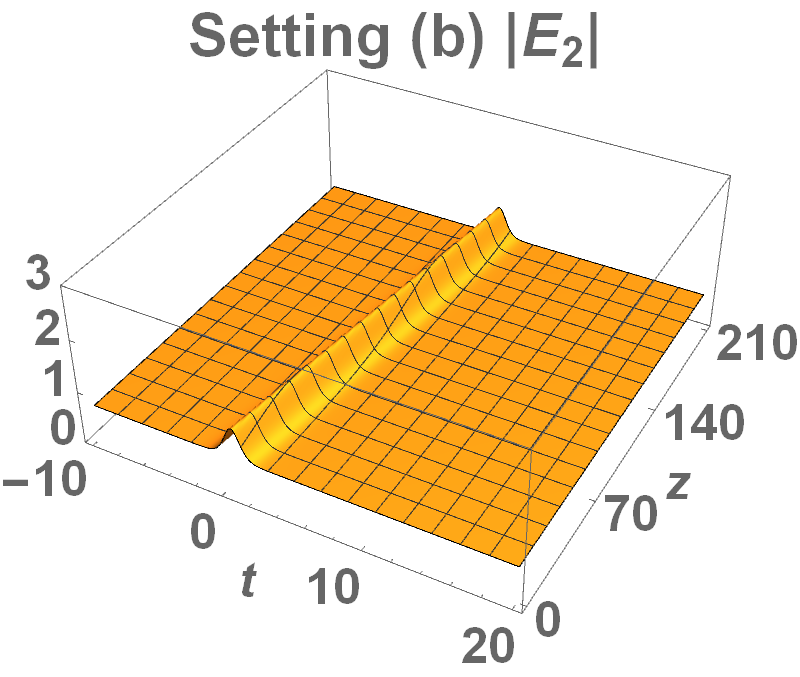}\quad
\includegraphics[width=0.24\textwidth]{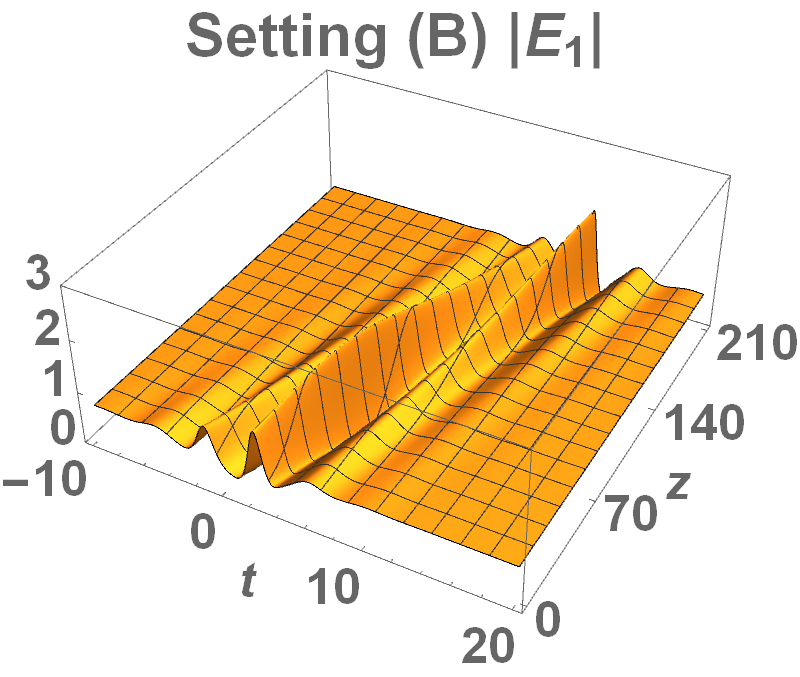}
\includegraphics[width=0.24\textwidth]{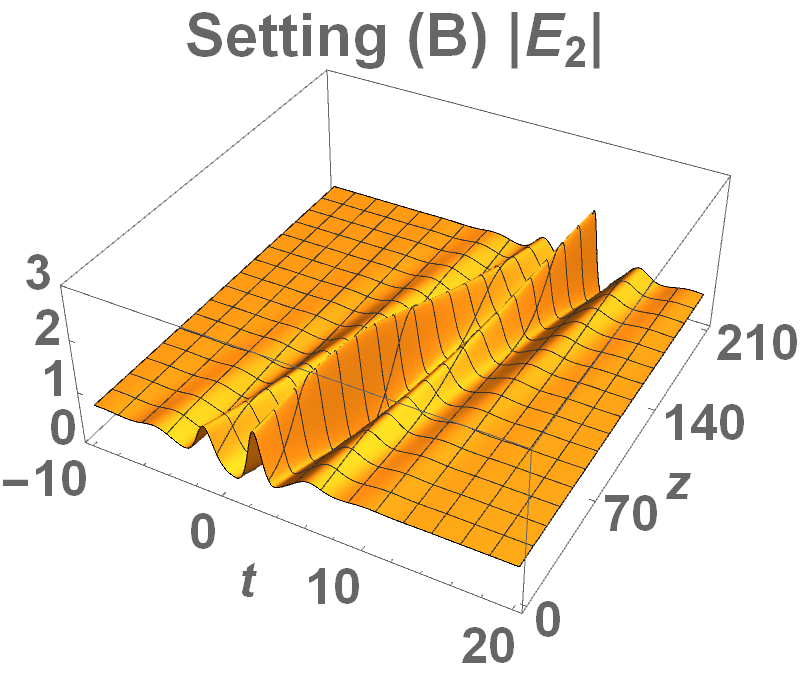}\\[2ex]
\includegraphics[width=0.24\textwidth]{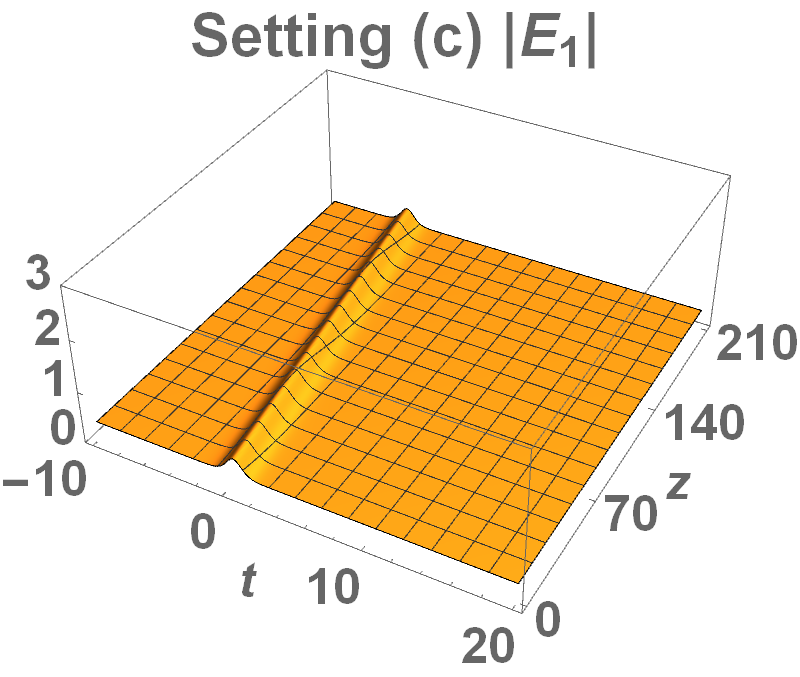}
\includegraphics[width=0.24\textwidth]{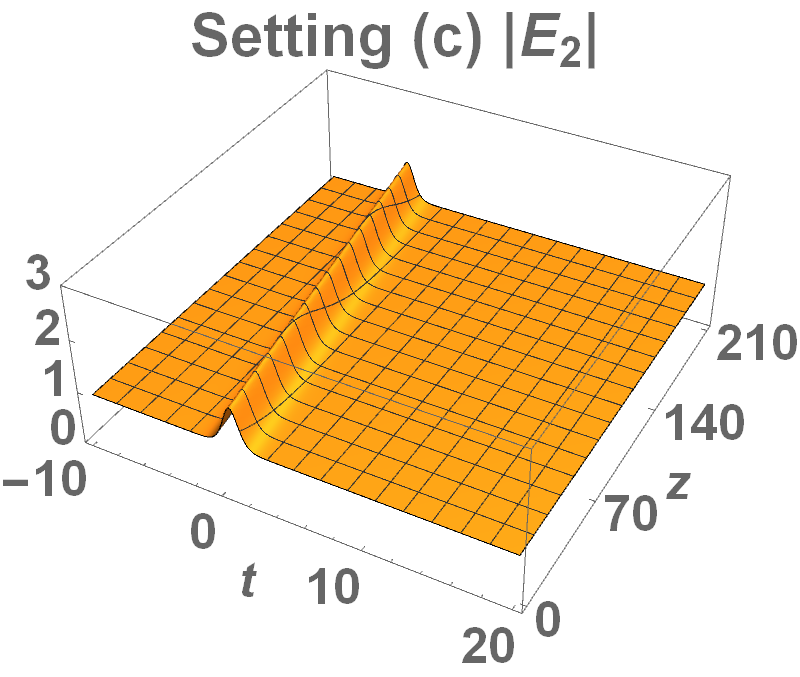}\quad
\includegraphics[width=0.24\textwidth]{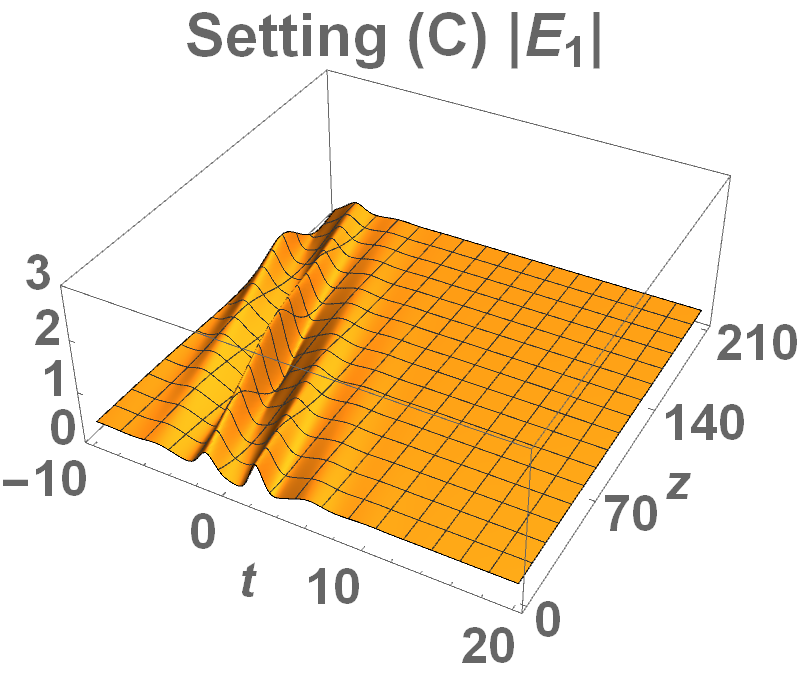}
\includegraphics[width=0.24\textwidth]{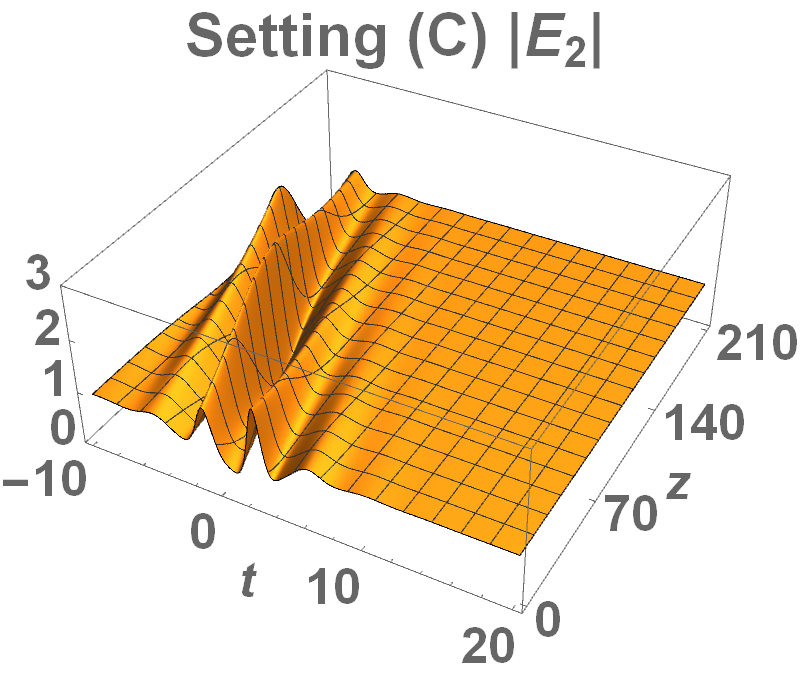}
\caption{Type I one-soliton solutions with settings (a/A).
First column: $|E_1(t,z)|$ with setting (a/b/c).
Second column: $|E_2(t,z)|$ with setting (a/b/c).
Third column: $|E_1(t,z)|$ with setting (A/B/C).
Last column: $|E_2(t,z)|$ with setting (A/B/C).
First row: solitons with settings (a/A).
Second row: solitons with settings (b/B).
Bottom row: solitons with settings (c/C).
The density matrix $\brho(t,z,\zeta)$ is omitted due to space constraint.}
\label{f:solitontypeI1}
\end{figure}

\subsubsection{Periodic solutions}
\label{s:periodic-ih}

Recall that one is able to derive periodic solutions from type-I solitons,
as shown in Section~\ref{s:periodic-solution}.
With the known spectral-line shape~\eqref{e:Lorentzian}, two such solutions are shown in Figure~\ref{f:periodic1}, where settings (A) and (C) are used, except that the discrete eigenvalues are chosen as $\e^{\ii\pi/4}$ on $\Sigma_\circ$. We call them settings (A') and (C').
The periodic solution with the setting (B'), which has the same parameters as setting (B) but a discrete eigenvalue $\e^{\ii\pi/4}$, is omitted for brevity, as it mimics the one with the setting (A').

We can see that setting (A'), corresponding to an uninverted medium, comprises waves traveling slower than the speed of light, but setting (C'), corresponding to a fully inverted medium, comprises waves traveling faster than the speed of light, indicating essential differences based on the initial state of the medium.
The solutions also show both temporal and spatial periodicity. Hence, these types of solutions exist inside the medium forever, and thus can be regarded as a type of background, just as the elliptic solutions of the focusing nonlinear Schr\"{o}dinger equation.

\begin{figure}[t!]
\centering
\includegraphics[width=0.24\textwidth]{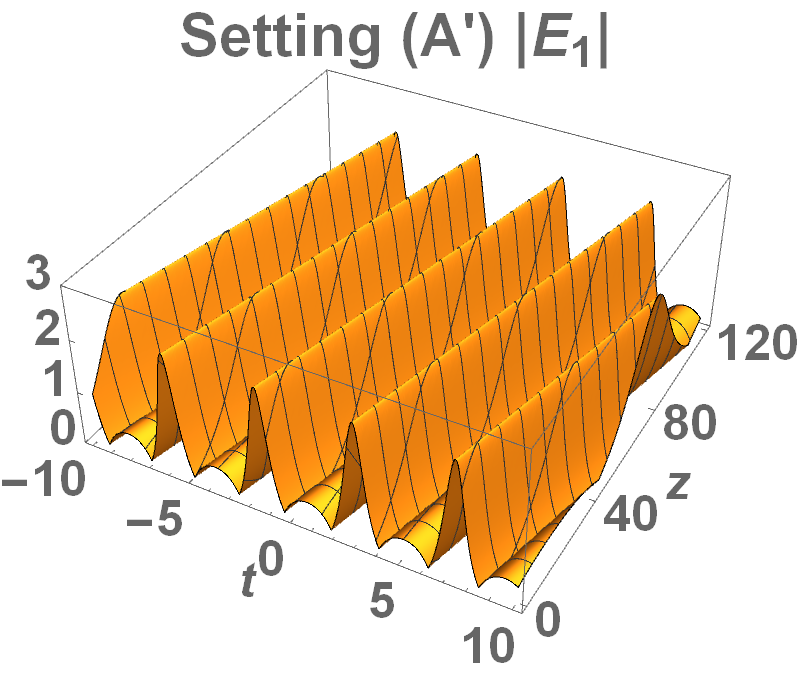}
\includegraphics[width=0.24\textwidth]{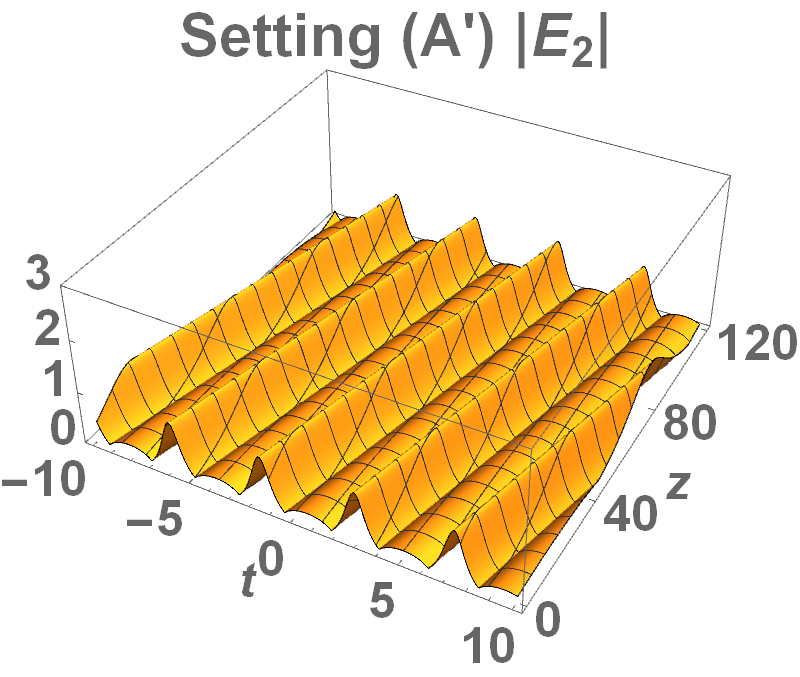}\quad
\includegraphics[width=0.24\textwidth]{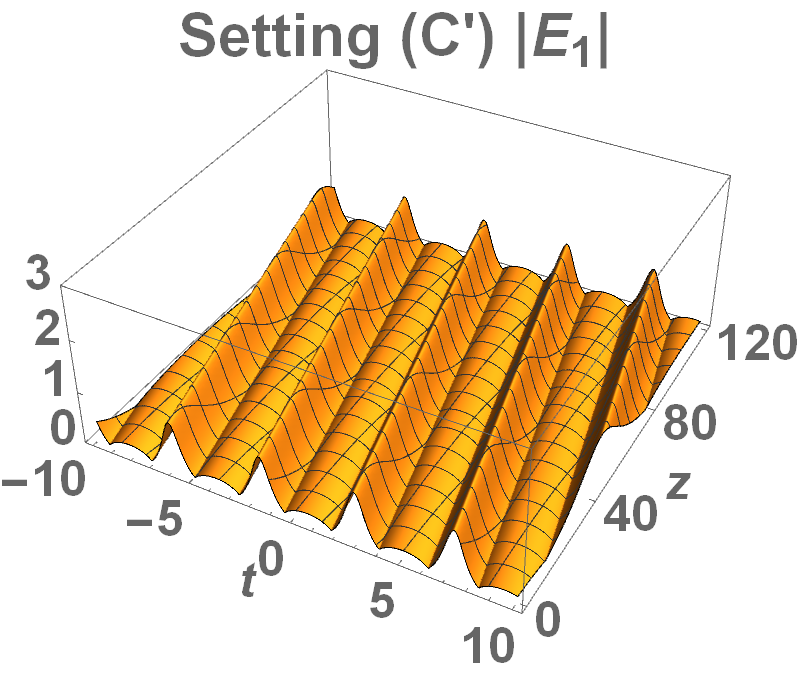}
\includegraphics[width=0.24\textwidth]{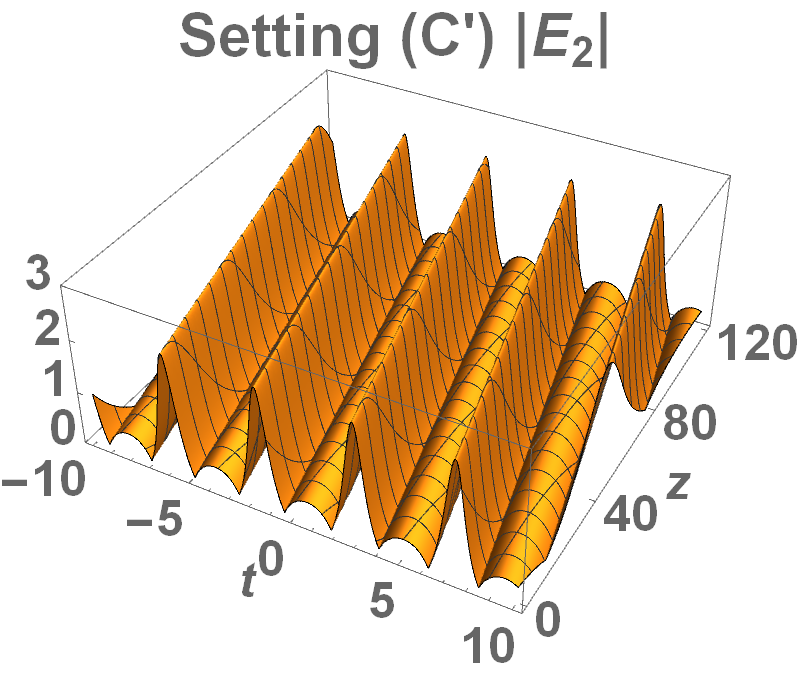}
\caption{Two periodic solutions with settings (A) and (C),
except that the discrete eigenvalues are chosen as $\e^{\ii\pi/4}$ on $\Sigma_\circ$. Therefore, the two cases are called settings (A') and (C').
The solution with setting (B), in which the discrete eigenvalue is similarly changed to $\e^{\ii\pi/4}$, exhibits rather similar behavior to the one with setting (A'), and is omitted here for brevity.
}
\label{f:periodic1}
\end{figure}

\subsubsection{Rational solutions}
\label{s:rational-ih}

We can also derive rational solutions from type-I soliton solutions as shown in Section~~\ref{s:rational-solution}.
Two such solutions are shown in Figure~\ref{f:rational1}.

Similarly to the classic two-level case~\cite{bgkl2019}, the rational solutions of CMBE exhibit traveling-wave instead of isolated rogue-wave structure.
For an initially uninverted medium, such as the setting (a'') in Figure~\ref{f:rational1}, the solution travels more slowly than the speed of light,
whereas for an initially fully inverted medium, such as the setting (c'') in Figure~\ref{f:rational1}, the solution travels faster than the speed of light.
The rational solution in a partially inverted medium (setting (b'')) also travels more slowly than the speed of light, and so is omitted.

It is also worth noting that rational solutions of CMBE look quite similar to the ones in the classic two-level case, so the solutions here can be regarded as a ``vectorization" of the rational solutions discussed in~\cite{bgkl2019}. However, note that another famous rational solution, but of the focusing nonlinear Schr\"{o}dinger equation, i.e., the Peregrine soliton, is an isolated rogue wave. This shows a fundamental difference between the focusing nonlinear Schr\"{o}dinger equation and the Maxwell-Bloch system. It also implies that rational solutions are not necessarily equivalent to rogue waves in a general integrable system.

\begin{figure}[t!]
\centering
\includegraphics[width=0.24\textwidth]{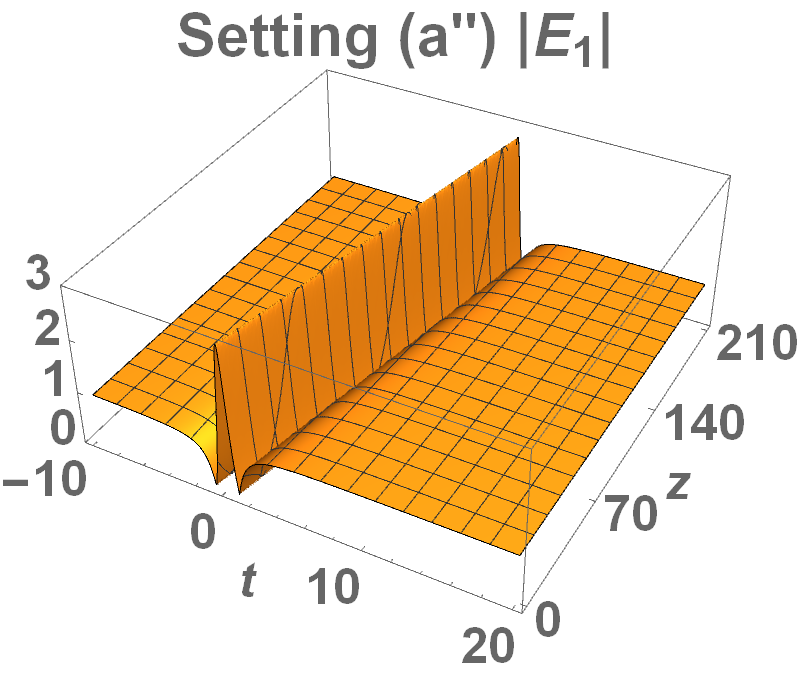}
\includegraphics[width=0.24\textwidth]{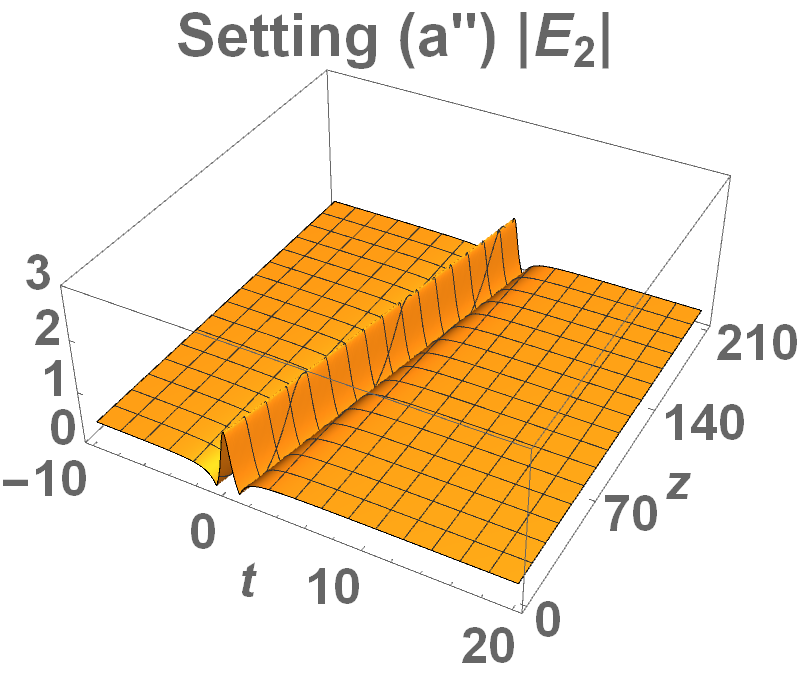}\quad
\includegraphics[width=0.24\textwidth]{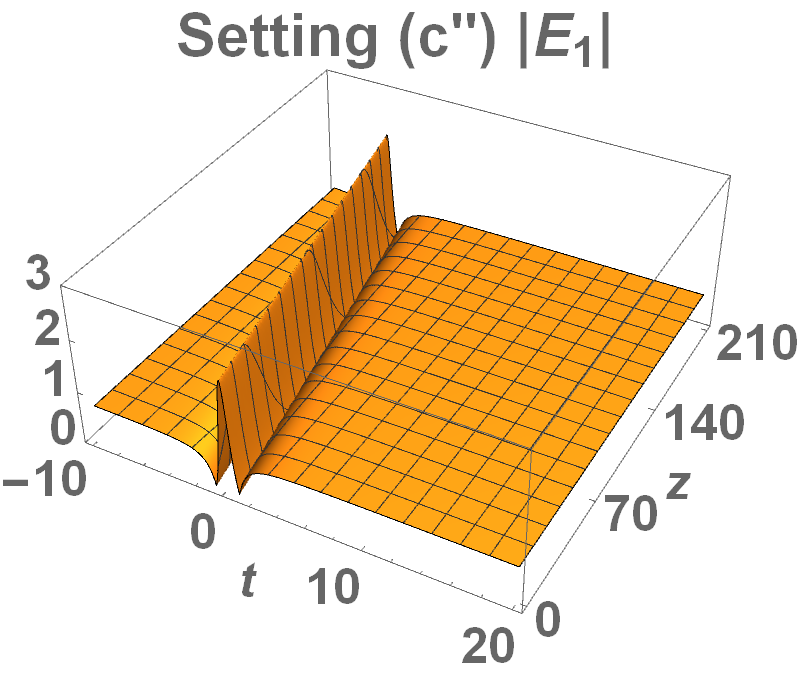}
\includegraphics[width=0.24\textwidth]{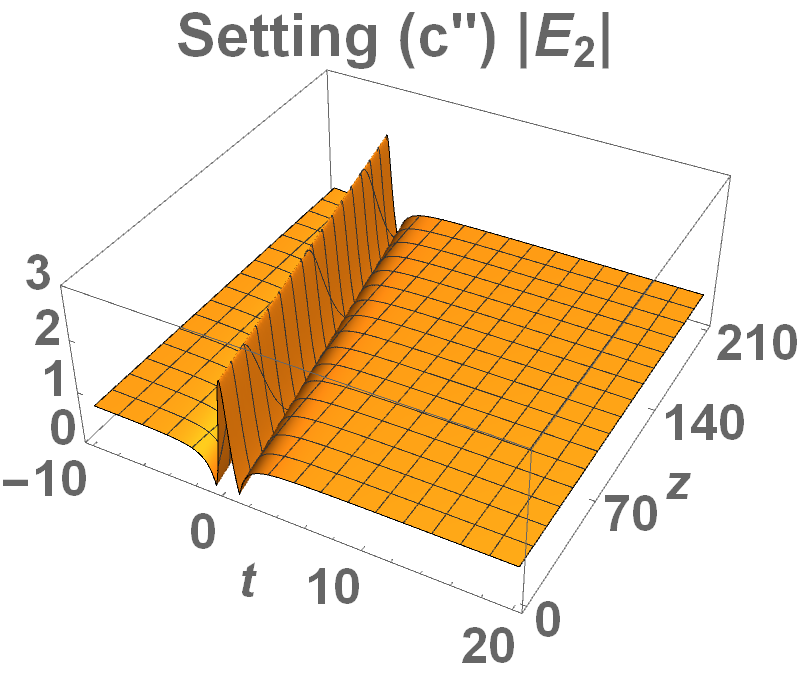}
\caption{Two rational solutions with settings (a) and (c),
except that the discrete eigenvalues are chosen as $\ii$. Therefore, the two cases are called settings (a'') and (c'').
Similarly to what happens to the periodic solutions in Figure~\ref{f:periodic1}, the setting analogous to the (b) case is omitted, because it is quite similar to (a'').
}
\label{f:rational1}
\end{figure}

\subsubsection{One-Soliton solutions: type II}

In this case, Equation~\eqref{e:normingtimeevolution} indicates that we need to substitute $z_1$ into $\R_{-,\d}(\zeta)$ and, therefore,
we define the following expression for $\zeta\in D_1$:
\begin{multline}
\label{e:RII}
R_\II(\zeta) \coloneq \frac{\ii}{2}(R_{-,2,2}(\zeta) - R_{-,3,3}(\zeta))  \\
= \frac{\ii}{4}\frac{\varrho_{-,1,1}-2\varrho_{-,2,2}+\varrho_{-,3,3}}{k+\ii\epsilon} - \frac{\ii}{2} \varrho_-^- \,g(k)
\left[ \log \left(\frac{E_0+\lambda }{E_0 - \lambda}\right) +
\frac{\lambda}{\sqrt{E_0^2-\epsilon^2}}
\log \left(\frac{E_0 + \sqrt{E_0^2-\epsilon^2}}{E_0-\sqrt{E _o^2-\epsilon ^2}}\right) \right]\,.
\end{multline}
Equation~\eqref{e:normingtimeevolution} yields
\begin{equation}
\xi_\II(z) = \Re(R_\II(z_1))z + \xi_\II(0)\,,\qquad
\psi_\II(z) = \Im(R_\II(z_1)z + \psi_\II(0)\,.
\end{equation}
The above equations govern the propagation of type II soliton~\eqref{e:soliton2}.
Examples are shown in Figures~\ref{f:solitontypeII1}, which contains all six settings from Table~\ref{tab:numerics}.
Similarly to the solitons of type I, type-II solitons with a purely imaginary discrete eigenvalue are traveling waves, illustrated in settings (a/b/c) in the first two columns. For a general complex value, i.e., settings (A/B/C), the solitons exhibit internal oscillations, and can be called breathers.

Moreover, if the medium is initially in an uninverted/fully inverted medium, the soliton travels more slowly/faster than the speed of light, respectively, just as type-I solitons discussed in previous sections.
However, when the medium initially is in a partially inverted medium, i.e., settings (b/B), it is evident that type-II solitons are traveling superluminally, contrary to the ones in Figure~\ref{f:solitontypeI1}.
Thus, we conclude that type-I and -II solitons differ in an essential way. 

\begin{figure}[t!]
\centering
\includegraphics[width=0.24\textwidth]{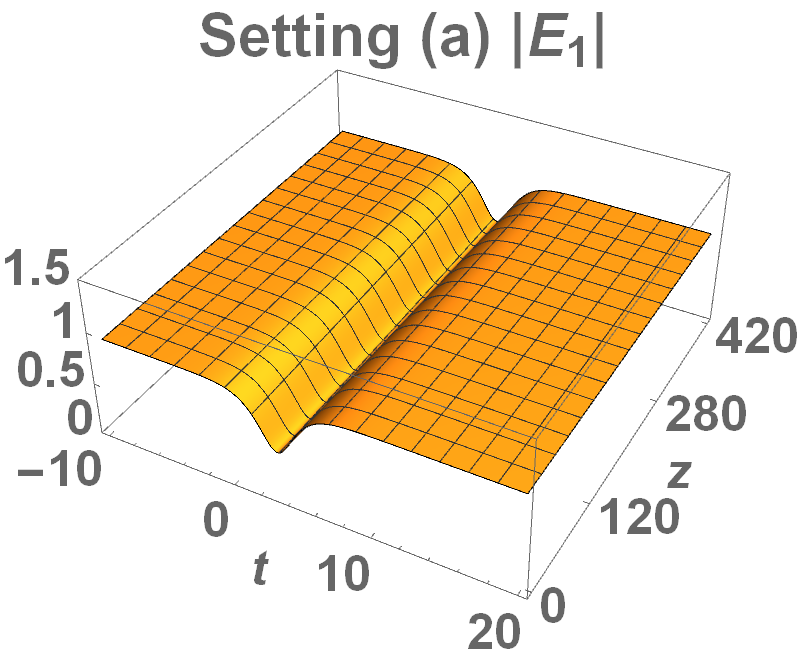}
\includegraphics[width=0.24\textwidth]{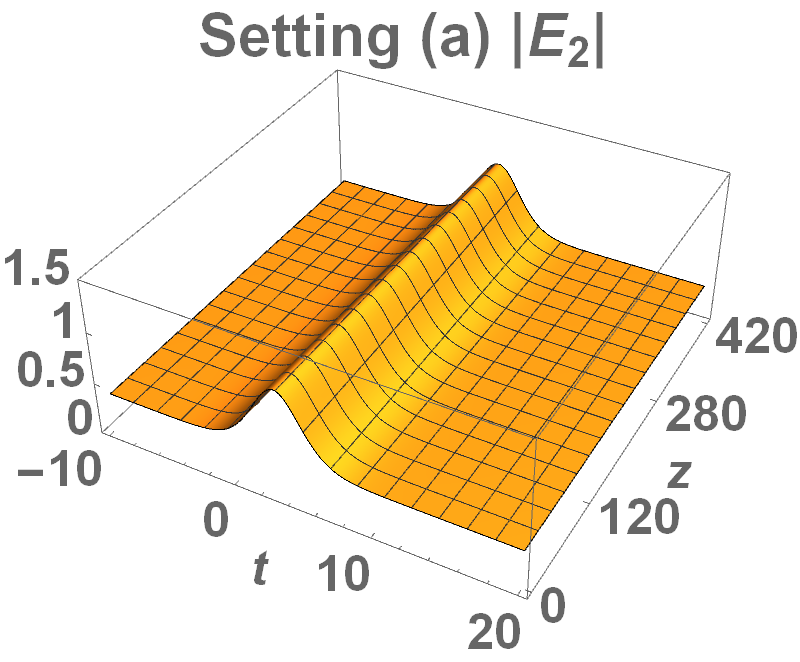}\quad
\includegraphics[width=0.24\textwidth]{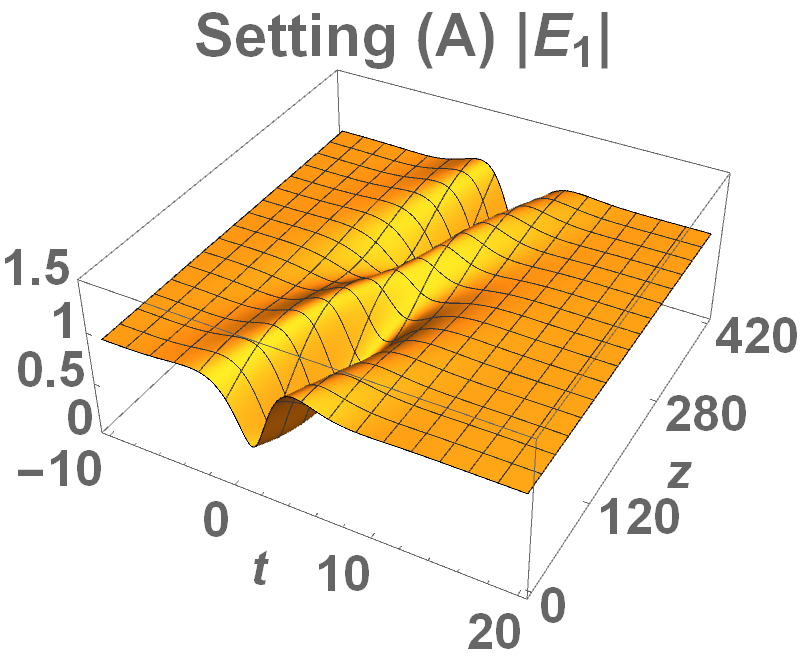}
\includegraphics[width=0.24\textwidth]{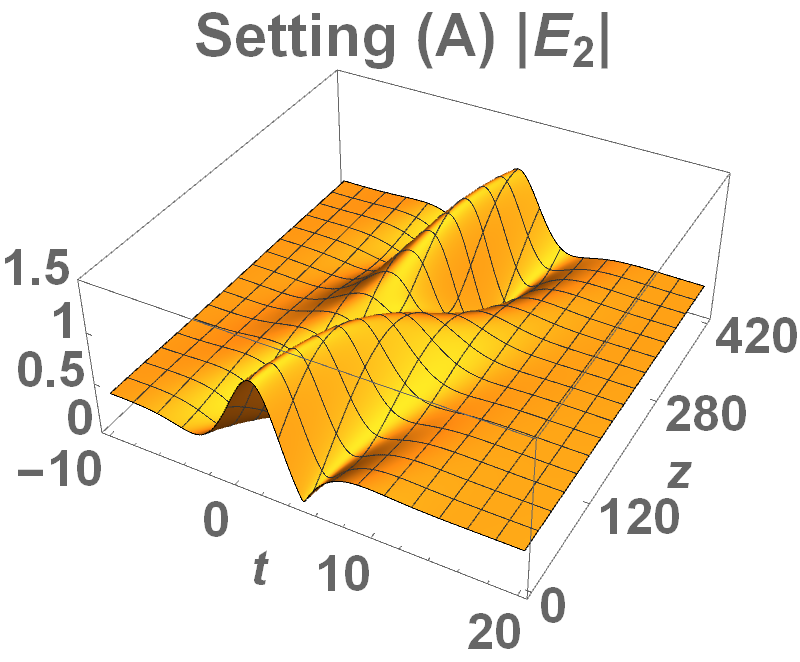}\\[2ex]
\includegraphics[width=0.24\textwidth]{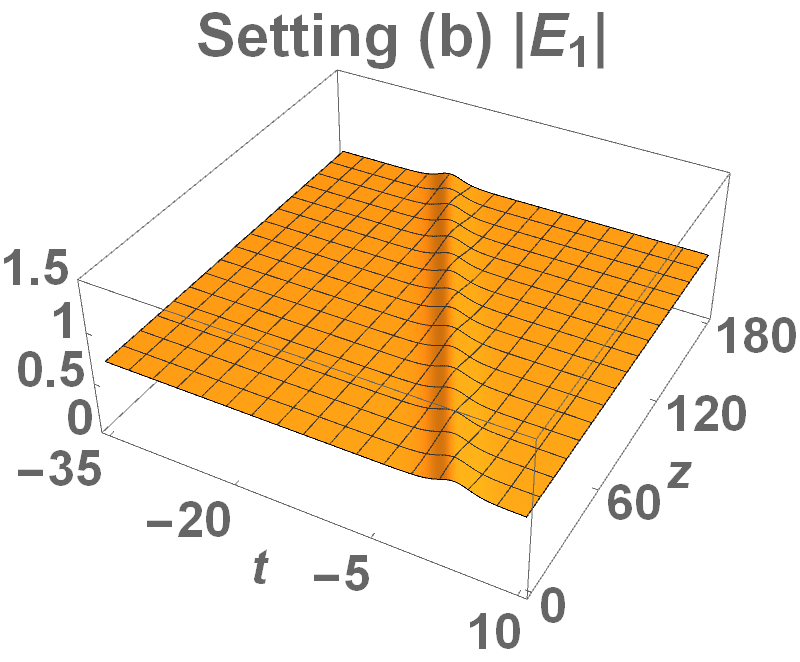}
\includegraphics[width=0.24\textwidth]{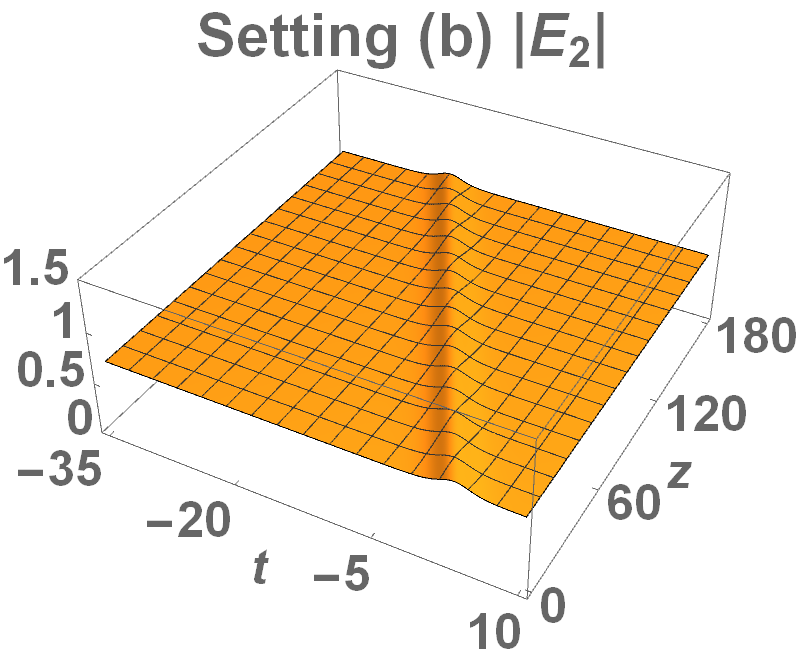}\quad
\includegraphics[width=0.24\textwidth]{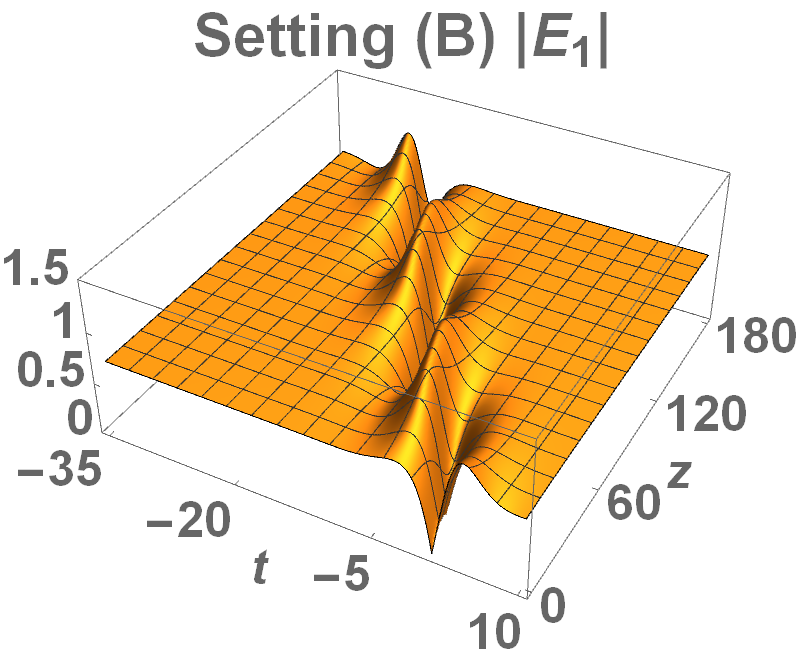}
\includegraphics[width=0.24\textwidth]{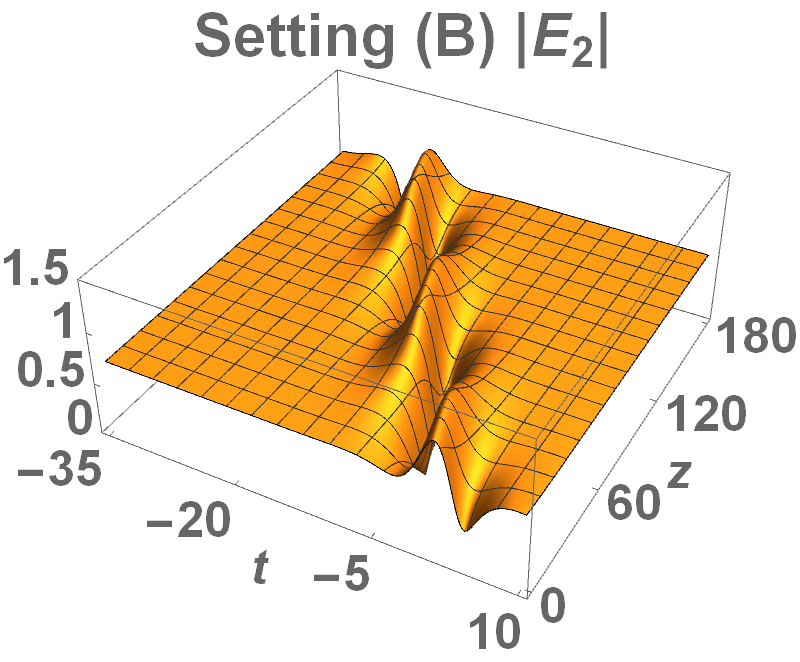}\\[2ex]
\includegraphics[width=0.24\textwidth]{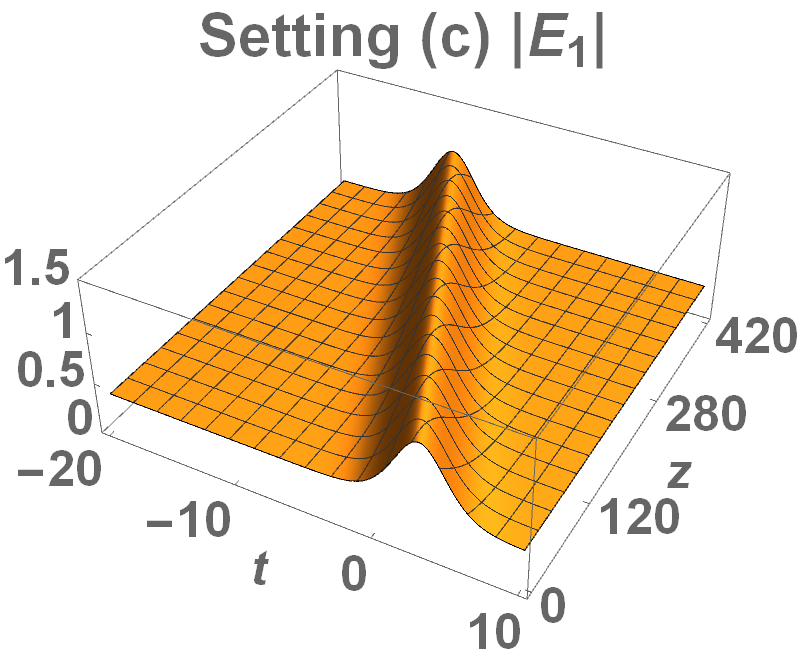}
\includegraphics[width=0.24\textwidth]{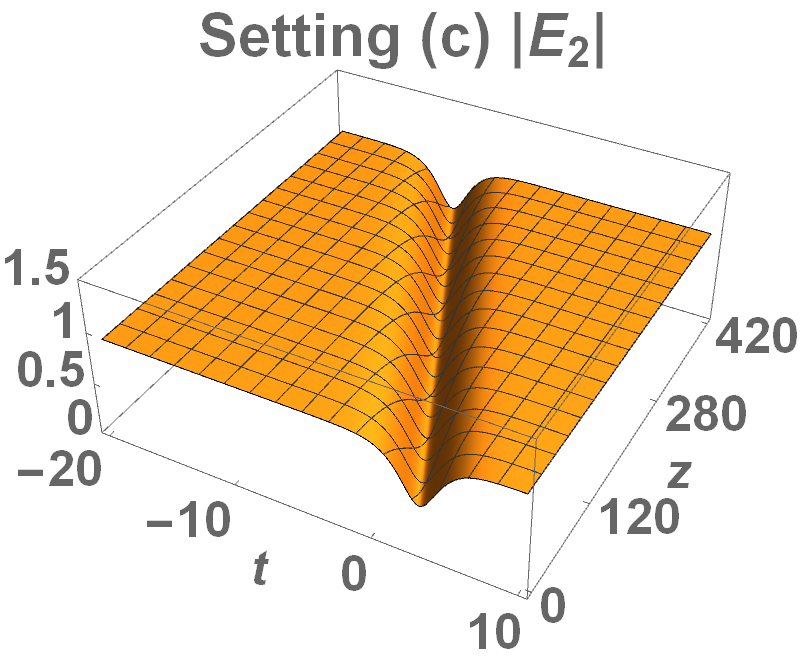}\quad
\includegraphics[width=0.24\textwidth]{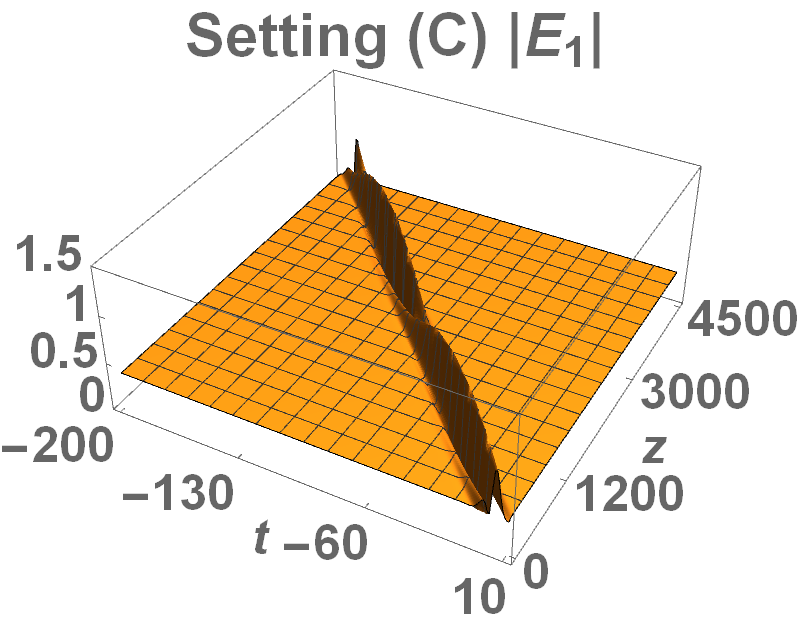}
\includegraphics[width=0.24\textwidth]{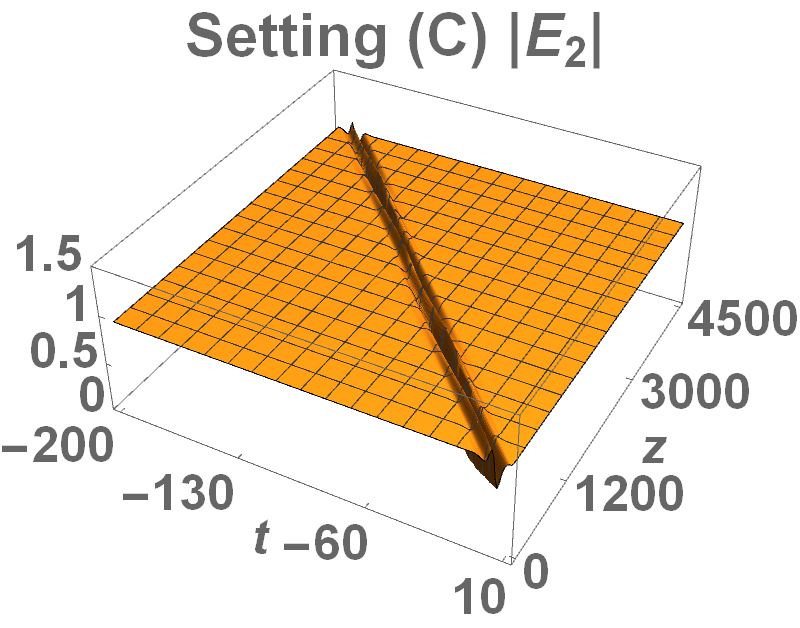}
\caption{Similar situation as in Figure~\ref{f:solitontypeI1}, but with type II solitons of settings (a/A/b/B/c/C). The density matrix $\brho(t,z,\zeta)$ is omitted due to space constraint.}
\label{f:solitontypeII1}
\end{figure}

\subsubsection{One-Soliton solutions: type III}

In this case, Equation~\eqref{e:normingtimeevolution} indicates that we need to substitute $\zeta_1^*$ into $\R_{-,\d}(\zeta)$ and therefore,
we define the following quantity for $\zeta\in D_2$:
\begin{multline}
\label{e:RIII}
R_\III(\zeta) \coloneq \frac{\ii}{2}(R_{-,1,1}(\zeta) - R_{-,2,2}(\zeta))  \\
= -\frac{\ii}{4}\frac{\varrho_{-,1,1}-2\varrho_{-,2,2}+\varrho_{-,3,3}}{k-\ii\epsilon} - \frac{\ii}{2} \varrho_-^- \,g(k)
\left[ \log \left(\frac{E_0+\lambda }{E_0 - \lambda}\right) +
\frac{\lambda}{\sqrt{E_0^2-\epsilon^2}}
\log \left(\frac{E_0 + \sqrt{E_0^2-\epsilon^2}}{E_0-\sqrt{E _o^2-\epsilon ^2}}\right) \right]\,.
\end{multline}

Six examples are shown in Figure~\ref{f:solitontypeIII1}.
It is evident that type-III solitons are somewhat similar to those of type II. They exhibit oscillatory internal structure in general and are sensitive to the initial state of the medium, because they only travel subluminally in an initially uninverted medium.

\begin{figure}[t!]
\centering
\includegraphics[width=0.24\textwidth]{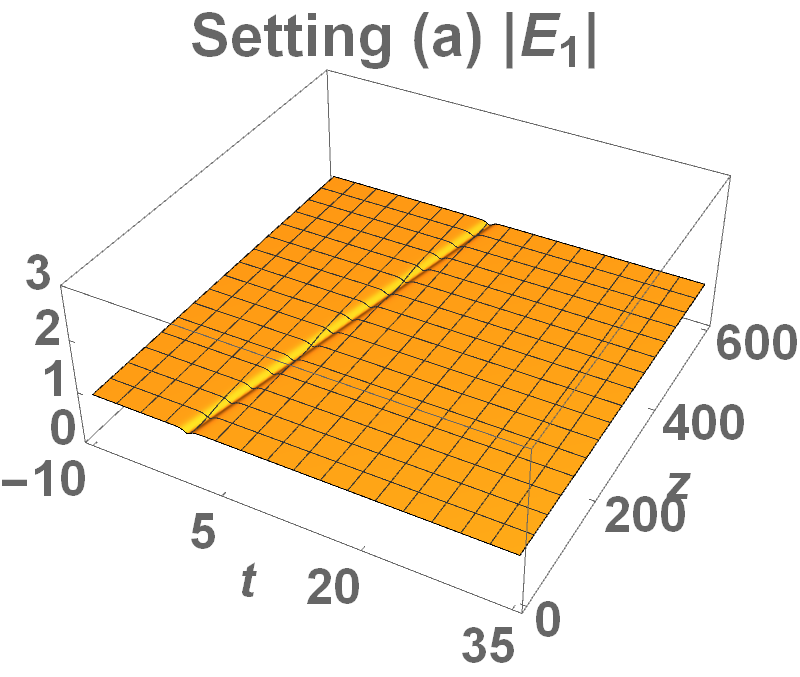}
\includegraphics[width=0.24\textwidth]{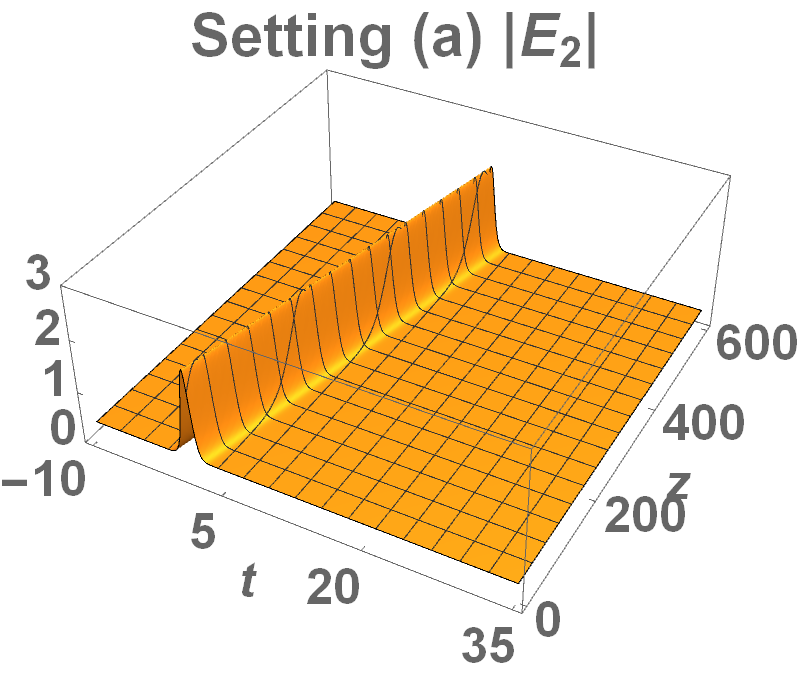}\quad
\includegraphics[width=0.24\textwidth]{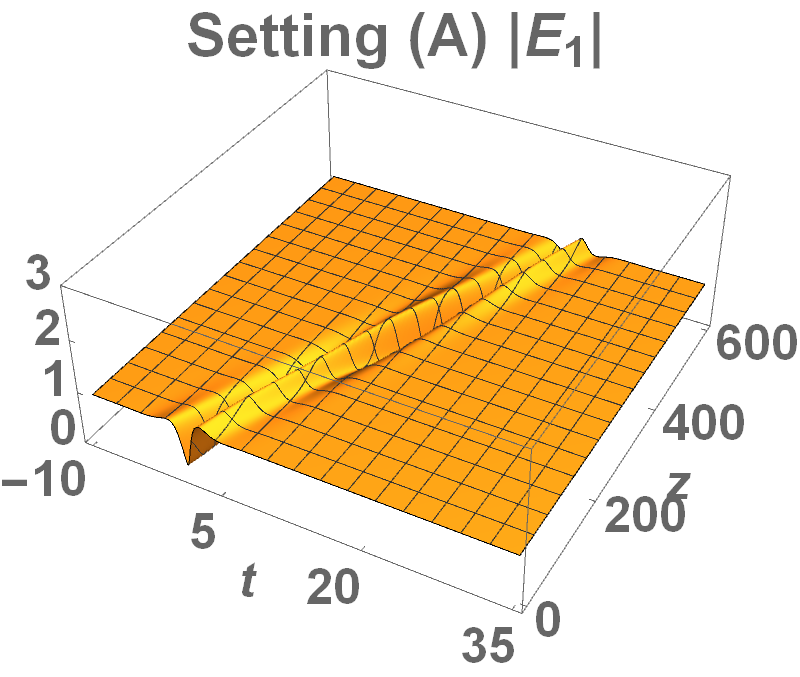}
\includegraphics[width=0.24\textwidth]{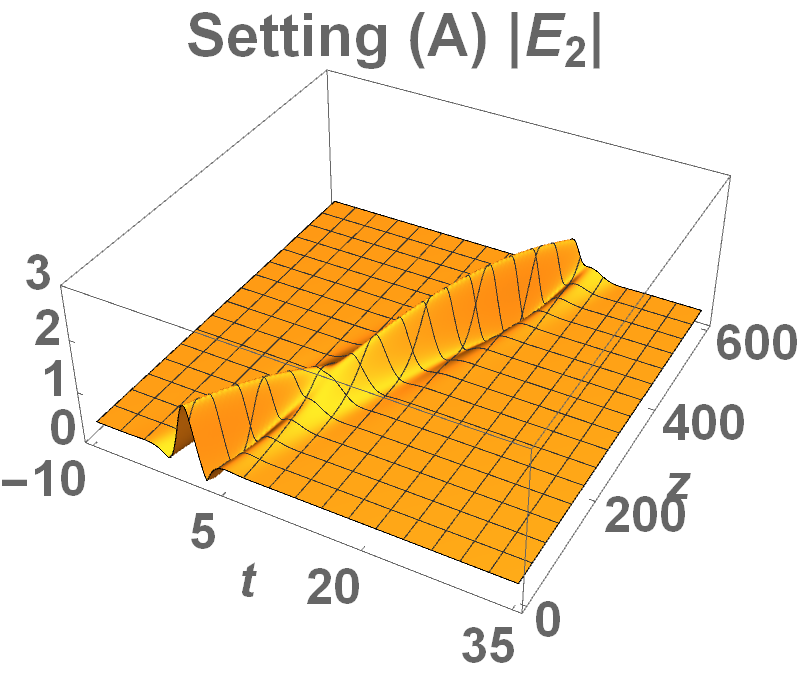}\\[2ex]
\includegraphics[width=0.24\textwidth]{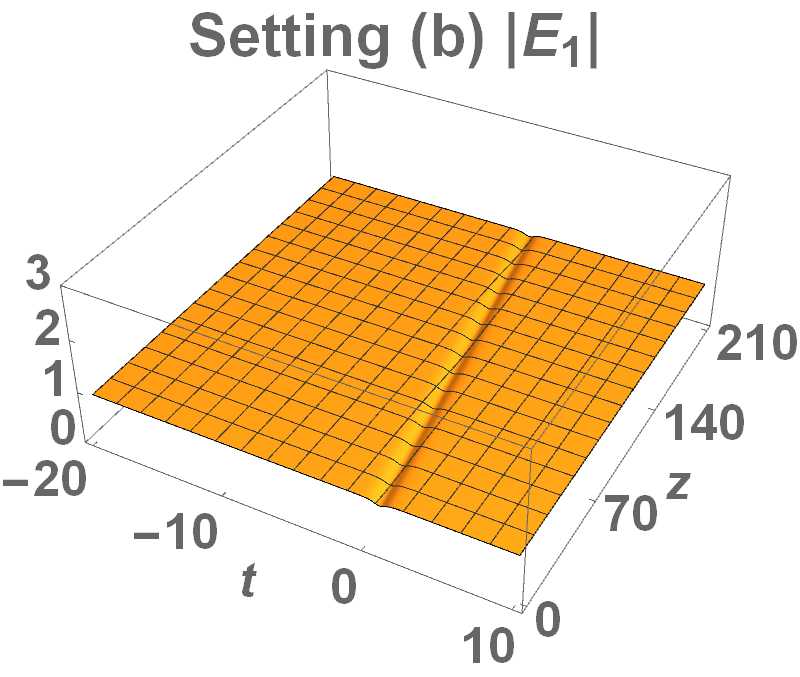}
\includegraphics[width=0.24\textwidth]{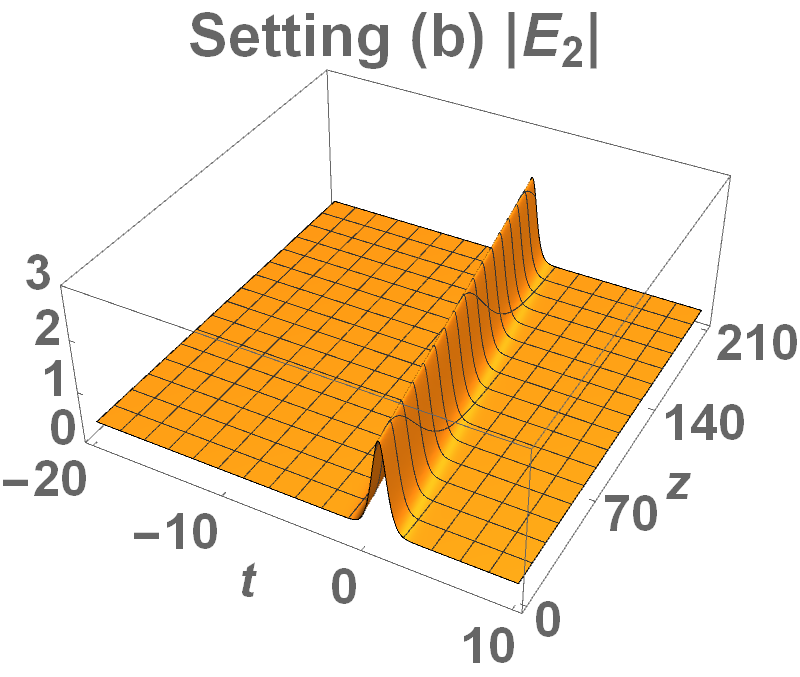}\quad
\includegraphics[width=0.24\textwidth]{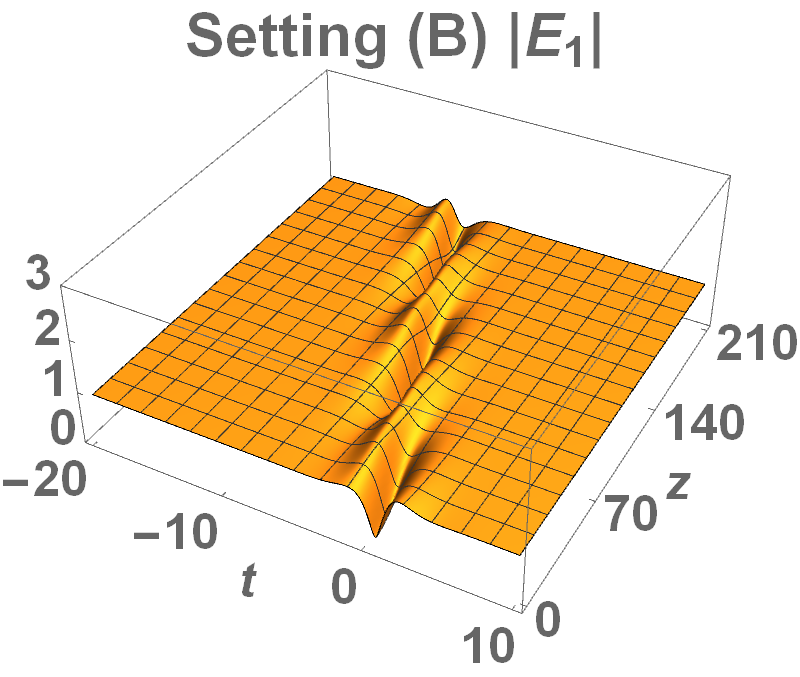}
\includegraphics[width=0.24\textwidth]{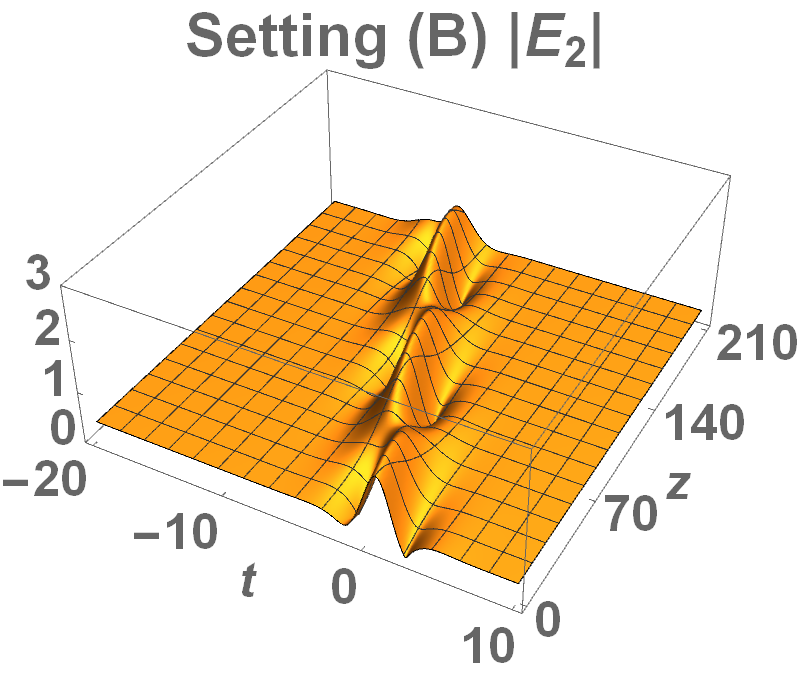}\\[2ex]
\includegraphics[width=0.24\textwidth]{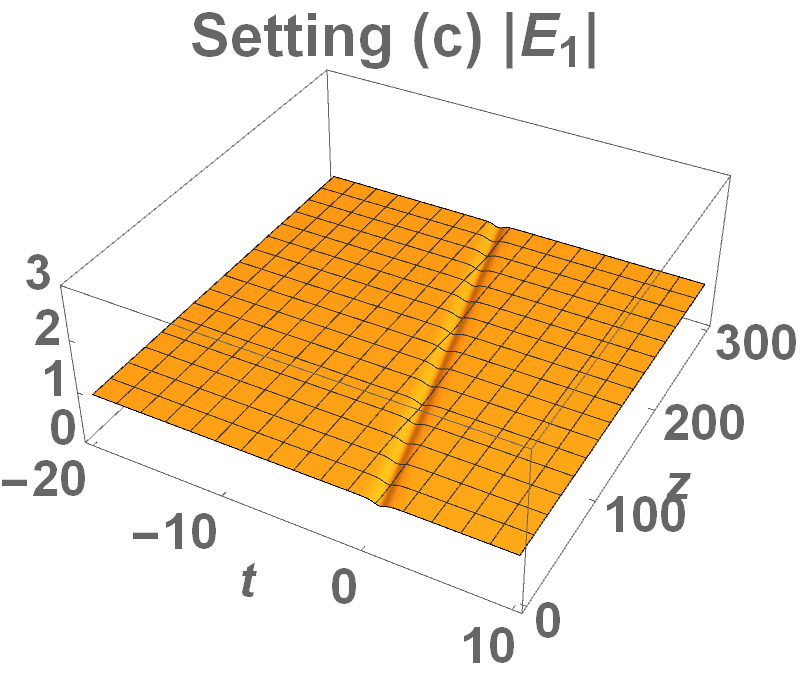}
\includegraphics[width=0.24\textwidth]{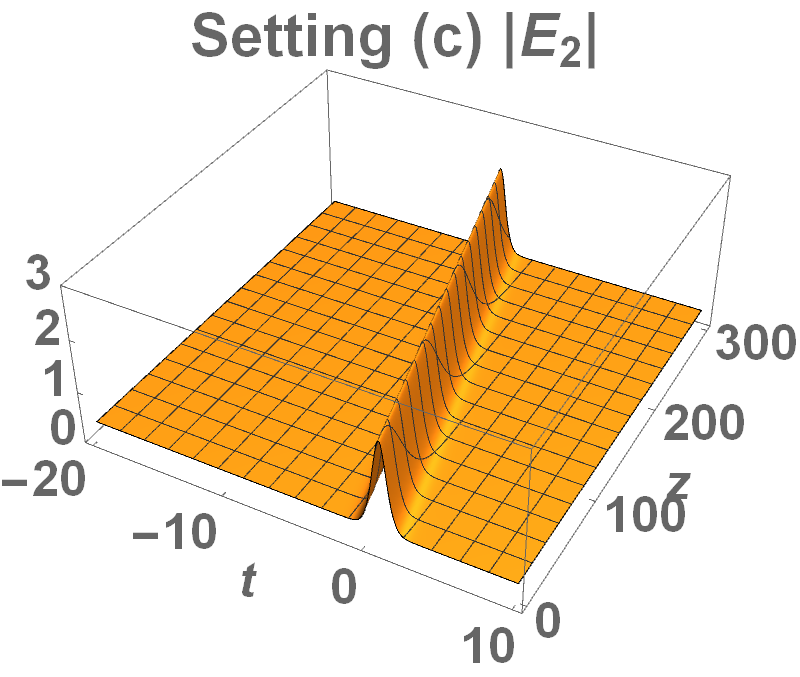}\quad
\includegraphics[width=0.24\textwidth]{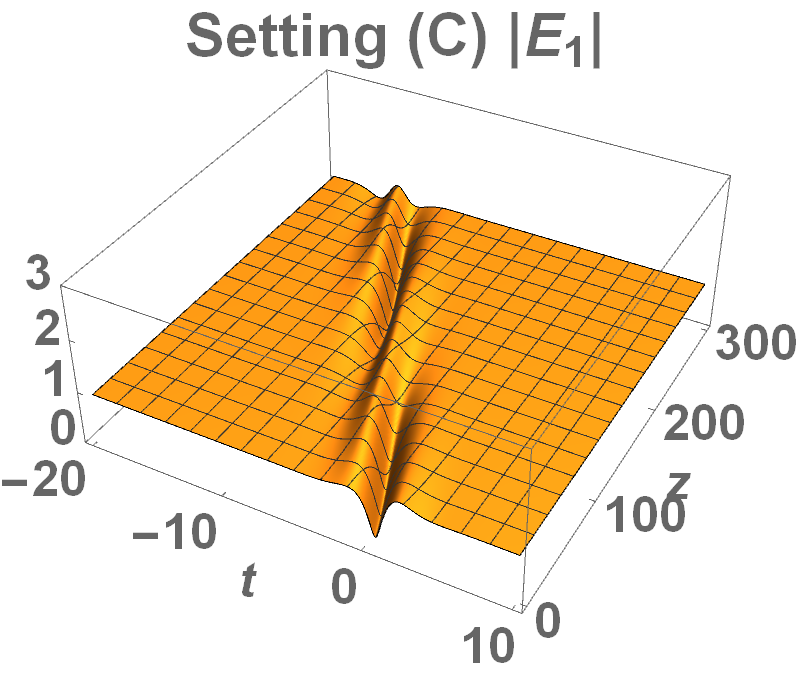}
\includegraphics[width=0.24\textwidth]{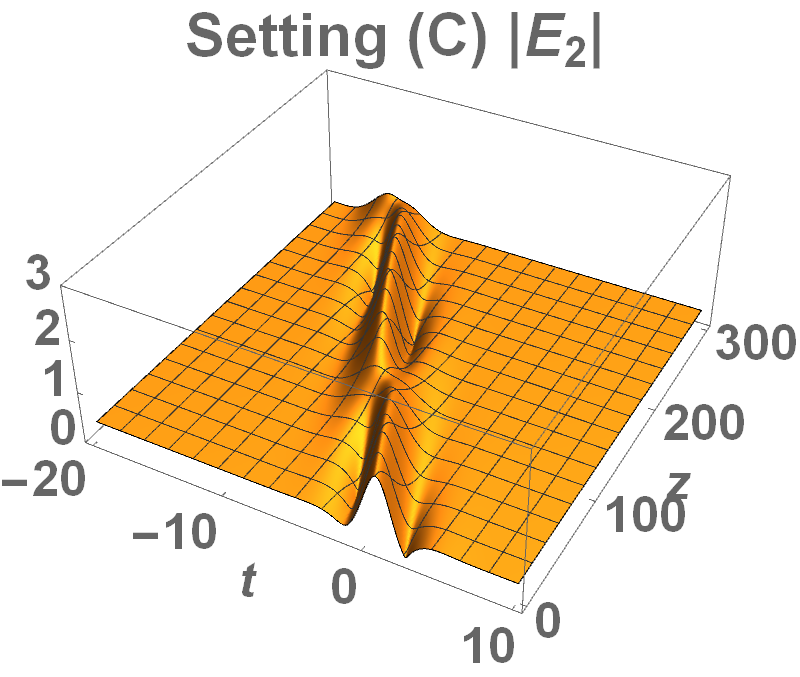}
\caption{Similar situation as in Figure~\ref{f:solitontypeI1}, but for type III solitons with settings (a/A/b/B/c/C). The density matrix $\brho(t,z,\zeta)$ is omitted due to space constraint.}
\label{f:solitontypeIII1}
\end{figure}

\subsection{Dark-state solutions}
\label{s:dark-state}

In this section, we describe several solutions in which the medium is in a dark state for the impinging light beam.
Inspired by the dark-state background solution in Section~\ref{s:dark-state-background}, we impose the following condition
\begin{equation}
\label{e:dark-state-condition}
\varrho_{-,1,1} = \varrho_{-,3,3} = 0\,,\qquad
\varrho_{-,2,2} \ne 0\,,
\end{equation}
and re-investigate all three types of soliton solutions.
(Note that the above condition implies that $\varrho_-^\pm = 0$ in Equation~\eqref{e:rhopmpm}).

For type-I soliton and its derivatives, Equation~\eqref{e:RI} yields a trivial solution, i.e., $R_\I(\zeta) = 0$. Hence, the corresponding solutions  all propagate trivially.
One can then verify that $\rho_{1,1}(t,z,\zeta)\equiv0$ which indeed shows that the medium is in a dark-state for the optical pulse involved.
One example is shown in Figure~\ref{f:darkstate}.

\begin{figure}[t!]
\centering
\includegraphics[width=0.24\textwidth]{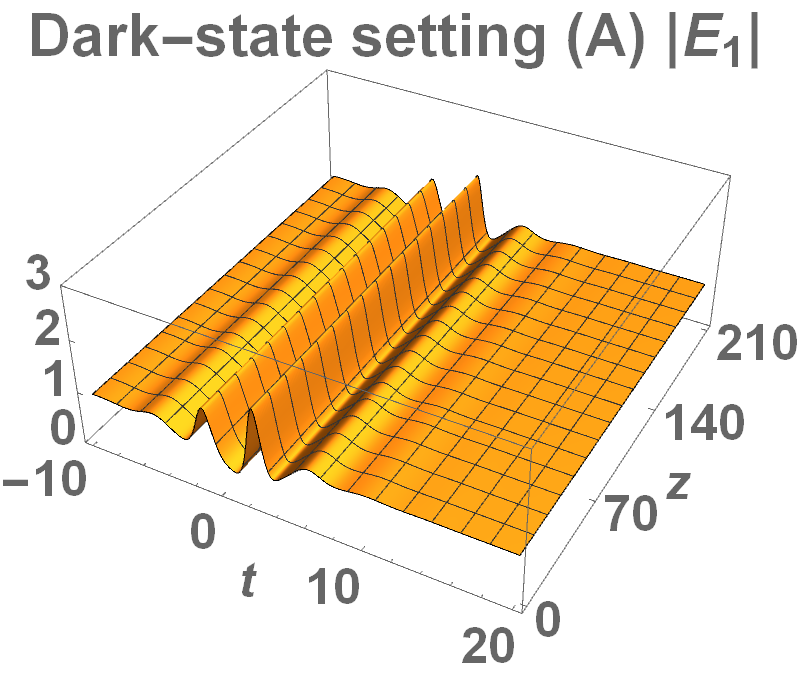}\quad
\includegraphics[width=0.24\textwidth]{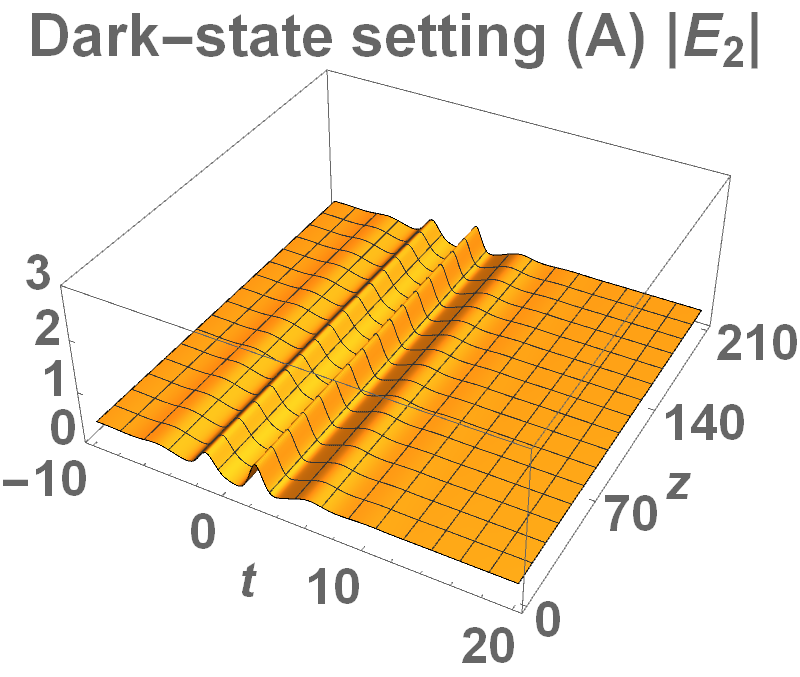}
\caption{Envelopes of electric fields that do not interact with the medium in type-I soliton with setting (A). The asymptotic boundary value $\varrho_-$ is chosen from Equation~\eqref{e:dark-state-condition} instead of Table~\ref{tab:numerics}.}
\label{f:darkstate}
\end{figure}

On the other hand, condition~\eqref{e:dark-state-condition}, together with Equations~\eqref{e:RII} and~\eqref{e:RIII}, implies that
\begin{equation}
R_\II(\zeta) = -\frac{\ii}{2}\frac{\varrho_{-,2,2}}{k+\ii\epsilon}\,,\qquad
R_\III(\zeta) = \frac{\ii}{2}\frac{\varrho_{-,2,2}}{k-\ii\epsilon}\,.
\end{equation}
So, type-II and type-III solitons retain nontrivial propagation. In other words, there are interactions between the electric field and the medium. The medium is thus not in a dark state, and the corresponding solutions are not dark-state solutions.

\subsection{Stability of solitons}
\label{s:stability}

In this section, we discuss the stability of soliton solutions of MBEs~\eqref{e:cmbe}.
We restrict our discussion to the case of the spectral-line shape described by $g(k;\epsilon)$ given in Equation~\eqref{e:Lorentzian}.
For greater generality,
we use the general asymptotic condition~\eqref{e:BC}.
The stability of solitons is determined by decaying/growing radiation as $z$ increases,
which in turn is determined by the jump matrix in RHP~\eqref{e:RHP} or its solution~\eqref{e:RHPsol}.
It is evident that the $z$-dependence of the jump matrix is given by all three reflection coefficients.
Thus, we turn to characterizing the reflection coefficients by analyzing the two systems~\eqref{e:timeevoeqnB1} and~\eqref{e:timeevoeqnB2}.

First, we rewrite the two systems as one,
\bse
\label{e:drdz}
\begin{equation}
\begin{aligned}
\frac{\partial\@r(z,\zeta)}{\partial z}
 & = \frac{\ii}{2} \A(z,\zeta)\,\@r(z,\zeta) + \@b(z,\zeta)\,,\qquad
\zeta\in\Sigma\backslash(-E_0,E_0)\,,\\
\@r(z,\zeta) & \coloneq (r_1(z,\zeta),r_2(z,\zeta),r_3(z,\zeta),r_3(z,\hat{\zeta}))^\top\,,
\end{aligned}
\end{equation}
\begin{equation}
\label{e:drdzA}
\begin{aligned}
\A(z,\zeta) & \coloneq \bpm\displaystyle
Y_{-,2,2} - Y_{-,1,1} & Y_{-,2,3} & 0 & 0 \\
0 & Y_{-,3,3} - Y_{-,1,1} & 0 & 0 \\
0 & 0 & \widetilde Y_{-,3,3} - \widetilde Y_{-,2,2} & 0 \\
0 & 0 & \ii \zeta\widetilde Y_{-,1,3}/E_0 & \widetilde Y_{-,1,1} - \widetilde Y_{-,2,2}
\epm\,,\\
\@b & \coloneq \nu_0 \pi g(k)\bpm\displaystyle
-\varrho_{-,2,1} \\ - \varrho_{-,3,1} \\ \varrho_{-,3,2} \\ \ii\zeta\varrho_{-,1,2}/E_0
\epm\,,\\
\Y_-(z,\zeta)
 & = (Y_{-,i,j})_{3\times3}(z,\zeta)
 \coloneq \R_-(z,\zeta) + \ii\nu_0\pi g(k) \bvarrho_-\,,\\
\widetilde \Y_-(z,\zeta)
 & = (\widetilde Y_{-,i,j})_{3\times3}(z,\zeta)
 \coloneq \R_-(z,\zeta) - \ii\nu_0\pi g(k) \bvarrho_-\,,
\end{aligned}
\end{equation}
\ese
where the quantity $\nu_0$ is defined in Equation~\eqref{e:nu0-def}.

%
The situation is complicated due to the fact that, in general, both $\A(z,\zeta)$ and $\@b(z,\zeta)$ depend on $z$.
For simplicity,
we assume that $\bvarrho_-$ in Equation~\eqref{e:rhopminfty} is independent of $z$,
which implies that the two quantities $\A$ and $\@b$ are independent of $z$ as well. This assumption is reasonable, because it is equivalent to assuming that the optical medium is homogeneous throughout $z\in[0,+\infty)$ in the distant past.
In this simpler case, we can write
\begin{equation}
\label{e:rsolconstant}
\@r(z,\zeta) = 2\ii \A^{-1} \@b + \e^{\ii z\A/2}\big[\@r(0,\zeta) - 2\ii \A^{-1}\@b\big]\,.
\end{equation}
Since the $z$ dependence of $\@r(z,\zeta)$ comes from the exponential $\e^{\ii z\A/2}$,
we next turn to analyzing the matrix $\A$.
Because $\A$ is block-diagonal and each block is triangular,
the diagonal entries are the eigenvalues,
which are denoted as
\begin{equation}
\nonumber
\lambda_1 \coloneq Y_{-,2,2} - Y_{-,1,1}\,,\qquad
\lambda_2 \coloneq Y_{-,3,3} - Y_{-1,1}\,,\qquad
\lambda_3 \coloneq \widetilde Y_{-,3,3} - \widetilde Y_{-,2,2}\,,\qquad
\lambda_4 \coloneq \widetilde Y_{-,1,1} - \widetilde Y_{-,2,2}\,,
\end{equation}
where $\zeta\in\Sigma\backslash(-E_0,E_0)$, and the matrices $\Y$ and $\widetilde \Y$ are defined in Equation~\eqref{e:drdz}.
Explicitly, we write the eigenvalues as
\begin{equation}
\begin{aligned}
\lambda_1
 & = - \pi\lambda\H_k[\varrho_-^-g/\lambda'] - \pi\H_k[\varrho_-^+g] + \pi\H_k[\varrho_{-,2,2}g] - \ii\nu_0\pi g(k)\varrho_{-,1,1} + \ii\nu_0\pi g(k)\varrho_{-,2,2} - 6\omega_-\,,\\
\lambda_2
 & = -2\pi\lambda\H_k[\varrho_-^-g/\lambda']- \ii\nu_0\pi g(k)\varrho_{-,1,1} + \ii\nu_0\pi g(k)\varrho_{-,3,3} \,,\\
\lambda_3
 & = -\pi\lambda\H_k[\varrho_-^-g/\lambda'] + \pi\H_k[\varrho_-^+ g] - \pi \H_k[\varrho_{-,2,2}g] + \ii\nu_0\pi g(k)\varrho_{-,2,2} - \ii\nu_0\pi g(k)\varrho_{-,3,3} + 6\omega_-\,,\\
\lambda_4
 & = \pi\lambda\H_k[\varrho_-^-g/\lambda'] + \pi\H_k[\varrho_-^+g] - \pi\H_k[\varrho_{-,2,2}g] - \ii\nu_0\pi g(k)\varrho_{-,1,1} + \ii\nu_0\pi g(k)\varrho_{-,2,2} + 6\omega_-\,.
\end{aligned}
\end{equation}
Evidently, all four eigenvalues are distinct in general. Therefore, there exists a constant invertible matrix $\B$ such that
\begin{equation}
\nonumber
\B^{-1}\A\B = \diag(\lambda_1,\lambda_2,\lambda_3,\lambda_4)\,,
\end{equation}
and that
\begin{equation}
\nonumber
\e^{\ii z\A/2} = \B^{-1}\diag(\e^{\ii z\lambda_1/2},\e^{\ii z\lambda_2/2},\e^{\ii z\lambda_3/2},\e^{\ii z\lambda_4/2}) \B\,.
\end{equation}
So, $\@r(z,\zeta)$ is bounded as $z\to\infty$ if and only if $\Im\lambda_j\ge0$ for all $j = 1,\dots,4$,
and $\@r(z,\zeta)$ is unbounded as $z\to\infty$ if and only if there exists at least one eigenvalue such that its imaginary part is negative.
We next investigate all eigenvalues of the matrix $\A$,
and then characterize the boundedness of the reflection coefficients.

One should recall that in the jump condition in RHP~\ref{rhp:inverse}:
i) $r_1(z,\zeta)$ only appears when $\zeta\in\Sigma_1$;
ii) $r_2(z,\zeta)$ appears when $\zeta\in\Sigma\backslash\Sigma_3$,
because $r_2(z,\hat{\zeta})$ in $\Sigma_3$ is equivalent to $r_2(z,\zeta)$ in $\Sigma_1$;
iii) $r_3(z,\zeta)$ appears when $\zeta\in\Real$.

By the block nature of $\A$ in Equation~\eqref{e:drdzA},
we only need to consider four situations:
\begin{enumerate}[label=\roman*),align=right]
\item
$\lambda_1$ on $\zeta\in\Sigma_1$;
\item
$\lambda_2$ on $\zeta\in\Sigma\backslash\Sigma_3$;
\item
$\lambda_3$ on $\zeta\in\Sigma_1$;
\item
$\lambda_4$ on $\zeta\in\Sigma_1$.
\end{enumerate}

\subsubsection{Calculation of eigenvalues.}

Recall that the auxiliary matrix $\R_\pm$ in Equation~\eqref{e:R-dexplicit} contains $\sqrt{E_0^2 - \epsilon^2}$ in the case of nontrivial inhomogeneous broadening.
It is easy to show that the quantity $1/\sqrt{E_0^2-\epsilon^2} \log\big[(E_0 + \sqrt{E_0^2-\epsilon^2})/(E_0-\sqrt{E_0^2-\epsilon ^2})\big]$
is always real no matter if $E_0<\epsilon$ or $E_0>\epsilon$.
Thus, we compute the eigenvalues in the two cases $E_0\lg \epsilon$ together.
Note that if $E_0<\epsilon$ the quantity $\sqrt{E_0^2-\epsilon^2}$ is purely imaginary,
which introduces an additional branch cut in the complex plane.
However,
as we discussed before,
we are only interested in the imaginary parts of the eigenvalues $\lambda_j$,
which will be shown to be independent of the choice of the new branch cut.

Without showing all the tedious calculations
(which are straightforward using Equations~\eqref{e:timeevoeqnB1},
\eqref{e:timeevoeqnB2} and~\eqref{e:R-dexplicit}),
we present the explicit expressions for all eigenvalues as follows:
\begin{equation}
\label{e:eigenvalues}
\begin{aligned}
\Im\lambda_1 & = \pi g(k)(\varrho_{-,2,2} - \varrho_{-,1,1})\,, \quad \zeta\in \Sigma_1\,,\\
\Im\lambda_2 & = \begin{cases}
  \pi g(k)(\varrho_{-,3,3} - \varrho_{-,1,1})\,, & \zeta\in \Sigma_1\,,\\
  0\,, & \zeta\in \Sigma_\circ \,,
\end{cases}\\
\Im\lambda_3 & = \pi g(k)(\varrho_{-,2,2} - \varrho_{-,3,3})\,,\qquad \zeta\in \Sigma_1\,,\\
\Im\lambda_4 & = \pi g(k)(\varrho_{-,2,2} - \varrho_{-,1,1})\,,\qquad \zeta\in \Sigma_1\,,
\end{aligned}
\end{equation}
where $\Sigma_1 = (-\infty,-E_0]\cup[E_0,\infty)$ as in Section~\ref{s:RHP}.

\subsubsection{Results.}

As we mentioned earlier,
all the reflection coefficients are bounded as $z\to\infty$ if and only if
all four eigenvalues have non-negative imaginary parts.
The expression~\eqref{e:eigenvalues} yields that all the reflection coefficients are bounded if and only if
the eigenvalues~\eqref{e:BC} satisfy
\begin{equation}
\label{e:stabilitycondition}
\varrho_{-,2,2} \ge \varrho_{-,3,3} \ge \varrho_{-,1,1}\,.
\end{equation}
Again, we stress that the eigenvalues $\varrho_{-,j,j}$ do not have a direct physical meaning. Instead, one should look at $D_{-,j}$ in Equation~\eqref{e:rho-D}, as they indicate the initial state of the medium.
Thus, the system~\eqref{e:Drho-} yields the following system of inequalities:
\begin{equation}
\label{e:stability}
\begin{aligned}
D_{-2}\sin^2\alpha
    - D_{-,3}\cos^2\alpha
 & > - \frac{\hat\zeta}{2 k}D_{-,1} \cos2\alpha
    - \frac{\zeta\cos^2\alpha}{2k} D_{-,2}
    + \frac{\zeta\sin^2\alpha}{2k} D_{-,3}&\\
 & > - \frac{\zeta}{2k} D_{-,1}\cos2\alpha
    - \frac{\hat\zeta\cos^2\alpha}{2k} D_{-,2}
    + \frac{\hat\zeta\sin^2\alpha}{2k} D_{-,3}\,,\qquad&
\mbox{if } 0\le \alpha< \frac{\pi}{4}\,,\\
D_{-,3}\cos^2\alpha - D_{-2}\sin^2\alpha
 & > \frac{\hat\zeta}{2 k}D_{-,1} \cos2\alpha
    + \frac{\zeta\cos^2\alpha}{2k} D_{-,2}
    - \frac{\zeta\sin^2\alpha}{2k} D_{-,3}&\\
 & > \frac{\zeta}{2k}D_{-,1} \cos2\alpha
    + \frac{\hat\zeta\cos^2\alpha}{2k} D_{-,2}
    - \frac{\hat\zeta\sin^2\alpha}{2k} D_{-,3}\,,\qquad&
\mbox{if } \frac{\pi}{4}< \alpha\le \frac{\pi}{2}\,,
\end{aligned}
\end{equation}
where we take $|E_{-,1}| = E_0\cos\alpha$ and $|E_{-,2}| = E_0\sin\alpha$ with $0\le\alpha\le\pi/2$ in Equation~\eqref{e:E-j} and the fact that $w_-(z) = 0$ due to $g(k)$ being a Lorentzian.
Therefore, we conclude that
\textit{solitons of CMBE with NZBG
are stable (radiation decays) as $z\to\infty$ if and only if,
initially, the population of atoms in each state satisfies the inequality~\eqref{e:stability}.}

\begin{remark}
i) As one expects, the initial population of atoms in the two ground states
($D_{-,2}$ and $D_{-,3}$) are symmetric in the inequality~\eqref{e:stability}.
This shows that the formulation of IST does not have any preference between the two ground states,
meaning that the IST is consistent with physics.
ii) However, unlike the case of ZBG or unlike the two-level MBE with ZBG/NZBG,
the inequality~\eqref{e:stability} is coupled and cannot be solved easily.
\end{remark}

\section*{Acknowledgments}

We thank Ildar Gabitov for mentioning this problem to us and for many insightful conversations over the years.
Gregor Kova\v{c}i\v{c} was partially supported by the U.S.\ National Science Foundation under grant DMS-1615859.
Gino Biondini was partially supported by the U.S.\ National Science Foundation under grant DMS-2009487.
Sitai Li was partially supported by the National Natural Science Foundation of China (No.\ 12201526), the Fundamental Research Funds for the Central Universities (No.\ 20720220040), and the Natural Science Foundation of Fujian Province of China (No.\ 2022J01032).

\appendix
\section*{Appendix: Proofs and calculations}
\def\thesubsection{\Alph{subsection}}
\def\theequation{A.\arabic{equation}}
\setcounter{section}1
\setcounter{equation}0

\subsection{Background solutions}
\label{s:background}

A background solution of CMBE~\eqref{e:cmbe} is defined as one for which the electric field envelopes are independent of $t$.
Therefore, we write $E_j(t,z) = E_{\bg,j}(z)$ for $j = 1,2$ and $\Q(t,z) = \Q_\bg(z)$.
Next, we derive such a solution.

Let us define a $t$-independent matrix
$\X_\bg(z,k) \coloneq \ii k \J + \Q_\bg(z)$.
Then, CMBE~\eqref{e:cmbe} implies that
\begin{equation}
\label{e:rho_odef}
\brho_\bg(t,z,k)
 \coloneq \e^{t\X_\bg(z,k)} \, \widetilde \C(z,k)\,\e^{-t\X_\bg(z,k)}\,,
\end{equation}
where $\widetilde \C(z,k)$ is an as yet undetermined $3\times3$ matrix.
The matrix $\X_\bg$ has eigenvalues $\pm \ii\lambda$ and $ - \ii k$,
where $\lambda$ is defined in Equation~\eqref{e:lambda-def}.
In this appendix,
only the case $k\in\Real$ is needed,
so $\lambda\in\Real$ as well.
The eigenvector matrix of $\X_\bg$ and its inverse are computed to be
\begin{equation}
\nonumber
\Y_\bg(z,k) \coloneq \
\bpm1 & 0 & -\ii E_0/\zeta \\
-i\E_\bg^*/\zeta & (\E_\bg^\bot)^*/E_0 & \E_\bg^*/E_0
\epm\,,\qquad
\Y_\bg^{-1}(z,k) = \frac{1}{\gamma(\zeta)}
\bpm
1 & \ii\E_\bg^\top/\zeta \\
0 & (\E_\bg^\bot)^\top \gamma(\zeta)/E_0 \\
i E_0/\zeta & \E_\bg^\top/E_0
\epm\,,
\end{equation}
where $\gamma(\zeta)$ is given in Equation~\eqref{e:gamma-def},
the superscript $\bot$ is defined in Equation~\eqref{e:perp},
and one defines
\begin{equation}
\E_\bg(z) \coloneq (E_{\bg,1}, E_{\bg,2})\,.
\end{equation}
The above eigenvalues and eigenvectors yield the following identity
\begin{equation}
\X_\bg \,\Y_\bg = \ii\Y_\bg \,\bLambda\,,\qquad
\mbox{with }\bLambda \coloneq \diag(\lambda,-k,-\lambda)\,.
\end{equation}
Consequently, the density matrix~\eqref{e:rho_odef} can be rewritten as
\begin{equation}
\label{e:rho_otilderho}
\brho_\bg(t,z,k)
 = \Y_\bg(z,k)\,\e^{\ii \bLambda t}\,\bvarrho_\bg(z,k) \,\e^{-\ii \bLambda t} \,\Y_\bg^{-1}(z,k)\,,
\end{equation}
where $\bvarrho_\bg(z,k) \coloneq \Y_\bg^{-1}(z,k)\,\widetilde \C(z,k)\, \Y_\bg(z,k)$ is another as yet undetermined matrix independent of $t$.

\begin{remark}
\label{rmk:varrhobg-properties}
The matrix $\bvarrho_\bg(z,k)$ must satisfy three conditions:
i) $\tr \bvarrho_\bg(z,k) = \tr\brho_\bg(t,z,k) = 1$;
ii) $\bvarrho_\bg(z,k)$ ensures that $\brho_\bg^\dagger = \brho_\bg$;
iii) $\bvarrho_\bg(z,k)$ ensures that the integral $\int[\J,\brho_\bg(t,z,k)]\,g(k)\,\d k$ is independent of $t$.
The first two conditions follow naturally from the properties of the density matrix $\brho_\bg$,
and the last ensures the consistency of the background solution, i.e.,
that $\Q_\bg$ is independent of $t$ by CMBE~\eqref{e:cmbe}.
\end{remark}

We next look for such a matrix $\bvarrho_\bg(z,k)$.
It is worth noting that, in general, $\bvarrho_\bg(z,k)$ itself is not a Hermitian matrix.
In order to simplify the situation,
we decompose the unknown matrix $\bvarrho_\bg$ as
\begin{equation}
\nonumber
\bvarrho_\bg(z,k) =\sum_{j = 1}^4 \bvarrho_\bg^{(j)}(z,k)\,,
\end{equation}
where
\begin{equation}
\begin{aligned}
\bvarrho_\bg^{(1)} & \coloneq \bpm \varrho_{\bg,1,1} & 0 & 0 \\ 0 & \varrho_{\bg,2,2} & 0 \\ 0 & 0 & \varrho_{\bg,3,3} \epm\,,\qquad&
\bvarrho_\bg^{(2)} & \coloneq \bpm 0 & \varrho_{\bg,1,2} & 0 \\ \varrho_{\bg,2,1} & 0 & 0 \\ 0 & 0 & 0 \epm\,,\\
\bvarrho_\bg^{(3)} & \coloneq \bpm 0 & 0 & \varrho_{\bg,1,3} \\ 0 & 0 & 0 \\ \varrho_{\bg,3,1} & 0 & 0 \epm\,,\qquad&
\bvarrho_\bg^{(4)} & \coloneq \bpm 0 & 0 & 0 \\ 0 & 0 & \varrho_{\bg,2,3} \\ 0 & \varrho_{\bg,3,2} & 0 \epm\,.
\end{aligned}
\end{equation}
Correspondingly,
Equation~\eqref{e:rho_otilderho} yields the decomposition of the density matrix,
\begin{equation}
\brho_\bg(t,z,k) = \sum_{j=1}^4 \brho_{\bg}^{(j)}(t,z,k)\,,\qquad
\brho_{\bg}^{(j)}(t,z,k) \coloneqq \Y_\bg(z,k)\,\e^{\ii \bLambda t}\,\bvarrho_\bg^{(j)}(z,k) \,\e^{-\ii \bLambda t} \,\Y_\bg^{-1}(z,k)\,.
\end{equation}
Following CMBE~\eqref{e:cmbe},
one decomposes the electric field envelope $\Q_\bg$ as
\begin{equation}
\label{e:Qbgj-def}
\frac{\partial \Q_\bg^{(j)}}{\partial z} = -\frac{1}{2}\int\big[\J,\brho_\bg^{(j)}\big]\,g(k)\d k\,.
\end{equation}
After tedious but straightforward calculations,
we obtain explicit expressions for each $\brho_{\bg}^{(j)}$.
The formulas are omitted here for brevity.
As we discussed before,
the background solution requires a $t$-independent $\Q_\bg$.
Hence, from Equation~\eqref{e:Qbgj-def},
one needs to investigate the quantity $[\J,\brho_\bg^{(j)}]$,
which should ensure that $\int[\J,\brho_\bg^{(j)}]g(k)\d k$ is $t$-independent.
Here, we present the explicit expressions for $[\J,\brho_\bg^{(j)}]$:
\begin{equation}
\label{e:[j,rhooj]}
\begin{aligned}
[\J,\brho_{\bg}^{(1)}]
 & = \frac{\ii}{\lambda} (\varrho_{\bg,1,1} - \varrho_{\bg,3,3})\,\Q_\bg\,\J \,,\\
[\J,\brho_{\bg}^{(2)}]
 & = \bpm
 0 & 2 \e^{\ii  (k+\lambda)t} \varrho_{\bg,1,2} (\E_\bg^\bot)^\dagger/E_0 \\
 - \e^{-\ii  (k + \lambda)t} \varrho_{\bg,2,1} \zeta \E_\bg^\bot/( E_0 \lambda) & 0\\
\epm\,,\\
[\J,\brho_{\bg}^{(3)}]
 & = \frac{1}{\zeta\lambda}\bpm
  0 & (E_0^2 \varrho_{\bg,3,1}\e^{-2\ii \lambda t} + \e^{2\ii \lambda t}\varrho_{\bg,1,3}\zeta^2)\E_\bg^\top\\
  -(\e^{2\ii \lambda t}E_0^2 \varrho_{\bg,1,3} + \e^{-2\ii \lambda t} \varrho_{\bg,3,1}\zeta^2)\E_\bg^* & 0
\epm\,,\\
[\J,\brho_{\bg}^{(4)}]
 & = \bpm
  0 & -\frac{2\ii }{\zeta}\e^{\ii (k-\lambda)t}\varrho_{\bg,3,2}(\E_\bg^\bot)^\dagger\\
  -\frac{\ii}{\lambda}\e^{-\ii (k-\lambda)t}\varrho_{\bg,2,3}\E_\bg^\bot & 0
\epm\,,
\end{aligned}
\end{equation}
where, of course, some $\varrho_{\bg,i,j}$ are not arbitrary,
because they have to ensure that Remark~\ref{rmk:varrhobg-properties} holds.
Nonetheless, we know that all $\varrho_{\bg,i,j}$ are independent of $t$,
meaning that all the $t$ dependence of the quantities in Equation~\eqref{e:[j,rhooj]} are expressed in the exponential functions.
The latter three matrices in Equation~\eqref{e:[j,rhooj]} are not independent of $t$ in general even upon integration and summation.
Therefore, the electric field envelope $\Q_\bg^{(j)}$, with $j = 2,3,4$, depends on $t$ and is not a background solution.
It is obvious that the first commutator in Equation~\eqref{e:[j,rhooj]} is independent of $t$.
Hence, following Remark~\ref{rmk:varrhobg-properties},
we conclude that $\bvarrho_\bg$ must be a real diagonal matrix.
Thereafter,
we obtain the differential equation
\begin{equation}
\label{e:Qbg-z}
\frac{\partial \Q_\bg}{\partial z} = -\frac{\ii}{2}\Q_\bg \J\int_{-\infty}^\infty (\varrho_{\bg,1,1} - \varrho_{\bg,3,3})\frac{g(k)}{\lambda(k)}\d k\,.
\end{equation}

\subsection{Computation of the auxiliary matrix \texorpdfstring{$\R_\pm(z,\zeta)$}{R±(z,ζ)}}
\label{a:computeRpm}

The following lemma will be useful in the later calculations.
\begin{lemma}
\label{thm:integralidentities}

If $k\in\Real$,
i.e.,
$\zeta\in\Real$,
and if $k$ and $k'$ are on the same Riemann sheet,
then the following identities hold:
\begin{gather*}
\lim_{t\to\pm\infty}\!\dashint_{-\infty}^\infty\e^{\pm \ii(\lambda'-\lambda)t}f(\zeta,\zeta')\frac{\d k'}{k'-k}
 = \pm \ii\nu\pi f(\zeta,\zeta)\,,\qquad
\lim_{t\to\pm\infty}\!\dashint_{-\infty}^\infty\e^{\pm \ii(\hat\zeta'-\hat\zeta)t}f(\zeta',\zeta)\frac{\d k'}{k'-k}
 = \pm \ii\pi f(\zeta,\zeta)\,,\\
 \lim_{t\to\pm\infty}\!\dashint_{-\infty}^\infty\e^{\pm \ii(\zeta'-\zeta)t}f(\zeta',\zeta)\frac{\d k'}{k'-k}
 = \pm \ii\pi f(\zeta,\zeta)\,,\\
\lim_{t\to\pm\infty}\dashint_{-\infty}^\infty\e^{\pm \ii(\hat\zeta'-\zeta)t}f(\zeta,\zeta')\frac{\d k'}{k'-k}=0 \,,\qquad
\lim_{t\to\pm\infty}\dashint_{-\infty}^\infty\e^{\pm \ii(\zeta'-\hat\zeta)t}f(\zeta',\zeta)\frac{\d k'}{k'-k}=0\,,
\end{gather*}
where we recall that $\hat\zeta$ is defined in Equation~\eqref{e:k-lambda}. We also use the shorthand notation $\lambda=\lambda(k)$, $\lambda'=\lambda(k')$, $\zeta=\zeta(k)$ and $\zeta'=\zeta(k')$,
and we define $\nu = \pm1$ when $k$ is on sheet I or II, respectively.

If $k\in i[E_0,0)\cup \ii(0,E_0]$, i.e., $-E_0 < \lambda < E_0$ and $\zeta\in \Sigma_\circ$,
then all the above limits are zero.
\end{lemma}

\begin{remark}
Lemma~\ref{thm:integralidentities} depends crucially on whether $k$ and $k'$ are on the same Riemann sheet or not,
because their corresponding $\lambda$ and $\lambda'$, respectively,
admit different values.
We only consider the case when $k$ and $k'$ are on the same sheet,
because the results in Lemma~\ref{thm:integralidentities} reduce to that in the case of ZBG naturally as $E_0\to0$.
\end{remark}

\begin{remark}
In the first limit in Lemma~\ref{thm:integralidentities}, if the term $(\lambda'-\lambda)t$ is replaced with $(\lambda'+\lambda)t$, or by $\lambda' t$,
one can show that the corresponding limit vanishes by the Riemann-Lebesgue lemma.
\end{remark}

\begin{proof}

We calculate the five limits separately below.

\begin{enumerate}[leftmargin = *]
\item
Consider the first integral in Lemma~\ref{thm:integralidentities}.
Let the left hand-side (LHS) be
\begin{equation}
\nonumber
I_\pm=\lim_{t\to\pm\infty}\dashint_{-\infty}^\infty \e^{\pm \ii(\lambda'-\lambda)t}f(\zeta,\zeta')\frac{\d k'}{k'-k}\,.
\end{equation}
Note that $\lambda'$ and $\lambda$ have opposite signs on sheets I and II,
so that $I_\pm$ have different values on each sheet.
By a change of variables $k'\to\lambda'$,
the domain changes as $\Real\to L\coloneq(-\infty,-E_0)\cup(E_0,\infty)$, and the integral becomes
\begin{equation}
\nonumber
I_\pm = \nu\lim_{t\to\pm\infty}\pvint_L\e^{\pm \ii(\lambda'-\lambda)t}f(\zeta,\zeta')
  \frac{\lambda'-\lambda}{k'-k}\frac{\d k'}{\d\lambda'}\frac{\d\lambda'}{\lambda'-\lambda}\,,
\end{equation}
where $\nu = \pm1$ given in Lemma~\ref{thm:integralidentities}.
Let $y'=\lambda'-\lambda$.
The integration domain becomes $L\to L'\coloneq(-\infty,-E_0-\lambda)\cup(E_0-\lambda,\infty)$, and the integral becomes
\begin{equation}
\nonumber
I_\pm = \nu \lim_{t\to\pm\infty}\pvint_{L'}\e^{\pm iy't}f(\zeta,\zeta')\frac{\lambda'-\lambda}{k'-k}\frac{\d k'}{\d\lambda'}\frac{\d y'}{y'}\,.
\end{equation}
From here,
we need to discuss two cases depending on whether $0\in L'$ or $0\not\in L'$.
\begin{itemize}[leftmargin = *]
\item
If $k\in\Real$,
i.e.,
$\zeta\in\Real$,
then $\lambda\in (-\infty,-E_0]\cup [E_0,\infty)$ and consequently $0\in L'$.
We obtain
\begin{equation}
\nonumber
I_\pm(\zeta)=\pm \nu \ii\pi f(\zeta,\zeta)\lim_{y'\to0}\frac{\lambda'-\lambda}{k'-k}\frac{\d k'}{\d \lambda'}\,.
\end{equation}
It is obvious that $\lim_{y'\to0}\frac{\lambda'-\lambda}{k'-k} = \frac{\d\lambda'}{\d k'}$,
so $I_\pm(\zeta) = \pm \ii\nu\pi  f(\zeta,\zeta)$.

\item
If $k\in \ii[-E_0,0)\cup \ii(0, E_0]$,
i.e.,
$\zeta\in \Sigma_\circ$ and $-E_0 < \lambda < E_0$,
then $0\not\in L'$.
By using the Riemann-Lebesgue lemma,
we find $I_\pm \to 0$ as $t\to\pm\infty$.
\end{itemize}

\item
Let us consider the second integral in Lemma~\ref{thm:integralidentities} and call the LHS $I_\pm$,
with some abuse of notation.
We first compute this integral on sheet I.
Since $k'=\frac{1}{2}(\zeta'+\hat\zeta')$,
we conclude $\d k'=\frac{1}{2}(1+E_0^2/\hat\zeta')\d\hat\zeta'$.
Correspondingly,
the integration domain becomes $L \coloneq (E_0^-,0^-)\cup(0^+,E_0^+)$,
again, with some abuse of notation.
Therefore,
\begin{equation}
\nonumber
I_\pm
 = \lim_{t\to\pm\infty} \pvint_L \e^{\pm \ii(\hat\zeta' - \hat\zeta)t}f(\zeta',\zeta) \frac{1 + E_0^2/\hat\zeta'}{\hat\zeta'-\hat\zeta - E_0^2/\hat\zeta' + E_0^2/\hat\zeta}
\d\hat\zeta'\,,
\end{equation}
Define $y' \coloneq \hat\zeta'-\hat\zeta$.
We know that
$0\in L' \coloneq (E_0^--\hat\zeta,-\hat\zeta^-)\cup(-\hat\zeta^+,-E_0^+-\hat\zeta)$.
Thus
\begin{equation}
\nonumber
I_\pm
 = \lim_{t\to\pm\infty} \pvint_{L'}\e^{\ii y't} f(\zeta',\zeta) \frac{\hat\zeta'^2\hat\zeta  + E_0^2\hat\zeta}{\hat\zeta'^2\hat\zeta + E_0^2\hat\zeta'}\frac{\d y'}{y'}\,.
\end{equation}
By a similar argument to the first case,
we conclude that if $\zeta\in\Real$,
$I_\pm \to \pm \ii\pi f(\zeta,\zeta)$,
and if $\zeta\in \Sigma_\circ$, $I_\pm \to 0$ as $t\to\pm\infty$.
Similarly,
one can compute the same integral on sheet II,
and obtain the same result.

\item
The third integral in Lemma~\ref{thm:integralidentities} can be computed similarly to the second,
and has the same value on both sheets.

\item
Consider the fourth integral in Lemma~\ref{thm:integralidentities} and let the LHS be $I_\pm$ again,
with some abuse of notation.
Similarly to the previous cases,
$I_\pm$ takes different values depending on which sheet it is evaluated.
Because
$k'=\frac{1}{2}(\zeta'-E_0^2/\zeta')$,
we know that $\d k'=\frac{1}{2}(1+E_0^2/\zeta'^2)\d\zeta'$,
and $\zeta'\in L \coloneq (-\infty,-E_0)\cup(E_0,+\infty)$.
Thus, the integral becomes
\begin{equation}
\label{e:integral1}
I_\pm
 = \nu\lim_{t\to\pm\infty}\pvint_L\e^{\mp \ii(E_0^2/\zeta'+\zeta)t} f(\zeta',\zeta) \frac{1+E_0^2/\zeta'}{\zeta'-\zeta-E_0^2/\zeta'+E_0^2/\zeta} \d \zeta'\,.
\end{equation}
Let $y' \coloneq E_0^2/\zeta'+\zeta$,
so $\d y' = -E_0^2/\zeta'^2 \d\zeta'$ and $L\mapsto L'\coloneq (\zeta-E_0,\zeta)\cup(\zeta,\zeta+E_0)$,
provided $\zeta\in\Real$.
Also,
note that there is an additional minus sign due to the change of integration limits,
\begin{equation}
\nonumber
I_\pm
 = \nu\lim_{t\to\infty}\pvint_{L'}\e^{\mp \ii y't} f(\zeta',\zeta) \frac{1}{y'-\zeta}\,\frac{E_0^2+(y'-\zeta)^2}{\zeta+E_0^2/\zeta-y'}\,\frac{\d y'}{y'}\,.
\end{equation}
Note also that $0\not\in L'$.
Consequently,
the above quantity vanishes by the Riemann-Lebesgue lemma.
Finally note that if $\zeta\in \Sigma_\circ$,
the integral~\eqref{e:integral1} is zero by the Riemann-Lebesgue lemma.

\item
The last integral in Lemma~\ref{thm:integralidentities} is zero by an argument similar to the fourth case.
\end{enumerate}

\end{proof}

We are ready to compute the matrix $\R_\pm$ in Equation~\eqref{e:Rpm-1}.
We compute the two terms in the square brackets separately,
and we consider the second first.
We know that, as $t\to\pm\infty$,
\begin{equation}
\begin{aligned}
\frac{\partial \bphi_\pm}{\partial z}
 = &
 \bpm 2\ii w_\pm\e^{2\ii W_\pm} & \@0 \\ \@0 & \@0 \epm
 \Y_\pm(0,\zeta)
 \bpm \e^{-2\ii W_\pm} & \@0 \\ \@0 & \e^{2\ii W_\pm\sigma_3} \epm
 \e^{\ii \Lambda t}\\
 &\qquad\qquad\qquad + \bpm \e^{2\ii W_\pm} & \@0 \\ \@0 & \bbI_2\epm
 \Y_\pm(0,\zeta)
 \bpm - 2 \ii w_\pm \e^{-2\ii W_\pm} & \@0  \\ \@0 & 2\ii w_\pm\e^{2\ii W_\pm\sigma_3}\sigma_3 \epm
 \e^{\ii \Lambda t} + o(1)\,.
\end{aligned}
\end{equation}
After some calculations,
the above equation yields
\begin{equation}
\label{e:phiphiz}
\bphi_\pm^{-1}(\bphi_\pm)_z
 = \left(\bphi_\pm^{-1}(\bphi_\pm)_z\right)_\mathrm{eff} + o(1)
 =\frac{\ii w_\pm}{\lambda}
 \bpm -E_0^2/\zeta & 0 & -\ii E_0\e^{-2\ii\lambda t} \\ 0 & 2\lambda & 0 \\
\ii E_0\e^{2\ii\lambda t} & 0 & -\zeta\epm + o(1)\,.
\end{equation}

Now let us compute the first term $\bphi_\pm^{-1}\V\bphi_\pm$ in Equation~\eqref{e:Rpm-1} as $t\to\pm\infty$.
One can write
\begin{equation}
\nonumber
\bphi_\pm^{-1}\V\bphi_\pm
 = \frac{\ii\pi}{2}\H_k[\C(k,k') g(k')] + o(1)\,.
\end{equation}
where Equation~\eqref{e:Jostsol} yields the leading term as
\begin{equation}
\@C(k,k')
 \coloneq \e^{-\ii \bLambda t} \Y_{\pm}^{-1}(z,\zeta) \Y_\pm(z,\zeta') \e^{\ii\bLambda t} \bvarrho_\pm(z,\zeta') \e^{-\ii\bLambda t} \Y_\pm^{-1}(z,\zeta')  \Y_\pm(z,\zeta) \e^{\ii \bLambda t}\,.
\end{equation}
After tedious calculations,
one can obtain each entry of $\C(k,k')$ from this matrix product.
They are omitted here for brevity.
Correspondingly,
each entry of $\bphi_\pm^{-1}\V\bphi_\pm$ can be calculated by taking the Hilbert transform and using Lemma~\ref{thm:integralidentities}.
We list some intermediate steps of the calculations below as $t\to\pm\infty$,
and we use the shorthand notation $f' \coloneq f(k')$ for a function $f(\cdot)$ containing the integration variable $k'$ of the Hilbert transform. The entries of $\bphi_\pm^{-1}\V\bphi_\pm$ are listed below:
\bse
\begin{equation}
\begin{aligned}
\left(\bphi_\pm^{-1}\V\bphi_\pm\right)_{1,1}
 & = \frac{\ii}{2}\pvint\nolimits_{\!\!\!\Real} C_{1,1}\frac{g(k')\d k'}{k'-k} + o(1)\\
 & = \frac{\ii}{2} \pvint\nolimits_{\!\!\!\Real} \frac{1}{\gamma\gamma'} \left[\varrho_{\pm,1,1}'\left(1+\frac{E_0^2}{\zeta\zeta'}\right)^2
+ E_0^2\varrho_{\pm,3,3}'\left(\frac{1}{\zeta'} - \frac{1}{\zeta} \right)^2 \right] \frac{g(k')\d k'}{k'-k} + o(1)\,.
\end{aligned}
\end{equation}
\begin{equation}
\begin{aligned}
\left(\bphi_\pm^{-1}\V\bphi_\pm\right)_{1,2}
 & = \frac{\ii}{2}\pvint\nolimits_{\!\!\!\Real} C_{1,2}\frac{g(k')\d k'}{k'-k} + o(1)\\
 & = \pvint\nolimits_{\!\!\!\Real} \frac{1}{2\gamma}\bigg[\ii \e^{-\ii (\zeta-\zeta')t} \varrho_{\pm,1,2}' \bigg(1+\frac{E_0^2}{\zeta\zeta'}\bigg) + E_0 \e^{-\ii (\zeta-\hat\zeta')t}\varrho_{\pm,3,2}'\left(\frac{1}{\zeta'}-\frac{1}{\zeta}\right)\bigg]\frac{g(k')\d k'}{k'-k} + o(1)\\
 & = \mp\frac{\pi \nu}{2}\varrho_{\pm,1,2}g(k) + o(1)\,.
\end{aligned}
\end{equation}
\begin{equation}
\label{e:phivphi13}
\begin{aligned}
& \left(\bphi_\pm^{-1}\V\bphi_\pm\right)_{1,3}\\
 = & \frac{\ii}{2}\pvint\nolimits_{\!\!\!\Real} C_{1,3}\frac{g(k')\d k'}{k'-k} + o(1)\\
 = & \frac{\ii}{2}\pvint\nolimits_{\!\!\!\Real} \bigg[\frac{\e^{-2\ii(\lambda-\lambda')t}}{\gamma\gamma'}\varrho_{\pm,1,3}'\bigg(\frac{E_0^2}{\zeta\zeta'} + 1\bigg)^2 - \frac{\ii  E_0}{\gamma\gamma'}\e^{-2\ii\lambda t}(\varrho_{\pm,1,1}' - \varrho_{\pm,3,3}') \bigg(\frac{1}{\zeta}-\frac{1}{\zeta'}\bigg)\bigg(1+\frac{E_0^2}{\zeta\zeta'}\bigg)\bigg]\frac{g'\d k'}{k'-k} + o(1)\\
 = & \frac{E_0}{\lambda}\e^{-2\ii \lambda t}w_\pm \mp\frac{\pi \nu}{2\gamma^2}\brho_{\pm1,3}\left(\frac{E_0^2}{\zeta^2}+1\right)^2g(k) + o(1)\,.
\end{aligned}
\end{equation}
Recall that the matrix $\R_\pm$ in Equation~\eqref{e:Rpm-1} contains two terms,
and we are calculating one of them right now.
The other term is in Equation~\eqref{e:phiphiz}.
The first term in Equation~\eqref{e:phivphi13} cancels the corresponding term appearing in Equation~\eqref{e:phiphiz}.
Thus, we finally obtain
\begin{equation}
\nonumber
R_{\pm,1,3} = \pm\ii\pi \nu\varrho_{\pm,1,3}\,g(k)\,.
\end{equation}

Let us continue calculating the entries of $\bphi_\pm^{-1}\V\bphi_\pm$:
\begin{equation}
\begin{aligned}
\left(\bphi_\pm^{-1}\V\bphi_\pm\right)_{2,1}
 & = \frac{\ii}{2}\pvint\nolimits_{\!\!\!\Real} C_{2,1}g(k')\frac{\d k'}{k'-k} + o(1)\\
 & = \pvint\nolimits_{\!\!\!\Real} \frac{1}{2\gamma'}\bigg[ \ii\e^{\ii (\zeta-\zeta')t} \left(1+\frac{E_0^2}{\zeta\zeta'}\right) \varrho_{\pm,2,1}' + E_0 \e^{\ii (\zeta-\hat\zeta')t} \varrho_{\pm,2,3}' \left(\frac{1}{\zeta}-\frac{1}{\zeta'}\right)\bigg] \frac{g(k')\d k'}{k'-k} + o(1)\\
 & = \pm\frac{\pi}{2}\varrho_{\pm,2,1}g(k) + o(1)\,,
\end{aligned}
\end{equation}
\begin{equation}
\begin{aligned}
& \left(\bphi_\pm^{-1}\V\bphi_\pm\right)_{3,1}\\
 = & \frac{\ii}{2} \pvint\nolimits_{\!\!\!\Real} C_{3,1}g(k')\frac{\d k'}{k'-k} + o(1)\\
 = & \frac{\ii}{2} \pvint\nolimits_{\!\!\!\Real} \left[\frac{\e^{-2\ii (\lambda'-\lambda)t}}{\gamma\gamma'}\varrho_{\pm,3,1}'\left(1+\frac{E_0^2}{\zeta\zeta'}\right)^2 - \frac{\ii E_0}{\gamma\gamma'}\e^{2\ii \lambda t}\left(\frac{1}{\zeta}-\frac{1}{\zeta'}\right)\left(1+\frac{E_0^2}{\zeta\zeta'}\right)(\varrho_{\pm,3,3}'-\varrho_{\pm,1,1}')\right]\frac{g'\d k'}{k'-k} + o(1)\\
 = & \pm\frac{\pi\nu}{2}\varrho_{\pm,3,1}g(k) - \frac{E_0}{\lambda}w_\pm\e^{2\ii \lambda t} + o(1)\,,
\end{aligned}
\end{equation}
where the last term will cancel the term appearing in Equation~\eqref{e:phiphiz},
similarly to the $(3,1)$ component.
The rest of the entries of $\bphi_\pm^{-1}\V\bphi_\pm$ are given by
\begin{equation}
\nonumber
\left(\bphi_\pm^{-1}\V\bphi_\pm\right)_{2,2}
 = \frac{\ii}{2}\pvint\nolimits_{\!\!\!\Real} C_{2,2}\frac{g(k')\d k'}{k'-k} + o(1)
 = \frac{\ii}{2}\pvint\nolimits_{\!\!\!\Real} \varrho_{\pm,2,2}'\frac{g(k')\d k'}{k'-k} + o(1)\,,
\end{equation}
\begin{equation}
\begin{aligned}
\left(\bphi_\pm^{-1}\V\bphi_\pm\right)_{2,3}
 & = \frac{\ii}{2}\pvint\nolimits_{\!\!\!\Real} C_{2,3}\frac{g(k')\d k'}{k'-k} + o(1)\\
 & = \frac{\ii}{2} \pvint\nolimits_{\!\!\!\Real} \frac{1}{\gamma'} \left[\e^{-\ii (\hat\zeta'-\hat\zeta)t}\left(1+\frac{E_0^2}{\zeta\zeta'}\right)\varrho_{\pm,2,3}' - \ii E_0\e^{-\ii (\zeta'-\hat\zeta)t} \left(\frac{1}{\zeta}-\frac{1}{\zeta'}\right)\varrho_{\pm,2,1}'\right]\frac{g(k')\d k'}{k'-k} + o(1)\\
 & = \pm\frac{\pi}{2}\varrho_{\pm,2,3}g(k) + o(1)\,,
\end{aligned}
\end{equation}
\begin{equation}
\begin{aligned}
\left(\bphi_\pm^{-1}\V\bphi_\pm\right)_{3,2}
 & = \frac{\ii}{2}\pvint\nolimits_{\!\!\!\Real} C_{3,2}\frac{g(k')\d k'}{k'-k} + o(1)\\
 & = \frac{\ii}{2} \pvint\nolimits_{\!\!\!\Real} \frac{1}{\gamma} \left[\e^{\ii (\hat\zeta'-\hat\zeta)t} \left(1+\frac{E_0^2}{\zeta\zeta'}\right)\varrho_{\pm,3,2}' - \ii E_0 \e^{\ii (\zeta'-\hat\zeta)t} \left(\frac{1}{\zeta'} - \frac{1}{\zeta}\right)\varrho_{\pm,1,2}'\right] \frac{g(k')\d k'}{k'-k} + o(1)\\
 & = \mp\frac{\pi}{2}\varrho_{\pm,3,2}g(k) + o(1)\,,
\end{aligned}
\end{equation}
\begin{equation}
\begin{aligned}
\left(\bphi_\pm^{-1}\V\bphi_\pm\right)_{3,3}
 & = \frac{\ii}{2}\pvint\nolimits_{\!\!\!\Real} C_{3,3}\frac{g(k')\d k'}{k'-k} + o(1)\\
 & = \frac{\ii}{2}\pvint\nolimits_{\!\!\!\Real} \frac{1}{\gamma\gamma'}\left[E_0^2\varrho_{\pm,1,1}'\left(\frac{1}{\zeta}-\frac{1}{\zeta'}\right)^2 + \varrho_{\pm,3,3}'\left(\frac{E_0^2}{\zeta\zeta'}+1\right)^2\right]\frac{g(k')\d k'}{k'-k} + o(1)\,.
\end{aligned}
\end{equation}
\ese

Finally,
combining all the above components,
we obtain the matrix $\R_\pm$
\begin{equation}
\everymath{\displaystyle}
\R_\pm =
\bpm
\pvint\nolimits_{\!\!\!\Real} C_{1,1}\frac{g(k')\d k'}{k'-k} - 2w_\pm\hat\zeta/\lambda &
\pm \ii\pi\varrho_{\pm,1,2}g(k) & \pm \ii\pi\nu\varrho_{\pm,1,3}g(k)\\
\mp \ii \pi \varrho_{\pm,2,1}g(k) &
\pvint\nolimits_{\!\!\!\Real} \varrho_{\pm,2,2}'\frac{g(k')\d k'}{k'-k} - 4w_{\pm} & \mp \ii\pi\varrho_{\pm,2,3}g(k) \\
\mp \ii\pi\nu\varrho_{\pm,3,1}g(k) & \pm \ii\pi\varrho_{\pm,3,2}g(k) &
\pvint\nolimits_{\!\!\!\Real} C_{3,3}\frac{g(k')\d k'}{k'-k} + 2w_\pm \zeta/\lambda
\epm\,,
\end{equation}
where $C_{1,1}$ and $C_{3,3}$ are given by
\begin{equation}
\begin{aligned}
C_{1,1} & = \frac{1}{\gamma\gamma'} \left[\varrho_{\pm,1,1}'\left(1+\frac{E_0^2}{\zeta\zeta'}\right)^2
+ E_0^2\varrho_{\pm,3,3}'\left(\frac{1}{\zeta'} - \frac{1}{\zeta} \right)^2 \right]\,,\\
C_{3,3} & = \frac{1}{\gamma\gamma'}\left[E_0^2\varrho_{\pm,1,1}'\left(\frac{1}{\zeta}-\frac{1}{\zeta'}\right)^2 + \varrho_{\pm,3,3}'\left(\frac{E_0^2}{\zeta\zeta'}+1\right)^2\right]\,.
\end{aligned}
\end{equation}

Next, we show how to simplify $R_{\pm,1,1}$ and $R_{\pm,3,3}$.
First, we recall the shorthand notation defined in Equation~\eqref{e:rhopmpm}.
With the help of these quantities one can rewrite $w_\pm$ in Equation~\eqref{e:wpm} as
\begin{equation}
\nonumber
w_\pm
 = \frac{1}{2}\int_\Real \varrho_\pm^- g(k) \frac{\d k}{\lambda}
 = \pvint\nolimits_{\!\!\!\Real} \bigg(\varrho_\pm^- \frac{k'-k}{2\lambda'}\bigg) g(k')\frac{\d k'}{k'-k}\,.
\end{equation}
Using Equation~\eqref{e:rhopmpm},
we can simplify the function $C_{1,1}$ in $R_{\pm,1,1}$ to become
\begin{equation}
\nonumber
C_{1,1} = \rho_\pm^+ + \frac{k' k+ E_0^2}{\lambda' \lambda}\varrho_\pm^-\,.
\end{equation}
Consequently, we have
\begin{equation}
\nonumber
R_{\pm,1,1}
 = \pvint\nolimits_{\!\!\!\Real} \bigg(\rho_\pm^+ + \frac{\lambda}{\lambda'}\varrho_\pm^-\bigg)g(k') \frac{\d k'}{k'-k} + \int_\Real \frac{\varrho_\pm^-}{\lambda'}g(k') \d k'\,.
\end{equation}
Note that the second integral is $2 w_\pm$,
so we have obtained the desired result in Equation~\eqref{e:Rpm}.
One applies similar simplifications to $R_{\pm,3,3}$ and obtains
\begin{equation}
\nonumber
R_{\pm,3,3}
 = \pvint\nolimits_{\!\!\!\Real} \bigg(\rho_\pm^+ - \frac{\lambda}{\lambda'}\varrho_\pm^-\bigg)g(k') \frac{\d k'}{k'-k} + \int_\Real \frac{\varrho_\pm^-}{\lambda'}g(k') \d k'\,,
\end{equation}
yielding the final result in Equation~\eqref{e:Rpm}.

\subsection{Derivation of propagation equations for norming constants}
\label{a:norming}

Using symmetries~\eqref{e:mnormingsymmetry},
it suffices to compute the propagation equations for the three norming constants $\overline C_n$, $D_n$ and $\overline F_n$.
We do this in the next three paragraphs.

\subsubsection{Evolution equation for $\overline C_n$.}

Recall the definition for $\overline C_n$ in Equation~\eqref{e:mnormingdef}.
Simple calculations yield
\begin{equation}
\label{e:dCndz}
\frac{\partial \overline C_n}{\partial z}
 = \frac{\partial \overline c_n}{\partial z}\frac{\e^{2\ii\lambda(w_n^*)t}}{b'_{1,1}(w_n^*)} - \overline C_n\lim_{\zeta\to w_n^*}\frac{1}{b_{1,1}(w_n^*)}\frac{\partial b_{1,1}(\zeta)}{\partial z}\,.
\end{equation}
Thus, we have to compute the two $z$-derivatives.

Using Equation~\eqref{e:dSdz},
we write $\partial_z b_{1,1}$ explicitly as
\begin{equation}
\nonumber
\frac{\partial b_{1,1}}{\partial z}
 = \frac{\ii}{2} \left(R_{+,1,1}b_{1,1} + R_{+,1,2}b_{2,1} + R_{+,1,3}b_{3,1} - R_{-,1,1}b_{1,1} - R_{-,2,1}b_{1,2} - R_{-3,1}b_{1,3}\right)\,.
\end{equation}
We first analytically continue every term to $D_2$ according to Equation~\eqref{e:Roff},
because it is necessary to substitute the discrete eigenvalue $w_n^*\in D_2$.
This yields
\begin{equation}
\nonumber
\frac{\partial b_{1,1}}{\partial z}
 = \frac{\ii}{2}\left(R_{+,1,1} - R_{-,1,1}\right)b_{1,1}\,,\qquad
\zeta\in D_2\,.
\end{equation}
Substituting $\zeta = w_n^*\in D_2$,
we obtain the second derivative term in Equation~\eqref{e:dCndz}.

Now, let us focus on the first derivative term.
Recall the definition~\eqref{e:defRpm} of the matrix $\R_\pm$,
which can be rewritten as
\begin{equation}
\nonumber
\frac{\partial \bphi_\pm}{\partial z}
 = \V \bphi_\pm -\frac{\ii}{2}\bphi_\pm  \R_\pm\,.
\end{equation}
Using the analyticity properties~\eqref{e:Roff} of the entries of $\R_\pm$,
again, we obtain
\begin{equation}
\label{e:phi-1-evo}
\frac{\partial \bphi_{-,1}}{\partial z}
 = \V\bphi_{-,1} - \frac{\ii}{2}R_{-,1,1}\bphi_{-,1}\,,\qquad
 \zeta\in D_1\,.
\end{equation}
Moreover, we need the propagation equation for the auxiliary eigenfunction $\bchi_1$.
Using the decomposition~\eqref{e:phidecompose} of the eigenfunctions,
we know that $\bchi_2 = b_{1,1}\bphi_{-,3} - b_{1,3}\bphi_{-,1}$ for $\zeta\in\Sigma$.
The chain rule yields
\begin{equation}
\label{e:chi2-evo}
\frac{\partial\bchi_2}{\partial z}
 = \frac{\partial b_{1,1}}{\partial z}\bphi_{-,3} + b_{1,1}\frac{\partial \bphi_{-,3}}{\partial z} - \frac{\partial b_{1,3}}{\partial z}\bphi_{-,1}-b_{1,3}\frac{\partial \bphi_{-,1}}{\partial z}\,.
\end{equation}
Using the definition~\eqref{e:defRpm} of the matrix $\R_\pm$ and~\eqref{e:dSdz},
again,
we obtain the following for $\zeta\in\Sigma$:
\begin{equation}
\label{e:phi-3-evo}
\begin{aligned}
\frac{\partial \bphi_{-,3}}{\partial z}
 & = \V\bphi_{-,3} - \frac{\ii}{2} \left(R_{-,1,3}\bphi_{-,1} + R_{-,2,3}\bphi_{-,2} + R_{-,3,3}\bphi_{-,3}\right)\,,\\
\frac{\partial b_{1,3}}{\partial z}
 & = \frac{\ii}{2} \left(R_{+,1,1}b_{1,3} + R_{+,1,2}b_{2,3} + R_{-,1,3}b_{3,3} - R_{-,1,3}b_{1,1} - R_{-,2,3}b_{1,2} - R_{-,3,3}b_{1,3}\right)\,.
\end{aligned}
\end{equation}
Therefore,
combining the above ingredients in Equations~\eqref{e:phi-1-evo},~\eqref{e:chi2-evo} and~\eqref{e:phi-3-evo}, and using Equation~\eqref{e:Roff} to extend all the terms to $D_2$,
we find
\begin{equation}
\nonumber
\frac{\partial\bchi_2}{\partial z}
 = \frac{\ii}{2}\left(R_{+,1,1} - R_{-,1,1} - R_{-,3,3}\right)\bchi_2 + \V\bchi_2\,,\qquad
 \zeta\in D_2\,.
\end{equation}

At $\zeta=w_n$,
one knows that $\bchi_2 = \overline c_n\bphi_{-,1}$ from the definition~\eqref{e:norming1} of the norming constant,
which implies that
$\partial_z\bchi_2 = \bphi_{-,1}\partial_z \overline c_n + \overline c_n\partial_z\bphi_{-,1}$.
By inserting the two derivatives,
we obtain the derivative of $\overline c_n$,
\begin{equation}
\nonumber
\frac{\partial \overline c_n}{\partial z}
 = \frac{\ii}{2}\left(R_{+,1,1} - R_{-,3,3}\right)\overline c_n\,,\qquad \zeta = w_n\,.
\end{equation}
This is the end result of the first derivative term in Equation~\eqref{e:dCndz}.

Substituting the above two derivative terms into Equation~\eqref{e:dCndz},
we finally obtain the propagation equation
\begin{equation}
\nonumber
\frac{\partial \overline C_n}{\partial z} = \frac{\ii}{2}\left(R_{-1,1}(w_n^*) - R_{-3,3}(w_n^*)\right)\overline C_n\,.
\end{equation}

\subsubsection{Evolution equation for $D_n$.}

Differentiating the expression for $D_n$ in Equation~\eqref{e:mnormingsymmetry},
we obtain
\begin{equation}
\label{e:dDndz}
\frac{\partial D_n}{\partial z}
 = \frac{\partial d_n}{\partial z}\e^{-\ii  \hat z_n t}\frac{1}{b'_{2,2}(z_n)} - D_n\lim_{\zeta\to z_n}\frac{1}{b_{2,2}(\zeta)}\frac{\partial b_{2,2}(\zeta)}{\partial z}\,.
\end{equation}
Similarly to the previous case,
we need to calculate the two derivatives separately.

It is easy to compute $\partial_z b_{2,2}$ on the continuous spectrum from Equation~\eqref{e:dSdz},
and to extend every part to $\Complex^+$ by Equation~\eqref{e:Roff},
resulting in
\begin{equation}
\nonumber
\frac{\partial b_{2,2}}{\partial z}
 = \frac{\ii}{2}\left(R_{+,2,2} - R_{-,2,2}\right)b_{2,2}\,,\qquad \zeta\in\Complex^+\,.
\end{equation}
This gives the second term in Equation~\eqref{e:dDndz}.

Next, we compute $\partial_z\chi_1$ on the continuous spectrum from Equations~\eqref{e:phidecompose} and~\eqref{e:defRpm}:
\begin{equation}
\begin{aligned}
\frac{\partial \bchi_1}{\partial z}
 & = \frac{\ii}{2} \left(R_{+,2,2} - R_{-,2,2} - R_{-,3,3}\right)\bchi_1 + \V\bchi_1 + \frac{\ii}{2}R_{-,1,3}\bchi_4+\@f\,,\qquad \zeta\in\Sigma\,,\\
\@f
 & \coloneq \frac{\ii}{2}R_{+,2,1} \left(b_{1,2}\bphi_{-,3} - b_{1,3}\bphi_{-,2}\right) + \frac{\ii}{2} R_{+,2,3} \left(b_{3,2}\bphi_{-,3} - b_{3,3}\bphi_{-,2}\right) + \frac{\ii}{2}R_{-,1,2} \left(b_{2,3}\bphi_{-,1} - b_{2,1}\bphi_{-,3}\right)\,.
\end{aligned}
\end{equation}
We would like to extend every term to $D_1$.
However, this cannot be done like in the ZBG case or in the classic two-level case,
because $R_{-,1,3} = 0$ only holds in $D_4$,
and $\bchi_4$ is analytic only in $D_4$.
In order to continue,
we apply the Cauchy projector $\P(f)$ defined in Equation~\eqref{e:Cauchyprojector} along the integration contour $\Sigma_0 = (-\infty,-E_0]\cup\{E_0\e^{\ii \theta}| 0\le\theta\le \pi\}\cup[E_0,\infty)\subset\Sigma$.
Note that $\Sigma_0$ is a subset of the continuous spectrum,
so applying this projector is valid for every value of the parameter $\zeta\in\Sigma_0$.
Now,
assuming that the $z$-derivative and the projector commute,
one obtains
\begin{equation}
\begin{aligned}
\frac{1}{2\pi\ii } \int_{\Sigma_0}\frac{\partial \bchi_1(\eta)}{\partial z} \frac{\d\eta}{\eta-\zeta}
 = & \frac{1}{4\pi}\int_{\Sigma_0}(R_{+,2,2} - R_{-,2,2} - R_{-,3,3})\bchi_1(\eta) + \V(\eta)\bchi_1(\eta)\frac{\d\eta}{\eta-\zeta}\\
 & + \frac{1}{4\pi}\int_{\Sigma_0}R_{-,1,3}\bchi_4(\eta)\frac{\d\eta}{\eta-\zeta}
 + \frac{1}{2\pi\ii }\int_{\Sigma_0}\@f(\eta)\frac{\d\eta}{\eta-\zeta}\,.
\end{aligned}
\end{equation}
Now, we insert the discrete eigenvalue $\zeta = z_n\in D_1$.
Since some of the terms are analytic in $D_1$,
we can use the Residue Theorem to calculate the corresponding integrals,
which yields
\begin{equation}
\nonumber
\frac{\partial \bchi_1(z_n)}{\partial z}
 = \frac{\ii}{2} (R_{+,2,2} - R_{-,2,2} - R_{-3,3})_{\zeta=z_n} \bchi_1(z_n) + \V(z_n)\bchi_1(z_n) + \frac{1}{4\pi} \int_{\Sigma_0} \frac{R_{-,1,3}\bchi_4}{\eta-z_n}\d\eta\,.
\end{equation}
Since $R_{-,1,3} = 0$ in $D_4$ from Equation~\eqref{e:Roff},
it is possible to deform the last integral and thus to obtain
\begin{equation}
\label{e:dchidz}
\frac{\partial \bchi_1(z_n)}{\partial z}
 = \frac{\ii}{2} (R_{+,2,2} - R_{-,2,2} - R_{-,3,3})_{\zeta=z_n} \bchi_1(z_n) + \V(z_n)\bchi_1(z_n) + \frac{1}{4\pi} \int_{\Real} \frac{R_{-,1,3}\bchi_4}{\eta-z_n}\d\eta\,.
\end{equation}

Moreover,
due to the fact that $\bchi_1(z_n) = d_n\bphi_{-,2}(z_n)$ from Equation~\eqref{e:norming2},
we find the following relation:
\begin{equation}
\nonumber
\frac{\partial \bchi_1(z_n)}{\partial z}
 = \frac{\partial d_n}{\partial z} \bphi_{-,2}(z_n) + d_n\frac{\partial \bphi_{-,2}(z_n)}{\partial z}\,.
\end{equation}
Note that the explicit expression for $\partial \bphi_{-,2}(z_n)/\partial z$ can be calculated using Equation~\eqref{e:defRpm}.
So, Equation~\eqref{e:dchidz} and the explicit formula for $\partial \bphi_{-,2}(z_n)/\partial z$ together yield
\begin{equation}
\label{e:chi1int}
\left[\frac{\partial d_n}{\partial z} - \frac{\ii}{2}\left(R_{+,2,2} - R_{-,3,3}\right)_{\zeta=z_n} d_n\right]\bchi_1(z_n)
 = \frac{d_n}{4\pi}\int_\Real \frac{R_{-,1,3} \bchi_4(\eta)}{\eta-z_n} \d\eta\,.
\end{equation}
In the above equation,
only the eigenfunctions $\bchi_1$ and $\bchi_4$ depend on the variable $t$,
so we can evaluate it via the limit $t\to\infty$.
The asymptotic behavior of the eigenfunctions is shown in Equation~\eqref{e:masym}.

Note that the eigenfunctions $\bchi_1$ and $\bchi_4$ are vectors. Let us look closely at the first component of the integral on the right hand side of Equation~\eqref{e:chi1int} in the limit
\begin{equation}
\nonumber
\lim_{t\to\infty} \bigg(\int_\Real  \frac{R_{-,1,3}\bchi_4(\eta)}{\eta-z_n} \d\eta\bigg)_1
 = \lim_{t\to\infty} \int_\Real \bigg(-\ii\pi\nu\bvarrho_{-,1,3}g(k(\eta)) a_{3,3}(\eta) \e^{\ii \frac{E_0^2t}{2\eta}}\bigg) \e^{\ii \eta t/2} \frac{\d\eta}{\eta-z_n}\,.
\end{equation}
The term $-\ii\pi\nu\bvarrho_{-,1,3}g(k(\eta)) a_{3,3}(\eta)\exp[\ii E_0^2t/(2\eta)]/(\eta-z_n^*)$
remains finite as $t\to\infty$.
Assuming that this term is in $L_1(\Real)$,
then the Riemann-Lebesgue lemma implies that this integral vanishes.
In other words,
\begin{equation}
\nonumber
\lim_{t\to\infty}
\int_\Real \bigg(-\ii\pi\nu\bvarrho_{-,1,3}g(k(\eta))a_{3,3}(\eta) \e^{\ii \frac{E_0^2t}{2\eta}}\bigg) \e^{\ii \eta t/2}\frac{\d\eta}{\eta-z_n} = 0\,.
\end{equation}
Hence, one obtains the following ODE from Equation~\eqref{e:chi1int}
\begin{equation}
\nonumber
\frac{\partial d_n}{\partial z}
 = \frac{\ii}{2}\left(R_{+,2,2} - R_{-,3,3}\right)_{\zeta=z_n}d_n\,.
\end{equation}
Finally, the propagation equation for the norming constant $D_n$ is obtained as
\begin{equation}
\nonumber
\frac{\partial D_n}{\partial z}
 = \frac{\ii}{2}\left(R_{-,2,2} - R_{-,3,3}\right)_{\zeta=z_n}D_n\,.
\end{equation}

\subsubsection{Propagation equation for $\overline F_n$.}

The calculation in this case is similar to the one for $\overline C_n$.
Therefore, the propagation equation for $\overline F_n$ in Equation~\eqref{e:normingtimeevolution} can be obtained by following the first case step by step.

\subsection{Calculation of trace formula and asymptotic phase difference}
\label{a:tracetheta}

We start from the simple fact $\S\cdot\S^{-1} = \bbI$ with $\zeta\in\Sigma$.
Its components yield
\begin{equation}
\nonumber
b_{2,1}a_{1,2}+b_{2,2}a_{2,2}+b_{2,3}a_{3,2} = 1\,,\qquad \zeta\in\Sigma\,,
\end{equation}
which reduces to
\begin{equation}
\nonumber
\log a_{2,2} - \log 1/b_{2,2} = -\log\left[\gamma(\zeta)(\gamma(\zeta)-1)r_3(\hat \zeta)r_3^*(\hat\zeta^*)
+\gamma(\zeta) r_3(\zeta)r_3^*(\zeta^*)\right]\,,\qquad
\zeta\in\Sigma\,,
\end{equation}
where we recall $r_3$ defined in Equation~\eqref{e:reflection-def}.
In order to remove the zeros of $a_{2,2}$ and $b_{2,2}$ and to fix the limits as $\zeta\to\infty$,
we define two more functions
\begin{equation}
\begin{aligned}
\beta^-(\zeta)
 & \coloneq a_{2,2}\e^{\ii \Delta \theta}
\prod_{n=1}^{N_2}\frac{\zeta-z_n}{\zeta-z_n^*}\frac{\zeta-\hat z_n}{\zeta-\hat z_n^*}
\prod_{n=1}^{N_3}\frac{\zeta-\zeta_n}{\zeta-\zeta_n^*}\frac{\zeta-\hat \zeta_n}{\zeta-\hat \zeta_n^*}\,,\qquad&&
\zeta\in\Complex^-\,,\\
\beta^+(\zeta)
 & \coloneq 1/b_{2,2}\e^{\ii \Delta \theta}
\prod_{n=1}^{N_2}\frac{\zeta-z_n}{\zeta-z_n^*}\frac{\zeta-\hat z_n}{\zeta-\hat z_n^*}
\prod_{n=1}^{N_3}\frac{\zeta-\zeta_n}{\zeta-\zeta_n^*}\frac{\zeta-\hat \zeta_n}{\zeta-\hat \zeta_n^*}\,,\qquad&&
\zeta\in\Complex^+\,.
\end{aligned}
\end{equation}
Clearly,
the two functions $\beta^\pm$ are analytic in $\Complex^\pm$, respectively,
and have no zeros and no poles in the corresponding region.
The additional as yet undetermined factor $\pm\Delta\theta$ ensures that the two functions tend to $1$ as $\zeta\to\infty$.
Thus, we obtain the following jump condition:
\begin{equation}
\nonumber
\log\beta^-(\zeta)-\log\beta^+(\zeta) = J_0\,,\qquad \zeta\in\Real\,,
\end{equation}
where the quantity $J_0$ is defined in Equation~\eqref{e:J0J-def}.
By applying Plemelj's formulas,
one can solve for $\beta^\pm$ and consequently $a_{2,2}$ and $b_{2,2}$.
The final results are given by
\begin{equation}
\begin{aligned}
a_{2,2}(\zeta)
 & = \e^{-\ii \Delta\theta}\e^{\frac{1}{2\pi\ii }\int_\Real\frac{J_0}{\eta-\zeta}\d\eta}
\prod_{n=1}^{N_2}\frac{\zeta-z_n^*}{\zeta-z_n}\frac{\zeta-\hat\zeta_n^*}{\zeta-\hat \zeta_n}
\prod_{n=1}^{N_3}\frac{\zeta-\zeta_n^*}{\zeta-\zeta_n}\frac{\zeta-\hat\zeta_n^*}{\zeta-\hat\zeta_n}\,,\\
b_{2,2}(\zeta)
 & = \e^{\ii \Delta\theta}\e^{-\frac{1}{2\pi\ii }\int_\Real\frac{J_0}{\eta-\zeta}\d \eta}
\prod_{n=1}^{N_2}\frac{\zeta-z_n}{\zeta-z_n^*}\frac{\zeta-\hat\zeta_n}{\zeta-\hat \zeta_n^*}
\prod_{n=1}^{N_3}\frac{\zeta-\zeta_n}{\zeta-\zeta_n^*}\frac{\zeta-\hat\zeta_n}{\zeta-\hat\zeta_n^*}\,.
\end{aligned}
\end{equation}

The following four identities can also be obtained from the simple fact $\S\cdot \S^{-1} = \bbI$ with $\zeta\in\Sigma$:
\begin{gather*}
1+\frac{a_{2,1}}{a_{1,1}}\frac{b_{1,2}}{b_{1,1}}+\frac{a_{3,1}}{a_{1,1}}\frac{b_{1,3}}{b_{1,1}} = \frac{1}{a_{1,1}b_{1,1}}\,,\qquad
\frac{a_{1,3}}{a_{3,3}}\frac{b_{3,1}}{b_{3,3}}+\frac{a_{2,3}}{a_{3,3}}\frac{b_{3,2}}{b_{3,3}}+1 = \frac{1}{a_{3,3}b_{3,3}}\,,\\
\frac{a_{2,2}}{b_{1,1}b_{3,3}} = 1-\frac{b_{1,3}}{b_{1,1}}\frac{b_{3,1}}{b_{3,3}}\,,\qquad
\frac{b_{2,2}}{a_{1,1}a_{3,3}} = 1-\frac{a_{1,3}}{a_{3,3}}\frac{a_{3,1}}{a_{1,1}}\,.
\end{gather*}
Using these identities, we write down four more jump conditions,
\begin{equation}
\label{e:ab-jumps}
\begin{aligned}
\log b_{1,1} - \log \frac{1}{a_{1,1}}
 & = J_1\,,\\
\log a_{3,3} - \log\frac{1}{b_{3,3}}
 & = J_3\,,\\
\log b_{1,1} - \log\frac{1}{b_{3,3}}
 & = \log a_{2,2} - \log\left(1-r_2^*(\zeta^*)r_2^*(\hat\zeta^*)\right)\,,\\
\log a_{3,3} - \log\frac{1}{a_{1,1}}
 & = \log b_{2,2} - \log\left(1-r_2(\hat\zeta)r_2(\zeta)\right)\,,
\end{aligned}
\end{equation}
where recall $J_1$ and $J_3$ are defined in Equation~\eqref{e:J0J-def}.
We define the following identities in order to remove the zeros of the scattering data:
\begin{equation}
\begin{aligned}
\beta_1(\zeta)
 & \coloneq \frac{1}{a_{1,1}} \prod_{n=1}^{N_1}\frac{\zeta-w_n}{\zeta-w_n^*}\frac{\zeta-\hat w_n^*}{\zeta-\hat w_n} \prod_{n=1}^{N_2}\frac{\zeta-\hat z_n}{\zeta-\hat z_n^*} \prod_{n=1}^{N_3}\frac{\zeta-\zeta_n}{\zeta-\zeta_n^*}\,,\qquad&&
\zeta\in D_1\,,\\
\beta_2(\zeta)
 & \coloneq b_{1,1} \prod_{n=1}^{N_1}\frac{\zeta-w_n}{\zeta-w_n^*}\frac{\zeta-\hat w_n^*}{\zeta-\hat w_n} \prod_{n=1}^{N_2}\frac{\zeta-\hat z_n}{\zeta-\hat z_n^*} \prod_{n=1}^{N_3}\frac{\zeta-\zeta_n}{\zeta-\zeta_n^*}\,,\qquad&&
\zeta\in D_2\,,\\
\beta_3(\zeta)
 & \coloneq \frac{1}{b_{3,3}}\e^{\ii \Delta\theta} \prod_{n=1}^{N_1}\frac{\zeta-w_n}{\zeta-w_n^*}\frac{\zeta-\hat w_n^*}{\zeta-\hat w_n} \prod_{n=1}^{N_2}\frac{\zeta-z_n^*}{\zeta-z_n} \prod_{n=1}^{N_3}\frac{\zeta-\hat\zeta_n^*}{\zeta-\hat\zeta_n}\,,\qquad&&
\zeta\in D_3\,,\\
\beta_4(\zeta)
 & \coloneq a_{3,3}\e^{\ii \Delta\theta} \prod_{n=1}^{N_1}\frac{\zeta-w_n}{\zeta-w_n^*}\frac{\zeta-\hat w_n^*}{\zeta-\hat w_n}  \prod_{n=1}^{N_2}\frac{\zeta-z_n^*}{\zeta-z_n} \prod_{n=1}^{N_3}\frac{\zeta-\hat\zeta_n^*}{\zeta-\hat\zeta_n}\,,\qquad&&
\zeta\in D_4\,,
\end{aligned}
\end{equation}
where the regions $D_j$ are defined in Definition~\ref{def:DSigmaL}.
The four functions $\beta_j(\zeta)$, with $j=1,2,3,4$, are analytic in $D_j$, respectively,
and have no zeros and no poles in their corresponding analyticity regions.
Moreover, they tend to $1$ as $\zeta\to0$ and $\zeta\to\infty$ in these analytic regions.
With new functions
\begin{equation}
\nonumber
\bar\beta(\zeta)
 \coloneqq \begin{cases}
 \beta_1(\zeta)\,, &\zeta\in D_1\,,\\
 \beta_2(\zeta)\,, &\zeta\in D_2\,,\\
 \beta_3(\zeta)\,, &\zeta\in D_3\,,\\
 \beta_4(\zeta)\,, &\zeta\in D_4\,,
 \end{cases}
\end{equation}
we can consolidate all the relations in Equation~\eqref{e:ab-jumps} as the jump conditions
\begin{equation}
\nonumber
\log\bar\beta^--\log\bar\beta^+ = J_j\,,\qquad \zeta\in\Sigma_j\,,
\end{equation}
where all the jumps are defined in Equation~\eqref{e:J0J-def} for $j = 1,\dots,4$.
Plemelj's formula yields
\begin{gather*}
\bar\beta(\zeta)=\e^{-\frac{1}{2\pi\ii }\int_\Sigma \frac{J}{\eta-\zeta}\d\eta}\,,
\end{gather*}
where $J$ is also defined in Equation~\eqref{e:J0J-def}.
In particular, the scattering data $b_{1,1}(\zeta)$ has the following explicit expression,
\begin{equation}
\nonumber
b_{1,1}(\zeta)
 = \e^{-\frac{1}{2\pi\ii }\int_\Sigma\frac{J}{\eta-\zeta}\d\eta}
\prod_{n=1}^{N_1}\frac{\zeta-w_n^*}{\zeta-w_n}\frac{\zeta-\hat w_n}{\zeta-\hat w_n^*}
\prod_{n=1}^{N_2}\frac{\zeta-\hat z_n^*}{\zeta-\hat z_n}
\prod_{n=1}^{N_3}\frac{\zeta-\zeta_n^*}{\zeta-\zeta_n}\,.
\end{equation}
Next, we let $\zeta\to0$,
compare the leading order of the asymptotic behavior of $b_{1,1}$, and obtain
\begin{equation}
\nonumber
\e^{-\frac{1}{2\pi\ii }\int_\Sigma\frac{J}{\eta}\d\eta}
\prod_{n=1}^{N_1}\frac{w_n^*}{w_n}\frac{\hat w_n}{\hat w_n^*}
\prod_{n=1}^{N_2}\frac{\hat z_n^*}{\hat z_n}\prod_{n=1}^{N_3}\frac{\zeta_n^*}{\zeta_n}
 = \E_+^\top\E_-^*/E_0^2\,.
\end{equation}
Simplifying the above expression yields the phase difference of the solution at the boundaries,
\begin{equation}
\nonumber
\Delta\theta
 = \theta_+-\theta_-
 = \frac{1}{2\pi}\int_\Sigma\frac{J}{\eta}\d\eta
 - 4\sum_{n=1}^{N_1}\arg w_n
 + 2\sum_{n=1}^{N_2}\arg z_n - 2\sum_{n=1}^{N_3}\arg \zeta_n\,.
\end{equation}

\subsection{Calculation of $\texorpdfstring{\R_{-,\dd}(z,\zeta)}{R{-,dd} (z,ζ)}$ with inhomogeneous broadening}
\label{a:Rnd}

According to Equation~\eqref{e:normingtimeevolution},
we must compute the diagonal part of the auxiliary matrix $\R_{-,\dd}(z,\zeta)$ from Equation~\eqref{e:Rpm} in order to compute the propagation of the norming constants.
Thus, this appendix focuses on the explicit computation of $\R_{-,\dd}(z,\zeta)$ with a known shape of the inhomogeneously broadened spectral line, given in Equation~\eqref{e:Lorentzian}.

Recall that it is necessary to pick $\bvarrho_-$ with opposite signs on $k$-sheets I and II.
The resulting auxiliary matrix $\R_-$ is uniquely defined on each sheet (cf. discussion in Section~\ref{s:reflectionnorming}).
Consequently,
without loss of generality,
we perform all the calculations on $k$-sheet I in this Appendix.

Upon inspecting Equation~\eqref{e:Rpm},
we see that we must compute the Hilbert transforms of $\varrho_-^\pm g(k)$.
Recall that we take $\bvarrho_{-}$ to be diagonal and independent of $k$ on either sheet from Section~\ref{s:reflectionless},
so that the quantity $\varrho_-^\pm$ in Equation~\eqref{e:rhopmpm} is also independent of $k$, and thus can be taken outside the integrals.
Therefore, it is sufficient to calculate the following two integrals to compute $\R_{-,\dd}$ with $g(k)$ in Equation~\eqref{e:Lorentzian}:
\begin{equation}
\label{e:I1I2-def}
I_1(k)
 \coloneq \pvint\nolimits_{\hspace{-0.4em}\Real} g(k') \frac{\d k'}{k' - k}\,,\qquad
I_2(k)
 \coloneq \pvint\nolimits_{\hspace{-0.4em}\Real} \frac{g(k')}{\lambda}\frac{\d k'}{k' - k}\,.
\end{equation}
Note that the above integrals differ when the parameter $k$ takes two possible values: $\Real$ or $\Complex\backslash\Real$.
Therefore,
we need to compute them separately.
We first compute the easier case in which $k\in\Complex\backslash\Real$.
Then, we use the following relation to compute the case in which $k\in\Real$:
\begin{equation}
\label{e:pvint}
\pvint\nolimits_{\hspace{-0.4em}\Real} f(k') \frac{\d k'}{k' - k}
 = \int_L f(k') \frac{\d k'}{k' - k} + \pi \ii \Res_{k' = k} \frac{f(k')}{k' - k}\,,
\end{equation}
where the contour is $L = (-\infty,k - r)\cup\{k' = r\,\e^{\ii \theta}|\,\theta\in[0,\pi]\}\cup(k + r,\infty)$,
provided that $r$ is sufficiently small and there are no singularities on $L$.

\subsubsection{Calculation of the integral $I_1$ in Equation~\eqref{e:I1I2-def} with $k\in\Complex\backslash\Real$.}

Explicitly, this integral is
\begin{equation}
\nonumber
I_1  = \int_\Real \frac{\epsilon }{\pi  (k' - k) (k'^2+\epsilon ^2 )} \d k'\,,\qquad
k\in\Complex\backslash\Real\,.
\end{equation}
It is easily evaluated using the Residue theorem.
The result is
\begin{equation}
\label{e:I1complex}
I_1
 = \begin{cases}
 - 1 /(k+i \epsilon)\,, & k\in\Complex^+\,,\\
 - 1/(k-i \epsilon)\,, & k\in\Complex^-\,.
\end{cases}
\end{equation}

\subsubsection{Calculation of the integral $I_1$ in Equation~\eqref{e:I1I2-def} with $k\in\Real$.}

The identity~\eqref{e:pvint} yields the integral $I_1$, with $k\in\Real$,
from Equation~\eqref{e:I1complex}:
\begin{equation}
\label{e:I1real}
I_1
 = - \frac{1}{k + \ii \epsilon} + \pi \ii \,\Res_{k' = k}\frac{\epsilon }{\pi  (k' - k) (k'^2+\epsilon ^2 )}
 = -\frac{k}{k^2+\epsilon ^2}\,,\qquad
k\in\Real\,.
\end{equation}

\subsubsection{Calculation of the integral $I_2$ in Equation~\eqref{e:I1I2-def} with $k\in\Complex\backslash\Real$.}

Explicitly,
with $k\in\Complex\backslash\Real$,
this integral is
\begin{equation}
\nonumber
I_2
 = \int_\Real \frac{\epsilon}{\pi}\frac{1}{\lambda(\epsilon^2+k'^2)}\frac{\d k'}{k' - k}\,.
\end{equation}
Similarly to what has been done in~\cite{bgkl2019},
it is convenient to compute this integral in the uniformization variable $\zeta$.
Recalling Equation~\eqref{e:k-lambda},
we can write $k' = (\zeta' - E_0^2/\zeta')/2$,
$\lambda' = (\zeta'+E_0^2/\zeta')/2$,
$k = (\zeta-E_0^2/\zeta)/2$,
and $\lambda = (\zeta + E_0^2/\zeta)/2$.
Thus, we find $\d k' = (1 + E_0^2/\zeta'^2)/2\d\zeta'$ with a new integration contour: $L \coloneq (-\infty,-E_0)\cup(E_0,\infty)$.
We substitute all these parts into $I_2$,
and find
\begin{equation}
\label{e:I2f}
I_2(k)
 = -\int_L f(\zeta,\zeta')\d\zeta'\,,\qquad
f(\zeta,\zeta')
 \coloneq \frac{8\epsilon \zeta  \zeta'^2  }{\pi  (\zeta - \zeta') (\zeta  \zeta'+E_0^2)
\left[\,\zeta'^4 + E_0^4 - 2\zeta'^2 (E_0^2 - 2 \epsilon^2)\,\right]}\,.
\end{equation}
The denominator of $f(\zeta,\zeta')$ has six simple roots $r_j$ with $j = 1,\dots,6$,
whose explicit expressions are omitted for brevity,
so the rational function $f(\zeta,\zeta')$ can be decomposed as
\begin{equation}
\nonumber
f(\zeta,\zeta')
 = \sum_{j = 1}^6 \frac{R_j}{\zeta'-r_j}\,,\qquad
R_j \coloneq \Res_{\zeta' = r_j} f(\zeta,\zeta')\,.
\end{equation}
Correspondingly, the integral can be rewritten using the Residue theorem as
\begin{equation}
\nonumber
I_2 = -\sum_{j=1}^6 R_j \log(\zeta' - r_j)\big|_L\,.
\end{equation}
After some tedious but straightforward calculations,
we finally obtain
\begin{equation}
\label{e:I2complex}
I_2
 = - \frac{g(k)}{\lambda} \left[ \log \left(\frac{E_0-\lambda }{E_0 + \lambda}\right) + \frac{\lambda}{\sqrt{E_0^2-\epsilon^2}} \log \left(\frac{E_0 + \sqrt{E_0^2-\epsilon^2}}{E_0-\sqrt{E_0^2-\epsilon ^2}}\right) \right]\,,\qquad
k\in\Complex\backslash\Real\,.
\end{equation}

\subsubsection{Calculations of the integral $I_2$ in Equation~\eqref{e:I1I2-def} with $k\in\Real$.}

We apply the relation~\eqref{e:pvint} again, and compute
\begin{align}
\nonumber
I_2   = & - \frac{g(k)}{\lambda}
\left[ \log \left(\frac{E_0-\lambda }{E_0 + \lambda}\right) +
\frac{\lambda}{\sqrt{E_0^2-\epsilon^2}} \log \left(\frac{E_0 + \sqrt{E_0^2-\epsilon^2}}{E_0-\sqrt{E_0^2-\epsilon ^2}}\right) \right]
 + \pi \ii\,\Res_{\zeta' = \zeta}f(\zeta,\zeta')\\
 = &  - \frac{g(k)}{\lambda}
\left[ \log \left(\frac{\lambda -E_0}{\lambda + E_0}\right) +
\frac{\lambda}{\sqrt{E_0^2-\epsilon^2}} \log \left(\frac{E_0 + \sqrt{E_0^2-\epsilon^2}}{E_0-\sqrt{E_0^2-\epsilon ^2}}\right)\right]\,,\qquad
k\in\Real\,.
\label{e:I2real}
\end{align}

\subsubsection{Calculation of $\R_{-,\dd}$.}

Using the values of the above two special integrals,
we are able to compute the matrix $\R_{-,\dd}$.
By examining Equation~\eqref{e:Rpm},
we obtain the following relations:
\begin{equation}
\nonumber
R_{-,1,1}(z,\zeta) = \rho_-^+ I_1 + \lambda \varrho_-^- I_2\,,\qquad
R_{-,2,2}(z,\zeta) = \rho_{-2,2}I_1\,,\qquad
R_{-,3,3}(z,\zeta) = \rho_-^+ I_1 - \lambda \varrho_-^- I_2\,,
\end{equation}
where $\varrho_-^\pm$ are defined in Equation~\eqref{e:rhopmpm}.

\bibliographystyle{unsrt}
\bibliography{bibliography/references.bib,bibliography/books,bibliography/optics,%
bibliography/newoptics}

\end{document}